\pdfoutput=1
\documentclass[english, 12pt]{article}
\newcommand{\blind}{0}

\input{macros}

\defaultbibliographystyle{apalike}
\defaultbibliography{main}

\newif\ifseparatebib

\def\separatebibflag{0} 

\ifnum\separatebibflag=1
\separatebibtrue
\else
\separatebibfalse
\fi

\usepackage{titlesec}

\usepackage{booktabs}
\usepackage{graphicx}
\usepackage{subcaption}
\usepackage{placeins}
\graphicspath{{figures/}}

\definecolor{DSgray}{cmyk}{0,1,0,0}

\newcommand{\hbtheta}{\widehat{\btheta}}
\newcommand{\tbtheta}{\widetilde{\btheta}}

\newcommand{\tDelta}{\widetilde{\Delta}}

\newcommand{\calD}{\mathcal{D}}
\newcommand{\calN}{\mathcal{N}}
\newcommand{\calB}{\mathcal{B}}

\newcommand{\E}{\mathbb{E}}

\newcommand{\ip}[2]{\left\langle #1,\, #2 \right\rangle}
\newcommand{\normQ}[1]{\left\lVert #1 \right\rVert_{\bQ}}

\renewcommand{\baselinestretch}{1.7}
\usepackage{microtype}
\AtBeginDocument{\microtypesetup{activate=false}}
\usepackage{etoolbox}

\RequirePackage{xcolor}
\makeatletter
\protected\long\def\REVreplace#1#2{%
  \ifmmode
    {\color{red}#1}\;{\color{black}\mbox{(original: }#2\mbox{)}}%
  \else
    {\color{red}#1}{\color{black}\,(original: #2)}%
  \fi
}
\protected\long\def\REVadd#1{%
  \ifmmode
    {\color{red}#1}\;{\color{black}\mbox{(original: no corresponding text)}}%
  \else
    {\color{red}#1}{\color{black}\,(original: no corresponding text)}%
  \fi
}
\protected\long\def\REVdelete#1{%
  \ifmmode
    {\color{red}\mbox{[deleted]}}\;{\color{black}\mbox{(original: }#1\mbox{)}}%
  \else
    {\color{red}[deleted]}{\color{black}\,(original: #1)}%
  \fi
}

\makeatother

\begin{document}

\title{Tuning-Free Covariance-Aware Sequential Shrinkage for Multi-Source Estimation}
\renewcommand{\thefootnote}{\arabic{footnote}}
\if0\blind
{
\author{
	Wenbo Jing\thanks{College of Business, City University of Hong Kong.
		Email: \href{mailto:wenbjing@cityu.edu.hk}{\texttt{wenbjing@cityu.edu.hk}}.}
	\hspace{2ex}
	Xi Chen\thanks{Stern School of Business, New York University.
		Email: \href{mailto:xc13@stern.nyu.edu}{\texttt{xc13@stern.nyu.edu}}.}
	\hspace{2ex}
	Yaqi Duan\thanks{Stern School of Business, New York University.
		Email: \href{mailto:yaqi.duan@stern.nyu.edu}{\texttt{yaqi.duan@stern.nyu.edu}}.}
	\hspace{2ex}
	Kaizheng Wang\thanks{Department of IEOR and Data Science Institute, Columbia University.
		Email: \href{mailto:kaizheng.wang@columbia.edu}{\texttt{kaizheng.wang@columbia.edu}}.}
	\hspace{2ex}
	Yichen Zhang\thanks{Daniels School of Business, Purdue University.
		Email: \href{mailto:zhang@purdue.edu}{\texttt{zhang@purdue.edu}}.}
}
} \fi

\if1\blind
{
	\author{}
} \fi

\date{}

\maketitle

\vspace{-20pt}

\begin{abstract}
	\linespread{1.5}\selectfont
Modern statistical learning problems often involve multiple related data sets, in which learning efficiency on a target set can be improved by borrowing information from related sources, while source heterogeneity may introduce bias.
Existing approaches are limited by suboptimal performance in multi-source settings, insufficient use of covariance information, or the computational burden of tuning procedures.
We propose a covariance-aware shrinkage framework that constructs shrinkage directions using covariance information to improve efficiency.
We establish finite-sample risk bounds that yield an explicit risk-improving interval for the shrinkage size, making the transfer strength fully data-driven.
When multiple source sets are available, we further propose a sequential algorithm that applies source-specific shrinkage updates according to a priority score.
Under a mixed-source separation condition, the estimator asymptotically attains the oracle risk and, when the pooled discrepancy is large, strictly improves over single-step shrinkage.
The framework is further extended to a broad class of smooth \(M\)-estimation problems.
Numerical studies show substantial gains over competing methods, especially when the source data sets are highly heterogeneous.
\end{abstract}


\ifseparatebib
\begin{bibunit}
	\fi

	\newpage

	\renewcommand{\baselinestretch}{1.7}

\section{Introduction}

Modern statistical learning problems often involve multiple heterogeneous data sources.
Such structures arise when data are collected across different studies, populations, or platforms, and the learning performance on any single data set can be improved by leveraging information from related source data sets.
This perspective is increasingly important in applications such as survival risk estimation \citep{li2023cox},
multi-site clinical prediction \citep{hickey2024recast},  and neuroimaging analysis \citep{ma2024multitask}.
The statistical value of multi-source learning comes from the possibility of borrowing strength across related data sets, but the usefulness of each source depends on its relationship with the data set or task of primary interest.
The central question is how to use external information to improve the target estimator while controlling the bias caused by heterogeneous sources.

This question has motivated extensive work in transfer learning and multi-task learning.
Recent statistical transfer-learning methods focus on specific problems, including binary and nonparametric classification \citep{reeve2021adaptive, cai2021transfer}, high-dimensional regression \citep{li2022transfer}, and graphical models \citep{li2023graphical}. These methods typically impose model-specific similarity structures, such as sparse parameter differences, posterior drift, or restrictions on the source-target transfer relationship. In parallel, multi-task learning takes a complementary view by estimating related tasks jointly through shared regularized, clustered, or low-rank structures \citep{evgeniou2005learning,jacob2008clustered,pong2010trace}.

Among the many ways to combine information across related data sets, \textit{shrinkage} offers a particularly attractive and analytically transparent principle.
This idea originated from the James--Stein estimator in multivariate normal mean estimation,  which shows that shrinking a noisy estimator toward a suitable target can reduce risk \citep{james1961estimation,efron1973stein,green1991james}.
Recent work has adapted this principle to transfer and data integration problems.
For example, \citet{han2024improving} use James--Stein shrinkage to integrate external summary information for prediction in linear regression, \citet{dempsey2025improving} develop shrinkage estimators for general \(M\)-estimation with external auxiliary information, and \citet{abba2026bayesian} study Bayesian shrinkage priors for transfer learning in normal means and multiple linear regression.

Although the existing methods highlight the appeal of safe transfer across heterogeneous data sets, they do not fully address the multi-source target-estimation problem.
First, existing shrinkage constructions are mainly designed for the two-set setting, that is, the setting with one target data set and one source data set.
When more than one source is available, a common strategy in related shrinkage-based approaches is to aggregate the source sets first and then shrink the target estimator once along the resulting pooled direction \citep{wang2026divide}. This strategy works when the average source discrepancy is small but can be distorted when multiple source sets are highly different from each other.  In cases where several sources are close to the target but a few are far away, pooling mixes the close and far sources before shrinkage. The far sources can then contaminate the pooled direction and make the aggregated source estimator less helpful for transfer. The resulting shrinkage can be overly mild, leaving useful close sources underutilized.

Second, multi-task learning methods  typically focus on the rate of parameter estimation, while not explicitly pursuing the optimal statistical efficiency (e.g., \citealp{evgeniou2005learning, duan2023adaptive}). The role of covariance in determining such efficiency is ignored. Indeed, two source estimators may be similarly close to the target estimator in parameter value, but their covariance matrices can be substantially different. A smaller covariance matrix indicates a less variable and hence more efficient estimator, which should provide more reliable information for estimating the target parameter. Therefore, a low-covariance source estimator should contribute more than a high-covariance source estimator when they have small bias.

Third, many existing methods require tuning parameters that control the strength of regularization or information transfer across data sets, such as regularization parameters and testing thresholds \citep{evgeniou2005learning,li2022transfer,yang2022elastic,duan2023adaptive}. These parameters are typically selected by data-driven tuning procedures, but the statistical effect of this additional tuning step is often not explicitly accounted for. This issue can be particularly severe when the target sample size is small, where tuning may be unstable and sample splitting for validation further reduces the effective target information available for estimation. These three issues motivate a tuning-free shrinkage approach that can handle multiple sources correctly and account for covariance information to improve efficiency. Throughout the paper, \emph{tuning-free} refers to the transfer step: once the single-set and covariance estimators are available, all shrinkage sizes and directions are determined analytically, with no validation-based tuning or testing threshold.

\subsection{Main Contributions}

This paper develops a covariance-aware shrinkage framework for transfer learning with multiple source data sets, with a tuning-free transfer step.
We start with the two-set Gaussian mean problem, where we propose a data-adaptive shrinkage estimator that shrinks the target-only estimator toward a pooled estimator.
The shrinkage direction is constructed using the covariance matrices of the single-set estimators, which quantify the estimation uncertainty, rather than only their point estimates.
Intuitively, this covariance-aware design improves efficiency by assigning larger weights to source estimators with smaller covariance.
We establish finite-sample risk bounds showing that the proposed shrinkage estimator is guaranteed to improve over the target-only estimator with an explicit interval of shrinkage sizes. The shrinkage size is then selected in closed form from this interval, without cross-validation or sample splitting.

One major contribution of this paper is the extension of two-set shrinkage to multiple source data sets.
To tackle this challenge, we propose a sequential algorithm that compares source-specific updates through a priority score, updates the estimator and a covariance proxy, and repeats this process.
This source-by-source construction allows the method to absorb information from homogeneous sources while reducing the influence of heterogeneous ones.
We establish rigorous theoretical guarantees for this procedure. Under a mixed-source separation condition, the sequential path gives priority to all homogeneous sources, the subsequent heterogeneous-source updates are asymptotically negligible, and the final estimator attains the risk of the corresponding precision-weighted oracle. When the pooled discrepancy is large, the procedure also achieves a strictly smaller quadratic risk than the pool-then-shrink estimator. We additionally prove that these prioritization and risk results are stable to estimated covariance matrices.

Furthermore, we extend the framework to a broad class of smooth $M$-estimation problems. 
For local \(M\)-estimators, the covariance-aware shrinkage direction can be constructed using the asymptotic covariance matrices of the corresponding estimators. 
We show that, under suitable regularity and separation conditions, the prioritization and oracle-risk guarantees established for the Gaussian mean model continue to hold.
These results justify the proposed shrinkage and sequential methods for smooth \(M\)-estimation models under the regularity conditions.

In numerical studies, the proposed estimators deliver consistent improvements over existing benchmarks across both simulations and real-data analysis.
The advantage is especially clear in multi-set settings with heterogeneous covariance structures, where the risk of the sequential estimator remains nearly stable as covariance heterogeneity increases, while competing methods deteriorate substantially.
In our simulations, the proposed sequential estimator reduces risk by up to \(48\%\) relative to the strongest competing method in the covariance-heterogeneity settings.
In the real-data analysis, it reduces MSE by 22.1\% relative to the strongest competing method and by 37.6\% relative to the single-task estimator.

\subsection{Related Work}\label{sec:related-work}

\vspace{5pt}
\noindent\textbf{Transfer learning.}
A large statistical literature on transfer learning covers a broad range of tasks.
One line of work studies classification and nonparametric transfer under distributional changes \citep{cai2021transfer,reeve2021adaptive,fan2025robust}, and another line focuses on high-dimensional generalized linear models, where source-target similarity is often encoded through sparse contrasts or related structural assumptions \citep{li2022transfer,tian2023transfer,li2024knowledge, gu2025angle,zhao2026riw}.
Transfer ideas have also been developed for graphical models \citep{li2023graphical}, survival risk estimation \citep{li2023cox}, functional and structured mean estimation \citep{cai2024functional,wang2025piecewise}, and sequential decision-making \citep{cai2024bandit,chen2025transferq}.  In addition, recent work characterizes the intrinsic cost of adaptation in transfer learning \citep{chakraborty2026statistical}.

\vspace{5pt}
\noindent\textbf{Multi-task learning.}
Multi-task learning solves related estimation tasks by exploiting shared structure across task-specific parameters.
Classical methods include regularized multi-task learning \citep{evgeniou2005learning}, 
clustered multi-task learning \citep{jacob2008clustered}, and low-rank formulations \citep{pong2010trace}.
More recent statistical work studies representation learning with adaptivity to unknown task similarity and robustness to outlying tasks \citep{duan2023adaptive, knight2023multitask, tian2025similar, kim2026multitask}, as well as learning with block-wise missing data  \citep{sui2025heterogeneous}.
Notably, \cite{duan2023adaptive} develop ARMUL, an adaptive and robust multi-task framework for general convex losses. \citet{knight2023multitask} study multi-task learning from summary statistics, which is useful when individual-level data are unavailable across all tasks.
\citet{kim2026multitask} further proposes matrix-weighted regularization for multi-task linear regression, using task-specific second-moment geometry in the penalty to weaken eigenvalue lower-bound conditions.
However, most existing multi-task methods are not covariance-aware in the sense that the covariance matrix of each single estimator does not directly determine the borrowing direction and borrowing size.
Although many procedures achieve optimal or near-optimal rates, the optimal rate alone does not guarantee statistical efficiency.
Our method uses covariance information directly in the shrinkage direction and size, leading to a more efficiency-oriented construction.

\vspace{5pt}
\noindent\textbf{Shrinkage and data integration.}
Classical James--Stein theory shows that shrinking an estimator can reduce the risk in multivariate mean estimation \citep{james1961estimation,efron1973stein,green1991james}.
This principle has been adapted to modern data integration problems in several directions.
For example, safe-integration procedures such as elastic integrative analysis and the information-sharing dial control the extent of borrowing from heterogeneous data sources \citep{yang2022elastic,hector2024turning}.
Recent James--Stein-type developments include convex aggregation for high-dimensional multiple mean estimation \citep{blanchard2024estimation}, shrinkage integration of external reduced-model summaries for linear prediction \citep{han2024improving}, and Bayesian shrinkage priors for  regression models \citep{abba2026bayesian, lai2026bayesian}.
Two related concurrent works, \citet{dempsey2025improving}
and \citet{wang2026divide}, also develop shrinkage-based procedures. However, they mainly focus on the two-set setting, and the multi-source extension in \cite{wang2026divide} first aggregates the sources before applying shrinkage, which can reduce the effective contribution of informative sources when the source collection is heterogeneous. Our sequential construction instead evaluates source-specific directions before they are collapsed into a single aggregate and propagates a covariance proxy after each update. Under the conditions in Section~\ref{sec:multiset}, it asymptotically attains the risk of the oracle estimator that knows in advance which sources are homogeneous.


The rest of the paper is organized as follows. Section~\ref{sec:twoset} develops the covariance-aware shrinkage estimator for the two-set problem, and Section~\ref{sec:equivalence} connects it with a regularized multi-task learning method.
 Section~\ref{sec:multiset} introduces the sequential shrinkage algorithm for multiple source sets, and Section~\ref{sec:general-M-estimation} further extends the framework to smooth $M$-estimation problems. Section~\ref{sec:numerical_study} presents simulation studies and the real-data analysis, followed by conclusions in Section~\ref{sec:conc}.

\section{Data-Adaptive Shrinkage Estimator}\label{sec:twoset}

We begin with the two-set Gaussian mean estimation problem. Throughout the theoretical analysis, the data sets are assumed to be mutually independent. Suppose that we observe i.i.d.~samples on the target set $\calD_1=\{\bx_{1i}\}_{i=1}^{n_1} \sim \calN\left(\btheta_1^{\star}, \bSigma_1\right)$ and the source set $\calD_2=\{\bx_{2i}\}_{i=1}^{n_2} \sim \calN\left(\btheta_2^{\star}, \bSigma_2\right)$, where the population means $\btheta_1^{\star}, \btheta_2^{\star} \in \R^p$ are unknown, and the covariance matrices $\bSigma_1, \bSigma_2  \in \R^{p \times p}$ are positive definite and known. Our goal is to estimate  $\btheta_1^{\star}$. For an estimator $\hbtheta$, we evaluate its performance using the risk $\EE\big[\big\|{\hbtheta-\btheta_1^{\star}}\big\|_{\bQ}^2\big]$, where $\bQ$ is a prespecified positive definite matrix satisfying $0<q_-\leq \lambda_{\min}\left(\bQ\right)\leq\lambda_{\max}\left(\bQ\right)\leq q_+<\infty$ for some constants $q_-$ and $q_+$, with $\lambda_{\min}$ and $\lambda_{\max}$ denoting the minimum and maximum eigenvalues of a matrix. The matrix $\bQ$ defines the target loss and remains fixed throughout the transfer procedure. Throughout the paper, for any positive integer \(k\), let \(\left[k\right]:=\{1,2,\dots,k\}\).

\begin{remark}\itshape
	The risk metric $\EE\big[\big\|{\hbtheta-\btheta_1^{\star}}\big\|_{\bQ}^2\big]$ is a generalized version of the classical mean squared error (MSE), where the matrix $\bQ$ provides flexibility to tailor the metric for different purposes. For example, choosing $\bQ=\bI_{p}$, the identity matrix, recovers the standard MSE. In other settings, specifying  a different $\bQ$ yields a prespecified metric aligned with the target objective. In a regression model $y=\bx^{\top}\btheta_1^{\star}+\epsilon$, when the second moment $\EE\big[\bx\bx^{\top}\big]$ is available, choosing $\bQ=\EE\big[\bx\bx^{\top}\big]$ converts the metric into $\EE\big[\big\|{\bx^{\top}\hbtheta-\bx^{\top}\btheta_1^{\star}}\big\|_{2}^2\big]$, which quantifies the prediction error of the fitted model.
\end{remark}

If the two populations share the same mean, i.e.,  $\btheta_1^{\star}=\btheta_2^{\star}$, then the Maximum Likelihood Estimator (MLE) of $\btheta_1^{\star}$ is
\begin{equation}\label{eq:MLE}
	\overline{\btheta}=\argmin_{\btheta} \frac{1}{2} \sum_{j=1}^{2}\sum_{i=1}^{n_j}\norm{\bx_{ji}-\btheta}_{\bSigma_j^{-1}}^2=\big(n_1\bSigma_1^{-1}+n_2\bSigma_2^{-1}\big)^{-1} \big(n_1\bSigma_1^{-1}\tbtheta_1+n_2\bSigma_2^{-1}\tbtheta_2\big),
\end{equation}
where $\tbtheta_j=\frac{1}{n_j}\sum_{i=1}^{n_j}\bx_{ji}$ $ \left(j=1, 2\right)$ denotes the sample mean of $\calD_j$ and is the single-set MLE for $\btheta_j^{\star}$. Compared to the single-set MLE $\tbtheta_1$, the pooled estimator $\overline{\btheta}$ reduces variance by transferring information across both data sets, since $\mathrm{Cov}\left(\overline{\btheta}\right)=\left(n_1\bSigma_1^{-1}+n_2\bSigma_2^{-1}\right)^{-1}\preceq n_1^{-1}\bSigma_1=\mathrm{Cov}\big(\tbtheta_1\big)$.

For a general setting where $\btheta_1^{\star}$ may be different from $\btheta_2^{\star}$, pooling introduces bias and may perform worse than using the single-set estimator $\tbtheta_1$ alone. The challenge is therefore to design an estimator that adapts to the unknown similarity between populations, which should borrow strength from $\calD_2$ when information borrowing is beneficial, yet remain faithful to $\calD_1$ when the two data sets disagree. This motivates the following shrinkage estimator:
\begin{equation}\label{eq:shrinkage-estimator}
	\overline{\btheta}_{t}:=\left(1-t\right)\tbtheta_1+t\overline{\btheta} = \tbtheta_1 + t\bW_2\big(\tbtheta_2-\tbtheta_1\big),
\end{equation}
where $\bW_2:=n_2 \big(n_1\bSigma_1^{-1}+n_2\bSigma_2^{-1}\big)^{-1}\bSigma_2^{-1}$ and $t \in \left[0, 1\right]$.

The shrinkage estimator \eqref{eq:shrinkage-estimator} shrinks the single-set estimator $\tbtheta_1$ to the pooled estimator $\overline{\btheta}$, and the amount of shrinkage is controlled by the parameter $t$. By adjusting $t$ properly, the estimator behaves similarly to the pooled MLE when $\btheta_1^{\star}\approx\btheta_2^{\star}$, while avoiding excessive bias when the two populations differ. We next develop an adaptive method for selecting the shrinkage size $t$.

\subsection{Adaptive Shrinkage Size} \label{sec:adaptive-t}

In this section, we develop an adaptive procedure for selecting the shrinkage parameter $t$ in \eqref{eq:shrinkage-estimator}.
Our strategy is to use Stein's unbiased risk estimation (SURE) to construct an unbiased estimator of the risk $\E\normQ{\overline\btheta_t-\btheta_1^{\star}}^2$ for each $t \in \left[0, 1\right]$ and minimize it to select the optimal $t$. The following lemma provides an identity for evaluating the risk of a broad class of estimators obtained by perturbing a Gaussian vector, and will be applied in particular to the shrinkage estimator $\overline\btheta_t$.

\begin{lemma}\label{lem:sure}
	Let $\bx\sim\mathcal{N}\left(\btheta,\bSigma\right)$. Suppose $\bg:\R^p\to\R^p$ is weakly differentiable and
	$\E\left|\left(x_k-\theta_k\right) g_{k'}\left(\bx\right)\right| + \E\left|\frac{\partial g_{k'}\left(\bx\right)}{\partial x_{k}}\right| < \infty$ for all $k, k'\in \left[p\right]$.
	Denote the Jacobian matrix of $\bg$ by $\bJ$. Then, for all $\bQ\succ 0$,
	\[
	\E \normQ{\bx+\bg\left(\bx\right)-\btheta}^2
	= \ip{\bSigma}{\bQ} + \E\!\left[2\ip{\bJ\left(\bx\right)\bSigma}{\bQ} + \normQ{\bg\left(\bx\right)}^2 \right].
	\]
\end{lemma}

To apply Lemma \ref{lem:sure}, we view $\tbtheta_1 \sim \calN\left(\btheta_1^{\star}, n_1^{-1}\bSigma_1\right)$ as the basic observation $\bx$, treat $\tbtheta_2$ and $t$ as fixed, and let $\bg\left(\bx\right)=t\bW_2\big(\tbtheta_2-\bx\big)$, which yields
\begin{equation}\label{eq:MSE-in-t}
	\E \normQ{\overline{\btheta}_t-\btheta_1^{\star}}^2
	= \ip{n_1^{-1}\bSigma_1}{\bQ}  -2t\ip{n_1^{-1}\bW_2\bSigma_1}{\bQ} +\E\!\left[ t^2 \normQ{\bW_2\big(\tbtheta_2-\tbtheta_1\big)}^2\right].
\end{equation}
Hence, the risk is a quadratic function of $t$ with coefficients that are either known or can be unbiasedly estimated from the data. Specifically, $q\big(t;\tbtheta_1,\tbtheta_2\big)=\ip{n_1^{-1}\bSigma_1}{\bQ}\allowbreak-2t\ip{n_1^{-1}\bW_2\bSigma_1}{\bQ}\allowbreak+t^2\big\|\bW_2\big(\tbtheta_2-\tbtheta_1\big)\big\|_{\bQ}^2$ is an unbiased estimator of $\E\normQ{\overline\btheta_t-\btheta_1^{\star}}^2$ for every fixed $t\in\left[0,1\right]$. The minimizer of $q\big(t;\tbtheta_1,\tbtheta_2\big)$ over $\left[0, 1\right]$ is $\widehat t=\min\big\{1,\ip{n_1^{-1}\bW_2\bSigma_1}{\bQ}/\big\|\bW_2\big(\tbtheta_2-\tbtheta_1\big)\big\|_{\bQ}^2\big\}$, and the corresponding SURE-type estimator is
\begin{equation}
	\overline\btheta_{\widehat t}
	=
	\tbtheta_1
	+
	\widehat t \,
	\bW_2\big(\tbtheta_2-\tbtheta_1\big).
	\label{eq:shrinkage-tstar}
\end{equation}

Although \eqref{eq:shrinkage-tstar} minimizes an unbiased risk estimate over $\left[0,1\right]$, it does not necessarily minimize the risk itself since the quantity $\widehat{t}$ is data-dependent. This necessitates a more rigorous risk analysis and motivates considering a possible refinement. Specifically, we consider a family of estimators
\begin{equation}\label{eq:shrinkage-estimator-s}
	\widehat{\btheta}_s
	=
	\tbtheta_1
	+
	\frac{s}{\big\|\bW_2\big(\tbtheta_2-\tbtheta_1\big)\big\|_{\bQ}^2\vee s}
	\bW_2\left(\tbtheta_2-\tbtheta_1\right),
\end{equation}
where $s\geq0$. Equivalently, the data-dependent shrinkage size is
$
\widehat t_s
=
\min\big\{1,s/\big\|\bW_2\big(\tbtheta_2-\tbtheta_1\big)\big\|_{\bQ}^2\big\}\in\left[0,1\right]
$. The SURE-type estimator in \eqref{eq:shrinkage-tstar} is the special case $s=\overline s:=\ip{n_1^{-1}\bW_2\bSigma_1}{\bQ}$.
The following theorem provides a non-asymptotic upper bound on the risk of $\widehat{\btheta}_s$ and identifies a range of $s$ that guarantees improvement over the single-set estimator.
\begin{theorem}\label{thm:risk-bound-s}
	Let $\bS=n_1^{-1}\bQ^{1/2}\bW_2\bSigma_1\bQ^{1/2}$. If $\Tr\left(\bS\right)>2\left\lVert\bS\right\rVert_2$, then, for any $s\in\left(0,2\Tr\left(\bS\right)-4\left\lVert\bS\right\rVert_2\right)$, we have
	\begin{equation}
		\E\normQ{\widehat{\btheta}_s-\btheta_1^{\star}}^2
		-\ip{n_1^{-1}\bSigma_1}{\bQ}\le
		-s\,\left[2\Tr\left(\bS\right)-4\left\lVert\bS\right\rVert_2-s\right]
		\E\left[
		\frac{1}{\normQ{\bW_2\left(\tbtheta_2-\tbtheta_1\right)}^2\vee s}
		\right].
		\label{eq:risk-upper}
	\end{equation}
	Consequently,
	$\E\normQ{\widehat{\btheta}_s-\btheta_1^{\star}}^2
	<\E\normQ{\tbtheta_1-\btheta_1^{\star}}^2
	=\ip{n_1^{-1}\bSigma_1}{\bQ}$.
\end{theorem}

Theorem~\ref{thm:risk-bound-s} asserts that, under the condition $\Tr\left(\bS\right)>2\left\|\bS\right\|_2$, the shrinkage estimator \eqref{eq:shrinkage-estimator-s} strictly improves on the single-set estimator $\tbtheta_1$ for every $s$ in the interval $\left(0,2\Tr\left(\bS\right)-4\left\|\bS\right\|_2\right)$. The quadratic factor in \eqref{eq:risk-upper} is maximized at
$
\underline{s}=\Tr\left(\bS\right)-2\left\|\bS\right\|_2,
$
which gives
\begin{equation}
	\widehat{\btheta}_{\underline{s}}
	=
	\tbtheta_1
	+
	\frac{\Tr\left(\bS\right)-2\left\|\bS\right\|_2}
	{\normQ{\bW_2\left(\tbtheta_2-\tbtheta_1\right)}^2\vee\left[\Tr\left(\bS\right)-2\left\|\bS\right\|_2\right]}
	\bW_2\left(\tbtheta_2-\tbtheta_1\right).
	\label{eq:shrinkage-sopt}
\end{equation}
Meanwhile, recall that $\overline\btheta_{\widehat{t}}$ in \eqref{eq:shrinkage-tstar} corresponds to $\widehat{\btheta}_{\overline{s}}$ with $\overline{s} = \ip{n_1^{-1}\bW_2\bSigma_1}{\bQ}=\Tr\left(\bS\right)$.  Theorem \ref{thm:risk-bound-s} ensures that both $\widehat{\btheta}_{\underline{s}}$ and $\widehat{\btheta}_{\overline{s}}$ lead to risk improvement over $ \tbtheta_1$, provided that $\Tr\left(\bS\right)>2\left\lVert\bS\right\rVert_2$ and $\Tr\left(\bS\right)>4\left\lVert\bS\right\rVert_2$, respectively. In the homoscedastic setting $\bSigma_j=\sigma_j^2\bI_p$ and $\bQ=\bI_p$, the untruncated version of $\widehat{\btheta}_{\underline{s}}$ agrees with the estimator proposed in \cite{green1991james}.


\begin{remark}[Tuning-free Property]\label{rmk:tuning-free}\itshape
A key feature of our shrinkage construction is that the transfer strength is determined analytically rather than by cross-validation, sample splitting, or a testing threshold. When $\Tr\left(\bS\right)>4\left\lVert\bS\right\rVert_2$, the preceding analysis yields an explicit population interval $\left[\underline{s}, \overline{s}\right]$, determined by $\bS$, over which $\widehat{\btheta}_s$ is guaranteed to improve on the single-set estimator. Theorem \ref{thm:risk-bound-shat} below further provides a common risk upper bound for all shrinkage sizes in this interval.

When $\bSigma_1$ and $\bSigma_2$ are unknown, we compute covariance estimators $\big\{\widehat\bSigma_j\big\}$, replace $\bW_2$ and $\bS$ by their plug-in analogues, and set
$
\widehat{s}
=
\Tr\big(\widehat{\bS}\big)-\big\|\widehat{\bS}\big\|_2
$, the midpoint of the data-driven interval
$
\big[
\Tr\big(\widehat{\bS}\big)-2\big\|\widehat{\bS}\big\|_2,\,
\Tr\big(\widehat{\bS}\big)
\big]
$. Among the points in the population interval, the midpoint maximizes the minimum distance to the two endpoints, and hence its plug-in analogue has the largest margin against covariance-estimation perturbations. Section~\ref{sec:plugin-covariance-stability} provides a risk bound for this fully plug-in estimator under consistent covariance estimation, which differs from Theorem \ref{thm:risk-bound-shat} by only a lower-order term and does not require $\big\{\widehat\bSigma_j\big\}$ to be independent of $\big\{\widetilde\btheta_j\big\}$. Details on computing $\big\{\widehat\bSigma_j\big\}$ for smooth $M$-estimators are given in Section~\ref{sec:general-M-estimation}.
\end{remark}

In the next theorem, we establish a common upper bound on the risk of $\widehat{\btheta}_{s}$ for all
$s \in \left[\underline{s},\overline{s}\right]$.

\begin{theorem}\label{thm:risk-bound-shat}
	Suppose that $\Tr\left(\bS\right)>4\left\|\bS\right\|_2$. For every $s\in\left[\underline{s},\overline{s}\right]$,
	\[
	\begin{aligned}
		\E\normQ{\widehat{\btheta}_s-\btheta_1^{\star}}^2
		&\le
		\ip{\big(n_1\bSigma_1^{-1}+n_2\bSigma_2^{-1}\big)^{-1}}{\bQ}
		+
		\frac{\Tr\left(\bS\right)\,\normQ{\bW_2\left(\btheta_2^{\star}-\btheta_1^{\star}\right)}^2}
		{\Tr\left(\bS\right)+\normQ{\bW_2\left(\btheta_2^{\star}-\btheta_1^{\star}\right)}^2}
		+4\left\|\bS\right\|_2.
	\end{aligned}
	\]
\end{theorem}
The upper bound in Theorem \ref{thm:risk-bound-shat} consists of three terms. The first term is the risk of the oracle estimator \eqref{eq:MLE} that would be attainable if the two populations shared the same mean.
The second term captures the effect of heterogeneity between the two means. It is bounded above by
$
\Tr\left(\bS\right)\wedge
\left\|\bW_2\left(\btheta_2^\star-\btheta_1^\star\right)\right\|_{\bQ}^2
$, which increases with the discrepancy between $\btheta_1^\star$ and $\btheta_2^\star$ but is capped by \(\Tr\left(\bS\right)\).
The third term is the price of selecting the shrinkage magnitude from the same data. It scales with the operator norm $\left\|\bS\right\|_2$ and is therefore dimension-free, becoming negligible relative to the first term when the effective dimension grows.

To further illustrate the upper bound in Theorem~\ref{thm:risk-bound-shat}, consider a specific setting
$\bQ=\bSigma_1=\bSigma_2=\bI_p$ and $n_2/n_1=\kappa$. In this case, the first term becomes
$\frac{p}{n_1+n_2}=\frac{p}{n_1\left(1+\kappa\right)}$, which corresponds to the risk of the pooled estimator \eqref{eq:MLE} when $\btheta_1^{\star}=\btheta_2^{\star}$. The second term can be upper bounded by
$\left\|\btheta_2^\star-\btheta_1^\star\right\|_2^2 \wedge \frac{p}{n_1}$, which captures the bias induced by source heterogeneity while remaining capped at the single-task rate $p/n_1$, indicating that the proposed estimator limits the amount of transfer. The third term is $\frac{4\kappa}{n_1\left(1+\kappa\right)}$ which, relative to the first term $\frac{p}{n_1\left(1+\kappa\right)}$, amounts to an effective-dimension inflation from $p$ to $p+4\kappa$. For fixed $\kappa$, the inflation is dimension-free and is negligible compared with the first term as $p$ increases.

\begin{remark}\itshape
	The pooled estimator in \eqref{eq:MLE} corresponds to the full-shrinkage case \(t=1\) in \eqref{eq:shrinkage-estimator}.
	Although it attains the oracle variance under the common-mean model, its risk under
	heterogeneous populations contains an additional bias term. Specifically,
	$
	\mathbb E\left\|\overline\btheta-\btheta_1^\star\right\|_{\bQ}^2
	=
	\bigl\langle
	\bigl(n_1\bSigma_1^{-1}+n_2\bSigma_2^{-1}\bigr)^{-1},\bQ
	\bigr\rangle
	+
	\left\|\bW_2\left(\btheta_2^\star-\btheta_1^\star\right)\right\|_{\bQ}^2
	$.
A sufficient condition under which the risk upper bound in Theorem \ref{thm:risk-bound-shat} is smaller than the
	risk of the pooled estimator is
	$
		\left\|\bW_2\left(\btheta_2^\star-\btheta_1^\star\right)\right\|_{\bQ}^4
	>
	4\left\|\bS\right\|_2\big(\Tr\left(\bS\right)+
	\left\|\bW_2\left(\btheta_2^\star-\btheta_1^\star\right)\right\|_{\bQ}^2\big)
	$.
	This condition shows that the proposed adaptive shrinkage estimator is guaranteed to improve
	over the pooled estimator whenever the  discrepancy term
	\(\left\|\bW_2\left(\btheta_2^\star-\btheta_1^\star\right)\right\|_{\bQ}^2\) is sufficiently large. In the simple case
	\(\bQ=\bSigma_1=\bSigma_2=\bI_p\) with \(n_2/n_1=\kappa\),
	the sufficient condition above reduces to
	$
	\left\|\btheta_2^\star-\btheta_1^\star\right\|_2^2
	>
	\frac{2\left(1+\kappa\right)\left(1+\sqrt{p+1}\right)}{n_1\kappa}
	$. Thus, the proposed estimator is guaranteed to outperform the pooled estimator once the squared discrepancy exceeds the order of
$
	\frac{\left(1+\kappa\right)\sqrt{p}}{n_1\kappa}
$.
\end{remark}

\subsection{Connection with a Regularized MLE}\label{sec:equivalence}

This section provides an optimization interpretation of the shrinkage estimator \eqref{eq:shrinkage-estimator}. The regularized estimator introduced below is not used to define or compute our proposed estimator. It serves only to show that the same covariance-aware direction arises from a familiar penalized-likelihood formulation. To formalize this, define the regularized estimator
\begin{equation}\label{eq:regularized-MSE}
	\big(\hbtheta_{1}^{\lambda}, \hbtheta_2^{\lambda}\big):= \argmin_{\btheta_1,\btheta_2} \bigg\{ \frac{1}{2}
	\sum_{j=1}^{2}\sum_{i=1}^{n_j} \left\| \bx_{ji} - \btheta_j \right\|_{ \bSigma_j^{-1} }^2 +\lambda \norm{\left(n_1^{-1} \bSigma_1 + n_2^{-1} \bSigma_2\right)^{-1/2}\left(\btheta_1-\btheta_2\right)}_2 \bigg\}.
\end{equation}

The regularization parameter $\lambda$ controls the extent to which the two means are encouraged to be close. When $\lambda$ is large, the penalty enforces $\btheta_1 \approx \btheta_2$, leading to an estimator similar to the pooled MLE  in \eqref{eq:MLE}. When $\lambda$ is small, the penalty is weak, and each estimate $\hbtheta_j$ remains close to its own empirical mean $\tbtheta_j$. The regularization therefore implements an interpolation between full pooling and complete separation.  This idea is similar to the multi-task objective in \citet{duan2023adaptive}. However, our penalty term is reparameterized by the inverse-covariance metric $\left(n_1^{-1}\bSigma_1+n_2^{-1}\bSigma_2\right)^{-1/2}$, which aligns the shrinkage strength with statistical uncertainty and yields a closed-form solution that is closely related to the proposed shrinkage estimator.

\begin{theorem}\label{thm:regularized-MLE}
	Let $\bz = \left(n_1^{-1}\bSigma_1 + n_2^{-1}\bSigma_2\right)^{-1/2}\Big(\tbtheta_{1} - \tbtheta_{2}\Big).
	$
	The solution to \eqref{eq:regularized-MSE} satisfies
	\begin{align*}
		& \widehat\btheta_1^{\lambda} = \widetilde{\btheta}_1 + \min \Big\{ \frac{\lambda}{\left\| \bz \right\|_2}, 1 \Big\}\cdot \bW_2\Big(\tbtheta_2-\tbtheta_1\Big)
		=  \overline{\btheta} -  \Big(1-\frac{\lambda}{\left\| \bz \right\|_2}\Big)_+ \cdot \bW_2\Big(\tbtheta_2-\tbtheta_1\Big)
		, \\
		& \widehat\btheta_2^{\lambda} = \widetilde{\btheta}_2 +\min \Big\{ \frac{\lambda}{\left\| \bz \right\|_2}, 1 \Big\}\cdot \bW_1\Big(\tbtheta_1-\tbtheta_2\Big)
		=  \overline{\btheta} -  \Big(1-\frac{\lambda}{\left\| \bz \right\|_2}\Big)_+ \cdot \bW_1\Big(\tbtheta_1-\tbtheta_2\Big)
		,
	\end{align*}
	where $\bW_j:=n_j \big(n_1\bSigma_1^{-1}+n_2\bSigma_2^{-1}\big)^{-1}\bSigma_j^{-1} $ for $j\in\{1, 2\}$, and $\overline{\btheta}$ is defined in \eqref{eq:MLE}.
	Consequently, $\widehat{\btheta}_1^{\lambda} = \widehat{\btheta}_2^{\lambda} = \overline{\btheta}$ holds if and only if $\left\| \bz \right\|_2 \leq \lambda$. 
\end{theorem}

Theorem~\ref{thm:regularized-MLE} reveals that the regularized estimator $\hbtheta_1^{\lambda}$ shares the same path as the shrinkage estimator $\overline{\btheta}_t$ in \eqref{eq:shrinkage-estimator}. Both estimators shrink the single-set estimator $\tbtheta_1$ along the same random direction $\bW_2\big(\tbtheta_2-\tbtheta_1\big)$, with the amount of shrinkage determined by a scalar factor. For a given $t$, choosing the data-dependent regularization parameter $\lambda=t\left\|\bz\right\|_2$ yields $\hbtheta_1^{\lambda}=\overline{\btheta}_t$.

\begin{remark}\itshape \label{rmk:regularization}
	Theorem  \ref{thm:regularized-MLE} provides a pre-test interpretation of
	the regularized estimator $\big(\hbtheta_{1}^{\lambda}, \hbtheta_2^{\lambda}\big)$. The statistic $\bz$ can be used to test \(H_0:\btheta_1^\star=\btheta_2^\star\) against
	\(H_1:\btheta_1^\star\neq\btheta_2^\star\). Suppose we set the critical value to be $\lambda$, yielding a rejection region $\{ \norm{\bz}_2 > \lambda \}$. If $\norm{\bz}_2 \leq \lambda$, then $H_0$ is not rejected, and $\hbtheta_{1}^{\lambda}= \hbtheta_2^{\lambda}=\overline{\btheta}$. Under $H_0$, $\bz \sim \calN\left(\bm{0},\bI_p\right)$, and thus the significance level can be controlled at $\alpha$ if one chooses $\lambda^2$ to be the $1-\alpha$ quantile of a $\chi^2_p$ distribution.  When $H_0$ is rejected, the estimates $\hbtheta_1^{\lambda}$ and $\hbtheta_2^{\lambda}$ are not the same but shrunk toward each other. Theorem  \ref{thm:regularized-MLE} implies that the distance $\left\| \widehat{\btheta}^{\lambda}_j - \widetilde\btheta_j \right\|_2$  is always $O\left(\lambda/\sqrt{n_j}\right)$. This provides a safety net for the regularized estimator as any discrepancy between the two sets can only have limited impact on $\widehat{\btheta}_j^{\lambda}$.
\end{remark}
In addition, we discuss the relationship between our proposed estimator and the classical James--Stein (JS) estimators in Section~\ref{sec:relationship} of the supplementary materials. In short, we show that, in the homoscedastic Gaussian mean setting, our proposed shrinkage estimator is guaranteed to have strictly smaller Euclidean risk than the JS estimator under mild conditions.
\begin{remark}\itshape
	The bounded shrinkage path also admits an empirical Bayes interpretation. 
	Consider the Gaussian hierarchical model with the improper prior $p\left(\btheta_1^\star\right)\propto 1$ and
	$
	\btheta_2^\star-\btheta_1^\star \mid \gamma
	\sim
	\calN\big(\bm{0},\,\gamma\left(n_1^{-1}\bSigma_1+n_2^{-1}\bSigma_2\right)\big)
	$
	for some $\gamma\ge 0$. Under this model, the posterior mean of $\btheta_1^\star$ is
	$
	\EE\big[\btheta_1^\star\mid \tbtheta_1,\tbtheta_2,\gamma\big]
	=
	\tbtheta_1+\frac{1}{1+\gamma}\bW_2\big(\tbtheta_2-\tbtheta_1\big).
	$
	Therefore, any realized multiplier in $\left[0,1\right]$, including the multiplier used in \eqref{eq:shrinkage-estimator-s}, can be represented by a corresponding empirical value of $\left(1+\gamma\right)^{-1}$. This is again an interpretation of the borrowing factor rather than the definition of the estimator.
\end{remark}

\section{Sequential Shrinkage for Multiple Data Sets}\label{sec:multiset}\label{sec:sequential}

This section proposes the sequential shrinkage procedure under the multi-set Gaussian mean model and establishes its theoretical guarantees. Suppose that we observe i.i.d.~samples on the target set $\calD_1=\{\bx_{1i}\}_{i=1}^{n_1} \sim \calN\left(\btheta_1^{\star}, \bSigma_1\right)$ and $m-1$ source sets $\calD_j=\{\bx_{ji}\}_{i=1}^{n_j} \sim \calN\left(\btheta_j^{\star}, \bSigma_j\right)$ for $j\in\{2,3,\dots, m\}$ with unknown population means $\{\btheta_j^{\star}\}_{j\in\left[m\right]} \subset \R^p$ and known covariance matrices $\{\bSigma_j\}_{j \in \left[m\right]} \subset \R^{p \times p}$. Similar to the two-set setting, our goal is to estimate the population mean of the target set, $\btheta_1^{\star}$, and we use the risk $\EE\big[\big\|{\hbtheta-\btheta_1^{\star}}\big\|_{\bQ}^2\big]$ to evaluate an estimator $\hbtheta$. If the $m$ populations share the same mean, i.e.,  $\btheta_1^{\star}=\btheta_2^{\star}=\dots=\btheta_m^{\star}$, then the MLE of $\btheta_1^{\star}$ is
\begin{equation}\label{eq:MLE-multi-set}
	\overline{\btheta}=\argmin_{\btheta} \frac{1}{2} \sum_{j=1}^{m}\sum_{i=1}^{n_j}\left(\btheta-\bx_{ji}\right)^{\top}\bSigma_j^{-1}\left(\btheta-\bx_{ji}\right)=\Big(\sum_{j=1}^{m}n_j\bSigma_j^{-1}\Big)^{-1} \Big(\sum_{j=1}^{m}n_j\bSigma_j^{-1}\tbtheta_j\Big),
\end{equation}
where $\tbtheta_j=\frac{1}{n_j}\sum_{i=1}^{n_j}\bx_{ji}$ $ \left(j\in\left[m\right]\right)$ denotes the sample mean of $\calD_j$.
In heterogeneous settings, however, directly using \(\overline{\btheta}\) can be biased if some source means are far from \(\btheta_1^\star\). A direct multi-source analogue of the two-set estimator shrinks \(\tbtheta_1\) along the pooled direction in \eqref{eq:MLE-multi-set}. Specifically, with \(\bW_j=n_j\left(\sum_{k=1}^{m}n_k\bSigma_k^{-1}\right)^{-1}\bSigma_j^{-1}\), \(j\in\{2,3,\dots, m\}\), define the single-step shrinkage estimator by
\begin{equation}\label{eq:shrinkage-estimator-multi-set}
	\widehat{\btheta}_{\rm ss}
	=
	\tbtheta_1
	+
	\frac{s}{
		\left\|
		\sum_{j=2}^{m}
		\bW_j\left(\tbtheta_j-\tbtheta_1\right)
		\right\|_{\bQ}^2\vee s
	}
	\sum_{j=2}^{m}
	\bW_j\left(\tbtheta_j-\tbtheta_1\right),
	\qquad s\in \left[\underline{s}, \overline{s}\right],
\end{equation}
where $\underline{s}=\Tr\left(\bS\right)-2\norm{\bS}_2$ and $\overline{s}=\Tr\left(\bS\right)$ with $\bS
=
\sum_{j=2}^{m}
n_1^{-1}
\bQ^{1/2}
\bW_j
\bSigma_1
\bQ^{1/2}
$.
This single-step estimator is a useful baseline and also motivates the sequential construction below. The detailed theoretical analysis for \eqref{eq:shrinkage-estimator-multi-set} is deferred to Section~\ref{sec:single-step-shrinkage} of the supplementary materials.

The single-step estimator in \eqref{eq:shrinkage-estimator-multi-set} shrinks the target estimator
\(\tbtheta_1\) along a single pooled direction formed by all source data sets. This construction is effective when
the weighted population discrepancy
$
\sum_{j=2}^{m}\bW_j\left(\btheta_j^\star-\btheta_1^\star\right)
$
is small. In heterogeneous multi-source problems, however, some sources may be close to the target, whereas others may have population parameters far from \(\btheta_1^\star\). A single pooled direction can then be dominated by heterogeneous sources and exhibit a large overall discrepancy from the target, which may lead the estimator to shrink only mildly even when some sources are highly informative. This motivates a sequential procedure that compares source-specific updates one at a time. At each step, the procedure evaluates a priority score for every remaining source and applies the update with the largest score. The shrinkage multiplier, rather than a hard screening decision, determines how much each source contributes: compatible sources receive substantial updates, whereas distant sources are automatically attenuated.

We now propose the sequential shrinkage estimator, which is summarized in Algorithm \ref{alg:sequential-shrinkage}. The procedure starts from the target estimator
\(\widehat\btheta^{\left[0\right]}=\tbtheta_1\) with working covariance matrix
\(\bV^{\left[0\right]}=n_1^{-1}\bSigma_1\) and a source index set $\mathcal A^{\left[0\right]}=\{2, 3, \dots, m\}$. Suppose that after \(\ell\) steps, the current estimator is
\(\widehat\btheta^{\left[\ell\right]}\), its working covariance matrix is \(\bV^{\left[\ell\right]}\), and the
remaining source index set is \(\mathcal A^{\left[\ell\right]}\). For each \(j\in\mathcal A^{\left[\ell\right]}\), define
$
	\bW_j^{\left[\ell\right]}
	:=
	\big[
	\left(\bV^{\left[\ell\right]}\right)^{-1}
	+
	n_j\bSigma_j^{-1}
	\big]^{-1}
	n_j\bSigma_j^{-1}
$. The corresponding shrinkage direction is
$
	\tDelta_j^{\left[\ell\right]}
	:=
	\bW_j^{\left[\ell\right]}
	\big(
	\tbtheta_j-\widehat\btheta^{\left[\ell\right]}
	\big)
$.

To determine the shrinkage size, define
$
	\bS_j^{\left[\ell\right]}
	:=
	\bQ^{1/2}
	\bW_j^{\left[\ell\right]}
	\bV^{\left[\ell\right]}
	\bQ^{1/2}
$,
$
	\underline s_j^{\left[\ell\right]}
	:=
	\big[\Tr\big(\bS_j^{\left[\ell\right]}\big)
	-
	2\big\|\bS_j^{\left[\ell\right]}\big\|_2\big]_+,
$
and
$
	\widehat s_j^{\left[\ell\right]}
	:=
	\Tr\big(\bS_j^{\left[\ell\right]}\big)
	-
	\big\|\bS_j^{\left[\ell\right]}\big\|_2.
$
The quantities \(\underline s_j^{\left[\ell\right]}\) and \(\widehat s_j^{\left[\ell\right]}\) are the local analogues of the lower endpoint and midpoint of the two-set interval. We set
$
	\gamma_j^{\left[\ell\right]}
	:=
	\frac{\widehat s_j^{\left[\ell\right]}}
	{\big\|\tDelta_j^{\left[\ell\right]}\big\|_{\bQ}^2\vee\widehat s_j^{\left[\ell\right]}},
$
with the convention that \(\gamma_j^{\left[\ell\right]}=0\) when \(\widehat s_j^{\left[\ell\right]}=0\). The candidate priority score is
\[\mathfrak S_j^{\left[\ell\right]}:=2\gamma_j^{\left[\ell\right]}\underline s_j^{\left[\ell\right]}-\left(\gamma_j^{\left[\ell\right]}\right)^2\big\|\tDelta_j^{\left[\ell\right]}\big\|_{\bQ}^2.\]
We note that, for the shrinkage size, we use the midpoint $\widehat s_j^{\left[\ell\right]}$, since it provides the largest margin to the interval endpoints as in the two-set construction. In contrast, for the priority score $\mathfrak S_j^{\left[\ell\right]}$, we use the conservative choice $\underline s_j^{\left[\ell\right]}$ appearing in Theorem~\ref{thm:risk-bound-s}. The sequential rule then chooses $j_{\ell+1}\in\argmax_{j\in\mathcal A^{\left[\ell\right]}}\mathfrak S_j^{\left[\ell\right]}$.

After choosing \(j_{\ell+1}\), we update the estimator by
$
\widehat\btheta^{\left[\ell+1\right]}
=
\widehat\btheta^{\left[\ell\right]}
+
\gamma_{j_{\ell+1}}^{\left[\ell\right]}
\tDelta_{j_{\ell+1}}^{\left[\ell\right]}
$, and the working covariance proxy is updated as
\begin{equation}\label{eq:V-update}
	\bV^{\left[\ell+1\right]}
	=
	\left(
	\bI_p-\gamma_{j_{\ell+1}}^{\left[\ell\right]}\bW_{j_{\ell+1}}^{\left[\ell\right]}
	\right)
	\bV^{\left[\ell\right]}
	\left(
	\bI_p-\gamma_{j_{\ell+1}}^{\left[\ell\right]}\bW_{j_{\ell+1}}^{\left[\ell\right]}
	\right)^\top+
	\left(\gamma_{j_{\ell+1}}^{\left[\ell\right]}\right)^2
	\bW_{j_{\ell+1}}^{\left[\ell\right]}
	n_{j_{\ell+1}}^{-1}\bSigma_{j_{\ell+1}}
	\left(\bW_{j_{\ell+1}}^{\left[\ell\right]}\right)^\top .
\end{equation}
The matrix $\bV^{\left[\ell\right]}$ in \eqref{eq:V-update} is used as a working covariance proxy rather than the unconditional covariance of $\widehat\btheta^{\left[\ell\right]}$.
Finally, we remove the incorporated source from the candidate set and repeat the same procedure until all sources have been considered.  We deliberately let the algorithm consider all sources rather than introduce a separate test or stopping rule to preserve any residual useful contribution from all sources. Theorem~\ref{thm:seq-large-separation-trace-condition} shows that the updates for the heterogeneous sources are asymptotically negligible under a mixed-source separation regime, indicating that a separate stopping rule is not essential in this setting.


\begin{algorithm}[!t]
	\caption{Sequential Shrinkage}
	\label{alg:sequential-shrinkage}
	\small
	\setlength{\baselineskip}{0.92\baselineskip}
	\setlength{\abovedisplayskip}{2pt}
	\setlength{\belowdisplayskip}{2pt}
	\setlength{\abovedisplayshortskip}{1pt}
	\setlength{\belowdisplayshortskip}{1pt}
	\begin{algorithmic}[1]
		\Statex \hspace{-18pt}\textbf{Input:} Local estimators \(\{\tbtheta_j\}_{j=1}^m\), covariance matrices
		\(\{\bSigma_j\}_{j=1}^m\), sample sizes \(\{n_j\}_{j=1}^m\), and matrix \(\bQ\).
		\Statex\hspace{-18pt}\textbf{Output:} sequential shrinkage estimator \(\widehat\btheta_{\rm seq}\).

		\State Initialize
		$
		\widehat\btheta^{\left[0\right]}=\tbtheta_1$,
		$\bV^{\left[0\right]}=n_1^{-1}\bSigma_1$, and
		$\mathcal A^{\left[0\right]}=\{2,3,\dots, m\}$.

		\For{\(\ell=0,\ldots,m-2\)}
		\For{each \(j\in\mathcal A^{\left[\ell\right]}\)}
		\State Compute \(\bW_j^{\left[\ell\right]}=\{\left(\bV^{\left[\ell\right]}\right)^{-1}+n_j\bSigma_j^{-1}\}^{-1}n_j\bSigma_j^{-1}\),  \(\tDelta_j^{\left[\ell\right]}=\bW_j^{\left[\ell\right]}\big(\tbtheta_j-\widehat\btheta^{\left[\ell\right]}\big)\),  and  \(\bS_j^{\left[\ell\right]}=\bQ^{1/2}\bW_j^{\left[\ell\right]}\bV^{\left[\ell\right]}\bQ^{1/2}\).
		\State Compute \(\underline s_j^{\left[\ell\right]}=\big[\Tr\big(\bS_j^{\left[\ell\right]}\big)-2\big\|\bS_j^{[\ell]}\big\|_2\big]_+\), \(\widehat s_j^{\left[\ell\right]}=\Tr\big(\bS_j^{\left[\ell\right]}\big)-\big\|\bS_j^{\left[\ell\right]}\big\|_2\), and \(\gamma_j^{\left[\ell\right]}=\widehat s_j^{\left[\ell\right]}/\big(\big\|\tDelta_j^{\left[\ell\right]}\big\|_{\bQ}^2\vee\widehat s_j^{\left[\ell\right]}\big)\), with \(0/0=0\).
		\State Compute \(\mathfrak S_j^{\left[\ell\right]}=2\gamma_j^{\left[\ell\right]}\underline s_j^{\left[\ell\right]}-\left(\gamma_j^{\left[\ell\right]}\right)^2\big\|\tDelta_j^{\left[\ell\right]}\big\|_{\bQ}^2\).
		\EndFor

		\State Choose the next update \(j_{\ell+1}\in\argmax_{j\in\mathcal A^{\left[\ell\right]}}\mathfrak S_j^{\left[\ell\right]}\).

		\State Update
		$
		\widehat\btheta^{\left[\ell+1\right]}
		=
		\widehat\btheta^{\left[\ell\right]}
		+
		\gamma_{j_{\ell+1}}^{\left[\ell\right]}
		\tDelta_{j_{\ell+1}}^{\left[\ell\right]}
		$ and update $\bV^{\left[\ell+1\right]}$ by \eqref{eq:V-update}.

		\State Set \(\mathcal A^{\left[\ell+1\right]}=\mathcal A^{\left[\ell\right]}\setminus\{j_{\ell+1}\}\).
		\EndFor

		\State Return \(\widehat\btheta_{\rm seq}=\widehat\btheta^{\left[m-1\right]}\).
	\end{algorithmic}
\end{algorithm}

The intuitive advantage of the sequential rule is most transparent in the presence of both homogeneous and heterogeneous sources. Homogeneous sources have population parameters that are close to the target parameter, whereas heterogeneous sources have large population discrepancies. A single-step estimator forms one pooled shrinkage direction using all sources simultaneously. When the source collection contains heterogeneous sources, this pooled direction can be dominated by cross-source discrepancy, leading to a small shrinkage size and limited information transfer. The sequential estimator instead repeatedly reassesses source-specific updates. Homogeneous sources tend to yield smaller discrepancies \(\left\|\tDelta_j^{\left[\ell\right]}\right\|_{\bQ}\), larger priority scores \(\mathfrak S_j^{\left[\ell\right]}\), and larger shrinkage, whereas heterogeneous sources receive small shrinkage. The procedure therefore differentiates the amount of shrinkage across sources without hard source screening.

We next formalize this intuition and establish the theoretical properties of the sequential shrinkage estimator under a mixed-source setting.

\begin{assumption}\label{assump:seq-1}
Let \(\mathcal I=\{2,3,\ldots,m\}\), and let \(\mathcal I_0\) and \(\mathcal I_1\) be nonempty disjoint sets satisfying \(\mathcal I=\mathcal I_0\cup\mathcal I_1\). Assume that \(n_j\asymp n\) uniformly over \(j\in\left[m\right]\), 
$
\max_{j\in\mathcal I_0}
\left\|\btheta_j^\star-\btheta_1^\star\right\|_{\bQ} \ll \sqrt{p/n}
$, and $\min_{j\in\mathcal I_1}\left\|\btheta_j^\star-\btheta_1^\star\right\|_{\bQ} \gg \sqrt{p/n}$, as \(n,p\to\infty\).
The indices in \(\mathcal I_0\) and  \(\mathcal I_1\) are referred to as homogeneous and heterogeneous sources, respectively.
\end{assumption}

Assumption \ref{assump:seq-1} formalizes the mixed-source separation setting. The set $\mathcal I_0$ contains the locally homogeneous sources, whose discrepancies from the target are negligible relative to the estimation error scale. The set \(\mathcal I_1\) contains the remaining heterogeneous sources.

\begin{assumption}\label{assump:seq-2}
Assume that, for some constants \(v_{-},v_{+}>0\), \(v_{-}n^{-1}\bI_p \preceq n_j^{-1}\bSigma_j \preceq v_{+}n^{-1}\bI_p\) for all \(j\in\left[m\right]\). For \(r\in\left[m\right]\), define \(\mathscr V_r:=\{\bV\succ0: v_- \left(rn\right)^{-1}\bI_p\preceq\bV\preceq v_+ n^{-1}\bI_p\}\). For \(\bV\in\mathscr V_r\) and \(j\in\left[m\right]\), let \(\bW_{j,\bV}:=\left(\bV^{-1}+n_j\bSigma_j^{-1}\right)^{-1}n_j\bSigma_j^{-1}\) and \(\bS_{j,\bV}:=\bQ^{1/2}\bW_{j,\bV}\bV\bQ^{1/2}\). Then assume that
\begin{equation*}
	d_{\rm eff}:=\min_{1\le r\le m}\min_{1\le j\le m}\inf_{\bV\in\mathscr V_r}\frac{\Tr\left(\bS_{j,\bV}\right)}{\left\|\bS_{j, \bV}\right\|_2} > 3. \end{equation*}
\end{assumption}

Assumption \ref{assump:seq-2} imposes a uniform effective-dimension condition along the sequential path. For each possible current covariance matrix \(\bV\) and candidate source \(j\), \(\bS_{j,\bV}=\bQ^{1/2}\bW_{j,\bV}\bV\bQ^{1/2}\) is the sequential analogue of the matrix \(\bS\) in the single-step risk bound. The ratio \(\Tr\left(\bS_{j,\bV}\right)/\left\|\bS_{j,\bV}\right\|_2\) measures the effective number of eigendirections contributing to the trace term. It is bounded by \(p\), and is of order \(p\) when the eigenvalues of \(\bS_{j,\bV}\) are not overly concentrated. We theoretically show that the sets \(\{\mathscr V_r\}_{r \in \left[m\right]}\) contain the entire
sequential covariance path in Algorithm \ref{alg:sequential-shrinkage}, that is,
$
\bV^{\left[\ell\right]}\in\mathscr V_{\ell+1}$
for all
$\ell =0,1,\ldots,m-1
$. It follows that every candidate source $j$ in Algorithm \ref{alg:sequential-shrinkage} satisfies
$
\Tr\big(\bS_j^{\left[\ell\right]}\big) / \big\|\bS_j^{\left[\ell\right]}\big\|_2
\ge d_{\rm eff}>3
$.
The threshold $3$ arises from the use of the midpoint shrinkage scale. Indeed, the midpoint $\Tr\big(\bS_j^{\left[\ell\right]}\big)-\big\|\bS_j^{\left[\ell\right]}\big\|_2$ lies in the two-set risk-improving interval $\big(0,2\Tr\big(\bS_j^{\left[\ell\right]}\big)-4\big\|\bS_j^{\left[\ell\right]}\big\|_2\big)$ whenever $\Tr\big(\bS_j^{\left[\ell\right]}\big)>3\big\|\bS_j^{\left[\ell\right]}\big\|_2$. Moreover, the priority score $\mathfrak{S}^{\left[\ell\right]}_j$ is strictly positive under this condition.
	\begin{remark}\label{rmk:illustration-eff-dim}\itshape
	To further illustrate Assumption \ref{assump:seq-2}, we consider a simplified setting where $\bQ=\bI_p$,
	$n_j=n$, $\bSigma_j=\bI_p$ for all $j\in\left[m\right]$,
	and \(v_-=v_+=1\). Under this setting, for any $\bV \in \mathscr V_r$,
	we can show that
	$
	\frac{\Tr\left(\bS_{j,\bV}\right)}
	{\left\|\bS_{j,\bV}\right\|_2}
	\ge
	1+\frac{2\left(p-1\right)}{r\left(r+1\right)}
	$,
	and hence
	$
	d_{\rm eff}
	=
	1+\frac{2\left(p-1\right)}{m\left(m+1\right)}
	$.
	Therefore, \(d_{\rm eff}\asymp p\) when \(m=O\left(1\right)\). The proof of this result is provided in Section \ref{sec:proof-multiset} of the supplementary materials.
\end{remark}

Under the two assumptions above, the following result shows that the procedure concentrates its non-negligible shrinkage on the homogeneous sources, suppresses the contribution of separated heterogeneous sources, and attains the corresponding oracle risk.

\begin{theorem}
\label{thm:seq-large-separation-trace-condition}
Suppose Assumptions~\ref{assump:seq-1} and \ref{assump:seq-2} hold. Define \(\delta_{\max}:=\max_{j\in\mathcal I_1}\left\|\btheta_j^\star-\btheta_1^\star\right\|_{\bQ}\). For any \(0<\rho<1\), suppose that \(m=O\left(1\right)\), \(p\gg\log\left(n\delta_{\max}^2/\rho\right)\), \(d_{\rm eff}\asymp p\), as \(n,p\to\infty\). Then the following statements hold.
\begin{enumerate}
\item The first \(m_0\) update indices of the update sequence are exactly the elements of \(\mathcal I_0\) with probability at least \(1-\rho\). In particular, \(\PP\left(\{j_1,j_2,\ldots,j_{m_0}\}=\mathcal I_0\right)\ge1-\rho\), where \(m_0:=\left|\mathcal I_0\right|\ge1\).

\item The risk of the sequential shrinkage estimator satisfies
\begin{equation}\label{eq:seq-risk-upper-clean}
\EE\left\|\widehat\btheta_{\rm seq}-\btheta_1^\star\right\|_{\bQ}^2
=
\ip{\bV_{\mathcal I_0}}{\bQ}+o\left(\frac pn\right),
\end{equation}
where \(\bV_{\mathcal I_0}:=\left(n_1\bSigma_1^{-1}+\sum_{j\in\mathcal I_0}n_j\bSigma_j^{-1}\right)^{-1}\).

\item Let \(\widehat\btheta_{\rm ss}\) denote the single-step shrinkage estimator in \eqref{eq:shrinkage-estimator-multi-set} with any \(s\in\left[\underline s,\overline s\right]\), and suppose the conditions in Theorem~\ref{thm:risk-bound-s-multiset} hold. Define \(\delta_{\mathcal I}:=\big\|\sum_{j\in\mathcal I}\bW_j\left(\btheta_j^\star-\btheta_1^\star\right)\big\|_{\bQ}\). If \(\delta_{\mathcal I}\gg\sqrt{p/n}\), then \(\EE\big\|\widehat\btheta_{\rm seq}-\btheta_1^\star\big\|_{\bQ}^2<\EE\big\|\widehat\btheta_{\rm ss}-\btheta_1^\star\big\|_{\bQ}^2\) for all sufficiently large \(n\).
\end{enumerate}
\end{theorem}

Theorem~\ref{thm:seq-large-separation-trace-condition} has three implications. First, with probability tending to one, the update sequence gives priority to all homogeneous sources before incorporating any heterogeneous source. This prioritization property follows from the separation condition in Assumption \ref{assump:seq-1}. For homogeneous sources, the observed discrepancies are at the estimation-noise scale, whereas for heterogeneous sources the population discrepancy dominates the noise. The condition is the signal-to-noise requirement for the exact prioritization and oracle conclusions.

Second, after the homogeneous sources have been incorporated, the sequential estimator achieves the oracle risk
\(\ip{\bV_{\mathcal I_0}}{\bQ}\) up to a lower-order term. This oracle benchmark is the risk of the estimator that knows the set of homogeneous sources in advance and aggregates the target with exactly those sources. Hence \eqref{eq:seq-risk-upper-clean} shows that the cost of adaptively prioritizing the homogeneous sources is asymptotically negligible.

Third, we compare sequential shrinkage with single-step shrinkage. When the pooled discrepancy
\(\delta_{\mathcal I}\) is larger than the noise level, the single-step direction is contaminated by the heterogeneous sources and cannot attain the oracle risk based on \(\mathcal I_0\). In contrast, the sequential estimator first shrinks toward the homogeneous sources and assigns negligible shrinkage to the heterogeneous sources. This yields a strictly smaller quadratic risk than the single-step shrinkage estimator for sufficiently large \(n\). Specifically, under the conditions in Theorem \ref{thm:seq-large-separation-trace-condition}, we show that
$
\EE
\big\|
	\widehat\btheta_{\rm ss}-\btheta_1^\star
\big\|_{\bQ}^2
\ge
\ip{n_1^{-1}\bSigma_1}{\bQ}
-
o\left(\frac pn\right)
$.
By \eqref{eq:seq-risk-upper-clean}, we have
\[
\EE
\normQ{
	\widehat\btheta_{\rm ss}-\btheta_1^\star
}^2
-
\EE
\normQ{
	\widehat\btheta_{\rm seq}-\btheta_1^\star
}^2
\ge
\ip{
	n_1^{-1}\bSigma_1-\bV_{\mathcal I_0}
}{\bQ}
-
o\left( p / n\right)
\gtrsim p / n
.\] 
Hence, when the pooled discrepancy
\(\delta_{\mathcal I}\) is large, the proposed estimator asymptotically
attains the oracle risk,
whereas the single-step estimator remains asymptotically as risky as
the target-only estimator, and the risk difference is of order \(p/n\).

\section{Extension to Smooth $M$-estimation}
\label{sec:general-M-estimation}

We next extend the proposed shrinkage framework to a broad class of smooth $M$-estimation problems. Suppose that, for each $j\in\left[m\right]$, the data set $\cD_j=\{\bx_{ji}\}_{i=1}^{n_j}$ consists of i.i.d. samples from a distribution $\cP_j$ over the sample space $\cX_j$, and let $\ell_j\left(\cdot,\bx\right):\RR^p\to\RR$ be a loss function. The population loss minimizer is $\btheta_j^\star\in\argmin_{\btheta\in\RR^p}\EE_{\bx\sim\cP_j}\ell_j\left(\btheta,\bx\right)$, and our goal is to learn $\btheta_1^{\star}$. For a fixed constant $B_\Theta>0$, define the parameter space
$
\Theta_p
=
\left\{
\btheta\in\RR^p
\,\middle|\,
\left\|\btheta\right\|_2
\le
B_\Theta\sqrt{p}
\right\}
$.
For every $j\in\left[m\right]$, let
$
\tbtheta_j
\in
\argmin_{\btheta\in\Theta_p}
f_j\left(\btheta\right)$ be the single-task estimator, where
$f_j\left(\btheta\right)
=
\frac{1}{n_j}
\sum_{i=1}^{n_j}
\ell_j\left(\btheta,\bx_{ji}\right)$ is the empirical loss.
For fixed $p$ and under standard regularity conditions, \[\sqrt{n_j}\big(\tbtheta_j-\btheta_j^\star\big)\stackrel d\longrightarrow\calN\left(\bm0,\bSigma_j\right) \text{ with }\bSigma_j=\bH_j^{\star-1}\bV_j^\star\bH_j^{\star-1}, \] $\bH_j^\star=\EE\nabla^2\ell_j\left(\btheta_j^\star,\bx\right)$, and $\bV_j^\star=\EE\big[\nabla\ell_j\left(\btheta_j^\star,\bx\right)\nabla\ell_j\left(\btheta_j^\star,\bx\right)^\top\big]$.

When $\btheta_1^\star=\cdots=\btheta_m^\star$, an aggregated estimator can be formed as
\begin{equation}\label{eq:set}
\tbtheta_{\{\bW_j\}}:=\sum_{j=1}^m\bW_j\tbtheta_j,
\quad
\sum_{j=1}^m\bW_j=\bI_{p\times p},
\quad
\bW_j\in\RR^{p\times p} \text{ for } j  \in \left[m\right].
\end{equation}
For fixed $p$ and $m$, the following proposition identifies the optimal weights.

\begin{prop}\label{prop:optimal-aggregation}
Assume that $\btheta_1^\star=\cdots=\btheta_m^\star$ and $\sqrt{n_j}\big(\tbtheta_j-\btheta_j^\star\big)\stackrel d\longrightarrow\calN\left(\bm0,\bSigma_j\right)$ for all $j\in\left[m\right]$. For any deterministic matrix weights satisfying \eqref{eq:set}, $\big(\sum_{j=1}^m n_j^{-1}\bW_j\bSigma_j\bW_j^\top\big)^{-1/2}\big(\tbtheta_{\{\bW_j\}}-\btheta_1^\star\big)\stackrel d\longrightarrow\calN\left(\bm0,\bI_{p}\right)$. Moreover, the asymptotic covariance matrix is minimized at $\bW_j^\star=n_j\left(\sum_{k=1}^m n_k\bSigma_k^{-1}\right)^{-1}\bSigma_j^{-1}$ in the sense that $\sum_{j=1}^m n_j^{-1}\bW_j\bSigma_j\bW_j^\top\succeq\sum_{j=1}^m n_j^{-1}\bW_j^\star\bSigma_j\bW_j^{\star\top}$.
\end{prop}

Proposition~\ref{prop:optimal-aggregation} motivates the shrinkage direction used below, while Theorem~\ref{thm:risk-upper-general} establishes the corresponding risk guarantees allowing diverging $p$. We define 
\begin{equation*}
	\overline{\btheta}
	:=
	\tbtheta_{\{\bW_j^{\star}\}}
	=
	\Big(\sum_{k=1}^m n_k\bSigma_k^{-1}\Big)^{-1}
	\Big(\sum_{k=1}^m n_k\bSigma_k^{-1}\tbtheta_k\Big),
\end{equation*}
which reduces to $\eqref{eq:MLE}$ in the two-set Gaussian mean example.  In a general case where $\btheta_1^{\star}$ may be different from $\btheta_2^{\star},\dots,\btheta_m^{\star}$, we propose the shrinkage estimator $\overline{\btheta}_{t}$ similar to $\eqref{eq:shrinkage-estimator}$, but here we further handle the case where the covariance matrices $\{\bSigma_j\}$ are unknown. Specifically,
\begin{equation*}
	\overline{\btheta}_{t}
	:=
	\tbtheta_1 + t\sum_{j=2}^m \widehat{\bW}_j\Big(\tbtheta_j-\tbtheta_1\Big),
\end{equation*}
where $t\in\left[0,1\right]$ and
$
\widehat{\bW}_j
=
n_j\big\{\sum_{k=1}^m n_k\big(\widehat\bSigma_k^{\rm c}\big)^{-1}\big\}^{-1}
\big(\widehat\bSigma_j^{\rm c}\big)^{-1}
$.

To ensure invertibility, we clip the empirical Hessian and sandwich covariance matrices. For a symmetric matrix $\bA=\bU\operatorname{diag}\left(a_1,\ldots,a_p\right)\bU^\top$, define
$\mathcal C_{\left[a_-,a_+\right]}\left(\bA\right)
:=
\bU\operatorname{diag}\big(\left[a_1\right]_{a_-}^{a_+},\ldots,\left[a_p\right]_{a_-}^{a_+}\big)\bU^\top$, where
$\left[a\right]_{a_-}^{a_+}:=\min\{a_+,\max\{a_-,a\}\}$. 
For a fixed numerical constant $\varepsilon^{\rm c}\in\left(0,1\right)$, we compute
\begin{equation}\label{eq:clipped-sandwich}
\widehat\bH_j^{\rm c}
=
\mathcal C_{\left[\varepsilon^{\rm c},\left(\varepsilon^{\rm c}\right)^{-1}\right]}\big(\widehat\bH_j\big),
\quad
\widehat\bSigma_j^{\rm c}
=
\mathcal C_{\left[\varepsilon^{\rm c},\left(\varepsilon^{\rm c}\right)^{-1}\right]}
\Big(\big(\widehat\bH_j^{\rm c}\big)^{-1}\widehat\bV_j\big(\widehat\bH_j^{\rm c}\big)^{-1}\Big),
\end{equation}
where
$\widehat\bV_j
=
\frac{1}{n_j}\sum_{i=1}^{n_j}
\nabla\ell_j\big(\tbtheta_j,\bx_{ji}\big)
\nabla\ell_j\big(\tbtheta_j,\bx_{ji}\big)^\top$ and $
\widehat\bH_j
=
\frac{1}{n_j}\sum_{i=1}^{n_j}
\nabla^2\ell_j\big(\tbtheta_j,\bx_{ji}\big)$,
for $j\in\left[m\right]$. The constant $\varepsilon^{\rm c}$ is fixed before observing the data and is used only to keep the empirical matrices uniformly positive definite. 

The proposed shrinkage framework for Gaussian mean estimation can then be extended to a broad class of smooth $M$-estimation problems. The single-step estimator in \eqref{eq:shrinkage-estimator-multi-set} is extended by replacing each sample mean with the corresponding $M$-estimator and replacing $\bSigma_j$ with $\widehat\bSigma_j^{\rm c}$, and the sequential estimator $\widehat\btheta_{\rm seq}$ is obtained from Algorithm~\ref{alg:sequential-shrinkage} through the same replacements. 

We now impose regularity conditions commonly used in the analysis of smooth $M$-estimators. Let $ \mathbb S^{p-1}:=\{\bv\in\RR^p:\left\|\bv\right\|_2=1\}$ and $\mathcal B_j\left(r\right)=\{\btheta:\left\|\btheta-\btheta_j^\star\right\|_2\le r\}$.

\begin{assumption}\label{assump:1}
Assume that $\ell_j\left(\cdot,\bx\right):\RR^p\to\RR$ is convex and three times differentiable for every $j\in\left[m\right]$. Let $\bH_j^\star=\EE\nabla^2\ell_j\left(\btheta_j^\star,\bx\right)$ and $\bV_j^\star=\EE\big[\nabla\ell_j\left(\btheta_j^\star,\bx\right)\nabla\ell_j\left(\btheta_j^\star,\bx\right)^\top\big]$. There exist constants $0<\underline\lambda<\overline\lambda<\infty$ such that $\underline\lambda\le\lambda_{\min}\left(\bH_j^\star\right)\le\lambda_{\max}\left(\bH_j^\star\right)\le\overline\lambda$ and $\underline\lambda\le\lambda_{\min}\left(\bV_j^\star\right)\le\lambda_{\max}\left(\bV_j^\star\right)\le\overline\lambda$ for every $j\in\left[m\right]$. Assume that the clipping constant satisfies
$0<\varepsilon^{\rm c}<\frac{1}{2}\min\big\{\underline\lambda,\frac{\underline\lambda}{\overline\lambda^2},\frac{1}{\overline\lambda},\frac{\underline\lambda^2}{\overline\lambda}\big\}$.
\end{assumption}

\begin{assumption}\label{assump:2}
	Assume that there exist positive constants $\tau$ and $r$ such that, for all $j\in\left[m\right]$, it holds that $\norm{\bv^{\top}\nabla \ell_j\left(\btheta, \bx_{j1}\right)}_{\psi_2} \leq \tau$ and $\norm{\bv^{\top}\nabla^2 \ell_j\left(\btheta, \bx_{j1}\right)\bv}_{\psi_1} \leq \tau^2$  for all $\btheta \in \mathcal{B}_j\left(r\right)$ and $\bv \in \mathbb S^{p-1}$.
\end{assumption}

\begin{assumption}\label{assump:3} Assume that there exist positive constants $L$ and $c$ such that, for $j \in \left[m\right]$, it holds that
	$\sup_{\btheta \in \mathcal{B}_j\left(r\right)}\norm{\EE_{\bx \sim \mathcal{P}_j}\left[\nabla^3 \ell_j\left(\btheta, \bx\right)\right]}_{\mathrm{op}} \leq L$,
	and
	\[\EE_{\bx \sim \mathcal{P}_j}\Big[\sup_{\btheta \in \mathcal{B}_j\left(r\right)}\norm{\nabla^3 \ell_j\left(\btheta, \bx\right)}_{\mathrm{op}}+\sup_{\btheta \in \mathcal{B}_j\left(r\right)}\norm{\nabla \ell_j\left(\btheta, \bx\right) \otimes \nabla^2 \ell_j\left(\btheta, \bx\right)}_{\mathrm{op}}\Big] \leq Lp^{c}, \] 
	where $\norm{\bT}_{\mathrm{op}}:=\sup\limits_{\bx, \by, \bz \in \mathbb S^{p-1}} \sum\limits_{i, j, k \in \left[p\right]} T_{i, j, k}x_iy_jz_k$ for a third-order tensor $\bT \in \RR^{p \times p \times p}$.
\end{assumption}

\begin{assumption}\label{assump:4}
Assume that there exists a fixed constant $B^\star$ satisfying $0<B^\star<B_\Theta$ such that
$
\max_{j\in\left[m\right]}
\left\|\btheta_j^\star\right\|_2
\le
B^\star\sqrt{p}
$.
\end{assumption}

Assumption \ref{assump:1} requires uniform eigenvalue bounds on the population Hessian $\bH_j^\star$ and the score second-moment matrix $\bV_j^\star$, which guarantees local convexity and prevents the asymptotic covariance
$\bSigma_j = \bH_j^{\star-1} \bV_j^\star \bH_j^{\star-1}$ from becoming ill-conditioned.
Assumption \ref{assump:2} assumes tail conditions for the gradient and Hessian within a local neighborhood of $\btheta_j^\star$. Such conditions are standard for deriving concentration bounds for empirical processes.

Assumption \ref{assump:3} bounds the third derivative of the loss and plays two roles in our analysis. First, it justifies a second-order Taylor expansion of the empirical loss around $\btheta_j^\star$, with a remainder term that can be controlled uniformly over $\calB_j\left(r\right)$. Second, it helps quantify the additional error incurred by plugging $\widehat{\bSigma}_j^{\rm c}$ into the weights $\widehat{\bW}_j$, since estimating $\bH_j^\star$ and $\bV_j^\star$ relies on third-order smoothness.  Assumptions~\ref{assump:1}--\ref{assump:3} are standard and widely adopted in theoretical analyses of $M$-estimation \citep{mei2018landscape, ostrovskii2021finite, duan2023adaptive}.

 Assumption~\ref{assump:4} guarantees that the population parameters $\{\btheta_j^\star\}$ are inside the parameter space, which is mild since the radius $B^\star\sqrt p$ contains all parameters with uniformly bounded coordinates. Such constraints are also standard in general $M$-estimation and extremum-estimation theory \citep{wang2018optimal, han2019general,yiu2020inference}.
 Section~\ref{sec:M-examples} of the supplementary materials provides representative smooth $M$-estimation models that satisfy Assumptions~\ref{assump:1}--\ref{assump:4}.

We then establish theoretical guarantees for the sequential algorithm under the assumptions, with the same conclusions as Theorem \ref{thm:seq-large-separation-trace-condition} for Gaussian mean estimation.

\begin{theorem}\label{thm:risk-upper-general}
Suppose Assumptions~\ref{assump:1}--\ref{assump:4} hold and that the population parameters $\{\btheta_j^\star\}$ and asymptotic covariance matrices $\{\bSigma_j\}$ satisfy Assumptions~\ref{assump:seq-1} and~\ref{assump:seq-2}. If $m=O\left(1\right)$, $n_j\asymp n$ uniformly over $j\in\left[m\right]$, and $p\log^2\left(n\right)/n=o\left(1\right)$, then, under the remaining conditions of Theorem~\ref{thm:seq-large-separation-trace-condition}, the conclusions in parts 1 and 2 of Theorem~\ref{thm:seq-large-separation-trace-condition} continue to hold for $M$-estimation. Part 3 also continues to hold provided that
$\delta_{\mathcal I} \gg \sqrt{\frac{p\log n}{n}}\delta_{\max}$.
\end{theorem}

Theorem~\ref{thm:risk-upper-general} shows that, for  smooth \(M\)-estimation under the stated conditions, the sequential estimator achieves the same
prioritization, oracle-risk, and single-step improvement guarantees as in Gaussian mean estimation. The difference is that Theorem~\ref{thm:seq-large-separation-trace-condition} assumes known Gaussian covariance matrices, whereas Theorem~\ref{thm:risk-upper-general} handles estimated asymptotic covariance matrices for the local \(M\)-estimators. Assumptions~\ref{assump:1}--\ref{assump:4} ensure that the plug-in error induced by the covariance estimation is asymptotically negligible, thereby justifying the use of the same shrinkage rule and sequential construction.

\section{Numerical Study}\label{sec:numerical_study}

This section evaluates the proposed shrinkage estimators through numerical studies. Section~\ref{sec:simulation} studies two-set and multi-set designs with parameter and covariance heterogeneity, while Section~\ref{sec:real-data} conducts a real-data application to predicting completion rates of U.S. Service Schools. 

\subsection{Simulations}\label{sec:simulation}

\noindent\textbf{Two-set settings.} Let \(\mathcal D_1\) denote the target data set and \(\mathcal D_2\) denote the source data set. We set \(\btheta_1^\star=\mathbf 1_p/\sqrt p\) and $
\btheta_2^\star=\btheta_1^\star+\delta\bv_2
$, where \(\bv_2\in \mathbb{S}^{p-1}\) is drawn uniformly from the unit sphere. Specifically, we vary \(p\in\{20,50,100\}\) and \(\delta\in\left[0, 1\right]\), and report the risk, measured by the mean squared error (MSE), of each estimator for estimating $\btheta_1^\star$, averaged over $100$ independent replications. The competitor estimators include the \texttt{single-task} estimator \(\widetilde\btheta_1\), the \texttt{pooled} estimator in \eqref{eq:MLE}, and the \texttt{proposed} estimator in \eqref{eq:shrinkage-estimator-s} with
$
\widehat s=\Tr\left(\widehat\bS\right)-\left\|\widehat\bS\right\|_2$, $\widehat\bS=n_1^{-1}\bQ^{1/2}\widehat\bW_2\widehat\bSigma_1\bQ^{1/2}
$,
and
$
\widehat\bW_2=n_2\left(n_1\widehat\bSigma_1^{-1}+n_2\widehat\bSigma_2^{-1}\right)^{-1}\widehat\bSigma_2^{-1}
$.
The matrix \(\widehat\bSigma_j\) denotes the covariance or asymptotic covariance estimator of \(\widetilde\btheta_j\) obtained by the clipped construction in Section~\ref{sec:general-M-estimation}.

We consider three problems: Gaussian mean estimation, linear regression, and logistic regression. For Gaussian mean estimation, we generate independent observations $\bx_{ji}\sim \calN\left(\btheta_j^\star,\bSigma_j\right)$ for $j\in\{1,2\}$, where $\bSigma_j=\sigma_j^2\bR\left(\rho\right)$ with $R_{k, l}\left(\rho\right)=\rho^{\left|k-l\right|}$.
We specify $\left(\rho,\sigma_1,\sigma_2\right)=\left(0.25,1.5,1.0\right)$ and $\left(n_1,n_2\right)=\left(250,250\right)$. Details of the linear regression and logistic regression settings are relegated to Section \ref{sec:simulation-supplement} of the supplementary materials.

\begin{figure*}[!t]
    \centering
    \includegraphics[width=1.0\textwidth]{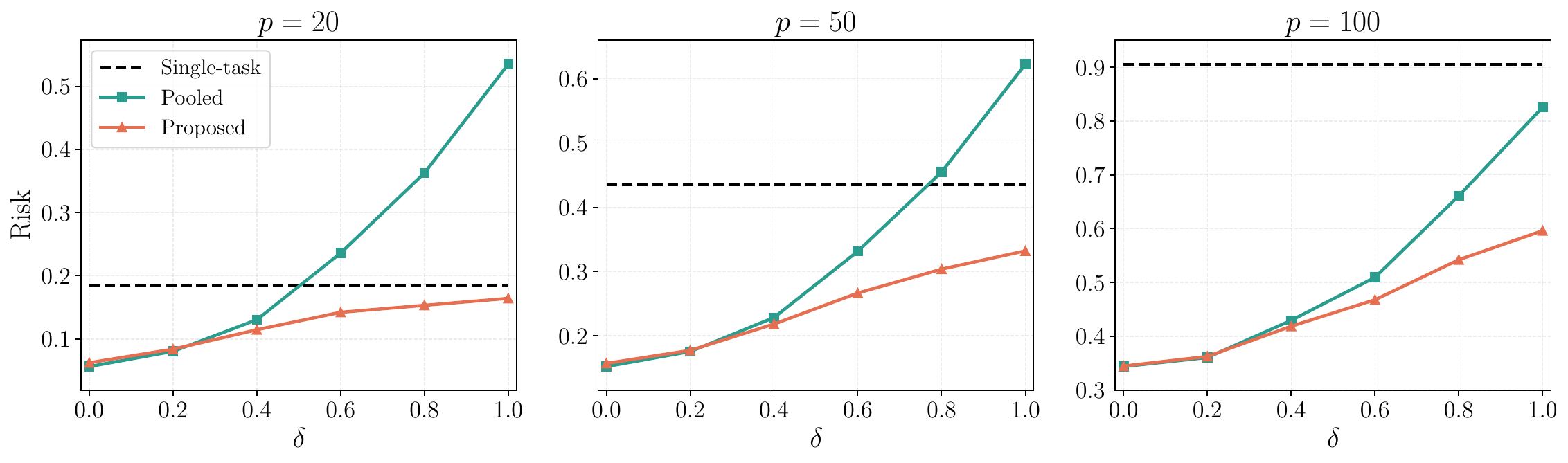}
    \caption{Two-set Gaussian mean simulation for \(p\in\{20,50,100\}\), with \(n_1=n_2=250\). The figure reports the risks of the \texttt{single-task}, \texttt{pooled}, and \texttt{proposed} estimators as \(\delta\) increases.}
    \label{fig:two-set-gaussian}
\end{figure*}

Figure~\ref{fig:two-set-gaussian}, together with Figures~\ref{fig:two-set-linear} and \ref{fig:two-set-logistic} in the supplementary materials, summarizes the risks of the estimators as the heterogeneity level $\delta$ increases. At $\delta=0$, \texttt{pooled} usually performs best, since it benefits from variance reduction under homogeneity. The performance of the \texttt{proposed} estimator stays close to that of \texttt{pooled}, with an absolute risk difference no larger than $0.011$ across all two-set experiments. As $\delta$ grows, the pooling bias of \texttt{pooled} increases, whereas the \texttt{proposed} estimator retains part of the variance reduction while limiting the movement toward the source. For large $\delta$, the \texttt{proposed} estimator attains the smallest risk in every setting.

A consistent pattern is that the risk reduction of the \texttt{proposed} estimator relative to the \texttt{single-task} estimator grows with the dimension $p$, showing that the empirical advantage of shrinkage becomes more pronounced in higher dimensions. This trend agrees with \eqref{eq:risk-upper}, where the risk difference typically increases with the dimension, since $\Tr\left(\bS\right)$ sums all eigenvalues of $\bS$ while $\left\|\bS\right\|_2$ depends only on the largest one.

\begin{figure}[!t]
    \centering
    \includegraphics[width=\linewidth]{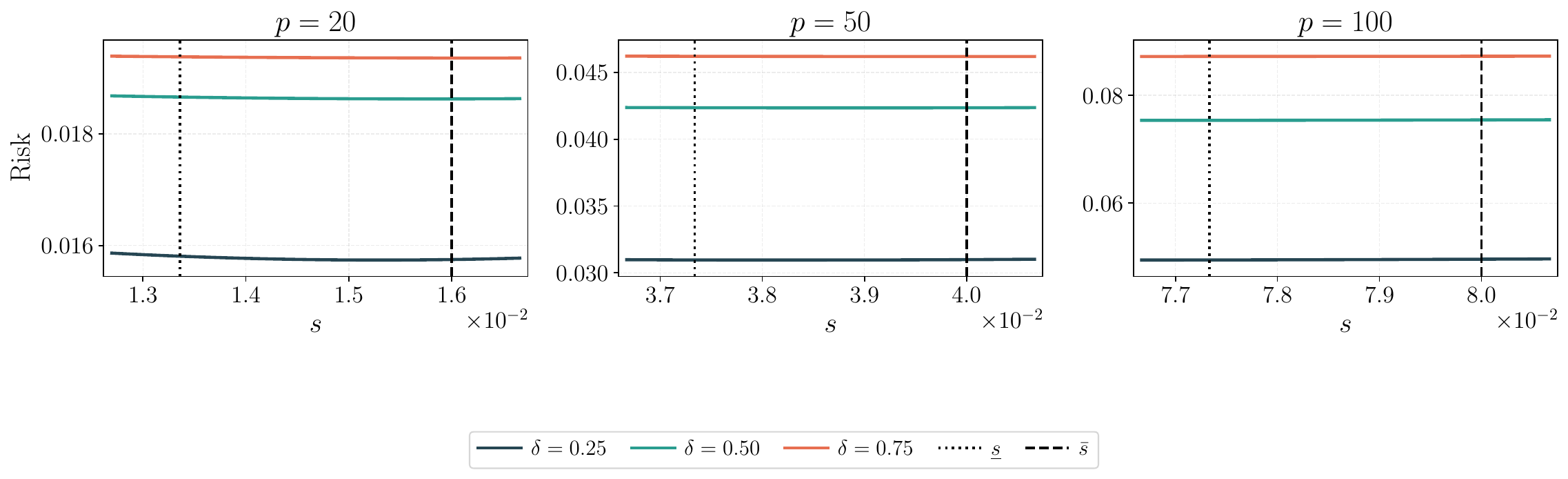}
    \caption{Sensitivity of the \texttt{proposed} estimator to the shrinkage size \(s\) for \(p\in\{20,50,100\}\) and \(\delta\in\{0.25,0.5,0.75\}\). The vertical lines mark \(s=\underline s\) and \(s=\bar s\).} 
    \label{fig:s-sweep}
\end{figure}

We next examine sensitivity to the shrinkage size. For each displayed dimension and heterogeneity level, Figure~\ref{fig:s-sweep} evaluates the risk over a grid of \(s\) values and marks the two endpoints $\underline{s}$ and $\overline{s}$. The curves are smooth and relatively flat within the interval $\left[\underline{s}, \overline{s}\right]$, and the two endpoints and their midpoint yield similar risks throughout the displayed range. 

\noindent\textbf{Multi-set settings.} We then consider five data sets, with \(\mathcal D_1\) as the target and \(\mathcal D_2,\ldots,\mathcal D_5\) as the sources. We draw \(\bv_2,\ldots,\bv_5\) independently and uniformly from the unit sphere $\mathbb{S}^{p-1}$. Let \(a=1\) for Gaussian mean estimation and linear regression and \(a=1/2\) for logistic regression. The five population parameters are
$
\btheta_1^\star=\mathbf 1_p/\sqrt p$,  
$\btheta_2^\star=\btheta_1^\star+\left(a/4\right)\sqrt{p/n_2}\bv_2$, 
$\btheta_3^\star=\btheta_1^\star+\left(a/4\right)\sqrt{p/n_3}\bv_3$, 
$\btheta_4^\star=\btheta_1^\star+a\delta\bv_4$, 
$\btheta_5^\star=\btheta_1^\star+a\delta\bv_5$.
Thus, \(\mathcal D_2\) and \(\mathcal D_3\) remain near the target at a fixed fraction of the local estimation-noise scale, whereas the separation of \(\mathcal D_4\) and \(\mathcal D_5\) is controlled by \(\delta\in\{0,0.5,\ldots,2.5\}\). This design examines a more ambiguous setting in which \(\mathcal D_2\) and \(\mathcal D_3\) differ from the target by an amount comparable to their estimation error, making them neither clearly homogeneous nor clearly heterogeneous. Details of the data generating process are relegated to Section \ref{sec:simulation-supplement} of the supplementary materials. In addition to \texttt{single-task}, \texttt{pooled}, and the \texttt{proposed} sequential algorithm, we also compare the \texttt{single-step} estimator in \eqref{eq:shrinkage-estimator-multi-set}.  The results are averaged over 100 independent replications.

\begin{figure*}[!t]
    \centering
    \includegraphics[width=\textwidth]{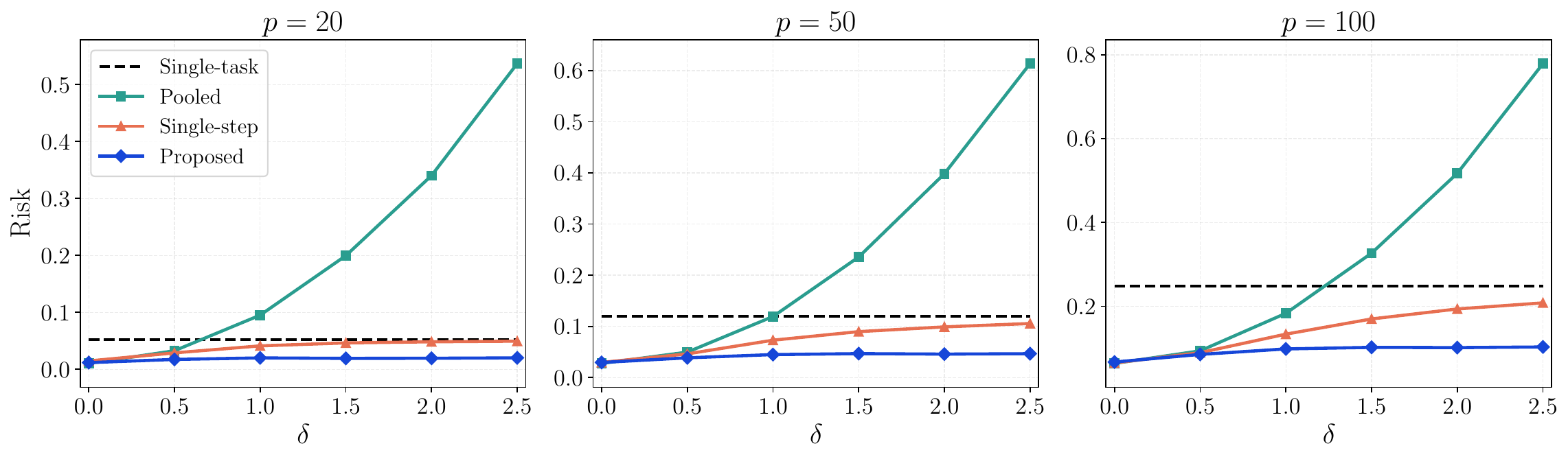}
    \caption{Multi-set Gaussian mean simulation for \(p\in\{20,50,100\}\), with \(n_1=n_2=n_3=n_4=n_5=400\).}
    \label{fig:multiset-gaussian-unknown}
\end{figure*}

\begin{figure*}[!t]
    \centering
    \includegraphics[width=\textwidth]{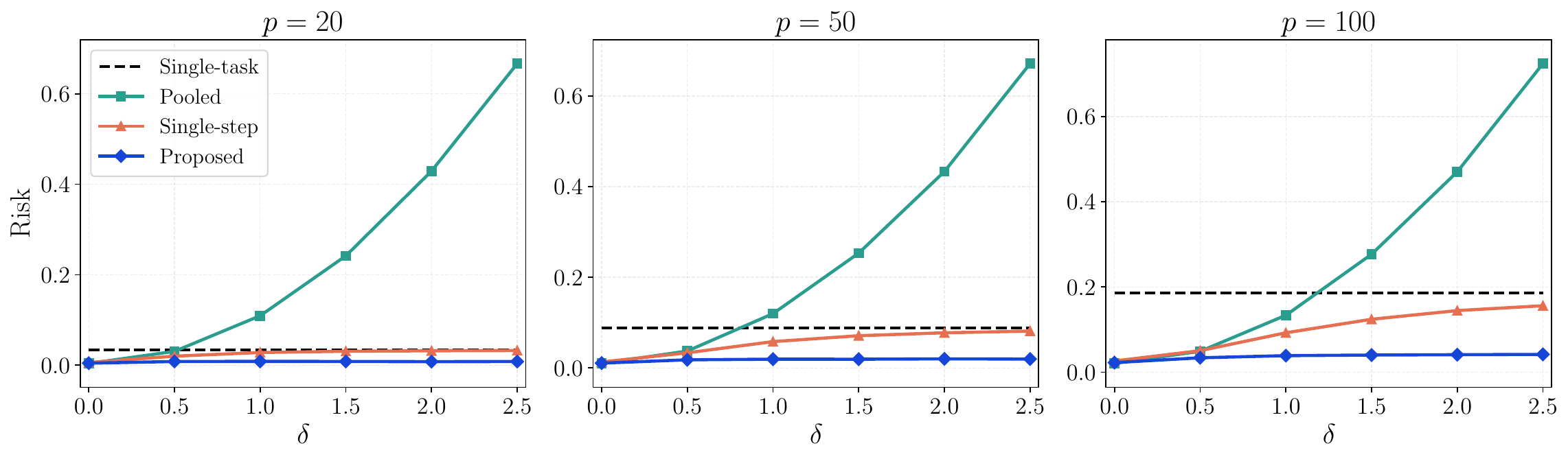}
    \caption{Multi-set linear regression simulation for \(p\in\{20,50,100\}\), with \(n_1=1000\) and \(n_2=n_3=n_4=n_5=2000\).}
    \label{fig:multiset-linear}
\end{figure*}

\begin{figure*}[!t]
    \centering
    \includegraphics[width=\textwidth]{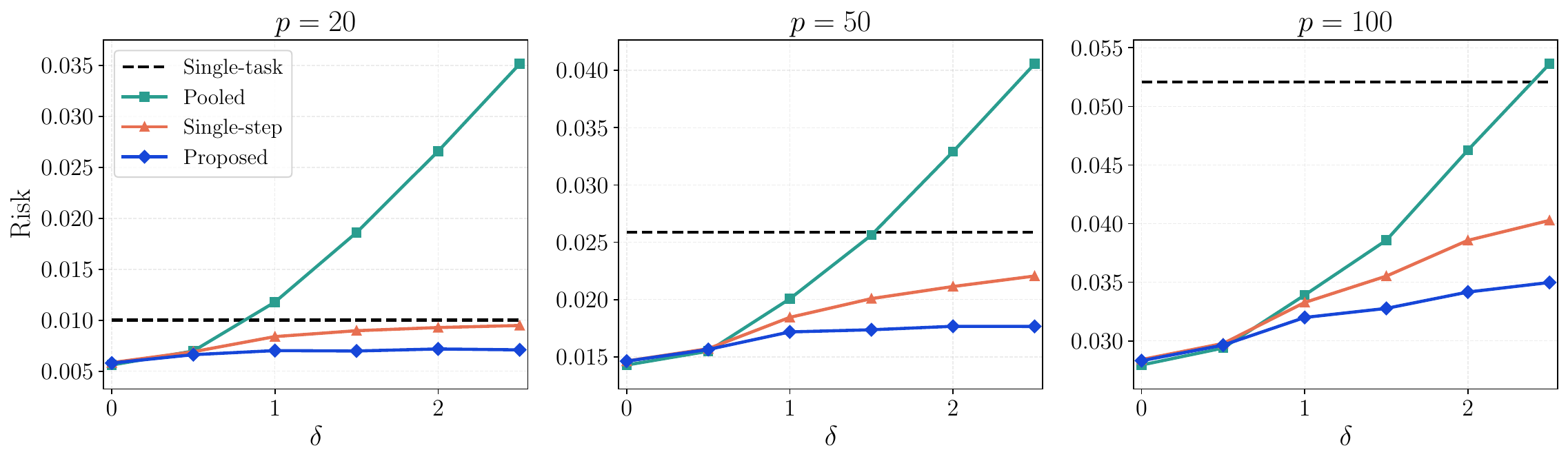}
    \caption{Multi-set logistic regression simulation for \(p\in\{20,50,100\}\), with \(n_1=20000\) and \(n_2=n_3=n_4=n_5=4000\).}
    \label{fig:multiset-logistic}
\end{figure*}

Figures~\ref{fig:multiset-gaussian-unknown}--\ref{fig:multiset-logistic} show how the two near-target sources and the two increasingly separated sources play different roles. At $\delta=0$, all four sources remain close to the target and \texttt{pooled} attains the smallest risk, with the absolute difference between the \texttt{proposed} and \texttt{pooled} estimators being at most $0.0034$. As $\delta$ increases, $\mathcal D_2$ and $\mathcal D_3$ continue to carry useful information, whereas $\mathcal D_4$ and $\mathcal D_5$ drift progressively farther from the target. The \texttt{proposed} method therefore outperforms \texttt{pooled} and \texttt{single-step}, as it allocates most of the borrowing to the informative sources while sharply attenuating the two sources that move away from the target. At $\delta=2.5$, the \texttt{proposed} estimator attains the smallest risk in all experiments, reducing the risk of \texttt{single-step} by 50.5\% to 59.3\% in Gaussian mean estimation, by 73.2\% to 76.0\% in linear regression, and by 13.1\% to 25.2\% in logistic regression.

\noindent\textbf{Covariance Heterogeneity.}
The multi-set simulations above impose a common covariance structure across all sets. We next consider a more heterogeneous setting in which the sources differ not only in their population parameters but also in their covariance matrices. Specifically, we retain the five-set organization of the Gaussian mean estimation experiment, hold the parameter discrepancies fixed, and vary only the covariance matrices. We set $n_1=250$ and $n_2=n_3=n_4=n_5=500$, draw $\bv_2,\ldots,\bv_5$ independently from the unit sphere, and define
$\btheta_1^\star=\mathbf 1_p/\sqrt p$,
$\btheta_2^\star=\btheta_1^\star+0.25\sqrt{p/n_2}\,\bv_2$,
$\btheta_3^\star=\btheta_1^\star+0.25\sqrt{p/n_3}\,\bv_3$,
$\btheta_4^\star=\btheta_1^\star+1.5\bv_4$, and
$\btheta_5^\star=\btheta_1^\star+1.5\bv_5$.
For a covariance heterogeneity level $\delta_\sigma$, let
$
\bD_A\left(\delta_\sigma\right)=\operatorname{diag}\!\big(\underbrace{1,\ldots,1}_{p/2},\underbrace{1+\delta_\sigma,\ldots,1+\delta_\sigma}_{p/2}\big)
\text{ and }
\bD_B\left(\delta_\sigma\right)=\operatorname{diag}\!\big(\underbrace{1+\delta_\sigma,\ldots,1+\delta_\sigma}_{p/2},\underbrace{1,\ldots,1}_{p/2}\big)
$.
We generate $\bx_{ji}\sim\mathcal N\left(\btheta_j^\star,\bSigma_j\right)$ with
$\bSigma_1=\bR\left(0.5\right)$,
$\bSigma_2=\bSigma_4=\bD_A\left(\delta_\sigma\right)\bR\left(0.5\right)\bD_A\left(\delta_\sigma\right)$, and
$\bSigma_3=\bSigma_5=\bD_B\left(\delta_\sigma\right)\bR\left(0.5\right)\bD_B\left(\delta_\sigma\right)$.
Accordingly, sources 2 and 4 are more precise on the first coordinate block, whereas sources 3 and 5 are more precise on the second. We vary $\delta_\sigma\in\{0,2,\ldots,10\}$ and average over 100 replications. Besides the \texttt{single-task}, \texttt{pooled}, \texttt{single-step}, and \texttt{proposed} estimators, we include three existing methods for comparison, namely the transfer-learning method \texttt{Trans-GLM} \citep{tian2023transfer}, the multi-task learning method \texttt{ARMUL} \citep{duan2023adaptive}, and the shrinkage method \texttt{dShrink} \citep{wang2026divide}. Implementation and tuning details are given in Section~\ref{sec:simulation-supplement}.

\begin{figure*}[!t]
    \centering
    \includegraphics[width=1.0\textwidth]{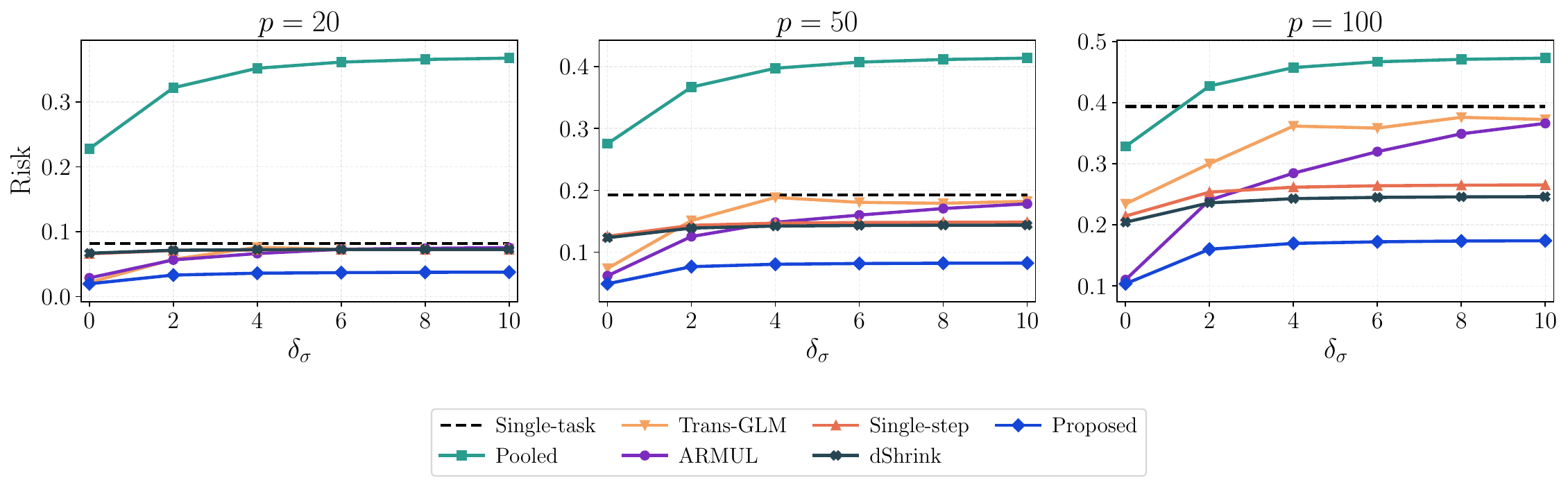}
    \caption{Multi-set simulation under covariance heterogeneity. The target has \(n_1=250\), and the sources have \(n_2=n_3=n_4=n_5=500\).}
    \label{fig:multiset-gaussian-hetero-sigma}
\end{figure*}

Figure~\ref{fig:multiset-gaussian-hetero-sigma} shows that the \texttt{proposed} estimator has the smallest risk in all configurations. At \(\delta_\sigma=10\), the \texttt{proposed} estimator reduces risk by 48.0\%, 42.7\%, and 29.3\% relative to the best competitor for \(p=20,50,100\), namely \texttt{single-step} for \(p=20\) and \texttt{dShrink} for \(p=50,100\). Across the full range of \(\delta_\sigma\), \texttt{dShrink} and \texttt{single-step} follow similar risk curves, since they both first reduce the source collection to a single aggregate and then apply shrinkage based on that aggregate. Compared to our sequential method, this aggregate-then-shrink strategy can reduce the contribution of informative sources when near-target and distant sources are mixed.

The contrast with \texttt{ARMUL} further highlights the role of covariance information. As \(\delta_\sigma\) increases from 0 to 10, the risk of \texttt{ARMUL} rises substantially, whereas the \texttt{proposed} estimator is much less sensitive to the changing covariance matrices. The deterioration of \texttt{ARMUL} indicates that adapting to parameter heterogeneity alone does not fully address large differences in covariance structure. \texttt{Trans-GLM} is also considerably less accurate in this setting. At \(\delta_\sigma=10\), the proposed estimator reduces the risk relative to \texttt{Trans-GLM} by 49.2\%, 54.8\%, and 53.3\% for the three dimensions. In this experiment, \texttt{Trans-GLM} screens each source as a whole, pools the retained sources, and then applies a target correction. It therefore does not use each source's covariance matrix to determine how much information should be borrowed along different coordinates. By contrast, our proposed covariance-aware sequential updates can borrow more from the precise coordinate block of each source while limiting the influence of its noisy block, allowing the method to combine the two complementary covariance structures more effectively.

\subsection{Real Data Analysis}\label{sec:real-data}

\begin{table}[!t]
	\centering
	\footnotesize
	\setlength{\tabcolsep}{4pt}
	\caption{Prediction error for the completion rate. RMSE and MAE are reported in percentage points, and parentheses contain Monte Carlo standard errors over 100 repetitions.}
	\label{tab:real-college-scorecard}
	\vspace{4pt}
	\begin{tabular}{lccccccc}
		\toprule
		Method & \texttt{proposed} & \texttt{single-task} & \texttt{pooled} & \texttt{Trans-GLM} & \texttt{ARMUL} & \texttt{single-step} & \texttt{dShrink}\\
		\midrule
		RMSE & \textbf{4.69 (0.12)} & 5.77 (0.21) & 8.33 (0.14) & 5.42 (0.20) & 5.39 (0.11) & 5.22 (0.17) & 5.74 (0.21)\\
		MAE  & \textbf{3.37 (0.06)} & 4.10 (0.10) & 6.32 (0.11) & 3.84 (0.09) & 4.00 (0.08) & 3.75 (0.08) & 4.09 (0.10)\\
		\bottomrule
	\end{tabular}
\end{table}

In this section, we apply our proposed method to a real-world data set, College Scorecard \citep{collegeScorecard}, which contains information on school characteristics, enrollment, completion, and student outcomes for U.S. Service Schools, treated as a separate group, and other postsecondary institutions located in the Plains, Southeast, Rocky Mountains, Southwest, and Far West regions of the United States.  The goal is to construct a model for the completion rate, that is, the proportion of students who complete within 150\% of normal program time, based on school characteristics. Such a model can reveal how completion rates are related to school characteristics and, as reported in \cite{collegeScorecardPerformance}, help institutions compare their outcomes with those of schools serving similar student populations and identify gaps in student support.

The group of U.S. Service Schools, which is the smallest one with only 100 records, is treated as the target set, while the Plains, Southeast, Rocky Mountains, Far West, and Southwest regions serve as the five source sets, each containing between 147 and 1,362 records. This is a natural multi-source setting in which a small target group may benefit from information in larger related groups. In each of 100 common repetitions, 60\% of the target records are randomly assigned to the training set and the remaining 40\% to the evaluation set. Using the target training data and source samples, we construct 20 principal components that explain 88.2\% of the variation in the training covariates. Further details on preprocessing, model fitting, and benchmark tuning are provided in Section~\ref{sec:college-scorecard-supplement}.

Table~\ref{tab:real-college-scorecard} reports the root mean squared error (RMSE) and mean absolute error (MAE) for the seven methods considered in the analysis. The \texttt{proposed} method achieves the smallest RMSE and MAE, while its strongest competitor is \texttt{single-step}, which also uses covariance-aware shrinkage but incorporates all source information through a single update rather than evaluating the sources sequentially. Relative to \texttt{single-step}, the \texttt{proposed} method reduces RMSE and MAE by 10.2\% and 10.1\%, respectively. 

\begin{figure*}[!t]
	\centering
	\begin{subfigure}[t]{0.45\textwidth}
		\centering
		\includegraphics[width=\linewidth]{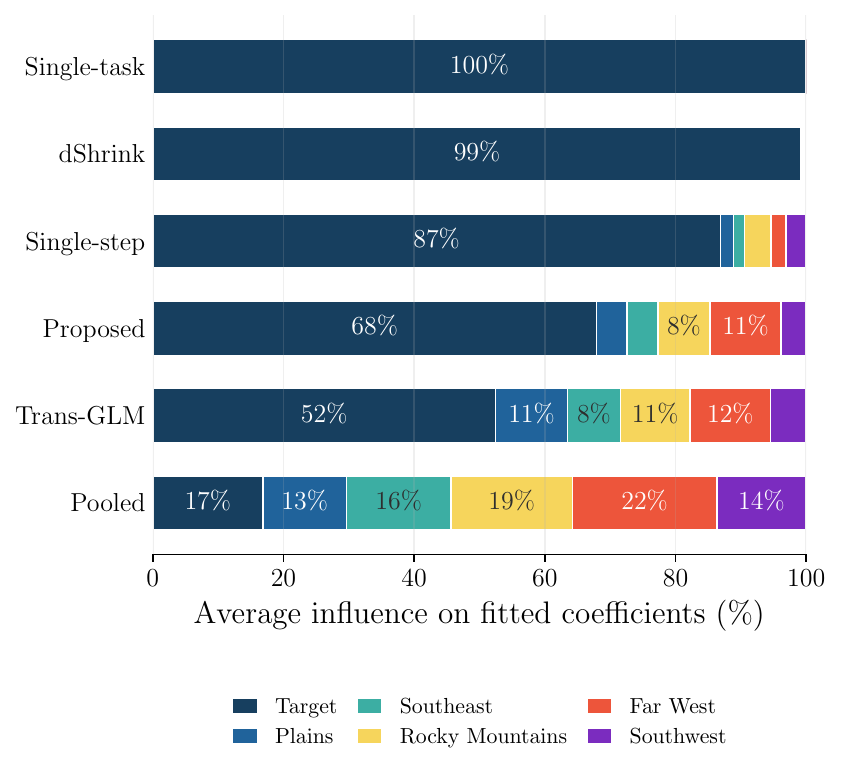}
		\caption{Contribution of each set to the fitted model}
		\label{fig:college-transfer-shares}
	\end{subfigure}\hfill
	\begin{subfigure}[t]{0.45\textwidth}
		\centering
		\includegraphics[width=\linewidth]{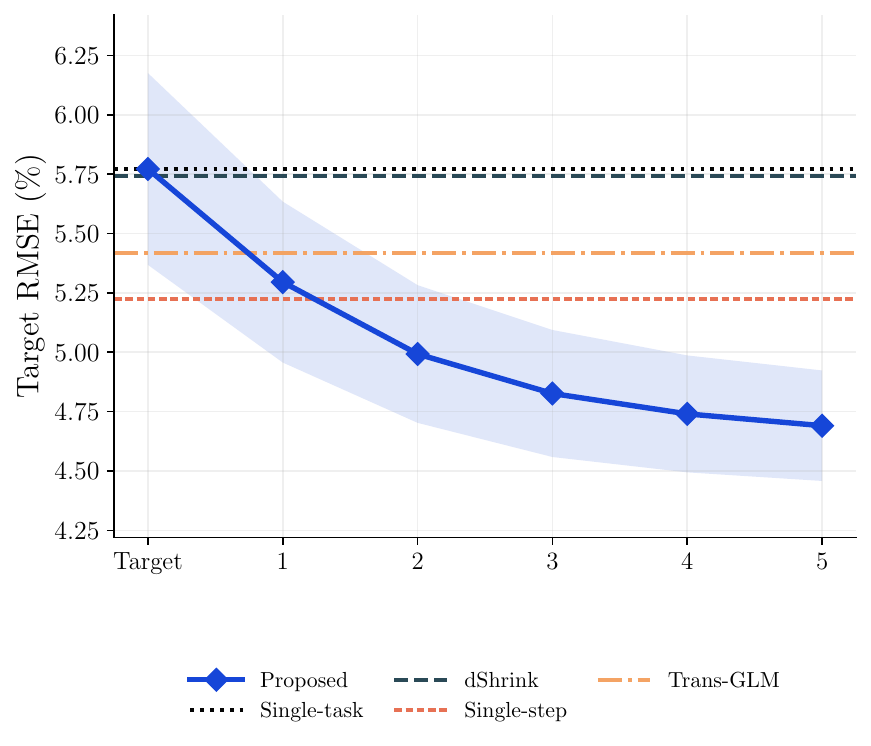}
		\caption{Target RMSE after each source update}
		\label{fig:college-sequential-path}
	\end{subfigure}
\caption{Target and source use in the College Scorecard analysis. (a) Normalized contribution shares of the target and five sources to the fitted model. (b) Target RMSE of the \texttt{proposed} method after sequential updates 1--5.} 
	\label{fig:college-transfer-story}
\end{figure*}

Figure~\ref{fig:college-transfer-story}(\subref{fig:college-transfer-shares}) summarizes the contribution of the target and each source in the final fitted model, expressed as percentages, and the formal calculations are provided in Section~\ref{sec:college-scorecard-supplement} of the supplementary materials. The \texttt{single-task} method uses no source information, and \texttt{dShrink} is close to the \texttt{single-task}  solution, indicating that its aggregate-then-shrink construction retains little source contribution in this application. The \texttt{single-step} estimator borrows more information and improves the fit, although it still compresses all sources into a single update. Our \texttt{proposed} estimator achieves the smallest RMSE with an intermediate level of source contribution distributed across all five sources, while the target remains the main contributor. In contrast, \texttt{Trans-GLM} relies more heavily on the sources but has a larger RMSE, and \texttt{pooled} depends mainly on the sources and performs worst. 

Therefore, the results do not favor either minimal or unrestricted borrowing. Our proposed sequential procedure avoids these two extremes by repeatedly reassessing source-specific shrinkage and updating both the fitted coefficients and their covariance proxies after each step, which allows useful source contributions to accumulate without displacing the target information.

In addition, Figure~\ref{fig:college-transfer-story}(\subref{fig:college-sequential-path}) traces the mean target RMSE after each update of the proposed method, averaged over 100 repetitions. The shaded region gives the pointwise 95\% Monte Carlo interval over the random splits, and the horizontal lines give the final RMSEs of the four comparison methods displayed. The mean RMSE is already below all comparison lines after the second update, and the later updates reduce it further. The gain is therefore accumulated through repeated reassessment of several sources rather than driven by one source or one update. 

\section{Conclusion}\label{sec:conc}

This paper develops a covariance-aware shrinkage framework for improving estimation efficiency with multiple related data sets. The proposed sequential algorithm compares sources and assigns different shrinkage strengths, automatically reducing the influence of heterogeneous sources. We establish oracle efficiency and strict asymptotic improvement over single-step shrinkage for the multi-set procedure under a separation condition, and extend the same construction to a broad class of smooth $M$-estimation problems. Future work may extend the framework to non-smooth losses and stricter privacy or communication constraints.



\ifseparatebib
\renewcommand{\baselinestretch}{1.45}\selectfont
\renewcommand{\refname}{References}
\putbib[main]
\end{bibunit}
\else
\bibliographystyle{apalike}
\bibliography{main}
\fi

\newpage
\appendix
\renewcommand{\baselinestretch}{1.65}\selectfont
\setlength{\abovedisplayskip}{10pt plus 2pt minus 5pt}
\setlength{\belowdisplayskip}{10pt plus 2pt minus 5pt}
\setlength{\abovedisplayshortskip}{0pt plus 2pt}
\setlength{\belowdisplayshortskip}{6pt plus 2pt minus 3pt}


\ifseparatebib
\begin{bibunit}
	\fi

\begin{center}
	\vspace*{1em}

	{\Large\scshape Supplementary Materials for}  \\[2.0em]

	{\Large\bfseries
		\begin{minipage}{0.86\textwidth}
			\centering
			Tuning-Free Covariance-Aware Sequential Shrinkage for Multi-Source Estimation
		\end{minipage}
	}

	\vspace{0.8em}
	\vspace{1.2em}
\end{center}

The supplementary materials provide additional theoretical and numerical details for the main paper. Section~\ref{sec:single-step-shrinkage} develops the direct single-step multi-source shrinkage estimator, which serves as the baseline motivating the sequential method in Section~\ref{sec:sequential} of the main text. Section~\ref{sec:relationship} gives a comparison with the classical James--Stein estimator. Section~\ref{sec:plugin-covariance-stability} provides theoretical results for the plug-in estimators with estimated covariance. Section~\ref{sec:numerical-appendix} gives details and additional results in the numerical studies. Section~\ref{sec:M-examples} verifies the assumptions for five representative smooth $M$-estimation models. Section~\ref{sec:proofs-appendix} provides the technical proofs of all theoretical results.

We use the notation \(\left[k\right]=\left\{1,\ldots,k\right\}\) for any positive integer \(k\).
For a vector \(\bv\), let \(\left\|\bv\right\|_p\) denote its $\ell_p$ norm. For a positive semidefinite matrix \(\bA\), let \(\left\|\bv\right\|_{\bA}^2=\bv^\top\bA\bv\), and let \(\left\|\bA\right\|_2\), \(\left\|\bA\right\|_{\mathrm F}\), \(\Tr\left(\bA\right)\), \(\lambda_{\min}\left(\bA\right)\), and \(\lambda_{\max}\left(\bA\right)\) denote the spectral norm, Frobenius norm, trace, smallest eigenvalue, and largest eigenvalue of $\bA$, respectively.
For symmetric matrices \(\bA\) and \(\bB\), \(\bA\preceq\bB\) if \(\bB-\bA\) is positive semidefinite, and \(\bA\prec\bB\) if \(\bB-\bA\) is positive definite. For matrices \(\bA\) and \(\bB\) of the same dimension,
let $\left\langle \bA,\bB\right\rangle=\operatorname{Tr}\left(\bA^{\top}\bB\right)$.
For sequences \(\left(a_n\right)\) and \(\left(b_n\right)\), \(a_n\lesssim b_n\) or $a_n=O\left(b_n\right)$ if \(a_n\le Cb_n\) for a universal constant \(C>0\), \(a_n\gtrsim b_n\) if \(b_n\lesssim a_n\), and \(a_n\asymp b_n\) if both \(a_n\lesssim b_n\) and \(a_n\gtrsim b_n\).
Also, \(a_n\ll b_n\) or \(a_n=o\left(b_n\right)\) if \(a_n/b_n\to0\).
For random variables \(\left(X_n\right)\) and positive numbers \(\left(a_n\right)\), \(X_n=O_{\PP}\left(a_n\right)\) if \(X_n/a_n\) is bounded in probability.
For a random variable \(X\), \(\left\|X\right\|_{\psi_1}\) and \(\left\|X\right\|_{\psi_2}\) denote the sub-exponential and sub-Gaussian Orlicz norms, respectively.

\section{Single-Step Shrinkage for Multiple Data Sets}\label{sec:single-step-shrinkage}

This section gives the detailed development of the direct multi-source analogue of the two-set shrinkage estimator. The pooled estimator in \eqref{eq:MLE-multi-set} can be rewritten as
\begin{equation*}
	\overline{\btheta}
	=
	\tbtheta_1
	+
	\sum_{j=2}^{m}
	\bW_j\left(\tbtheta_j-\tbtheta_1\right),
\end{equation*}
where the weights are
\begin{equation}\label{eq:weight-multiset}
	\bW_j
	=
	n_j
	\left(\sum_{k=1}^{m}n_k\bSigma_k^{-1}\right)^{-1}
	\bSigma_j^{-1},
	\qquad j\in\left\{2,3,\dots, m\right\}.
\end{equation}
The following estimator is referred to as the \emph{single-step shrinkage estimator}, since it performs one shrinkage update from the single-set estimator toward the pooled estimator:
\begin{equation*}
	\overline{\btheta}_{t}
	:=
	\left(1-t\right)\tbtheta_1+t\overline{\btheta}
	=
	\tbtheta_1
	+
	t\sum_{j=2}^{m}\bW_j\left(\tbtheta_j-\tbtheta_1\right),
	\qquad t\in\left[0,1\right].
\end{equation*}
Similar to \eqref{eq:MSE-in-t}, the risk of \(\overline{\btheta}_{t}\) admits the following quadratic representation:
\begin{equation}\label{eq:MSE-in-t-multiset}
	\begin{aligned}
		\E \normQ{\overline{\btheta}_{t}-\btheta_1^\star}^2
		=\ip{n_1^{-1}\bSigma_1}{\bQ}
		-
		2t
		\sum_{j=2}^{m}
		\ip{n_1^{-1}\bW_j\bSigma_1}{\bQ}+
		t^2
		\E\left[
		\left\|
		\sum_{j=2}^{m}
		\bW_j\left(\tbtheta_j-\tbtheta_1\right)
		\right\|_{\bQ}^2
		\right].
	\end{aligned}
\end{equation}
Therefore, the formal minimizer of the right-hand side suggests the plug-in estimator
\begin{equation}\label{eq:shrinkage-tstar-multiset}
	\overline{\btheta}_{\widehat t}
	=
	\tbtheta_1
	+
	\min\left(
	1,
	\frac{
		\sum_{j=2}^{m}
		\ip{n_1^{-1}\bW_j\bSigma_1}{\bQ}
	}{
		\left\|
		\sum_{j=2}^{m}
		\bW_j\left(\tbtheta_j-\tbtheta_1\right)
		\right\|_{\bQ}^2
	}
	\right)
	\sum_{j=2}^{m}
	\bW_j\left(\tbtheta_j-\tbtheta_1\right).
\end{equation}

The estimator in \eqref{eq:shrinkage-tstar-multiset} is obtained by plugging a data-dependent \(\widehat t\) into a risk identity that is valid only for fixed \(t\). Thus, although \eqref{eq:shrinkage-tstar-multiset} is a natural SURE-type choice, its risk is not directly controlled by \eqref{eq:MSE-in-t-multiset}. For this reason, we again introduce the refined family indexed by \(s\):
\begin{equation}\label{eq:shrinkage-estimator-s-multi-set}
	\widehat{\btheta}_{s}
	=
	\tbtheta_1
	+
	\frac{s}{
		\left\|
		\sum_{j=2}^{m}
		\bW_j\left(\tbtheta_j-\tbtheta_1\right)
		\right\|_{\bQ}^2\vee s
	}
	\sum_{j=2}^{m}
	\bW_j\left(\tbtheta_j-\tbtheta_1\right),
	\qquad s\ge0.
\end{equation}
The plug-in estimator in \eqref{eq:shrinkage-tstar-multiset} corresponds to the particular choice
$
\bar s
:=
\sum_{j=2}^{m}
\ip{n_1^{-1}\bW_j\bSigma_1}{\bQ}
$.

The next theorem provides an upper bound for the risk of \eqref{eq:shrinkage-estimator-s-multi-set}.

\begin{theorem}\label{thm:risk-bound-s-multiset}
	Define
	$
	\bS
	:=
	\sum_{j=2}^{m}
	n_1^{-1}
	\bQ^{1/2}
	\bW_j
	\bSigma_1
	\bQ^{1/2}
	$.
If \(\Tr\left(\bS\right)>2\left\|\bS\right\|_2\), then for any
	$
	s\in\left(0,\;2\Tr\left(\bS\right)-4\left\|\bS\right\|_2\right)
$,
	we have
	\begin{equation}\label{eq:risk-upper-multiset}
		\begin{aligned}
			&\E\normQ{\widehat{\btheta}_{s}-\btheta_1^\star}^2
			-
			\ip{n_1^{-1}\bSigma_1}{\bQ}
			\le
			-s\left[2\Tr\left(\bS\right)-4\left\|\bS\right\|_2-s\right]
			\E\left[
			\frac{1}{
			\left\|
			\sum_{j=2}^{m}
			\bW_j\left(\tbtheta_j-\tbtheta_1\right)
			\right\|_{\bQ}^{2}\vee s}
			\right],
		\end{aligned}
	\end{equation}
	which implies
	$
	\E\normQ{\widehat{\btheta}_{s}-\btheta_1^\star}^2
	<
	\E\normQ{\tbtheta_1-\btheta_1^\star}^2=	\ip{n_1^{-1}\bSigma_1}{\bQ}
	$.
\end{theorem}

Theorem~\ref{thm:risk-bound-s-multiset} shows that the single-step shrinkage estimator achieves uniform risk improvement over the single-set estimator \(\tbtheta_1\) throughout an explicit admissible interval of \(s\). Let
$
\underline s=
\Tr\left(\bS\right)-2\left\|\bS\right\|_2$.
The value \(\underline s\) maximizes the quadratic risk-reduction factor in \eqref{eq:risk-upper-multiset}, whereas \(\bar s=\Tr\left(\bS\right)
\) corresponds to the constrained SURE-type estimator in \eqref{eq:shrinkage-tstar-multiset}. The interval
$
s\in\left[\underline s,\bar s\right]
$
therefore provides an admissible shrinkage range in the multi-source setting when $
\Tr\left(\bS\right)>4\left\|\bS\right\|_2
$. The left endpoint \(\underline s\) is the most conservative choice suggested by the upper bound in \eqref{eq:risk-upper-multiset}, while the right endpoint recovers the plug-in choice induced by \(\widehat t\).

The next theorem provides a common risk upper bound for all \(s\) in this interval.

\begin{theorem}\label{thm:risk-bound-shat-multiset}
	Under the assumptions in Theorem~\ref{thm:risk-bound-s-multiset}, if
	$
	\Tr\left(\bS\right)>4\left\|\bS\right\|_2
	$,
	then for any
	$
	s\in\left[\underline s,\bar s\right]
	$,
	we have
	\begin{equation}\label{eq:risk-bound-shat-multiset}
		\begin{aligned}
			\E\normQ{\widehat{\btheta}_{s}-\btheta_1^\star}^2
			&\le
			\ip{
				\left(
				\sum_{j=1}^{m}
				n_j\bSigma_j^{-1}
				\right)^{-1}
			}{\bQ}
+
			\frac{
				\Tr\left(\bS\right)
				\left\|
				\sum_{j=2}^{m}
				\bW_j\left(\btheta_j^\star-\btheta_1^\star\right)
				\right\|_{\bQ}^2
			}{
				\Tr\left(\bS\right)
				+
				\left\|
				\sum_{j=2}^{m}
				\bW_j\left(\btheta_j^\star-\btheta_1^\star\right)
				\right\|_{\bQ}^2
			}
			+
			4\left\|\bS\right\|_2.
		\end{aligned}
	\end{equation}
\end{theorem}

The upper bound in Theorem~\ref{thm:risk-bound-shat-multiset} contains three terms. The first term is the oracle risk that would be attained if all data sets shared the same population mean. The second term quantifies the price of heterogeneity between the target and the sources. It depends on the weighted population discrepancy
$
\sum_{j=2}^{m}\bW_j\left(\btheta_j^\star-\btheta_1^\star\right)
$,
which is the population version of the pooled shrinkage direction. The third term is the price of using a data-adaptive shrinkage magnitude and depends only on the spectral size of \(\bS\).

To better understand the bound, consider the illustrative setting where \(\bQ=\bSigma_j=\bI_p\) for all \(j\in\left[m\right]\), and let \(\kappa_j:=n_j/n_1\) for \(j\in\left\{2,3,\dots, m\right\}\), with \(\kappa:=\sum_{j=2}^{m}\kappa_j\). Then
the first term in \eqref{eq:risk-bound-shat-multiset} reduces to
$
\frac{p}{n_1\left(1+\kappa\right)}
$,
which is the variance reduction obtained by pooling all data sets under homogeneity. The third term becomes
$
\frac{4\kappa}{n_1\left(1+\kappa\right)}
$,
which is of order \(n_1^{-1}\) and does not increase with the dimension \(p\). The second term is governed by
$
\left\|
\sum_{j=2}^{m}
\frac{\kappa_j}{1+\kappa}
\left(\btheta_j^\star-\btheta_1^\star\right)
\right\|_2^2
$.
Thus, if all source means are close to the target mean, the single-step shrinkage estimator can achieve substantial variance reduction while keeping the heterogeneity cost small.

However, the source data sets can be a mixture of homogeneous sources, whose population parameters are close to the target parameter \(\btheta_1^\star\), and heterogeneous sources, whose population parameters are substantially different from \(\btheta_1^\star\). The former can provide transferable information for estimating \(\btheta_1^\star\), whereas the latter may introduce bias. In this case, the weighted discrepancy can be large due to the heterogeneous sources. As a result, the single-step shrinkage estimator may choose a small shrinkage size and fail to fully use the transferable information from sources homogeneous with the target. This limitation motivates the sequential shrinkage estimator, which incorporates sources one at a time and updates the shrinkage direction after each source-specific~{update}.

\section{Relationship to the James--Stein Estimator}\label{sec:relationship}

In this section, we discuss the relationship between our shrinkage estimator and the classical James--Stein (JS) estimator \citep{james1961estimation}, which was proposed under a homoscedastic setting where $\tbtheta_1 \sim \calN\left(\btheta_1^{\star}, n_1^{-1}\sigma_1^2\bI_{p}\right)$ and $\bQ=\bI_p$. The JS estimator shrinks $\tbtheta_1$ toward the origin and takes the form
\[
\widehat{\btheta}_{\mathrm{JS}, s}=\tbtheta_1-\frac{s\tbtheta_1 }{\tbtheta_1^{\top}\tbtheta_1}, \qquad s>0.
\]
It is well known that the JS estimator improves upon the usual MLE when the dimension $p\geq3$ by shrinking individual estimates toward the origin, leading to a variance reduction that exceeds the additional bias and therefore lowers the overall risk. In contrast, our estimator shrinks $\tbtheta_1$
toward a direction informed by the source data set.

The risk of $\widehat{\btheta}_{\mathrm{JS}, s}$ follows directly from Lemma \ref{lem:sure}:
\[
\E\norm{\widehat{\btheta}_{\mathrm{JS}, s}-\btheta_1^{\star}}_2^2 = pn_1^{-1}\sigma_1^2+ s\left[-2n_1^{-1}\sigma_1^2\left(p-2\right)+s\right]\E\left[\left\|\tbtheta_1\right\|_2^{-2}\right],
\]
which is minimized when $\underline{s}_{\mathrm{JS}}=n_1^{-1}\sigma_1^{2}\left(p-2\right)$. Substituting this choice yields
\begin{equation}\label{eq:bound-JS}
	\E\norm{\widehat{\btheta}_{\mathrm{JS}, \underline{s}_{\mathrm{JS}}}-\btheta_1^{\star}}_2^2-pn_1^{-1}\sigma_1^2 = - n_1^{-2}\sigma_1^4\left(p-2\right)^2\E\left[\left\|\tbtheta_1\right\|_2^{-2}\right]\leq \frac{-n_1^{-2}\sigma_1^4\left(p-2\right)^2}{p\sigma_1^{2}n_1^{-1}+\norm{\btheta_1^{\star}}_2^2}.
\end{equation}
Under the same homoscedastic setting, i.e., $\tbtheta_2 \sim \calN\left(\btheta_2^{\star}, n_2^{-1}\sigma_2^2\bI_{p}\right)$, the pooled-direction variable in the decomposition \eqref{eq:positive-part-decomposition-detailed} has a spherical covariance matrix. Hence \eqref{eq:shrinkage-sopt} is the positive-part version of the corresponding untruncated Green--Strawderman rule. The positive-part rule dominates its untruncated counterpart under squared Euclidean loss \citep{nickerson1988dominance}. Therefore the sharper untruncated upper bound remains valid for our truncated estimator:
\begin{equation}\label{eq:bound-shrinkage-simple}
	\E\norm{\widehat{\btheta}_{\underline{s}}-\btheta_1^{\star}}_2^2 - pn_1^{-1}\sigma_1^2 \leq \frac{-n_1^{-2}w\sigma_1^4\left(p-2\right)^2}{p\sigma_1^{2}n_1^{-1}+w\norm{\btheta_1^{\star}-\btheta_2^{\star}}_2^2}, \text{ where } w=\frac{n_2\sigma_2^{-2}}{n_1\sigma_1^{-2}+n_2\sigma_2^{-2}}.
\end{equation}
When $\norm{\btheta_1^{\star}-\btheta_2^{\star}}_2 < \norm{\btheta_1^{\star}}_2$ and $n_2 > \frac{p\sigma_2^2}{\norm{\btheta_1^{\star}}_2^2-\norm{\btheta_1^{\star}-\btheta_2^{\star}}_2^2}$, the upper bound in  \eqref{eq:bound-shrinkage-simple} is strictly smaller than the bound  in \eqref{eq:bound-JS}, which indicates that our proposed estimator is likely to outperform the JS estimator when the two population means are sufficiently close and the source sample size is sufficiently large.  The next theorem strengthens this comparison by establishing a sufficient condition under which $\widehat{\btheta}_{\underline{s}}$ is guaranteed to achieve a lower risk than $\widehat{\btheta}_{\mathrm{JS}, s}$ for every $s>0$.

\begin{theorem}\label{thm:JS}
	Suppose that $p\geq 3$ and $\left\|\btheta_2^\star-\btheta_1^\star\right\|_2^2 < \left(1-\frac{2}{p}\right)\left\|\btheta_1^\star\right\|_2^2$. If the sample sizes satisfy
	$n_1
	> \frac{2p\sigma_1^2}{\left(p-2\right)\left\|\btheta_1^\star\right\|_2^2
		- p\left\|\btheta_2^\star-\btheta_1^\star\right\|_2^2}$ and  $n_2 >
	\frac{
		p^2\sigma_2^2
	}{
		\left(p-2\right)\left\|\btheta_1^\star\right\|_2^2
		- p\left\|\btheta_2^\star-\btheta_1^\star\right\|_2^2
		- 2p\sigma_1^2/n_1
	}$,
	then, for all $s>0$, we have
	\[\E\left\|\widehat{\btheta}_{\underline{s}}-\btheta_1^{\star}\right\|_2^2 < \E\left\|\widehat{\btheta}_{\mathrm{JS}, s}-\btheta_1^{\star}\right\|_2^2.\]
\end{theorem}

A generalized version of James--Stein's estimator for non-diagonal covariance $\bSigma_1$ was given in \cite{bock1975minimax}:
\[
\widehat{\btheta}_{\mathrm{JS}, s}=\tbtheta_1-\frac{s\tbtheta_1 }{n_1\tbtheta_1^{\top}\bSigma_1^{-1}\tbtheta_1},
\]
where $s>0$ and $\tbtheta_1 \sim \mathcal{N}\left(\btheta_1^{\star}, n_1^{-1}\bSigma_1\right)$. By Lemma \ref{lem:sure}, the risk of $\widehat{\btheta}_{\mathrm{JS}, s}$ with $\bQ=\bI_{p}$ becomes
\[
\E\norm{\widehat{\btheta}_{\mathrm{JS}, s}-\btheta_1^{\star}}_2^2 \leq n_1^{-1} \Tr\left(\bSigma_1\right)+ sn_1^{-1}\left[\left(s+4\right)\norm{\bSigma_1}_2-2\Tr\left(\bSigma_1\right)\right]\E\left[\left(n_1\tbtheta_1^{\top}\bSigma_1^{-1}\tbtheta_1\right)^{-1}\right].
\]
When $\Tr\left(\bSigma_1\right) > 2\left\|\bSigma_1\right\|_2$, the upper bound is minimized at $\underline{s}_{\mathrm{JS}}=\left(\Tr\left(\bSigma_1\right)-2\left\|\bSigma_1\right\|_2\right)/\left\|\bSigma_1\right\|_2$. Theorem~\ref{thm:risk-bound-shat} also gives a comparison for a general non-diagonal covariance matrix. A sufficient condition for $\widehat{\btheta}_{\underline{s}}$ to have smaller risk than $\widehat{\btheta}_{\mathrm{JS},s}$ for every $s>0$ is
\begin{equation*}
\frac{\Tr\left(\bS\right)^2}{\Tr\left(\bS\right)+\left\|\bW_2\left(\btheta_2^{\star}-\btheta_1^{\star}\right)\right\|_2^2}-4\left\|\bS\right\|_2
>
\frac{\left(\operatorname{Tr}\left(\bm{\Sigma}_{1}\right)-2\,\lambda_{\min}\left(\bm{\Sigma}_{1}\right)\right)^{2}}{n_1\lambda_{\min}\left(\bm{\Sigma}_{1}\right)}
\frac{p}{p-2}\frac{1}{p+n_1\btheta_1^{\star\top}\bSigma_1^{-1}\btheta_1^{\star}}.
\end{equation*}

\section{Plug-in Covariance Estimation and Risk Stability}
\label{sec:plugin-covariance-stability}

The risk bounds in Theorems~\ref{thm:risk-bound-s} and \ref{thm:risk-bound-shat} are stated for the Gaussian mean model with known covariance matrices, whereas the implementation uses estimated covariance matrices. This section analyzes the fully plug-in estimator, including both the estimated shrinkage direction and size. The argument is deterministic conditional on the covariance estimation error and therefore does not require the covariance estimators to be independent of the local estimators.

Throughout this section, the eigenvalues of the population covariance matrices are assumed to be bounded away from zero and infinity. We use positive definite covariance estimators whose eigenvalues obey the same fixed bounds. This condition can always be enforced by clipping the eigenvalues to a fixed interval containing the population spectra. Under consistent covariance estimation, the clipping is inactive with probability tending to one and has no first-order effect.

Given covariance estimators \(\widehat\bSigma_1\) and \(\widehat\bSigma_2\), define
$
\widehat\bW_2
=
n_2\left(n_1\widehat\bSigma_1^{-1}+n_2\widehat\bSigma_2^{-1}\right)^{-1}\widehat\bSigma_2^{-1}$ and
$\widehat\bS
=
n_1^{-1}\bQ^{1/2}\widehat\bW_2\widehat\bSigma_1\bQ^{1/2}$.
Let
$
s_0:=\Tr\left(\bS\right)-\left\|\bS\right\|_2$ and $\widehat s:=\Tr\left(\widehat\bS\right)-\left\|\widehat\bS\right\|_2$. The following proposition gives deterministic and mean-square perturbation bounds.

\begin{prop}\label{prop:plugin-stability}
Suppose that \(n_1\asymp n_2\asymp n\), and that the eigenvalues of \(\bSigma_j\) and \(\widehat\bSigma_j\), \(j=1,2\), are uniformly bounded away from zero and infinity. Let
$
\eta_n:=\max_{j=1,2}\left\|\widehat\bSigma_j-\bSigma_j\right\|_2
$.
There exists a constant \(C>0\), depending only on the eigenvalue bounds and \(\bQ\), such that, whenever \(\eta_n\) is sufficiently small,
\[
\left\|\widehat\bW_2-\bW_2\right\|_2\le C\eta_n,
\qquad
\left\|\widehat\bS-\bS\right\|_2\le C\frac{\eta_n}{n},
\qquad
\left|\widehat s-s_0\right|\le C\frac{p\eta_n}{n}.
\]
Consequently, if
$
\eta_{n,2}:=
\left[\E\left[\max_{j=1,2}\left\|\widehat\bSigma_j-\bSigma_j\right\|_2^2\right]\right]^{1/2}=o\left(1\right)
$,
then
\[
\left[\E\left\|\widehat\bW_2-\bW_2\right\|_2^2\right]^{1/2}\le C\eta_{n,2},
\qquad
\left[\E\left\|\widehat\bS-\bS\right\|_2^2\right]^{1/2}\le C\frac{\eta_{n,2}}{n},
\qquad
\left[\E\left|\widehat s-s_0\right|^2\right]^{1/2}\le C\frac{p\eta_{n,2}}{n}.
\]
\end{prop}

Define the fully plug-in Gaussian estimator
\begin{equation*}
\widehat\btheta_{\rm pl}
=
\tbtheta_1+
\frac{\widehat s}
{\left\|\widehat\bW_2\left(\tbtheta_2-\tbtheta_1\right)\right\|_{\bQ}^2\vee\widehat s}
\widehat\bW_2\left(\tbtheta_2-\tbtheta_1\right).
\end{equation*}

For \(s>0\), define $\mathcal T_s\left(\bu\right) = \frac{s}{\left\|\bu\right\|_{\bQ}^2\vee s}\bu, \quad \mathcal T_0\left(\bu\right)=\bm{\zero}$.

\begin{theorem}[Risk with estimated covariance matrices]\label{thm:plugin-risk}
Suppose the conditions of Proposition~\ref{prop:plugin-stability} hold, \(\eta_{n,2}=o\left(1\right)\), and \(\Tr\left(\bS\right)>4\left\|\bS\right\|_2\). Then
\[
\left[
\E\left\|
\widehat\btheta_{\rm pl}
-
\left[\tbtheta_1+\mathcal T_{s_0}\left(\bW_2\left(\tbtheta_2-\tbtheta_1\right)\right)\right]
\right\|_{\bQ}^2
\right]^{1/2}
\le C\sqrt{\frac pn}\,\eta_{n,2}.
\]
The corresponding risk satisfies
\[
\begin{aligned}
\E\left\|\widehat\btheta_{\rm pl}-\btheta_1^\star\right\|_{\bQ}^2
&\le
\ip{\left(n_1\bSigma_1^{-1}+n_2\bSigma_2^{-1}\right)^{-1}}{\bQ}
+
\frac{\Tr\left(\bS\right)\,\left\|\bW_2\left(\btheta_2^\star-\btheta_1^\star\right)\right\|_{\bQ}^2}
{\Tr\left(\bS\right)+\left\|\bW_2\left(\btheta_2^\star-\btheta_1^\star\right)\right\|_{\bQ}^2}
+4\left\|\bS\right\|_2
+C\frac pn\eta_{n,2}.
\end{aligned}
\]
The bound is uniform over the source discrepancy \(\btheta_2^\star-\btheta_1^\star\), and the fully plug-in estimator has risk \(O\left(p/n\right)\) for an arbitrary value of this discrepancy. If
\(\left\|\bW_2\left(\btheta_2^\star-\btheta_1^\star\right)\right\|_{\bQ}^2\le C_0\Tr\left(\bS\right)\) for a fixed \(C_0\) and \(\Tr\left(\bS\right)/\left\|\bS\right\|_2\to\infty\), then it has strictly smaller risk than \(\tbtheta_1\) for all sufficiently large \(n\).
\end{theorem}

The same stability argument applies along the entire sequential path. Let \(\widehat\btheta_{\rm seq,pl}\) denote Algorithm~\ref{alg:sequential-shrinkage} with every \(\bSigma_j\) replaced by \(\widehat\bSigma_j\), including the initialization and the recursion for the working covariance proxy.

\begin{theorem}[Plug-in sequential shrinkage]\label{thm:plugin-sequential}
Suppose the population covariance matrices satisfy Assumption~\ref{assump:seq-2}, the eigenvalues of \(\bSigma_j\) and \(\widehat\bSigma_j\), \(j\in\left[m\right]\), are uniformly bounded away from zero and infinity, and
$
\eta_{n,2}:=
\left[\E\left[\max_{j\in\left[m\right]}\left\|\widehat\bSigma_j-\bSigma_j\right\|_2^2\right]\right]^{1/2}=o\left(1\right)
$.
Under the remaining conditions of Theorem~\ref{thm:seq-large-separation-trace-condition},
\[
\PP\left(\left\{\widehat j_1,\ldots,\widehat j_{m_0}\right\}=\mathcal I_0\right)
\ge1-\rho,
\quad \text{and \quad}
\E\left\|\widehat\btheta_{\rm seq,pl}-\btheta_1^\star\right\|_{\bQ}^2
=
\ip{\bV_{\mathcal I_0}}{\bQ}+o\left(\frac pn\right).
\]
The strict comparison with the plug-in single-step estimator also holds under the pooled-discrepancy condition in part 3 of Theorem~\ref{thm:seq-large-separation-trace-condition}, provided additionally that
$
\delta_{\mathcal I} \gg \eta_{n,2}\delta_{\max}
$.
\end{theorem}

For smooth \(M\)-estimators, Lemmas~\ref{lem:consistency-HhatVhat} and~\ref{lem:deterministic-clipping} yield
\[
\left[\EE\max_{j\in\left[m\right]}\left\|\widehat\bSigma_j^{\rm c}-\bSigma_j\right\|_2^2\right]^{1/2}
\lesssim
\sqrt{\frac{p\log\left(n\right)}{n}},
\]
under Assumptions~\ref{assump:1}--\ref{assump:4}. The proof of Theorem~\ref{thm:risk-upper-general} uses this rate together with the fixed clipping bounds to control the complement of the covariance-consistency event. Given the covariance estimators, every weight, score, and shrinkage size in the sequential algorithm is determined analytically without cross-validation, sample splitting, a testing threshold, or any additional transfer-tuning parameter.

\section{Additional Details and Results in Numerical Studies}\label{sec:numerical-appendix}

All numerical studies use the fixed value $\varepsilon^{\rm c}=10^{-8}$ in \eqref{eq:clipped-sandwich}.

\subsection{Simulation}\label{sec:simulation-supplement}

\paragraph{Two-set settings.} We first provide details for the linear and logistic regression settings.

\emph{Linear regression.} For the linear regression experiment, we generate covariates
\(\bx_{ji}\sim \calN\left(\mathbf{0},\sigma_j^2\bR\left(\rho\right)\right)\) and responses
$
y_{ji}=\bx_{ji}^\top\btheta_j^\star+\varepsilon_{ji}$, $\varepsilon_{ji}\sim \calN\left(0,\sigma_\varepsilon^2\right)
$,
independently. We set
\(\left(\rho,\sigma_1,\sigma_2,\sigma_\varepsilon\right)=\left(0.25,1.0,1.0,2\right)\)
and \(\left(n_1,n_2\right)=\left(500,500\right)\). The loss function is
$
\ell\left(\btheta,\left(\bx,y\right)\right)
=
\frac{1}{2}\left(y-\bx^{\top}\btheta\right)^2
$.
The single-set estimator \(\widetilde{\btheta}_j\) is the ordinary least squares estimator, and the estimated covariance  is computed using the clipped sandwich form described in Section~\ref{sec:general-M-estimation} in the main text.

\emph{Logistic regression.} For the logistic regression experiment, we generate covariates
\(\bx_{ji}\sim \calN\left(\mathbf{0},\sigma_j^2\bR\left(\rho\right)\right)\). Conditional on \(\bx_{ji}\), the binary response satisfies
$
y_{ji}\mid \bx_{ji}\sim \mathrm{Bernoulli}\left(\pi_{ji}\right)$,
$\pi_{ji}
=
\left[1+\exp\left(-\bx_{ji}^\top\btheta_j^\star\right)\right]^{-1}.
$
We set
\(\left(\rho,\sigma_1,\sigma_2\right)=\left(0.25,1.0,0.25\right)\)
and \(\left(n_1,n_2\right)=\left(2500,20000\right)\). The loss function is the logistic negative log-likelihood
$
\ell\left(\btheta,\left(\bx,y\right)\right)
=
\log\!\left(1+\exp\left(\bx^{\top}\btheta\right)\right)-y\,\bx^{\top}\btheta
$.
The single-set \(M\)-estimator \(\widetilde{\btheta}_j\) is computed by Newton's method, and \(\widehat\bSigma_j\) is formed using the same clipped construction.

We further note that the Gaussian mean and logistic experiments instantiate the theoretical models, while standard Gaussian-design linear regression demonstrates that the proposed method remains effective beyond the sufficient sub-Gaussian score condition in Assumption~\ref{assump:2}.

\begin{figure*}[!t]
    \centering
    \includegraphics[width=0.95\textwidth]{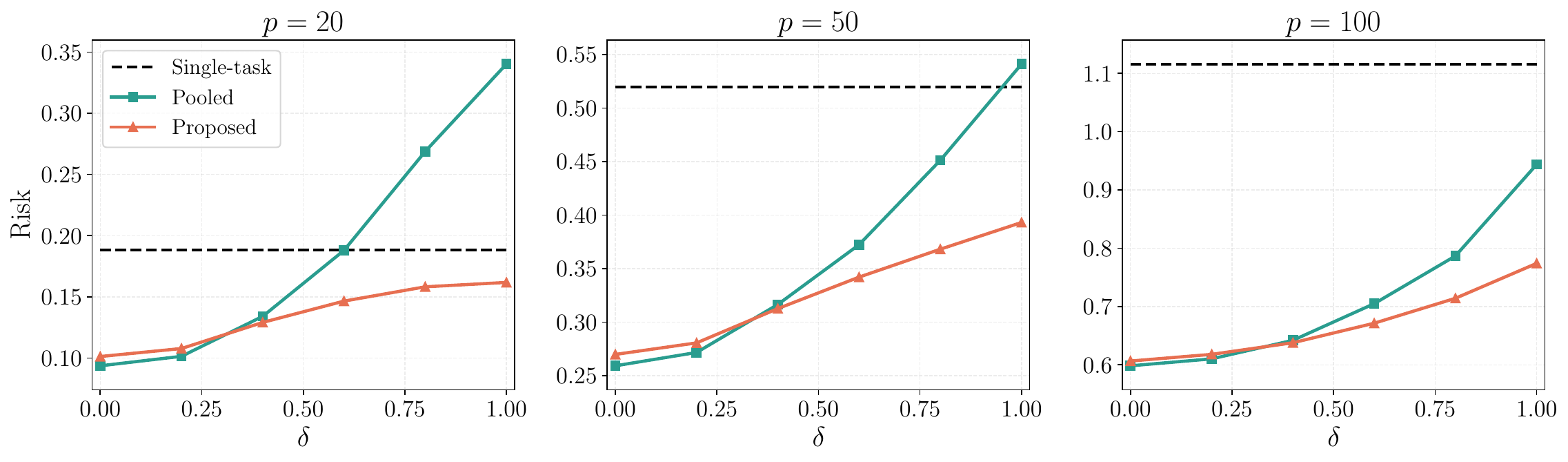}
    \caption{Two-set linear regression results for \(p\in\left\{20,50,100\right\}\), with \(n_1=n_2=500\). The figure compares the \texttt{single-task}, \texttt{pooled}, and \texttt{proposed} estimators.}
    \label{fig:two-set-linear}
\end{figure*}

\begin{figure*}[!t]
    \centering
    \includegraphics[width=0.93\textwidth]{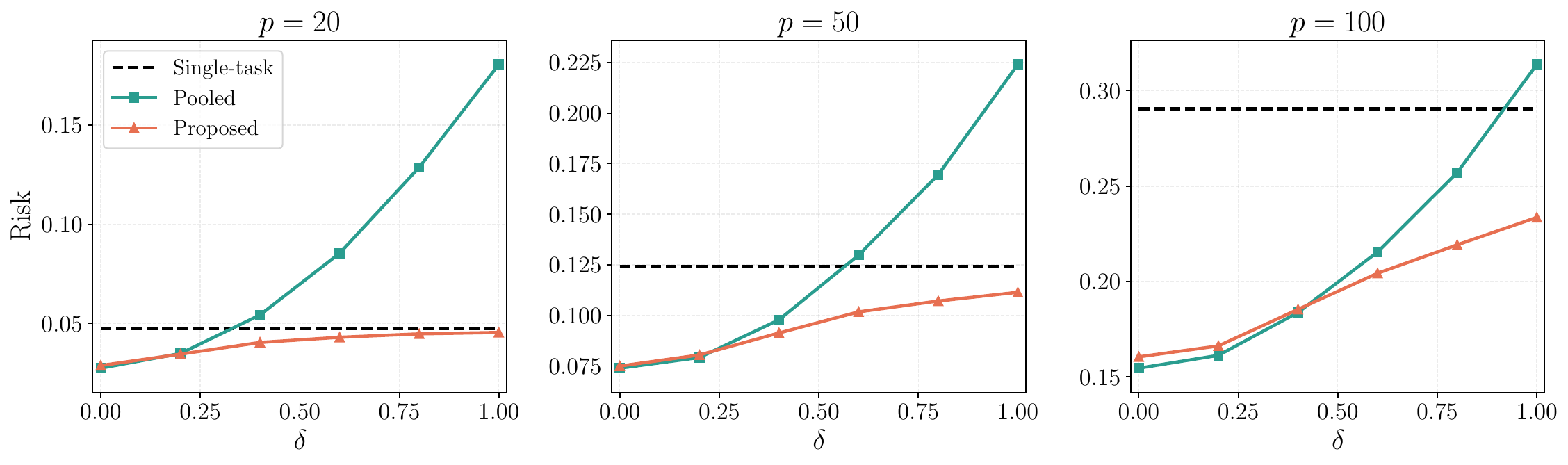}
    \caption{Two-set logistic regression results for \(p\in\left\{20,50,100\right\}\), with \(n_1=2500\) and \(n_2=20000\). The figure compares the \texttt{single-task}, \texttt{pooled}, and \texttt{proposed} estimators.}
    \label{fig:two-set-logistic}
\end{figure*}

Figures~\ref{fig:two-set-linear} and \ref{fig:two-set-logistic} complement Figure~\ref{fig:two-set-gaussian}. Near \(\delta=0\), the \texttt{proposed} estimator closely follows the variance reduction obtained by \texttt{pooled}. As \(\delta\) increases, its curve bends toward \texttt{single-task} rather than following the rapidly increasing pooling bias, and at \(\delta=1\) it improves on the stronger competitor by 14.1\%, 24.3\%, and 18.0\% in linear regression and by 3.7\%, 10.5\%, and 19.6\% in logistic regression for \(p=20,50,100\).

\paragraph{Multi-set settings.}

For the multi-set experiments with homogeneous covariance, we draw \(\bv_2,\ldots,\bv_5\) independently from the unit sphere, and set
\[
\begin{aligned}
\btheta_1^\star&=\mathbf 1_p/\sqrt p,\\
\btheta_2^\star&=\btheta_1^\star+\left(a/4\right)\sqrt{p/n_2}\bv_2,
&\btheta_3^\star&=\btheta_1^\star+\left(a/4\right)\sqrt{p/n_3}\bv_3,\\
\btheta_4^\star&=\btheta_1^\star+a\delta\bv_4,
&\btheta_5^\star&=\btheta_1^\star+a\delta\bv_5,
\end{aligned}
\]
where \(a=1\) for Gaussian mean estimation and linear regression and \(a=1/2\) for logistic regression. In the Gaussian experiment, \(\bx_{ji}\sim\mathcal N\left(\btheta_j^\star,\bR\left(0.5\right)\right)\) and \(n_1=n_2=n_3=n_4=n_5=400\). In the linear experiment, \(\bx_{ji}\sim\mathcal N\left(\mathbf 0,\bR\left(0.5\right)\right)\), \(y_{ji}=\bx_{ji}^{\top}\btheta_j^\star+\varepsilon_{ji}\), and \(\varepsilon_{ji}\sim\mathcal N\left(0,1\right)\), with \(n_1=1000\) and \(n_2=n_3=n_4=n_5=2000\). In the logistic experiment, \(\bx_{ji}\sim\mathcal N\left(\mathbf 0,\bR\left(0.5\right)\right)\) and the conditional response is Bernoulli with logistic probability, with \(n_1=20000\) and \(n_2=n_3=n_4=n_5=4000\). The directions remain fixed over \(\delta\in\left\{0,0.5,\ldots,2.5\right\}\) within a replication, all observations and noise variables are independent, and every method uses \(\bQ=\bI_p\).

\begin{table}[p]
	\centering
	\scriptsize
	\setlength{\tabcolsep}{9pt}
	\renewcommand{\arraystretch}{0.84}
	\caption{Risk at the common reference level in each simulation. The first six blocks report the results with $\delta=1$, while the heterogeneous covariance block reports results with $\delta_\sigma=10$. Each entry is the Monte Carlo mean with its standard error in parentheses over 100 replications.}
	\label{tab:simulation-reference-risks}
	\vspace{-7pt}
	\begin{tabular}{lccc}
		\toprule
		Method & $p=20$ & $p=50$ & $p=100$\\
		\midrule
		\multicolumn{4}{l}{\textit{Two-set Gaussian mean}}\\
		\texttt{single-task} & 0.1839 (0.0056) & 0.4354 (0.0100) & 0.9058 (0.0146)\\
		\texttt{pooled} & 0.5348 (0.0088) & 0.6226 (0.0092) & 0.8256 (0.0103)\\
		\texttt{proposed} & \textbf{0.1644 (0.0050)} & \textbf{0.3321 (0.0074)} & \textbf{0.5962 (0.0091)}\\
		\addlinespace[2pt]
		\multicolumn{4}{l}{\textit{Two-set linear regression}}\\
		\texttt{single-task} & 0.1884 (0.0063) & 0.5194 (0.0104) & 1.1160 (0.0194)\\
		\texttt{pooled} & 0.3402 (0.0089) & 0.5411 (0.0097) & 0.9437 (0.0151)\\
		\texttt{proposed} & \textbf{0.1618 (0.0053)} & \textbf{0.3931 (0.0079)} & \textbf{0.7740 (0.0134)}\\
		\addlinespace[2pt]
		\multicolumn{4}{l}{\textit{Two-set logistic regression}}\\
		\texttt{single-task} & 0.0473 (0.0015) & 0.1244 (0.0024) & 0.2906 (0.0051)\\
		\texttt{pooled} & 0.1805 (0.0032) & 0.2240 (0.0038) & 0.3138 (0.0038)\\
		\texttt{proposed} & \textbf{0.0455 (0.0015)} & \textbf{0.1114 (0.0021)} & \textbf{0.2337 (0.0037)}\\
		\addlinespace[2pt]
		\multicolumn{4}{l}{\textit{Multi-set Gaussian mean}}\\
		\texttt{single-task} & 0.0518 (0.0020) & 0.1198 (0.0027) & 0.2484 (0.0044)\\
		\texttt{pooled} & 0.0950 (0.0023) & 0.1193 (0.0025) & 0.1833 (0.0024)\\
		\texttt{single-step} & 0.0411 (0.0014) & 0.0730 (0.0016) & 0.1337 (0.0021)\\
		\texttt{proposed} & \textbf{0.0200 (0.0008)} & \textbf{0.0448 (0.0011)} & \textbf{0.0984 (0.0018)}\\
		\addlinespace[2pt]
		\multicolumn{4}{l}{\textit{Multi-set linear regression}}\\
		\texttt{single-task} & 0.0331 (0.0012) & 0.0882 (0.0018) & 0.1858 (0.0030)\\
		\texttt{pooled} & 0.1092 (0.0025) & 0.1198 (0.0020) & 0.1336 (0.0017)\\
		\texttt{single-step} & 0.0280 (0.0009) & 0.0580 (0.0012) & 0.0925 (0.0015)\\
		\texttt{proposed} & \textbf{0.0083 (0.0003)} & \textbf{0.0193 (0.0005)} & \textbf{0.0389 (0.0007)}\\
		\addlinespace[2pt]
		\multicolumn{4}{l}{\textit{Multi-set logistic regression}}\\
		\texttt{single-task} & 0.0100 (0.0003) & 0.0259 (0.0006) & 0.0521 (0.0008)\\
		\texttt{pooled} & 0.0118 (0.0004) & 0.0201 (0.0005) & 0.0339 (0.0005)\\
		\texttt{single-step} & 0.0084 (0.0003) & 0.0184 (0.0005) & 0.0333 (0.0005)\\
		\texttt{proposed} & \textbf{0.0070 (0.0002)} & \textbf{0.0172 (0.0005)} & \textbf{0.0320 (0.0005)}\\
		\addlinespace[2pt]
		\multicolumn{4}{l}{\textit{Multi-set with covariance heterogeneity}}\\
		\texttt{single-task} & 0.0817 (0.0032) & 0.1923 (0.0055) & 0.3936 (0.0072)\\
		\texttt{pooled} & 0.3676 (0.0102) & 0.4136 (0.0061) & 0.4727 (0.0060)\\
		\texttt{Trans-GLM} & 0.0743 (0.0030) & 0.1823 (0.0054) & 0.3722 (0.0076)\\
		\texttt{ARMUL} & 0.0754 (0.0031) & 0.1782 (0.0049) & 0.3660 (0.0069)\\
		\texttt{single-step} & 0.0725 (0.0028) & 0.1486 (0.0038) & 0.2653 (0.0045)\\
		\texttt{dShrink} & 0.0728 (0.0028) & 0.1438 (0.0038) & 0.2461 (0.0043)\\
		\texttt{proposed} & \textbf{0.0377 (0.0019)} & \textbf{0.0825 (0.0024)} & \textbf{0.1740 (0.0038)}\\
		\bottomrule
	\end{tabular}
\end{table}

For the multi-set experiments with heterogeneous covariance, \(n_1=250\), \(n_2=n_3=n_4=n_5=500\), and
\[
\begin{aligned}
\btheta_1^\star&=\mathbf 1_p/\sqrt p,\\
\btheta_2^\star&=\btheta_1^\star+0.25\sqrt{p/n_2}\,\bv_2,
&\btheta_3^\star&=\btheta_1^\star+0.25\sqrt{p/n_3}\,\bv_3,\\
\btheta_4^\star&=\btheta_1^\star+1.5\bv_4,
&\btheta_5^\star&=\btheta_1^\star+1.5\bv_5.
\end{aligned}
\]
For the even dimensions in the experiment, define
\[
\begin{aligned}
\bD_A\left(\delta_\sigma\right)
&=\operatorname{diag}\!\left(\bI_{p/2},\left(1+\delta_\sigma\right)\bI_{p/2}\right),\\
\bD_B\left(\delta_\sigma\right)
&=\operatorname{diag}\!\left(\left(1+\delta_\sigma\right)\bI_{p/2},\bI_{p/2}\right).
\end{aligned}
\]
With \(\bR=\bR\left(0.5\right)\), the five samples follow \(\bx_{ji}\sim\mathcal N\left(\btheta_j^\star,\bSigma_j\right)\), where \(\bSigma_1=\bR\), \(\bSigma_2=\bSigma_4=\bD_A\bR\bD_A\), and \(\bSigma_3=\bSigma_5=\bD_B\bR\bD_B\). Hence sources 2 and 4 are precise on the first coordinate block, sources 3 and 5 are precise on the second block, and the contrast grows with \(\delta_\sigma\in\left\{0,2,\ldots,10\right\}\).

The implementation of \texttt{Trans-GLM} screens sources through three-fold target validation and selects the multiplier for an \(\ell_1\) target correction from \(\left\{0,0.25,0.5,1,2\right\}\) \citep{tian2023transfer}. The implementation of \texttt{ARMUL} uses task-specific penalties \(\lambda_j=C_{\lambda}\sqrt{p/n_j}\) and prototype weights proportional to \(\left(n_j/N\right)\lambda_j\), with $N=\sum_j n_j$ and \(C_{\lambda}\in\left\{0.25,0.5,1,2,4,8\right\}\) selected through three-fold target validation \citep{duan2023adaptive}. The implementation of \texttt{dShrink} follows its decomposition into calibrated and discrepancy components after inverse-covariance source aggregation \citep{wang2026divide}. The \texttt{single-step}, \texttt{dShrink}, and \texttt{proposed} methods receive the same local covariance estimates. 

Table~\ref{tab:simulation-reference-risks} gives the numerical values at the common reference points. At \(\delta=1\), the \texttt{proposed} estimator has the smallest risk in all six standard experiments and every displayed dimension. The two-set gains range from 3.7\% to 27.8\%, while the multi-set gains are larger once the two heterogeneous sources have separated from the target, particularly in Gaussian mean estimation and linear regression. At \(\delta_\sigma=10\), the heterogeneous covariance experiment yields the same ordering across all three dimensions, consistent with the proposed method's ability to handle source information that is heterogeneous in both estimation precision and population parameters.

\subsection{Real Data Analysis}\label{sec:college-scorecard-supplement}

\paragraph{Preprocessing and implementation details.} The College Scorecard data set contains the U.S. Service Schools target and the five source data sets used in the main analysis. Each observation records information about a school or institution,  with the response defined as the completion rate within 150\% of normal program time for full-time, first-time degree- or certificate-seeking undergraduate cohorts \citep{ipedsGraduation}. Table~\ref{tab:college-data-set-sizes} reports the sample sizes for the target and each source.

\begin{table*}[!t]
\centering
\small
\setlength{\tabcolsep}{3.8pt}
\renewcommand{\arraystretch}{0.97}
\caption{Sample sizes in the College Scorecard data set.}
\label{tab:college-data-set-sizes}
\vspace{4pt}
\resizebox{0.97\textwidth}{!}{%
\begin{tabular}{lcccccc}
\toprule
Data set & U.S. Service Schools & Plains & Southeast & Rocky Mountains & Far West & Southwest\\
\midrule
Sample size & 100 & 147 & 1,362 & 598 & 276 & 162\\
\bottomrule
\end{tabular}%
}
\end{table*}

For each of 100 repetitions, a random permutation of the 100 target records assigns 60 records to the training set and 40 to the evaluation set. The feature set contains 37 numerical covariates that summarize enrollment, tuition and costs, staffing and resources, student composition, program-of-study shares, degree categories, and admission characteristics. In every repetition, preprocessing uses only the target training data and the five source samples. We compute the common feature means, standard deviations, and covariance matrix. We retain and whiten the first 20 principal components and then add an intercept, which~gives a 21-dimensional local coefficient. The retained components explain 88.22\% of the training-feature variation on average, with values between 87.84\% and 88.63\% across repetitions.

For a completion rate $y$, the transformed response is
$
 z=\Phi^{-1}\left(\min\left(0.99,\max\left(0.01,y\right)\right)\right)
$,
where $\Phi$ is the standard-normal distribution function. We fit each data set separately by squared-loss linear $M$-estimation with $\bQ=\bI_{21}$ on this transformed scale and estimate the covariance of each local coefficient with a sandwich estimator. Predictions are returned to the original scale through $\Phi$. The proposed method, \texttt{single-step}, and \texttt{dShrink} use the same local coefficient and covariance estimators. \texttt{Trans-GLM} and \texttt{ARMUL} select their tuning parameters by three-fold cross-validation within the target training data. 

For repetition $r$ and method $m$, let $\widehat y_{mri}$ denote the prediction for target evaluation record $i$. We compute
\[
 \operatorname{RMSE}_{mr}
 =\left[\frac{1}{40}\sum_{i=1}^{40}
 \left(y_{ri}-\widehat y_{mri}\right)^2\right]^{1/2},
 \qquad
 \operatorname{MAE}_{mr}
 =\frac{1}{40}\sum_{i=1}^{40}\left|y_{ri}-\widehat y_{mri}\right|.
\]
The tables report averages of these repetition-level quantities. Parentheses contain the sample standard deviation across the 100 common repetitions divided by $10$. Since the random evaluation sets overlap, these values describe Monte Carlo variation across splits rather than sampling uncertainty from 100 independent data sets.

 Table~\ref{tab:college-paired} compares repetition-level MSE. For competitor $m$, the relative reduction is $\left(\overline L_m-\overline L_{\mathrm{proposed}}\right)/\overline L_m$, where $\overline L_m:= \frac{1}{100}\sum_{r=1}^{100}\operatorname{RMSE}_{mr}^2$, and the win frequency is the fraction of repetitions in which the MSE of the proposed method is smaller. 

\begin{table}[!t]
\centering
\small
\setlength{\tabcolsep}{5.5pt}
\renewcommand{\arraystretch}{0.97}
\caption{Paired comparisons with the \texttt{proposed} estimator.}
\label{tab:college-paired}
\vspace{4pt}
\begin{tabular}{lrr}
\toprule
Competitor & Relative MSE reduction & Win frequency \\
\midrule
\texttt{single-task} & 37.6\% & 84\% \\
\texttt{pooled} & 67.2\% & 98\% \\
\texttt{Trans-GLM} & 29.7\% & 81\% \\
\texttt{ARMUL} & 22.4\% & 84\% \\
\texttt{single-step} & 22.1\% & 72\% \\
\texttt{dShrink} & 37.1\% & 83\% \\
\bottomrule
\end{tabular}
\end{table}

\paragraph{Calculation of the percentages in Figure~\ref{fig:college-transfer-story}(\subref{fig:college-transfer-shares}).} Index the target by \(j=1\) and the five sources by \(j=2,\ldots,6\), and let \(\widetilde\btheta_j\in\mathbb R^p\) denote the local least-squares coefficient from data set \(j\).
In each repetition, the coefficient responses below are calculated with the preprocessing, design matrices, and estimated covariance matrices held fixed. We also hold fixed the fitted shrinkage factors and, where applicable, the realized update path, selected source set, tuning parameters, and active correction set and its sign pattern, as specified for each method below. Under this convention, define the conditional response matrices \(\bA_{mj}\) for method \(m\) through $\mathrm d\widehat\btheta_m = \sum_{j=1}^{6} \bA_{mj}\,\mathrm d\widetilde\btheta_j$. Thus, \(\bA_{mj}\) records the response of the final coefficient vector to a perturbation of the local coefficient from data set \(j\), with the specified fitted quantities held fixed rather than re-estimated. These matrices therefore describe conditional coefficient responses, not the total derivatives of the full data-adaptive fitting procedure. For the \texttt{single-task}, \texttt{single-step}, \texttt{proposed}, and \texttt{pooled} estimators, the final coefficient is an exact linear combination of the six local coefficients under this convention, and hence the matrices \(\bA_{mj}\) are obtained directly from their matrix weights.

For \texttt{dShrink}, let \(\widehat\bV_j\) denote the estimated covariance matrix of \(\widetilde\btheta_j\). The five source fits are first combined through inverse-covariance weighting. Define
$
\widehat\bV_S
=
\left(
\sum_{j=2}^{6}\widehat\bV_j^{-1}
\right)^{-1}$,
$\bW_j^S
=
\widehat\bV_S\widehat\bV_j^{-1}$, $j=2,\ldots,6$,
and
$
\widetilde\btheta_S
=
\sum_{j=2}^{6}
\bW_j^S\widetilde\btheta_j
$.
Thus, \(\widetilde\btheta_S\) is the inverse-covariance-weighted source aggregate and \(\widehat\bV_S\) is its estimated covariance matrix. Define
$
\bK
=
\widehat\bV_1
\left(
\widehat\bV_1+\widehat\bV_S
\right)^{-1}
$,
together with the discrepancy and calibrated components
$
\widetilde\bd
=
\bK
\left(
\widetilde\btheta_1-\widetilde\btheta_S
\right)$ and
$\widetilde\bc
=
\widetilde\btheta_1-\widetilde\bd
$.
The final dShrink estimator is
$
\widehat\btheta_{\mathrm{dShrink}}
=
\widehat\lambda_c\widetilde\bc
+
\widehat\lambda_d\widetilde\bd
$,
where \(\widehat\lambda_c\) and \(\widehat\lambda_d\) are the fitted shrinkage factors for the calibrated and discrepancy components. Conditional on these two shrinkage factors and the estimated covariance matrices,
\[
	\mathrm d\widehat\btheta_{\mathrm{dShrink}}
	=\left[
	\widehat\lambda_c\bI_p
	+
	\left(
	\widehat\lambda_d-\widehat\lambda_c
	\right)\bK
	\right]
	\mathrm d\widetilde\btheta_1
	+
	\sum_{j=2}^{6}
	\left(
	\widehat\lambda_c-\widehat\lambda_d
	\right)
	\bK\bW_j^S
	\,\mathrm d\widetilde\btheta_j.
\]
Hence, the local coefficient-response matrices for \texttt{dShrink} are
$
\bA_{\mathrm{dShrink},1}
=
\widehat\lambda_c\bI_p
+
\left(
\widehat\lambda_d-\widehat\lambda_c
\right)\bK
$
and
$
\bA_{\mathrm{dShrink},j}
=
\left(
\widehat\lambda_c-\widehat\lambda_d
\right)
\bK\bW_j^S$, 
$j=2,\ldots,6$.

The calculation for \texttt{Trans-GLM} requires an additional step. Let \(\widehat{\mathcal I}\subseteq\left\{2,\ldots,6\right\}\) denote the sources selected in a given repetition and let \(\mathcal J=\left\{1\right\}\cup\widehat{\mathcal I}\). For \(j\in\mathcal J\), define $ \bG_j=\frac{\bX_j^\top\bX_j}{n_j}$, $\bH_j=\bG_j+\varepsilon\bI_p$, $\pi_j=\frac{n_j}{\sum_{k\in\mathcal J}n_k} $, where \(\varepsilon=10^{-8}\) is the numerical ridge used in the implementation. Also define $ \bH_{\mathcal J} = \sum_{j\in\mathcal J}\pi_j\bG_j+\varepsilon\bI_p $. The coefficient obtained by pooling the target and selected sources can then be written as $ \overline\btheta = \sum_{j=1}^{6}\bB_j\widetilde\btheta_j $, where $\bB_j=\pi_j\bH_{\mathcal J}^{-1}\bH_j$ for $j\in\mathcal J$ and $\bB_j=\bm{\zero}$ for $j\notin\mathcal J$. Starting from \(\overline\btheta\), \texttt{Trans-GLM} estimates a target correction \[ \widehat\bd = \argmin_{\bd\in\mathbb R^p} \left[ \frac{1}{2n_1} \left\| \by_1-\bX_1\left(\overline\btheta+\bd\right) \right\|_2^2 + \lambda\left\|\bd_{2:p}\right\|_1 \right], \] where the intercept is not penalized, and its final coefficient is $\widehat\btheta_{\mathrm{Trans}} = \overline\btheta+\widehat\bd$. Let $ \widehat{\mathcal A} = \left\{1\right\} \cup \left\{ k\in\left\{2,\ldots,p\right\} : \left|\widehat d_k\right|>10^{-7} \right\} $ be the active correction coordinates, including the intercept, and let \(\bP_{\widehat{\mathcal A}}\) be the matrix that selects these coordinates. Define $ \bL = \bP_{\widehat{\mathcal A}} \left( \bP_{\widehat{\mathcal A}}^\top \bG_1 \bP_{\widehat{\mathcal A}} \right)^\dagger \bP_{\widehat{\mathcal A}}^\top$ and $\bC=\bI_p-\bL\bG_1$, where \(^{\dagger}\) denotes the Moore--Penrose inverse. Conditional on the selected source set, the correction parameter, the active coordinates and their signs, and the observed design matrices, the KKT conditions imply \[ \mathrm d\widehat\btheta_{\mathrm{Trans}} = \left(\bC\bB_1+\bL\bH_1\right) \mathrm d\widetilde\btheta_1 + \sum_{j=2}^{6} \bC\bB_j\,\mathrm d\widetilde\btheta_j. \] Hence, the local coefficient-response matrices for \texttt{Trans-GLM} are $ \bA_{\mathrm{Trans},1} = \bC\bB_1+\bL\bH_1$ and $\bA_{\mathrm{Trans},j} = \bC\bB_j$, $j=2,\ldots,6$. An unselected source has \(\bB_j=\bm{\zero}\) and therefore has zero local influence. For a selected source, \(\bC\bB_j\) is the part of its pooled contribution that remains after the target correction. The additional term \(\bL\bH_1\) records the direct contribution of the target data through that correction.

For every method, define the unnormalized contribution of data set \(j\) as
$
q_{mj}
=
\sum_{\ell=1}^{p}
\left|
\left(\bA_{mj}\right)_{\ell\ell}
\right|
$.
The corresponding contribution share is
$
w_{mj}
=
\frac{q_{mj}}
{\sum_{k=1}^{6}q_{mk}}
$.
The diagonal entry \(\left(\bA_{mj}\right)_{\ell\ell}\) measures the response of the \(\ell\)-th final coefficient to a unit perturbation of the same local coefficient from data set \(j\). Thus, \(w_{mj}\) is data set \(j\)'s share of the total absolute same-coordinate responses, evaluated in the coefficient coordinates formed by the 20 whitened principal components and the intercept. This is the meaning of ``contribution share'' in the main text: it summarizes conditional coefficient responses, rather than assigning shares of prediction-error reduction or explained variance.
When all diagonal responses are nonnegative, this definition reduces to
$
q_{mj}
=
\operatorname{tr}\left(\bA_{mj}\right)
$.
Figure~\ref{fig:college-transfer-story}(\subref{fig:college-transfer-shares}) reports the averages of \(100\% \times w_{mj}\) over the 100 repetitions.

\section{Representative Smooth $M$-Estimation Models}
\label{sec:M-examples}

This section verifies Assumptions~\ref{assump:1}--\ref{assump:4} for five representative smooth \(M\)-estimation models. Let \(\mathbb S^{p-1}=\left\{\bv\in\RR^p\mid\left\|\bv\right\|_2=1\right\}\). All constants below are independent of \(j\), \(n\), and \(p\), and all bounds hold uniformly over \(j\in\left[m\right]\). In every example, assume that \(\max_{j\in\left[m\right]}\left\|\btheta_j^\star\right\|_2\le B^\star\sqrt p\), which verifies Assumption~\ref{assump:4}.

For the logistic, pseudo-Huber, and smoothed quantile regression examples, suppose that the covariate vector \(\bX\in\RR^p\) satisfies \(\sup_{\bv\in\mathbb S^{p-1}}\left\|\bv^\top\bX\right\|_{\psi_2}\le K_x\). Let \(M_3=3\sqrt{\pi}/2\). The definition of the sub-Gaussian Orlicz norm gives \(\EE\left|\bv^\top\bX\right|^3\le M_3K_x^3\) for every \(\bv\in\mathbb S^{p-1}\). Consequently, H\"older's inequality and \(\left\|\bx\right\|_2^3\le p^{1/2}\sum_{k=1}^p\left|x_k\right|^3\) give \(\sup_{\bu,\bv,\bw\in\mathbb S^{p-1}}\EE\left|\bu^\top\bX\right|\left|\bv^\top\bX\right|\left|\bw^\top\bX\right|\le M_3K_x^3\) and \(\EE\left\|\bX\right\|_2^3\le M_3K_x^3p^{3/2}\). We also use \(\left\|U-\EE U\right\|_{\psi_2}\le2\left\|U\right\|_{\psi_2}\), \(\left\|UV\right\|_{\psi_1}\le2\left\|U\right\|_{\psi_2}\left\|V\right\|_{\psi_2}\), and \(\left\|U^2\right\|_{\psi_1}\le2\left\|U\right\|_{\psi_2}^2\).

\paragraph{Gaussian mean estimation.}
Suppose that the observation is \(\bZ=\btheta_j^\star+\boldsymbol\varepsilon\), where \(\boldsymbol\varepsilon\sim\mathcal N\left(\bm0,\bSigma_{\varepsilon,j}\right)\), \(\underline\lambda_\varepsilon\bI_p\preceq\bSigma_{\varepsilon,j}\preceq\overline\lambda_\varepsilon\bI_p\), and \(\sup_{\bv\in\mathbb S^{p-1}}\left\|\bv^\top\boldsymbol\varepsilon\right\|_{\psi_2}\le K_\varepsilon\). For the quadratic loss \(\ell_j\left(\btheta,\bZ\right)=\left\|\bZ-\btheta\right\|_2^2/2\), the derivatives are \(\nabla\ell_j\left(\btheta,\bZ\right)=\btheta-\bZ\), \(\nabla^2\ell_j\left(\btheta,\bZ\right)=\bI_p\), and \(\nabla^3\ell_j\left(\btheta,\bZ\right)=\bm0\). Hence \(\bH_j^\star=\bI_p\) and \(\bV_j^\star=\bSigma_{\varepsilon,j}\). Assumption~\ref{assump:1} holds with \(\underline\lambda=\min\left(1,\underline\lambda_\varepsilon\right)\) and \(\overline\lambda=\max\left(1,\overline\lambda_\varepsilon\right)\). For \(\btheta\in\mathcal B_j\left(r\right)\) and \(\bv\in\mathbb S^{p-1}\), \(\left\|\bv^\top\nabla\ell_j\left(\btheta,\bZ\right)\right\|_{\psi_2}\le K_\varepsilon+r/\sqrt{\log 2}\) and \(\left\|\bv^\top\nabla^2\ell_j\left(\btheta,\bZ\right)\bv\right\|_{\psi_1}=1/\log 2\). Thus Assumption~\ref{assump:2} holds with \(\tau=\max\left(K_\varepsilon+r/\sqrt{\log 2},\left(\log 2\right)^{-1/2}\right)\). Finally, \(\sup_{\btheta\in\mathcal B_j\left(r\right)}\left\|\EE\nabla^3\ell_j\left(\btheta,\bZ\right)\right\|_{\mathrm{op}}=0\), while \(\EE\sup_{\btheta\in\mathcal B_j\left(r\right)}\left\|\nabla\ell_j\left(\btheta,\bZ\right)\otimes\nabla^2\ell_j\left(\btheta,\bZ\right)\right\|_{\mathrm{op}}\le\EE\left\|\boldsymbol\varepsilon\right\|_2+r\le\sqrt{\overline\lambda_\varepsilon p}+r\le\left(\sqrt{\overline\lambda_\varepsilon}+r\right)p^{1/2}\). Therefore, Assumption~\ref{assump:3} holds with \(c=1/2\) and \(L=\sqrt{\overline\lambda_\varepsilon}+r\).

\paragraph{Semi-supervised linear regression with known covariance.}
This is the total-information setting of semi-supervised linear regression, in which the covariate second-moment matrix is known or can be estimated from a much larger unlabeled sample \citep{chakrabortty2018efficient,azriel2022semi}. Let \(\bSigma_{X,j}=\EE\left(\bX\bX^\top\right)\), define the best linear projection coefficient by \(\btheta_j^\star=\bSigma_{X,j}^{-1}\EE\left(\bX Y\right)\), and take \(\ell_j\left(\btheta,\left(\bX,Y\right)\right)=\btheta^\top\bSigma_{X,j}\btheta/2-Y\bX^\top\btheta\). Assume that \(\underline\lambda_x\bI_p\preceq\bSigma_{X,j}\preceq\overline\lambda_x\bI_p\), \(\sup_{\bv\in\mathbb S^{p-1}}\left\|\bv^\top\bX\right\|_{\psi_2}\le K_x\), \(\left|Y\right|\le B_Y\), and \(\operatorname{Var}\left(Y\mid\bX\right)\ge\underline\sigma_Y^2\) almost surely. The derivatives are \(\nabla\ell_j\left(\btheta,\left(\bX,Y\right)\right)=\bSigma_{X,j}\btheta-Y\bX\), \(\nabla^2\ell_j\left(\btheta,\left(\bX,Y\right)\right)=\bSigma_{X,j}\), and \(\nabla^3\ell_j\left(\btheta,\left(\bX,Y\right)\right)=\bm0\). Consequently, \(\bH_j^\star=\bSigma_{X,j}\) and \(\bV_j^\star=\operatorname{Var}\left(Y\bX\right)\). For every \(\bv\in\mathbb S^{p-1}\), the conditional variance decomposition gives \(\bv^\top\bV_j^\star\bv\ge\EE\left[\left(\bv^\top\bX\right)^2\operatorname{Var}\left(Y\mid\bX\right)\right]\ge\underline\sigma_Y^2\underline\lambda_x\), while \(\bv^\top\bV_j^\star\bv\le\EE\left[Y^2\left(\bv^\top\bX\right)^2\right]\le B_Y^2\overline\lambda_x\). Thus Assumption~\ref{assump:1} holds with \(\underline\lambda=\underline\lambda_x\min\left(1,\underline\sigma_Y^2\right)\) and \(\overline\lambda=\overline\lambda_x\max\left(1,B_Y^2\right)\). Since \(\EE\left(Y\bX\right)=\bSigma_{X,j}\btheta_j^\star\), for \(\btheta\in\mathcal B_j\left(r\right)\) and \(\bv\in\mathbb S^{p-1}\), \(\left\|\bv^\top\left(\bSigma_{X,j}\btheta_j^\star-Y\bX\right)\right\|_{\psi_2}\le2B_YK_x\), and hence \(\left\|\bv^\top\nabla\ell_j\left(\btheta,\left(\bX,Y\right)\right)\right\|_{\psi_2}\le2B_YK_x+\overline\lambda_xr/\sqrt{\log 2}\). Moreover, \(\left\|\bv^\top\nabla^2\ell_j\left(\btheta,\left(\bX,Y\right)\right)\bv\right\|_{\psi_1}\le\overline\lambda_x/\log 2\). Therefore, Assumption~\ref{assump:2} holds with \(\tau=\max\left(2B_YK_x+\overline\lambda_xr/\sqrt{\log 2},\left(\overline\lambda_x/\log 2\right)^{1/2}\right)\). Furthermore, \(\EE\left\|\bSigma_{X,j}\btheta_j^\star-Y\bX\right\|_2\le\left[\operatorname{Tr}\left(\bV_j^\star\right)\right]^{1/2}\le B_Y\sqrt{\overline\lambda_xp}\). It follows that \(\EE\sup_{\btheta\in\mathcal B_j\left(r\right)}\left\|\nabla\ell_j\left(\btheta,\left(\bX,Y\right)\right)\otimes\nabla^2\ell_j\left(\btheta,\left(\bX,Y\right)\right)\right\|_{\mathrm{op}}\le\overline\lambda_x\left[B_Y\sqrt{\overline\lambda_xp}+\overline\lambda_xr\right]\le C_{\mathrm{lin}}p^{1/2}\), where \(C_{\mathrm{lin}}=B_Y\overline\lambda_x^{3/2}+\overline\lambda_x^2r\). Since the population third derivative is zero, Assumption~\ref{assump:3} holds with \(c=1/2\) and \(L=C_{\mathrm{lin}}\).

\paragraph{Logistic regression.}
Let \(Y\in\left\{0,1\right\}\), define \(b\left(t\right)=\log\left(1+e^t\right)\), and take \(\ell_j\left(\btheta,\left(\bX,Y\right)\right)=b\left(\bX^\top\btheta\right)-Y\bX^\top\btheta\). The derivatives are \(\nabla\ell_j=\bX\left(b'\left(\bX^\top\btheta\right)-Y\right)\), \(\nabla^2\ell_j=b''\left(\bX^\top\btheta\right)\bX\bX^\top\), and \(\nabla^3\ell_j=b'''\left(\bX^\top\btheta\right)\bX^{\otimes3}\). Under the correctly specified model, \(\bV_j^\star=\bH_j^\star=\EE\left[b''\left(\bX^\top\btheta_j^\star\right)\bX\bX^\top\right]\). Suppose that \(\underline h_{\mathrm{log}}\bI_p\preceq\bH_j^\star\preceq\overline h_{\mathrm{log}}\bI_p\). Then Assumption~\ref{assump:1} holds with \(\underline\lambda=\underline h_{\mathrm{log}}\) and \(\overline\lambda=\overline h_{\mathrm{log}}\). The inequalities \(\left|b'\left(t\right)-Y\right|\le1\), \(0<b''\left(t\right)\le1/4\), and \(\left|b'''\left(t\right)\right|\le1/4\) give \(\left\|\bv^\top\nabla\ell_j\right\|_{\psi_2}\le K_x\) and \(\left\|\bv^\top\nabla^2\ell_j\bv\right\|_{\psi_1}\le K_x^2/2\). Thus Assumption~\ref{assump:2} holds with \(\tau=K_x\). Moreover, \(\sup_{\btheta\in\mathcal B_j\left(r\right)}\left\|\EE\nabla^3\ell_j\right\|_{\mathrm{op}}\le M_3K_x^3/4\), while \(\EE\sup_{\btheta\in\mathcal B_j\left(r\right)}\left[\left\|\nabla^3\ell_j\right\|_{\mathrm{op}}+\left\|\nabla\ell_j\otimes\nabla^2\ell_j\right\|_{\mathrm{op}}\right]\le\EE\left\|\bX\right\|_2^3/2\le M_3K_x^3p^{3/2}/2\). Hence Assumption~\ref{assump:3} holds with \(c=3/2\) and \(L=M_3K_x^3/2\).

\paragraph{Pseudo-Huber robust regression.}
Pseudo-Huber regression is a smooth version of Huber regression and has been used in robust and Huberized Bayesian regression \citep{kawakami2023approximate}. For \(a_{\mathrm H}\in\left[\underline a_{\mathrm H},\overline a_{\mathrm H}\right]\), define \(\phi_{a_{\mathrm H}}\left(u\right)=a_{\mathrm H}^2\left[\left(1+u^2/a_{\mathrm H}^2\right)^{1/2}-1\right]\) and take \(\ell_j\left(\btheta,\left(\bX,Y\right)\right)=\phi_{a_{\mathrm H}}\left(Y-\bX^\top\btheta\right)\). Let \(R_j^\star=Y-\bX^\top\btheta_j^\star\), and suppose that the eigenvalues of \(\bH_j^\star=\EE\left[\phi_{a_{\mathrm H}}''\left(R_j^\star\right)\bX\bX^\top\right]\) belong to \(\left[\underline h_{\mathrm H},\overline h_{\mathrm H}\right]\), while those of \(\bV_j^\star=\EE\left[\phi_{a_{\mathrm H}}'\left(R_j^\star\right)^2\bX\bX^\top\right]\) belong to \(\left[\underline v_{\mathrm H},\overline v_{\mathrm H}\right]\). Assumption~\ref{assump:1} holds with \(\underline\lambda=\min\left(\underline h_{\mathrm H},\underline v_{\mathrm H}\right)\) and \(\overline\lambda=\max\left(\overline h_{\mathrm H},\overline v_{\mathrm H}\right)\). Writing \(R\left(\btheta\right)=Y-\bX^\top\btheta\), the derivatives are \(\nabla\ell_j=-\phi_{a_{\mathrm H}}'\left(R\left(\btheta\right)\right)\bX\), \(\nabla^2\ell_j=\phi_{a_{\mathrm H}}''\left(R\left(\btheta\right)\right)\bX\bX^\top\), and \(\nabla^3\ell_j=-\phi_{a_{\mathrm H}}'''\left(R\left(\btheta\right)\right)\bX^{\otimes3}\). The scalar derivative bounds \(\left|\phi_{a_{\mathrm H}}'\left(u\right)\right|\le\overline a_{\mathrm H}\), \(0<\phi_{a_{\mathrm H}}''\left(u\right)\le1\), and \(\left|\phi_{a_{\mathrm H}}'''\left(u\right)\right|\le3/\underline a_{\mathrm H}\) give \(\left\|\bv^\top\nabla\ell_j\right\|_{\psi_2}\le\overline a_{\mathrm H}K_x\) and \(\left\|\bv^\top\nabla^2\ell_j\bv\right\|_{\psi_1}\le2K_x^2\). Thus Assumption~\ref{assump:2} holds with \(\tau=K_x\max\left(\overline a_{\mathrm H},\sqrt2\right)\). Furthermore, \(\sup_{\btheta\in\mathcal B_j\left(r\right)}\left\|\EE\nabla^3\ell_j\right\|_{\mathrm{op}}\le3M_3K_x^3/\underline a_{\mathrm H}\), while \(\EE\sup_{\btheta\in\mathcal B_j\left(r\right)}\left[\left\|\nabla^3\ell_j\right\|_{\mathrm{op}}+\left\|\nabla\ell_j\otimes\nabla^2\ell_j\right\|_{\mathrm{op}}\right]\le M_3K_x^3\left(3/\underline a_{\mathrm H}+\overline a_{\mathrm H}\right)p^{3/2}\). Hence Assumption~\ref{assump:3} holds with \(c=3/2\) and \(L=M_3K_x^3\left(3/\underline a_{\mathrm H}+\overline a_{\mathrm H}\right)\).

\paragraph{Smoothed quantile regression.}
For \(\alpha\in\left(0,1\right)\) and \(a_{\mathrm Q}\in\left[\underline a_{\mathrm Q},\overline a_{\mathrm Q}\right]\), define \(q_{\alpha,a_{\mathrm Q}}\left(u\right)=\alpha u+a_{\mathrm Q}\log\left(1+e^{-u/a_{\mathrm Q}}\right)\) and take \(\ell_j\left(\btheta,\left(\bX,Y\right)\right)=q_{\alpha,a_{\mathrm Q}}\left(Y-\bX^\top\btheta\right)\). Let \(R_j^\star=Y-\bX^\top\btheta_j^\star\), and suppose that the eigenvalues of \(\bH_j^\star=\EE\left[q_{\alpha,a_{\mathrm Q}}''\left(R_j^\star\right)\bX\bX^\top\right]\) belong to \(\left[\underline h_{\mathrm Q},\overline h_{\mathrm Q}\right]\), while those of \(\bV_j^\star=\EE\left[q_{\alpha,a_{\mathrm Q}}'\left(R_j^\star\right)^2\bX\bX^\top\right]\) belong to \(\left[\underline v_{\mathrm Q},\overline v_{\mathrm Q}\right]\). Assumption~\ref{assump:1} holds with \(\underline\lambda=\min\left(\underline h_{\mathrm Q},\underline v_{\mathrm Q}\right)\) and \(\overline\lambda=\max\left(\overline h_{\mathrm Q},\overline v_{\mathrm Q}\right)\). Writing \(R\left(\btheta\right)=Y-\bX^\top\btheta\), the derivatives are \(\nabla\ell_j=-q_{\alpha,a_{\mathrm Q}}'\left(R\left(\btheta\right)\right)\bX\), \(\nabla^2\ell_j=q_{\alpha,a_{\mathrm Q}}''\left(R\left(\btheta\right)\right)\bX\bX^\top\), and \(\nabla^3\ell_j=-q_{\alpha,a_{\mathrm Q}}'''\left(R\left(\btheta\right)\right)\bX^{\otimes3}\). The scalar derivative bounds \(\left|q_{\alpha,a_{\mathrm Q}}'\left(u\right)\right|\le1\), \(0<q_{\alpha,a_{\mathrm Q}}''\left(u\right)\le1/\left(4\underline a_{\mathrm Q}\right)\), and \(\left|q_{\alpha,a_{\mathrm Q}}'''\left(u\right)\right|\le1/\left(4\underline a_{\mathrm Q}^2\right)\) give \(\left\|\bv^\top\nabla\ell_j\right\|_{\psi_2}\le K_x\) and \(\left\|\bv^\top\nabla^2\ell_j\bv\right\|_{\psi_1}\le K_x^2/\left(2\underline a_{\mathrm Q}\right)\). Thus Assumption~\ref{assump:2} holds with \(\tau=K_x\max\left(1,\left(2\underline a_{\mathrm Q}\right)^{-1/2}\right)\). Moreover, \(\sup_{\btheta\in\mathcal B_j\left(r\right)}\left\|\EE\nabla^3\ell_j\right\|_{\mathrm{op}}\le M_3K_x^3/\left(4\underline a_{\mathrm Q}^2\right)\), while \(\EE\sup_{\btheta\in\mathcal B_j\left(r\right)}\left[\left\|\nabla^3\ell_j\right\|_{\mathrm{op}}+\left\|\nabla\ell_j\otimes\nabla^2\ell_j\right\|_{\mathrm{op}}\right]\le M_3K_x^3\left[1/\left(4\underline a_{\mathrm Q}^2\right)+1/\left(4\underline a_{\mathrm Q}\right)\right]p^{3/2}\). Hence Assumption~\ref{assump:3} holds with \(c=3/2\) and \(L=M_3K_x^3\left[1/\left(4\underline a_{\mathrm Q}^2\right)+1/\left(4\underline a_{\mathrm Q}\right)\right]\).

\section{Technical Proofs of the Theoretical Results} \label{sec:proofs-appendix}


\subsection{Proof of the Results for Gaussian Mean Estimation}

\begin{proof}[Proof of Theorem~\ref{thm:risk-bound-s}]
Fix the source estimator \(\tbtheta_2\) and regard it as nonrandom throughout the conditional calculation. For \(s>0\), define the perturbation of the target sample mean directly by
\begin{equation}\label{eq:two-set-perturbation-full}
\bg_s\left(\bx\right)
=
\frac{s}{
\left\|\bW_2\left(\tbtheta_2-\bx\right)\right\|_{\bQ}^{2}
\vee s}
\bW_2\left(\tbtheta_2-\bx\right).
\end{equation}
The Jacobian is
\begin{equation}\label{eq:two-set-Jacobian-full}
\begin{aligned}
\bJ_s\left(\bx\right)&=-\bW_2,
&&\left\|\bW_2\left(\tbtheta_2-\bx\right)\right\|_{\bQ}^{2}<s,\\
\bJ_s\left(\bx\right)&=-\frac{s\bW_2}{\left\|\bW_2\left(\tbtheta_2-\bx\right)\right\|_{\bQ}^{2}}
&&\left\|\bW_2\left(\tbtheta_2-\bx\right)\right\|_{\bQ}^{2}>s,\\
&\quad+\frac{2s\bW_2\left(\tbtheta_2-\bx\right)
\left[\bW_2\left(\tbtheta_2-\bx\right)\right]^{\top}\bQ\bW_2}
{\left\|\bW_2\left(\tbtheta_2-\bx\right)\right\|_{\bQ}^{4}}.
\end{aligned}
\end{equation}
To verify Lemma~\ref{lem:sure}, let
$\boldsymbol\tau_s\left(\bu\right)=s\bu/\left(\left\|\bu\right\|_{\bQ}^2\vee s\right)$ and write
$\left\|\bA\right\|_{\bQ\to\bQ}=\left\|\bQ^{1/2}\bA\bQ^{-1/2}\right\|_2$.
The map is continuous at $\left\|\bu\right\|_{\bQ}^2=s$. Its derivative is $\bI_p$ inside this boundary, while outside it
\[
\bQ^{1/2}D\boldsymbol\tau_s\left(\bu\right)\bQ^{-1/2}
=\frac{s}{\left\|\bu\right\|_{\bQ}^2}
\left[\bI_p-2\frac{\left(\bQ^{1/2}\bu\right)\left(\bQ^{1/2}\bu\right)^\top}{\left\|\bu\right\|_{\bQ}^2}\right].
\]
The matrix in brackets is an orthogonal reflection, which~gives $\left\|D\boldsymbol\tau_s\left(\bu\right)\right\|_{\bQ\to\bQ}\le1$ wherever the derivative exists. A line segment intersects the boundary in at most two points unless its endpoints coincide. Integration on the resulting subintervals therefore gives
$\left\|\boldsymbol\tau_s\left(\bu\right)-\boldsymbol\tau_s\left(\bv\right)\right\|_{\bQ}\le\left\|\bu-\bv\right\|_{\bQ}$~and
\begin{equation}
\left\|\bg_s\left(\bx_1\right)-\bg_s\left(\bx_0\right)\right\|_{\bQ}
\le\left\|\bW_2\left(\bx_1-\bx_0\right)\right\|_{\bQ}.
\label{eq:f.3}
\end{equation}
Thus $\bg_s$ is globally Lipschitz and weakly differentiable. Moreover,
\begin{equation}
\left\|\bg_s\left(\bx\right)\right\|_{\bQ}\le\sqrt{s},\qquad
\left\|\bg_s\left(\bx\right)\right\|_2\le\sqrt{s/q_-},
\label{eq:two-set-truncated-bound}
\end{equation}
and, almost everywhere,
\begin{equation}
\left\|\bJ_s\left(\bx\right)\right\|_2\le\sqrt{q_+/q_-}\,\left\|\bW_2\right\|_2.
\label{eq:two-set-jacobian-bound}
\end{equation}
For every $k,k'$, Cauchy--Schwarz and these deterministic bounds imply
\begin{equation}
\begin{aligned}
\EE\left| \left(\tbtheta_{1,k}-\theta_{1,k}^\star\right)g_{s,k'}\left(\tbtheta_1\right)\right|
&\le\sqrt{\frac{\Sigma_{1,kk}}{n_1}}\sqrt{\frac{s}{q_-}}<\infty,\\
\EE\left|\frac{\partial g_{s,k'}\left(\tbtheta_1\right)}{\partial\tbtheta_{1,k}}\right|
&\le\sqrt{q_+/q_-}\,\left\|\bW_2\right\|_2<\infty.
\end{aligned}
\label{eq:two-set-integrability-derivative}
\end{equation}
These bounds hold conditionally on every fixed source estimator.

The identity
$
\bW_2\left(n_1^{-1}\bSigma_1\right)
=
\left(n_1^{-1}\bSigma_1\right)
\left(n_1^{-1}\bSigma_1+n_2^{-1}\bSigma_2\right)^{-1}
\left(n_1^{-1}\bSigma_1\right)
$
shows that \(\bW_2\left(n_1^{-1}\bSigma_1\right)\) is symmetric positive semidefinite. Hence, for every \(\bu\in\mathbb R^p\),
\begin{align}
\bu^{\top}\bQ\bW_2\left(n_1^{-1}\bSigma_1\right)\bQ\bu
&=
\left(\bQ^{1/2}\bu\right)^{\top}
\left[
\bQ^{1/2}\bW_2\left(n_1^{-1}\bSigma_1\right)\bQ^{1/2}
\right]
\left(\bQ^{1/2}\bu\right)
\notag\\
&\le
\left\|
\bQ^{1/2}\bW_2\left(n_1^{-1}\bSigma_1\right)\bQ^{1/2}
\right\|_2
\left\|\bu\right\|_{\bQ}^{2}
=
\left\|\bS\right\|_2
\left\|\bu\right\|_{\bQ}^{2}.
\label{eq:two-set-quadratic-bound-full}
\end{align}

Lemma~\ref{lem:sure}, applied conditionally on \(\tbtheta_2\), now gives
\begin{align}
\mathbb E\left[
\left\|\widehat{\btheta}_s-\btheta_1^{\star}\right\|_{\bQ}^{2}
\,
\middle|
\,
\tbtheta_2
\right]
=
\left\langle n_1^{-1}\bSigma_1, \bQ\right\rangle
+
\mathbb E\left[
2\left\langle \bJ_s\left(\tbtheta_1\right)
\left(n_1^{-1}\bSigma_1\right), \bQ\right\rangle
+
\left\|\bg_s\left(\tbtheta_1\right)\right\|_{\bQ}^{2}
\,
\middle|
\,
\tbtheta_2
\right].
\label{eq:two-set-SURE-conditional-full}
\end{align}
On the region
\(
\left\|\bW_2\left(\tbtheta_2-\tbtheta_1\right)\right\|_{\bQ}^{2}>s
\), substitution of the second line of \eqref{eq:two-set-Jacobian-full} into the integrand in \eqref{eq:two-set-SURE-conditional-full} gives the exact expression
\begin{equation}
\frac{s^2-2s\Tr\left(\bS\right)}
{\left\|\bW_2\left(\tbtheta_2-\tbtheta_1\right)\right\|_{\bQ}^{2}}
+
\frac{4s\left[\bW_2\left(\tbtheta_2-\tbtheta_1\right)\right]^{\top}
\bQ\bW_2\left(n_1^{-1}\bSigma_1\right)\bQ\bW_2\left(\tbtheta_2-\tbtheta_1\right)}
{\left\|\bW_2\left(\tbtheta_2-\tbtheta_1\right)\right\|_{\bQ}^{4}}.
\label{eq:two-set-outer-integrand-exact}
\end{equation}
Equation~\eqref{eq:two-set-quadratic-bound-full} bounds \eqref{eq:two-set-outer-integrand-exact} from above by
\begin{align}
&-
\frac{s}
{\left\|\bW_2\left(\tbtheta_2-\tbtheta_1\right)\right\|_{\bQ}^{2}}
\left[
2\Tr\left(\bS\right)
-
4\left\|\bS\right\|_2
-
s
\right].
\label{eq:two-set-outer-integrand-bound}
\end{align}
On the region
\(
\left\|\bW_2\left(\tbtheta_2-\tbtheta_1\right)\right\|_{\bQ}^{2}\le s
\), the first line of \eqref{eq:two-set-Jacobian-full} gives
\begin{align}
2\left\langle -\bW_2\left(n_1^{-1}\bSigma_1\right), \bQ\right\rangle
+
\left\|\bW_2\left(\tbtheta_2-\tbtheta_1\right)\right\|_{\bQ}^{2}
&=
-2\Tr\left(\bS\right)
+
\left\|\bW_2\left(\tbtheta_2-\tbtheta_1\right)\right\|_{\bQ}^{2}
\notag\\
&\qquad\le
-2\Tr\left(\bS\right)+s
\notag\\
&\qquad\leq
-
\frac{s\left[2\Tr\left(\bS\right)-4\left\|\bS\right\|_2-s\right]}
{s}.
\label{eq:two-set-inner-integrand-bound}
\end{align}
Equations~\eqref{eq:two-set-outer-integrand-bound} and \eqref{eq:two-set-inner-integrand-bound} yield the pointwise inequality
\begin{align}
&2\left\langle \bJ_s\left(\tbtheta_1\right)
\left(n_1^{-1}\bSigma_1\right), \bQ\right\rangle
+
\left\|\bg_s\left(\tbtheta_1\right)\right\|_{\bQ}^{2}
\le
-
\frac{s\left[2\Tr\left(\bS\right)-4\left\|\bS\right\|_2-s\right]}
{
\left\|\bW_2\left(\tbtheta_2-\tbtheta_1\right)\right\|_{\bQ}^{2}
\vee s}.
\label{eq:two-set-integrand-combined-full}
\end{align}
Taking the conditional expectation in \eqref{eq:two-set-integrand-combined-full}, substituting it into \eqref{eq:two-set-SURE-conditional-full}, and then taking expectation with respect to \(\tbtheta_2\) proves \eqref{eq:risk-upper}. For
\(
0<s<2\Tr\left(\bS\right)-4\left\|\bS\right\|_2
\), the coefficient in square brackets is positive and $0< \mathbb E\left[ \frac{1}{ \left\|\bW_2\left(\tbtheta_2-\tbtheta_1\right)\right\|_{\bQ}^{2} \vee s} \right] \le \frac{1}{s}$. The risk difference therefore satisfies
\begin{align*}
&\mathbb E
\left[
\left\|\widehat\btheta_s-\btheta_1^\star\right\|_{\bQ}^2
\right]
-
\frac{1}{n_1}
\Tr\left(\bQ\bSigma_1\right)
\le
-s
\left[
2\Tr\left(\bS\right)
-
4\left\|\bS\right\|_2
-
s
\right]
\mathbb E\left[
\frac{1}{
\left\|\bW_2\left(\tbtheta_2-\tbtheta_1\right)\right\|_{\bQ}^2
\vee s
}
\right]
<0.
\end{align*}
\end{proof}

\begin{proof}[Proof of Theorem~\ref{thm:risk-bound-shat}]
Let
$
\bV_{\rm pool}
=
\left(n_1\bSigma_1^{-1}+n_2\bSigma_2^{-1}\right)^{-1}
$.
The pooled estimator satisfies
\begin{align*}
\mathbb E\left[\overline{\btheta}\right]
=
\btheta_1^{\star}
+
\bW_2\left(\btheta_2^{\star}-\btheta_1^{\star}\right),
\quad
\operatorname{Cov}\left(\overline{\btheta}\right)
=
\bV_{\rm pool}.
\end{align*}
The relevant cross-covariances are
\begin{align}
\operatorname{Cov}\left(\overline{\btheta},\tbtheta_1\right)
=
\bV_{\rm pool}, \quad
\operatorname{Cov}\left(
\overline{\btheta},
\overline{\btheta}-\tbtheta_1
\right)
=
\operatorname{Cov}\left(\overline{\btheta}\right)
-
\operatorname{Cov}\left(\overline{\btheta},\tbtheta_1\right)
=
\bm0.
\label{eq:two-set-pooled-direction-cross-cov-full}
\end{align}
The pair
\(
\left(\overline{\btheta},\overline{\btheta}-\tbtheta_1\right)
\)
is jointly Gaussian, and thus Equation~\eqref{eq:two-set-pooled-direction-cross-cov-full} proves their independence. Direct calculation also gives
\begin{align}
\mathbb E\left[
\overline{\btheta}-\tbtheta_1
\right]
&=
\bW_2\left(\btheta_2^{\star}-\btheta_1^{\star}\right),
\notag\\
\operatorname{Cov}\left(
\overline{\btheta}-\tbtheta_1
\right)
&=
\bV_{\rm pool}+n_1^{-1}\bSigma_1-2\bV_{\rm pool}
=
 n_1^{-1}\bSigma_1-\bV_{\rm pool}
=
\bW_2\left(n_1^{-1}\bSigma_1\right).
\label{eq:two-set-direction-mean-cov-full}
\end{align}

The identity
\[
\widehat{\btheta}_s-\btheta_1^{\star}
=
\overline{\btheta}-\EE\left[\overline{\btheta}\right]
-
\left(
1-
\frac{s}{
\left\|\overline{\btheta}-\tbtheta_1\right\|_{\bQ}^{2}}
\right)_+
\left(\overline{\btheta}-\tbtheta_1\right)+\bW_2\left(\btheta_2^{\star}-\btheta_1^{\star}\right)
\]
and the independence established above imply
\begin{align}
\mathbb E\left[
\left\|\widehat{\btheta}_s-\btheta_1^{\star}\right\|_{\bQ}^{2}
\right]
&=
\left\langle \bV_{\rm pool}, \bQ\right\rangle
+
\mathbb E\left[
\left\|
\left(
1-
\frac{s}{
\left\|\overline{\btheta}-\tbtheta_1\right\|_{\bQ}^{2}}
\right)_+
\left(\overline{\btheta}-\tbtheta_1\right)
-
\bW_2\left(\btheta_2^{\star}-\btheta_1^{\star}\right)
\right\|_{\bQ}^{2}
\right].
\label{eq:positive-part-decomposition-detailed}
\end{align}

Define the transformed direction and mean by
\begin{align*}
\bZ_{\rm dir}
=
\bQ^{1/2}\left(\overline{\btheta}-\tbtheta_1\right),
\quad
\bmu_{\rm dir}
=
\bQ^{1/2}\bW_2\left(\btheta_2^{\star}-\btheta_1^{\star}\right).
\end{align*}
By \eqref{eq:two-set-direction-mean-cov-full},
$
\bZ_{\rm dir}
\sim
\mathcal N\left(\bmu_{\rm dir},\bS\right)
$.
For a deterministic \(\alpha\in\left[0,1\right]\), the quadratic risk of \(\alpha\bZ_{\rm dir}\) for \(\bmu_{\rm dir}\) is
\begin{align}
\mathcal R_{\rm lin}\left(\alpha\right)
=
\mathbb E\left[
\left\|\alpha\bZ_{\rm dir}-\bmu_{\rm dir}\right\|_2^2
\right]=
\alpha^2\Tr\left(\bS\right)
+
\left(1-\alpha\right)^2
\left\|\bmu_{\rm dir}\right\|_2^2.
\label{eq:two-set-fixed-alpha-risk-full}
\end{align}
Differentiating the right-hand side of \eqref{eq:two-set-fixed-alpha-risk-full} gives
$
\frac{\mathrm d}{\mathrm d\alpha}
\mathcal R_{\rm lin}\left(\alpha\right)
=
2\alpha\Tr\left(\bS\right)
-
2\left(1-\alpha\right)
\left\|\bmu_{\rm dir}\right\|_2^2
$.
The unique minimizer and its risk are therefore
\begin{align}
\alpha_{\rm lin}^{\star}
=
\frac{\left\|\bmu_{\rm dir}\right\|_2^2}
{\Tr\left(\bS\right)+\left\|\bmu_{\rm dir}\right\|_2^2},
\quad
\mathcal R_{\rm lin}\left(\alpha_{\rm lin}^{\star}\right)
=
\frac{\Tr\left(\bS\right)\left\|\bmu_{\rm dir}\right\|_2^2}
{\Tr\left(\bS\right)+\left\|\bmu_{\rm dir}\right\|_2^2}.
\label{eq:two-set-best-linear-risk-full}
\end{align}

For
\(
s\in
\left[
\Tr\left(\bS\right)-2\left\|\bS\right\|_2,
\Tr\left(\bS\right)
\right]
\), define
\begin{equation}\label{eq:two-set-data-criterion-full}
\widehat{\mathcal L}_s\left(\alpha,\bZ_{\rm dir}\right)
=
\left(1-\alpha\right)^2
\left\|\bZ_{\rm dir}\right\|_2^2
+
\left(2\alpha-1\right)s,
\quad
\alpha\in\left[0,1\right].
\end{equation}
The derivative with respect to \(\alpha\) is
$
2\left(\alpha-1\right)\left\|\bZ_{\rm dir}\right\|_2^2+2s
$.
Projection of the unconstrained minimizer onto \(\left[0,1\right]\) gives
\begin{equation}\label{eq:two-set-data-alpha-full}
\widehat\alpha_s
=
\left(
1-
\frac{s}{\left\|\bZ_{\rm dir}\right\|_2^2}
\right)_+.
\end{equation}

Apply Lemma~\ref{lem:sure} to $\bx=\bZ_{\rm dir}$ with covariance matrix \(\bSigma=\bS\), loss matrix \(\bQ=\bI_p\), and~\mbox{perturbation} $\bg\left(\bz\right)=\bh_s\left(\bz\right) := - \frac{s}{\left\|\bz\right\|_2^2\vee s}\bz$. The integrability verification in \eqref{eq:two-set-truncated-bound}--\eqref{eq:two-set-integrability-derivative} applies verbatim with \(\bW_2=\bI_p\) and \(\bQ=\bI_p\). The Jacobian equals \(-\bI_p\) when \(\left\|\bz\right\|_2^2<s\), and it equals
$
-
\frac{s}{\left\|\bz\right\|_2^2}\bI_p
+
\frac{2s}{\left\|\bz\right\|_2^4}\bz\bz^{\top}
$
when \(\left\|\bz\right\|_2^2>s\). Lemma~\ref{lem:sure} therefore gives the unbiased risk estimate
\begin{equation}\label{eq:two-set-exact-SURE-alpha-full}
\begin{aligned}
\widehat{\mathcal R}_s\left(\bZ_{\rm dir}\right)
&=\left\|\bZ_{\rm dir}\right\|_2^2-\Tr\left(\bS\right),
&&\left\|\bZ_{\rm dir}\right\|_2^2\le s,\\
\widehat{\mathcal R}_s\left(\bZ_{\rm dir}\right)
&=\Tr\left(\bS\right)
-\frac{2s\Tr\left(\bS\right)}{\left\|\bZ_{\rm dir}\right\|_2^2}
+\frac{s^2}{\left\|\bZ_{\rm dir}\right\|_2^2}
&&\left\|\bZ_{\rm dir}\right\|_2^2>s,\\
&\quad+\frac{4s\bZ_{\rm dir}^{\top}\bS\bZ_{\rm dir}}
{\left\|\bZ_{\rm dir}\right\|_2^4}.
\end{aligned}
\end{equation}

When \(\left\|\bZ_{\rm dir}\right\|_2^2\le s\), the minimizer in \eqref{eq:two-set-data-alpha-full} is zero and
\begin{align}
\widehat{\mathcal R}_s\left(\bZ_{\rm dir}\right)
-
\min_{\alpha\in\left[0,1\right]}
\widehat{\mathcal L}_s\left(\alpha,\bZ_{\rm dir}\right)
&=
\left\|\bZ_{\rm dir}\right\|_2^2-\Tr\left(\bS\right)
-
\left[\left\|\bZ_{\rm dir}\right\|_2^2-s\right]
=
-\left[\Tr\left(\bS\right)-s\right].
\label{eq:two-set-SURE-criterion-inner-full}
\end{align}
When \(\left\|\bZ_{\rm dir}\right\|_2^2>s\), substitution of \eqref{eq:two-set-data-alpha-full} into \eqref{eq:two-set-data-criterion-full} gives $\min_{\alpha\in\left[0,1\right]} \widehat{\mathcal L}_s\left(\alpha,\bZ_{\rm dir}\right) = s- \frac{s^2}{\left\|\bZ_{\rm dir}\right\|_2^2}$. Subtracting the two expressions gives
\begin{align}
&\widehat{\mathcal R}_s\left(\bZ_{\rm dir}\right)
-
\min_{\alpha\in\left[0,1\right]}
\widehat{\mathcal L}_s\left(\alpha,\bZ_{\rm dir}\right)
\notag\\
&\quad=
\Tr\left(\bS\right)-s
-
\frac{2s\left[\Tr\left(\bS\right)-s\right]}
{\left\|\bZ_{\rm dir}\right\|_2^2}
+
\frac{4s\bZ_{\rm dir}^{\top}\bS\bZ_{\rm dir}}
{\left\|\bZ_{\rm dir}\right\|_2^4}
\notag\\
&\quad\le
\Tr\left(\bS\right)-s
+
\left[
4\left\|\bS\right\|_2
-2\left[\Tr\left(\bS\right)-s\right]
\right]
\frac{s}{\left\|\bZ_{\rm dir}\right\|_2^2}
\notag\\
&\quad\le
4\left\|\bS\right\|_2
-
\left[\Tr\left(\bS\right)-s\right].
\label{eq:two-set-SURE-criterion-outer-full}
\end{align}
The last inequality uses
\(
0\le\Tr\left(\bS\right)-s\le2\left\|\bS\right\|_2
\), which makes the coefficient multiplying
\(
s/\left\|\bZ_{\rm dir}\right\|_2^2
\)
nonnegative, together with
\(
s/\left\|\bZ_{\rm dir}\right\|_2^2<1
\).
Equations~\eqref{eq:two-set-SURE-criterion-inner-full} through \eqref{eq:two-set-SURE-criterion-outer-full} prove the pointwise~bound
\begin{align}
\widehat{\mathcal R}_s\left(\bZ_{\rm dir}\right)
&\le
\min_{\alpha\in\left[0,1\right]}
\widehat{\mathcal L}_s\left(\alpha,\bZ_{\rm dir}\right)
+
4\left\|\bS\right\|_2
-
\left[\Tr\left(\bS\right)-s\right].
\label{eq:two-set-SURE-criterion-pointwise-full}
\end{align}

For every deterministic \(\alpha\in\left[0,1\right]\),
\begin{align*}
\mathbb E\left[
\widehat{\mathcal L}_s\left(\alpha,\bZ_{\rm dir}\right)
\right]
&=
\left(1-\alpha\right)^2
\left[
\Tr\left(\bS\right)+\left\|\bmu_{\rm dir}\right\|_2^2
\right]
+
\left(2\alpha-1\right)s
\\
&=
\mathcal R_{\rm lin}\left(\alpha\right)
+
\left(1-2\alpha\right)
\left[\Tr\left(\bS\right)-s\right].
\end{align*}
Taking expectations in \eqref{eq:two-set-SURE-criterion-pointwise-full}, using the unbiasedness of \eqref{eq:two-set-exact-SURE-alpha-full}, and evaluating the random minimum at the deterministic value \(\alpha_{\rm lin}^{\star}\) gives
\begin{align}
\mathbb E\left[
\left\|\widehat\alpha_s\bZ_{\rm dir}-\bmu_{\rm dir}\right\|_2^2
\right]
&\le
\mathbb E\left[
\widehat{\mathcal L}_s\left(\alpha_{\rm lin}^{\star},\bZ_{\rm dir}\right)
\right]
+
4\left\|\bS\right\|_2
-
\left[\Tr\left(\bS\right)-s\right]
\notag\\
&=
\mathcal R_{\rm lin}\left(\alpha_{\rm lin}^{\star}\right)
+
4\left\|\bS\right\|_2
-
2\alpha_{\rm lin}^{\star}
\left[\Tr\left(\bS\right)-s\right]
\notag\\
&\le
\mathcal R_{\rm lin}\left(\alpha_{\rm lin}^{\star}\right)
+
4\left\|\bS\right\|_2.
\label{eq:two-set-adaptive-alpha-risk-full}
\end{align}
Substituting \eqref{eq:two-set-best-linear-risk-full} into \eqref{eq:two-set-adaptive-alpha-risk-full}, then substituting the result into \eqref{eq:positive-part-decomposition-detailed}, proves the stated bound after using
\(
\left\|\bmu_{\rm dir}\right\|_2^2
=
\left\|\bW_2\left(\btheta_2^{\star}-\btheta_1^{\star}\right)\right\|_{\bQ}^2
\).
\end{proof}

		\begin{proof}[Proof of Theorem~\ref{thm:regularized-MLE}]
Let $\bC_j=n_j^{-1}\bSigma_j$, $\bA=\bC_1+\bC_2$, and
$\bz=\bA^{-1/2}\left(\tbtheta_1-\tbtheta_2\right)$. After discarding terms independent of the parameters, the objective is
\[
\frac12\sum_{j=1}^2\left\|\btheta_j-\tbtheta_j\right\|_{\bC_j^{-1}}^2
+\lambda\left\|\bA^{-1/2}\left(\btheta_1-\btheta_2\right)\right\|_2.
\]
Strict convexity guarantees a unique minimizer. With
$\bu=\bA^{-1/2}\left(\widehat\btheta_1^\lambda-\widehat\btheta_2^\lambda\right)$,
its optimality conditions are
\[
\bC_1^{-1}\left(\widehat\btheta_1^\lambda-\tbtheta_1\right)+\lambda\bA^{-1/2}\boldsymbol\eta=\bm0,
\qquad
\bC_2^{-1}\left(\widehat\btheta_2^\lambda-\tbtheta_2\right)-\lambda\bA^{-1/2}\boldsymbol\eta=\bm0,
\qquad
\boldsymbol\eta\in\partial\left\|\bu\right\|_2.
\]
Subtracting the first two solutions and multiplying by $\bA^{-1/2}$ gives
$\bu=\bz-\lambda\boldsymbol\eta$. If $\bu\ne\bm0$, then
$\boldsymbol\eta=\bu/\left\|\bu\right\|_2$, and hence
\[
\bz=\left(1+\frac{\lambda}{\left\|\bu\right\|_2}\right)\bu,
\qquad \left\|\bz\right\|_2=\left\|\bu\right\|_2+\lambda,
\qquad \bu=\left(1-\frac{\lambda}{\left\|\bz\right\|_2}\right)\bz.
\]
In particular, $\left\|\bz\right\|_2>\lambda$.
If $\bu=\bm0$, the subgradient condition is equivalent to $\left\|\bz\right\|_2\le\lambda$.
Thus, defining $t=\min\left[\lambda/\left\|\bz\right\|_2,1\right]$ for $\bz\ne\bm0$ and $t=1$ for $\bz=\bm0$, we have $\lambda\boldsymbol\eta=t\bz$. Consequently,
\[
\begin{aligned}
\widehat\btheta_1^\lambda
&=\tbtheta_1-t\bC_1\bA^{-1}\left(\tbtheta_1-\tbtheta_2\right)
=\tbtheta_1+t\bW_2\left(\tbtheta_2-\tbtheta_1\right),\\
\widehat\btheta_2^\lambda
&=\tbtheta_2+t\bC_2\bA^{-1}\left(\tbtheta_1-\tbtheta_2\right)
=\tbtheta_2+t\bW_1\left(\tbtheta_1-\tbtheta_2\right).
\end{aligned}
\]
Here $\bC_1\bA^{-1}=\bW_2$ and $\bC_2\bA^{-1}=\bW_1$ by the parallel-sum identity. Substituting
$\overline\btheta=\tbtheta_1+\bW_2\left(\tbtheta_2-\tbtheta_1\right)
=\tbtheta_2+\bW_1\left(\tbtheta_1-\tbtheta_2\right)$ gives the other displayed representations. The two estimates coincide exactly when $\bu=\bm0$, equivalently $\left\|\bz\right\|_2\le\lambda$.

The same optimality conditions also yield the displacement bound used in Remark~\ref{rmk:regularization}:
\[
\left\|\bC_j^{-1/2}\left(\widehat\btheta_j^\lambda-\tbtheta_j\right)\right\|_2
=\lambda\left\|\bC_j^{1/2}\bA^{-1/2}\boldsymbol\eta\right\|_2\le\lambda,
\qquad
\left\|\widehat\btheta_j^\lambda-\tbtheta_j\right\|_2
\le\lambda\sqrt{\frac{\lambda_{\max}\left(\bSigma_j\right)}{n_j}},
\]
since $\left\|\boldsymbol\eta\right\|_2\le1$ and $\bA^{-1/2}\bC_j\bA^{-1/2}\preceq\bI_p$.
\end{proof}

\begin{proof}[Proof of Theorem~\ref{thm:JS}]
Under the homoscedastic assumptions in the theorem,
\[
\tbtheta_1
\sim
\mathcal N
\left(
 \btheta_1^{\star},
 \frac{\sigma_1^2}{n_1}\bI_p
\right),
\qquad
\tbtheta_2
\sim
\mathcal N
\left(
 \btheta_2^{\star},
 \frac{\sigma_2^2}{n_2}\bI_p
\right),
\]
and the two vectors are independent. We first establish an upper bound for the inverse moment appearing in the James--Stein risk. Define the noncentral chi-square variable
$
U_{\rm nc}
=
\frac{n_1}{\sigma_1^2}
\left\|\tbtheta_1\right\|_2^2
$
and the noncentrality parameter
$
\Lambda_{\rm nc}
=
\frac{n_1}{\sigma_1^2}
\left\|\btheta_1^{\star}\right\|_2^2
$.
The moment generating function of \(U_{\rm nc}\) is
\begin{align}
\mathbb E
\left[
 \exp\left(tU_{\rm nc}\right)
\right]
&=
\left(
 1-2t
\right)^{-p/2}
\exp\left(
 \frac{\Lambda_{\rm nc}t}{1-2t}
\right)
\notag\\
&=
\exp\left(
 -\frac{\Lambda_{\rm nc}}{2}
\right)
\sum_{k=0}^{\infty}
\frac{\left(\Lambda_{\rm nc}/2\right)^k}{k!}
\left(
 1-2t
\right)^{-\left(p+2k\right)/2},
\qquad t<\frac12.
\label{eq:JS-noncentral-mixture-mgf}
\end{align}
The second equality follows from
$
\frac{\Lambda_{\rm nc}t}{1-2t}
=
-\frac{\Lambda_{\rm nc}}2
+
\frac{\Lambda_{\rm nc}}{2\left(1-2t\right)}
$.
Equation \eqref{eq:JS-noncentral-mixture-mgf} identifies the exact Poisson mixture representation
\[
U_{\rm nc}\mid N_{\rm mix}=k
\sim
\chi_{p+2k}^2,
\qquad
N_{\rm mix}
\sim
\operatorname{Poisson}
\left(
 \frac{\Lambda_{\rm nc}}2
\right).
\]
For every \(d>2\), the chi-square density and the identity \(\Gamma\left(d/2\right)=\left(d/2-1\right)\Gamma\left(d/2-1\right)\) give
\begin{align*}
\mathbb E
\left[
 \left(\chi_d^2\right)^{-1}
\right]
=
\frac{1}{2^{d/2}\Gamma\left(d/2\right)}
\int_0^{\infty}
 x^{d/2-2}e^{-x/2}\,\mathrm dx=
\frac{2^{d/2-1}\Gamma\left(d/2-1\right)}
{2^{d/2}\Gamma\left(d/2\right)}
=
\frac1{d-2}.
\end{align*}
It follows that
\begin{equation}
\mathbb E
\left[
 U_{\rm nc}^{-1}
\right]
=
\mathbb E
\left[
 \frac1{p+2N_{\rm mix}-2}
\right].
\label{eq:JS-mixture-inverse-moment}
\end{equation}
Define the two inverse moments by
\[
\mathcal I_{\rm inv}
=
\mathbb E
\left[
 \frac1{p+2N_{\rm mix}-2}
\right],
\quad
\mathcal J_{\rm inv}
=
\mathbb E
\left[
 \frac1{p+2N_{\rm mix}}
\right].
\]
The Poisson identity $\mathbb E \left[ N_{\rm mix}f\left(N_{\rm mix}\right) \right] = \frac{\Lambda_{\rm nc}}2 \mathbb E \left[ f\left(N_{\rm mix}+1\right) \right]$ with \(f\left(k\right)=\left(p+2k-2\right)^{-1}\) yields
\begin{align}
1
&=
\mathbb E
\left[
 \frac{p+2N_{\rm mix}-2}{p+2N_{\rm mix}-2}
\right]
=
\left(
 p-2
\right)
\mathcal I_{\rm inv}
+
2\mathbb E
\left[
 \frac{N_{\rm mix}}{p+2N_{\rm mix}-2}
\right]
=
\left(
 p-2
\right)
\mathcal I_{\rm inv}
+
\Lambda_{\rm nc}\mathcal J_{\rm inv}.
\label{eq:JS-Poisson-inverse-identity}
\end{align}
For every integer \(k\ge0\),
\begin{align*}
\frac1{p+2k}
-
\frac{p-2}{p}
\frac1{p+2k-2}
&=
\frac{4k}
{p\left(p+2k\right)\left(p+2k-2\right)}
\ge0.
\end{align*}
Thus
\(
\mathcal J_{\rm inv}
\ge
\left(p-2\right)\mathcal I_{\rm inv}/p
\).
Substitution into \eqref{eq:JS-Poisson-inverse-identity} gives
\begin{align}
1
&\ge
\left(
 p-2
\right)
\left(
 1+\frac{\Lambda_{\rm nc}}p
\right)
\mathcal I_{\rm inv},
\notag\\
\mathcal I_{\rm inv}
&\le
\frac{p}
{\left(p-2\right)\left(p+\Lambda_{\rm nc}\right)}.
\label{eq:JS-Poisson-inverse-bound}
\end{align}
Combining \eqref{eq:JS-mixture-inverse-moment} and \eqref{eq:JS-Poisson-inverse-bound} proves
\begin{equation}
\mathbb E
\left[
 \left\|\tbtheta_1\right\|_2^{-2}
\right]
\le
\frac{p}
{\left(p-2\right)
\left[
 p\sigma_1^2/n_1
 +
 \left\|\btheta_1^{\star}\right\|_2^2
\right]}.
\label{eq:JS-inverse-moment-final}
\end{equation}

For an arbitrary \(s>0\), Lemma~\ref{lem:sure} gives the exact James--Stein risk identity
\begin{align}
\mathbb E
\left[
 \left\|
 \widehat\btheta_{{\rm JS},s}-\btheta_1^{\star}
 \right\|_2^2
\right]
=
\frac{p\sigma_1^2}{n_1}
+
\left[
 s^2
 -
 \frac{2\sigma_1^2}{n_1}
 \left(
  p-2
 \right)s
\right]
\mathbb E
\left[
 \left\|\tbtheta_1\right\|_2^{-2}
\right].
\label{eq:JS-risk-exact-proof}
\end{align}
Completing the square gives
\begin{align}
s^2
-
\frac{2\sigma_1^2}{n_1}
\left(
 p-2
\right)s
=
\left[
 s-
 \frac{\sigma_1^2}{n_1}
 \left(
  p-2
 \right)
\right]^2
-
\frac{\sigma_1^4}{n_1^2}
\left(
 p-2
\right)^2
\ge
-
\frac{\sigma_1^4}{n_1^2}
\left(
 p-2
\right)^2.
\label{eq:JS-square-completion}
\end{align}
Equations \eqref{eq:JS-risk-exact-proof}, \eqref{eq:JS-square-completion}, and \eqref{eq:JS-inverse-moment-final} imply the lower bound, uniform over \(s>0\),
\begin{align}
\mathbb E
\left[
 \left\|
 \widehat\btheta_{{\rm JS},s}-\btheta_1^{\star}
 \right\|_2^2
\right]
&\ge
\frac{p\sigma_1^2}{n_1}
-
\frac{p\sigma_1^4\left(p-2\right)}
{n_1^2
\left[
 p\sigma_1^2/n_1
 +
 \left\|\btheta_1^{\star}\right\|_2^2
\right]}.
\label{eq:JS-risk-lower-uniform}
\end{align}

Under the same spherical covariance model, the upper bound in \eqref{eq:bound-shrinkage-simple} can be rewritten without the weight \(w\). Indeed,
\begin{align*}
w
&=
\frac{n_2\sigma_2^{-2}}
{n_1\sigma_1^{-2}+n_2\sigma_2^{-2}}
=
\frac{\sigma_1^2/n_1}
{\sigma_1^2/n_1+\sigma_2^2/n_2},
\end{align*}
and division of the numerator and denominator in \eqref{eq:bound-shrinkage-simple} by \(w\) gives
\begin{align}
\mathbb E
\left[
 \left\|
 \widehat\btheta_{\underline s}-\btheta_1^{\star}
 \right\|_2^2
\right]
\le
\frac{p\sigma_1^2}{n_1}
-
\frac{\sigma_1^4\left(p-2\right)^2/n_1^2}
{p\left(\sigma_1^2/n_1+\sigma_2^2/n_2\right)
+
\left\|\btheta_2^{\star}-\btheta_1^{\star}\right\|_2^2}.
\label{eq:proposed-risk-upper-JS-comparison}
\end{align}
The first sample size condition in the theorem and the assumed discrepancy bound imply
\begin{align*}
&
\left(
 p-2
\right)
\left\|\btheta_1^{\star}\right\|_2^2
-
p
\left\|\btheta_2^{\star}-\btheta_1^{\star}\right\|_2^2
-
\frac{2p\sigma_1^2}{n_1}
>0.
\end{align*}
The second sample size condition is therefore equivalent to
\begin{align}
\frac{p^2\sigma_2^2}{n_2}
&<
\left(
 p-2
\right)
\left\|\btheta_1^{\star}\right\|_2^2
-
p
\left\|\btheta_2^{\star}-\btheta_1^{\star}\right\|_2^2
-
\frac{2p\sigma_1^2}{n_1}.
\label{eq:JS-second-sample-equivalent}
\end{align}
After rearrangement, \eqref{eq:JS-second-sample-equivalent} becomes
\begin{align}
\left(
 p-2
\right)
\left[
 p\sigma_1^2/n_1
 +
 \left\|\btheta_1^{\star}\right\|_2^2
\right]
>
 p
\left[
 p\left(
 \sigma_1^2/n_1+\sigma_2^2/n_2
 \right)
 +
 \left\|\btheta_2^{\star}-\btheta_1^{\star}\right\|_2^2
\right].
\label{eq:JS-denominator-comparison}
\end{align}
All quantities in the two square brackets in \eqref{eq:JS-denominator-comparison} are positive. Dividing both sides by their product and multiplying by \(\sigma_1^4\left(p-2\right)/n_1^2\) yields
\begin{align}
\frac{\sigma_1^4\left(p-2\right)^2/n_1^2}
{p\left(\sigma_1^2/n_1+\sigma_2^2/n_2\right)
+
\left\|\btheta_2^{\star}-\btheta_1^{\star}\right\|_2^2}
>
\frac{p\sigma_1^4\left(p-2\right)/n_1^2}
{p\sigma_1^2/n_1
+
\left\|\btheta_1^{\star}\right\|_2^2}.
\label{eq:JS-risk-improvement-comparison}
\end{align}
The strict inequality \eqref{eq:JS-risk-improvement-comparison}, together with the upper bound \eqref{eq:proposed-risk-upper-JS-comparison} and the lower bound \eqref{eq:JS-risk-lower-uniform},~proves
\[
\mathbb E
\left[
 \left\|
 \widehat\btheta_{\underline s}-\btheta_1^{\star}
 \right\|_2^2
\right]
<
\mathbb E
\left[
 \left\|
 \widehat\btheta_{{\rm JS},s}-\btheta_1^{\star}
 \right\|_2^2
\right]
\]
for every \(s>0\).

We next verify the sufficient condition for the non-diagonal comparison stated in Section~\ref{sec:relationship}. The required distributional identities are
\[
\tbtheta_1
\sim
\mathcal N\left(\btheta_1^\star,n_1^{-1}\bSigma_1\right),
\qquad
n_1\tbtheta_1^\top\bSigma_1^{-1}\tbtheta_1
\sim
\chi_p^2\left(n_1\btheta_1^{\star\top}\bSigma_1^{-1}\btheta_1^\star\right).
\]
Applying Lemma~\ref{lem:sure} gives
\begin{align*}
\mathbb E\left\|\widehat\btheta_{{\rm JS},s}-\btheta_1^\star\right\|_2^2
&=
\frac{\Tr\left(\bSigma_1\right)}{n_1}
-
\frac{2s\Tr\left(\bSigma_1\right)}{n_1}
\mathbb E\left[\left(n_1\tbtheta_1^\top\bSigma_1^{-1}\tbtheta_1\right)^{-1}\right]
+
\left(s^2+4s\right)
\mathbb E\left[
\frac{\left\|\tbtheta_1\right\|_2^2}
{\left(n_1\tbtheta_1^\top\bSigma_1^{-1}\tbtheta_1\right)^2}
\right]
\\
&=
\frac{\Tr\left(\bSigma_1\right)}{n_1}
+
\frac{s\left[\left(s+4\right)\lambda_{\min}\left(\bSigma_1\right)-2\Tr\left(\bSigma_1\right)\right]}{n_1}
\mathbb E\left[\left(n_1\tbtheta_1^\top\bSigma_1^{-1}\tbtheta_1\right)^{-1}\right]
\\
&\quad+
\left(s^2+4s\right)
\mathbb E\left[
\frac{
\tbtheta_1^\top
\left(\bI_p-\lambda_{\min}\left(\bSigma_1\right)\bSigma_1^{-1}\right)
\tbtheta_1
}
{\left(n_1\tbtheta_1^\top\bSigma_1^{-1}\tbtheta_1\right)^2}
\right].
\end{align*}
Since $
\bI_p-\lambda_{\min}\left(\bSigma_1\right)\bSigma_1^{-1}
\succeq\bm0$, the last expectation is nonnegative. Moreover, 
\begin{align*}
&s\left[\left(s+4\right)\lambda_{\min}\left(\bSigma_1\right)-2\Tr\left(\bSigma_1\right)\right]
\\
&\quad=
\lambda_{\min}\left(\bSigma_1\right)
\left[
 s-
 \frac{\Tr\left(\bSigma_1\right)-2\lambda_{\min}\left(\bSigma_1\right)}
 {\lambda_{\min}\left(\bSigma_1\right)}
\right]^2
-
\frac{\left[\Tr\left(\bSigma_1\right)-2\lambda_{\min}\left(\bSigma_1\right)\right]^2}
{\lambda_{\min}\left(\bSigma_1\right)},\\
&\quad\ge
-
\frac{\left[\Tr\left(\bSigma_1\right)-2\lambda_{\min}\left(\bSigma_1\right)\right]^2}
{\lambda_{\min}\left(\bSigma_1\right)},
\end{align*}
while the inverse moment satisfies
\[
\mathbb E\left[\left(n_1\tbtheta_1^\top\bSigma_1^{-1}\tbtheta_1\right)^{-1}\right]
\le
\frac{p}
{\left(p-2\right)
\left[p+n_1\btheta_1^{\star\top}\bSigma_1^{-1}\btheta_1^\star\right]}.
\]
Combining these bounds yields
\begin{align*}
\mathbb E\left\|\widehat\btheta_{{\rm JS},s}-\btheta_1^\star\right\|_2^2
&\ge
\frac{\Tr\left(\bSigma_1\right)}{n_1}
-
\frac{\left[\Tr\left(\bSigma_1\right)-2\lambda_{\min}\left(\bSigma_1\right)\right]^2}
{n_1\lambda_{\min}\left(\bSigma_1\right)}
\frac{p}{p-2}
\frac{1}
{p+n_1\btheta_1^{\star\top}\bSigma_1^{-1}\btheta_1^\star},
\qquad
s>0.
\end{align*}
The target-only risk decomposes as
\[
\frac{\Tr\left(\bSigma_1\right)}{n_1}
=
\left\langle
\left(n_1\bSigma_1^{-1}+n_2\bSigma_2^{-1}\right)^{-1},
\bI_p
\right\rangle
+
\Tr\left(\bS\right),
\]
Theorem~\ref{thm:risk-bound-shat} and the displayed sufficient condition in Section~\ref{sec:relationship} now give
\begin{align*}
\mathbb E\left\|\widehat\btheta_{\underline s}-\btheta_1^\star\right\|_2^2
&\le
\frac{\Tr\left(\bSigma_1\right)}{n_1}
-
\frac{\Tr\left(\bS\right)^2}
{\Tr\left(\bS\right)+\left\|\bW_2\left(\btheta_2^\star-\btheta_1^\star\right)\right\|_2^2}
+
4\left\|\bS\right\|_2
\\
&<
\frac{\Tr\left(\bSigma_1\right)}{n_1}
-
\frac{\left[\Tr\left(\bSigma_1\right)-2\lambda_{\min}\left(\bSigma_1\right)\right]^2}
{n_1\lambda_{\min}\left(\bSigma_1\right)}
\frac{p}{p-2}
\frac{1}
{p+n_1\btheta_1^{\star\top}\bSigma_1^{-1}\btheta_1^\star}
\\
&\le
\mathbb E\left\|\widehat\btheta_{{\rm JS},s}-\btheta_1^\star\right\|_2^2,
\qquad
s>0.
\end{align*}
\end{proof}

\subsection{Proof of the Results for Multiple Sets}\label{sec:proof-multiset}

\begin{proof}[Proof of Remark~\ref{rmk:illustration-eff-dim}]
For $\bV\in\mathscr V_r$, let $x_1,\ldots,x_p\in\left[1/r,1\right]$ be the eigenvalues of $n\bV$. Since $\bQ=\bSigma_j=\bI_p$ and $n_j=n$, the eigenvalues of $\bS_{j,\bV}=\bV\left(\bV+n^{-1}\bI_p\right)^{-1}\bV$ are
$n^{-1}f\left(x_k\right)$, where $f\left(x\right)=x^2/\left(1+x\right)$ is increasing on $\left(0,\infty\right)$. Hence
\[
\frac{\Tr\left(\bS_{j,\bV}\right)}{\left\|\bS_{j,\bV}\right\|_2}
=\frac{\sum_{k=1}^p f\left(x_k\right)}{\max_k f\left(x_k\right)}
\ge1+\left(p-1\right)\frac{f\left(1/r\right)}{f\left(1\right)}
=1+\frac{2\left(p-1\right)}{r\left(r+1\right)}.
\]
Equality holds for $n\bV=\operatorname{diag}\left(1,1/r,\ldots,1/r\right)$. Taking the minimum over $r\in\left[m\right]$ gives
$d_{\rm eff}=1+2\left(p-1\right)/\left[m\left(m+1\right)\right]$, which is of order $p$ when $m=O\left(1\right)$.
\end{proof}

To prove Theorems \ref{thm:seq-large-separation-trace-condition}, \ref{thm:risk-upper-general}, and \ref{thm:plugin-sequential}, we first establish the following result for the sequential~\mbox{algorithm}.

\begin{theorem}\label{thm:sequential-unified}
Let \(\mathcal I_0\) and \(\mathcal I_1\) be nonempty disjoint subsets of
\(\left\{2,\ldots,m\right\}\) whose union is \(\left\{2,\ldots,m\right\}\).
For each \(j\in\left[m\right]\), let \(\tbtheta_j^{\circ}\in\mathbb R^p\) and let
\(\bC_j^{\circ}\) be positive definite. Apply Algorithm~\ref{alg:sequential-shrinkage}
to \(\left\{\tbtheta_j^{\circ},\bC_j^{\circ}\right\}_{j=1}^m\), where the covariance
input \(n_j^{-1}\bSigma_j\) in the algorithm is replaced by \(\bC_j^{\circ}\).
Write the resulting estimator, working covariance, selected index, multiplier, and score as
$
\widehat\btheta^{\circ\left[\ell\right]}$, 
$\bV^{\circ\left[\ell\right]}$,$j_{\ell}^{\circ}$,  
$\gamma_j^{\circ\left[\ell\right]}$, and
$\mathfrak S_j^{\circ\left[\ell\right]}$, respectively.
For every \(J\subseteq\mathcal I_0\), define
\begin{align*}
\bV_J^{\circ}
=
\left[
\left(\bC_1^{\circ}\right)^{-1}
+
\sum_{k\in J}
\left(\bC_k^{\circ}\right)^{-1}
\right]^{-1}, \quad
\overline\btheta_J^{\circ}
=
\bV_J^{\circ}
\left[
\left(\bC_1^{\circ}\right)^{-1}\tbtheta_1^{\circ}
+
\sum_{k\in J}
\left(\bC_k^{\circ}\right)^{-1}\tbtheta_k^{\circ}
\right].
\end{align*}
For \(j\notin J\cup\left\{1\right\}\), define
\begin{align*}
\bW_{j,J}^{\circ}
=
\bV_J^{\circ}
\left(\bV_J^{\circ}+\bC_j^{\circ}\right)^{-1}, \quad
\boldsymbol d_{j,J}^{\circ}
=
\bW_{j,J}^{\circ}
\left(\tbtheta_j^{\circ}-\overline\btheta_J^{\circ}\right), \quad
\bS_{j,J}^{\circ}
=
\bQ^{1/2}
\bW_{j,J}^{\circ}
\bV_J^{\circ}
\bQ^{1/2}.
\end{align*}
Moreover, let
\begin{align*}
\mathcal R_n^{\circ}
&=
\max_{J\subsetneq\mathcal I_0}
\max_{j\in\mathcal I_0\setminus J}
\left[
\left\|\tbtheta_j^{\circ}-\overline\btheta_J^{\circ}\right\|_{\bQ}
+
\Tr\left(\bS_{j,J}^{\circ}\right)^{1/2}
\right],
\\
\mathcal Z_n^{\circ}
&=
\max_{J\subsetneq\mathcal I_0}
\max_{j\in\mathcal I_0\setminus J}
\frac{
\left|
\left\|\boldsymbol d_{j,J}^{\circ}\right\|_{\bQ}^{2}
-
\Tr\left(\bS_{j,J}^{\circ}\right)
\right|
+
\left\|\bS_{j,J}^{\circ}\right\|_2
}{
\Tr\left(\bS_{j,J}^{\circ}\right)
},
\\
\mathcal U_n^{\circ}
&=
\max_{1\le j\le m}
\left\|\tbtheta_j^{\circ}-\btheta_j^{\star}\right\|_{\bQ}
+
\max_{J\subseteq\mathcal I_0}
\left\|\overline\btheta_J^{\circ}-\btheta_1^{\star}\right\|_{\bQ}.
\end{align*}
Assume the following conditions. 
\begin{itemize}
	\item \(m=O\left(1\right)\),  \(0<q_-=\lambda_{\min}\left(\bQ\right) \leq \lambda_{\max}\left(\bQ\right)= q_+<\infty\), and fixed constants \(a_->0\) and \(a_+<\infty\) satisfy
	$
	\frac{a_-}{n}\bI_p
	\preceq
	\bC_j^{\circ}
	\preceq
	\frac{a_+}{n}\bI_p
	$
	for every realization and every \(j\in\left[m\right]\).
	\item 
	$\frac{\delta_{\min,n}}{\sqrt{p/n}}
	\to
	\infty
	$, where $
	\delta_{\min,n}
	=
	\min_{j\in\mathcal I_1}
	\left\|\btheta_j^{\star}-\btheta_1^{\star}\right\|_{\bQ}$.
	\item There exist deterministic
	sequences \(\varepsilon_n=o\left(1\right)\) and \(\pi_n=o\left(1\right)\), a fixed constant \(k_0<\infty\), and events
	\(\mathcal G_n^{\circ}\) such that
$
	\mathbb P\left(\mathcal G_n^{\circ}\right)
	\ge
	1-\pi_n
$
	and, on \(\mathcal G_n^{\circ}\),
	\[
	\mathcal R_n^{\circ}
	\le
	k_0\sqrt{\frac{p}{n}},
	\qquad
	\mathcal Z_n^{\circ}
	\le
	\varepsilon_n,
	\qquad
	\mathcal U_n^{\circ}
	\le
	k_0\sqrt{\frac{p}{n}}.
	\]
	\item There exists  a fixed constant \(k_4<\infty\)  such that $
	\mathbb E
	\left\|
	\tbtheta_1^{\circ}
	-
	\overline\btheta_{\mathcal I_0}^{\circ}
	\right\|_{\bQ}^{4}
	\le
	k_4
	\left(\frac{p}{n}\right)^2
	$.
\end{itemize}

Then the following statements hold.
For all sufficiently large \(n\),
\[
\mathbb P
\left(
\left\{j_1^{\circ},\ldots,j_{m_0}^{\circ}\right\}
=
\mathcal I_0
\right)
\ge
1-\pi_n, 
\]
where $
m_0
=
\left|\mathcal I_0\right|$, and
\[
\mathbb E
\left\|
\widehat\btheta_{\rm seq}^{\circ}
-
\overline\btheta_{\mathcal I_0}^{\circ}
\right\|_{\bQ}^{2}
=
o\left(\frac{p}{n}\right).
\]

\end{theorem}

\begin{proof}
Throughout the proof, $c,C>0$ depend only on the fixed bounds for $m$, $a_-$, $a_+$, $q_-$, $q_+$, $k_0$, and $k_4$, and may change between occurrences. All bounds are uniform over the indicated subsets and candidate sources. Write $\left\|\bA\right\|_{\bQ\to\bQ}=\left\|\bQ^{1/2}\bA\bQ^{-1/2}\right\|_2$. Since $\mathcal Z_n^\circ\ge1/p$, the calibration assumption implies $p\to\infty$. We work with sufficiently large $n$ and $p$.

For a positive definite matrix $\bA$, write
\[
\bW_j\left(\bA\right)=\bA\left(\bA+\bC_j^\circ\right)^{-1},\qquad
\mathcal F_j\left(\bA\right)=\left[\bA^{-1}+\left(\bC_j^\circ\right)^{-1}\right]^{-1}.
\]
The covariance recursion satisfies
\[
\bV^{\circ\left[\ell+1\right]}
=\mathcal F_j\left(\bV^{\circ\left[\ell\right]}\right)
+\left(1-\gamma_j^{\circ\left[\ell\right]}\right)^2
\left[\bV^{\circ\left[\ell\right]}-\mathcal F_j\left(\bV^{\circ\left[\ell\right]}\right)\right],
\qquad j=j_{\ell+1}^\circ.
\]
Since this is a convex combination of $\bV^{\circ\left[\ell\right]}$ and $\mathcal F_j\left(\bV^{\circ\left[\ell\right]}\right)$, induction gives
\begin{equation}
\frac{a_-}{\left(\ell+1\right)n}\bI_p\preceq\bV^{\circ\left[\ell\right]}\preceq\frac{a_+}{n}\bI_p,
\qquad
\frac{a_-}{\left(\left|J\right|+1\right)n}\bI_p\preceq\bV_J^\circ\preceq\frac{a_+}{n}\bI_p.
\label{eq:compressed-sequential-covariances}
\end{equation}
For every $\bA$ in these spectral ranges,
$\bW_j\left(\bA\right)^{-1}=\bI_p+\bC_j^\circ\bA^{-1}$ and
$\bW_j\left(\bA\right)\bA=\bA\left(\bA+\bC_j^\circ\right)^{-1}\bA$. These identities and spectral bounds give
\begin{equation}
\begin{gathered}
c\left\|\bx\right\|_{\bQ}\le\left\|\bW_j\left(\bA\right)\bx\right\|_{\bQ}\le C\left\|\bx\right\|_{\bQ},\qquad
\frac{c}{n}\bI_p\preceq\bQ^{1/2}\bW_j\left(\bA\right)\bA\bQ^{1/2}\preceq\frac{C}{n}\bI_p,\\
c\frac pn\le\underline s_j\left(\bA\right)\le\widehat s_j\left(\bA\right)\le C\frac pn,
\end{gathered}
\label{eq:compressed-sequential-geometry}
\end{equation}
where the two scales are defined from $\bQ^{1/2}\bW_j\left(\bA\right)\bA\bQ^{1/2}$ as in the algorithm. The lower matrix bound follows from
$\lambda_{\min}\left[\bA\left(\bA+\bC_j^\circ\right)^{-1}\bA\right]\ge\lambda_{\min}\left(\bA\right)^2/\lambda_{\max}\left(\bA+\bC_j^\circ\right)$.

We next establish the homogeneous-source recursion and source prioritization jointly on $\mathcal G_n^\circ$. Let $J_\ell^\circ=\left\{j_1^\circ,\ldots,j_\ell^\circ\right\}$ and, whenever $J_\ell^\circ\subseteq\mathcal I_0$, define
\[
\mathcal E_{\theta,\ell}^\circ
=\frac{\left\|\widehat\btheta^{\circ\left[\ell\right]}-\overline\btheta_{J_\ell^\circ}^\circ\right\|_{\bQ}}{\mathcal R_n^\circ},\qquad
\mathcal E_{V,\ell}^\circ=n\left\|\bV^{\circ\left[\ell\right]}-\bV_{J_\ell^\circ}^\circ\right\|_2,\qquad
\mathcal E_\ell^\circ=\mathcal E_{\theta,\ell}^\circ+\mathcal E_{V,\ell}^\circ.
\]
The denominator is positive by \eqref{eq:compressed-sequential-geometry}. For brevity within one step, set $J=J_\ell^\circ$, $\bA=\bV^{\circ\left[\ell\right]}$, and $\bB=\bV_J^\circ$. The identities
\begin{equation}
\begin{aligned}
\bW_j\left(\bA\right)-\bW_j\left(\bB\right)
&=\bC_j^\circ\left(\bB+\bC_j^\circ\right)^{-1}\left(\bA-\bB\right)\left(\bA+\bC_j^\circ\right)^{-1},\\
\mathcal F_j\left(\bA\right)-\mathcal F_j\left(\bB\right)
&=\bC_j^\circ\left(\bB+\bC_j^\circ\right)^{-1}\left(\bA-\bB\right)\left(\bA+\bC_j^\circ\right)^{-1}\bC_j^\circ
\end{aligned}
\label{eq:compressed-sequential-resolvents}
\end{equation}
imply the perturbation bounds
\[
\left\|\bW_j\left(\bA\right)-\bW_j\left(\bB\right)\right\|_{\bQ\to\bQ}\le C\mathcal E_{V,\ell}^\circ,
\qquad
\left\|\mathcal F_j\left(\bA\right)-\mathcal F_j\left(\bB\right)\right\|_2\le C\left\|\bA-\bB\right\|_2.
\]
For every $j\in\mathcal I_0\setminus J$, let
$\bD_j=\bW_j\left(\bA\right)\left(\tbtheta_j^\circ-\widehat\btheta^{\circ\left[\ell\right]}\right)$.
Subtracting $\boldsymbol d_{j,J}^\circ=\bW_j\left(\bB\right)\left(\tbtheta_j^\circ-\overline\btheta_J^\circ\right)$ gives
\begin{equation}
\begin{aligned}
\bD_j-\boldsymbol d_{j,J}^\circ
&=\left[\bW_j\left(\bA\right)-\bW_j\left(\bB\right)\right]\left(\tbtheta_j^\circ-\overline\btheta_J^\circ\right)
+\bW_j\left(\bA\right)\left(\overline\btheta_J^\circ-\widehat\btheta^{\circ\left[\ell\right]}\right),\\
\left\|\bD_j-\boldsymbol d_{j,J}^\circ\right\|_{\bQ}
&\le C\mathcal E_\ell^\circ\mathcal R_n^\circ,
\qquad
\left\|\bS_j^{\circ\left[\ell\right]}-\bS_{j,J}^\circ\right\|_2\le\frac{C}{n}\mathcal E_{V,\ell}^\circ.
\end{aligned}
\label{eq:compressed-sequential-perturbations}
\end{equation}
Using $\mathcal R_n^\circ\le k_0\sqrt{p/n}$, $\mathcal Z_n^\circ\le\varepsilon_n$, and
$\left|\left\|\bu\right\|_{\bQ}^2-\left\|\bv\right\|_{\bQ}^2\right|\le\left\|\bu-\bv\right\|_{\bQ}\left(2\left\|\bv\right\|_{\bQ}+\left\|\bu-\bv\right\|_{\bQ}\right)$, we obtain, whenever $\mathcal E_\ell^\circ\le1$,
\begin{equation}
\begin{aligned}
\left|\left\|\bD_j\right\|_{\bQ}^2-\Tr\left(\bS_j^{\circ\left[\ell\right]}\right)\right|
+\left\|\bS_j^{\circ\left[\ell\right]}\right\|_2
&\le C\left(\varepsilon_n+\mathcal E_\ell^\circ\right)\frac pn,\\
0\le1-\gamma_j^{\circ\left[\ell\right]}
&=\frac{\left(\left\|\bD_j\right\|_{\bQ}^2-\widehat s_j^{\circ\left[\ell\right]}\right)_+}
{\left\|\bD_j\right\|_{\bQ}^2\vee\widehat s_j^{\circ\left[\ell\right]}}
\le C\left(\varepsilon_n+\mathcal E_\ell^\circ\right).
\end{aligned}
\label{eq:compressed-sequential-calibration}
\end{equation}

Suppose inductively that $J\subseteq\mathcal I_0$ and $\mathcal E_\ell^\circ\le K_\ell\varepsilon_n$, with $K_\ell$ fixed. For a homogeneous candidate, the score definition gives the exact identity
\begin{align*}
\mathfrak S_j^{\circ\left[\ell\right]}-\Tr\left(\bS_j^{\circ\left[\ell\right]}\right)
&=-\left(1-\gamma_j^{\circ\left[\ell\right]}\right)^2\Tr\left(\bS_j^{\circ\left[\ell\right]}\right)
-4\gamma_j^{\circ\left[\ell\right]}\left\|\bS_j^{\circ\left[\ell\right]}\right\|_2\\
&\quad-\left(\gamma_j^{\circ\left[\ell\right]}\right)^2
\left[\left\|\bD_j\right\|_{\bQ}^2-\Tr\left(\bS_j^{\circ\left[\ell\right]}\right)\right].
\end{align*}
Equations~\eqref{eq:compressed-sequential-geometry} and \eqref{eq:compressed-sequential-calibration} therefore imply, uniformly in $j$,
\[
\left|\mathfrak S_j^{\circ\left[\ell\right]}-\Tr\left(\bS_j^{\circ\left[\ell\right]}\right)\right|
\le C\left(\varepsilon_n+\mathcal E_\ell^\circ\right)\frac pn,
\qquad
\mathfrak S_j^{\circ\left[\ell\right]}\ge c\frac pn
\]
for sufficiently large $n$.
For every heterogeneous candidate, the definition of $\mathcal U_n^\circ$ gives
\[
\left\|\tbtheta_j^\circ-\widehat\btheta^{\circ\left[\ell\right]}\right\|_{\bQ}
\ge\delta_{\min,n}-\mathcal U_n^\circ-\mathcal E_{\theta,\ell}^\circ\mathcal R_n^\circ
\ge\frac{\delta_{\min,n}}2.
\]
Since $\delta_{\min,n}^2\gg p/n$, its multiplier is in the untruncated branch, and
\[
\mathfrak S_j^{\circ\left[\ell\right]}
=\frac{\widehat s_j^{\circ\left[\ell\right]}\left[\Tr\left(\bS_j^{\circ\left[\ell\right]}\right)-3\left\|\bS_j^{\circ\left[\ell\right]}\right\|_2\right]}
{\left\|\bW_j\left(\bA\right)\left(\tbtheta_j^\circ-\widehat\btheta^{\circ\left[\ell\right]}\right)\right\|_{\bQ}^2}
\le C\frac{\left(p/n\right)^2}{\delta_{\min,n}^2}=o\left(\frac pn\right).
\]
The maximizing source $j=j_{\ell+1}^\circ$ therefore belongs to $\mathcal I_0\setminus J$.

For this selected source, write $\gamma=\gamma_j^{\circ\left[\ell\right]}$ and $\boldsymbol e=\widehat\btheta^{\circ\left[\ell\right]}-\overline\btheta_J^\circ$. The exact mean and covariance recursions give
\begin{align*}
\widehat\btheta^{\circ\left[\ell+1\right]}-\overline\btheta_{J\cup\left\{j\right\}}^\circ
&=\left(\bI_p-\gamma\bW_j\left(\bA\right)\right)\boldsymbol e
\\
&\quad+\gamma\left[\bW_j\left(\bA\right)-\bW_j\left(\bB\right)\right]\left(\tbtheta_j^\circ-\overline\btheta_J^\circ\right)
+\left(\gamma-1\right)\boldsymbol d_{j,J}^\circ,\\
\bV^{\circ\left[\ell+1\right]}-\bV_{J\cup\left\{j\right\}}^\circ
&=\mathcal F_j\left(\bA\right)-\mathcal F_j\left(\bB\right)+\left(1-\gamma\right)^2\left[\bA-\mathcal F_j\left(\bA\right)\right].
\end{align*}
Consequently, \eqref{eq:compressed-sequential-geometry}--\eqref{eq:compressed-sequential-calibration} yield
\[
\mathcal E_{\theta,\ell+1}^\circ\le C\left(\mathcal E_\ell^\circ+\varepsilon_n\right),\qquad
\mathcal E_{V,\ell+1}^\circ\le C\mathcal E_{V,\ell}^\circ+C\left(\mathcal E_\ell^\circ+\varepsilon_n\right)^2,
\]
and hence $\mathcal E_{\ell+1}^\circ\le C\left(\mathcal E_\ell^\circ+\varepsilon_n\right)$ when $\mathcal E_\ell^\circ+\varepsilon_n\le1$. Fix one constant $C_{\rm rec}$ valid for the last recursion at every step, and choose $K_0=0$ and $K_{\ell+1}=C_{\rm rec}\left(K_\ell+1\right)$. Since $m$ is bounded and $\varepsilon_n\to0$, induction from $J_0^\circ=\varnothing$ and $\mathcal E_0^\circ=0$ proves on $\mathcal G_n^\circ$ that
\begin{equation}
\left\{j_1^\circ,\ldots,j_{m_0}^\circ\right\}=\mathcal I_0,
\qquad
\left\|\widehat\btheta^{\circ\left[m_0\right]}-\overline\btheta_{\mathcal I_0}^\circ\right\|_{\bQ}
\le C\varepsilon_n\sqrt{\frac pn}.
\label{eq:compressed-sequential-homogeneous-endpoint}
\end{equation}

Every remaining source is heterogeneous. If, after $r$ further updates, the distance to $\overline\btheta_{\mathcal I_0}^\circ$ is at most
$C\varepsilon_n\sqrt{p/n}+Cr\left(p/n\right)/\delta_{\min,n}$, then its distance from any remaining local estimator is at least
\[
\delta_{\min,n}-k_0\sqrt{p/n}-C\varepsilon_n\sqrt{p/n}-Cr\left(p/n\right)/\delta_{\min,n}
\ge\delta_{\min,n}/2.
\]
Its next update is therefore bounded by
\[
\left\|\gamma_j^{\circ\left[\ell\right]}\bW_j\left(\bV^{\circ\left[\ell\right]}\right)\left(\tbtheta_j^\circ-\widehat\btheta^{\circ\left[\ell\right]}\right)\right\|_{\bQ}
=\frac{\widehat s_j^{\circ\left[\ell\right]}}
{\left\|\bW_j\left(\bV^{\circ\left[\ell\right]}\right)\left(\tbtheta_j^\circ-\widehat\btheta^{\circ\left[\ell\right]}\right)\right\|_{\bQ}}
\le C\frac{p/n}{\delta_{\min,n}}.
\]
Induction over the bounded number of remaining sources, starting from \eqref{eq:compressed-sequential-homogeneous-endpoint}, gives
\begin{equation}
\left\|\widehat\btheta_{\rm seq}^\circ-\overline\btheta_{\mathcal I_0}^\circ\right\|_{\bQ}
\le C\left[\varepsilon_n\sqrt{\frac pn}+\frac{p/n}{\delta_{\min,n}}\right]
=o\left(\sqrt{\frac pn}\right)
\quad\text{on }\mathcal G_n^\circ.
\label{eq:compressed-sequential-good-event}
\end{equation}
For every realization, each update has norm at most $\sqrt{\widehat s_j^{\circ\left[\ell\right]}}\le C\sqrt{p/n}$. Thus
\[
\left\|\widehat\btheta_{\rm seq}^\circ-\overline\btheta_{\mathcal I_0}^\circ\right\|_{\bQ}
\le\left\|\tbtheta_1^\circ-\overline\btheta_{\mathcal I_0}^\circ\right\|_{\bQ}+C\sqrt{p/n},
\qquad
\mathbb E\left\|\widehat\btheta_{\rm seq}^\circ-\overline\btheta_{\mathcal I_0}^\circ\right\|_{\bQ}^4
\le C\left(\frac pn\right)^2.
\]
Splitting the expectation over $\mathcal G_n^\circ$ and its complement and applying Cauchy--Schwarz yields
\[
\mathbb E\left\|\widehat\btheta_{\rm seq}^\circ-\overline\btheta_{\mathcal I_0}^\circ\right\|_{\bQ}^2
\le C\left[\varepsilon_n^2+\frac{p/n}{\delta_{\min,n}^2}+\pi_n^{1/2}\right]\frac pn
=o\left(\frac pn\right).
\]
The prioritization probability follows from \eqref{eq:compressed-sequential-homogeneous-endpoint} and $\mathbb P\left(\mathcal G_n^\circ\right)\ge1-\pi_n$.
\end{proof}

\begin{proof}[Proof of Theorem~\ref{thm:seq-large-separation-trace-condition}]
Apply Theorem~\ref{thm:sequential-unified} with
\(\tbtheta_j^{\circ}=\tbtheta_j\) and
\(\bC_j^{\circ}=n_j^{-1}\bSigma_j\). Assumption~\ref{assump:seq-2} gives $\frac{v_-}{n}\bI_p \preceq \bC_j^{\circ} \preceq \frac{v_+}{n}\bI_p, \quad j\in\left[m\right]$, which verifies the deterministic covariance condition with
\(a_-=v_-\) and \(a_+=v_+\).

For \(J\subseteq\mathcal I_0\), define
\begin{equation}
\begin{aligned}
\bV_J
&=
\left(
 n_1\bSigma_1^{-1}
 +
 \sum_{k\in J}n_k\bSigma_k^{-1}
\right)^{-1},\\
\overline\btheta_J
&=
\bV_J
\left(
 n_1\bSigma_1^{-1}\tbtheta_1
 +
 \sum_{k\in J}n_k\bSigma_k^{-1}\tbtheta_k
\right),
\quad
\overline\btheta_J^\star
=
\bV_J
\left(
 n_1\bSigma_1^{-1}\btheta_1^\star
 +
 \sum_{k\in J}n_k\bSigma_k^{-1}\btheta_k^\star
\right).
\end{aligned}
\label{eq:g-subset-definitions}
\end{equation}
Let \(\bV_J^{\circ}=\bV_J\), \(\overline\btheta_J^{\circ}=\overline\btheta_J\), and 
\(\EE\left[\overline\btheta_J^{\circ}\right]=\overline\btheta_J^\star\) in Theorem~\ref{thm:sequential-unified}. The covariance bounds imply
\begin{equation}
\frac{v_-}{\left(\left|J\right|+1\right)n}\bI_p
\preceq
\bV_J
\preceq
\frac{v_+}{\left(\left|J\right|+1\right)n}\bI_p
\preceq
\frac{v_+}{n}\bI_p.
\label{eq:g-subset-covariance-bounds-new}
\end{equation}

Let
\(c_1^{\mathrm g}=\sqrt{q_+/q_-}\,m v_+/v_-\). Equations
\eqref{eq:g-subset-definitions} and \eqref{eq:g-subset-covariance-bounds-new} give
\begin{equation}
\left\|\overline\btheta_J^\star-\btheta_1^\star\right\|_{\bQ}
=
\left\|
\bV_J\sum_{k\in J}n_k\bSigma_k^{-1}
\left(\btheta_k^\star-\btheta_1^\star\right)\right\|_{\bQ}\le
c_1^{\mathrm g}
\max_{k\in\mathcal I_0}
\left\|\btheta_k^\star-\btheta_1^\star\right\|_{\bQ}
=
o\left(\sqrt{\frac{p}{n}}\right).
\label{eq:g-subset-bias}
\end{equation}
Independence across data sets and \eqref{eq:g-subset-definitions} also give
\begin{equation}
\operatorname{Cov}
\left(\overline\btheta_J-\overline\btheta_J^\star\right)
=
\bV_J
\left(n_1\bSigma_1^{-1}+\sum_{k\in J}n_k\bSigma_k^{-1}\right)
\bV_J
=
\bV_J.
\label{eq:g-subset-covariance}
\end{equation}

For $\bz\sim\mathcal N\left(\bm0,\bA\right)$, let $\bB=\bQ^{1/2}\bA\bQ^{1/2}$ and let $\lambda_k$ be its eigenvalues. Orthogonal diagonalization gives
$\left\|\bz\right\|_{\bQ}^2\stackrel{\mathrm d}=\sum_k\lambda_k G_k^2$ for independent $G_k\sim\mathcal N\left(0,1\right)$. Hence, with $X=\left\|\bz\right\|_{\bQ}^2-\Tr\left(\bB\right)$,
\begin{equation}
\begin{aligned}
\log\EE e^{uX}
&=\frac12\sum_k\left[-\log\left(1-2u\lambda_k\right)-2u\lambda_k\right]
\le\frac{u^2\Tr\left(\bB^2\right)}{1-2u\left\|\bB\right\|_2},
\quad 0<u<\left(2\left\|\bB\right\|_2\right)^{-1},\\
\log\EE e^{-uX}&\le u^2\Tr\left(\bB^2\right),\quad u>0.
\end{aligned}
\label{eq:g-gaussian-upper-mgf}
\end{equation}
Chernoff's inequality gives, for every $t>0$,
\begin{equation}
\PP\left[\left\|\bz\right\|_{\bQ}^2>
\Tr\left(\bQ\bA\right)+2\sqrt{\Tr\left[\left(\bQ^{1/2}\bA\bQ^{1/2}\right)^2\right]t}
+2\left\|\bQ^{1/2}\bA\bQ^{1/2}\right\|_2t\right]\le e^{-t},
\label{eq:g-gaussian-quadratic-tail}
\end{equation}
and the two-sided inequality
\begin{equation}
\PP\left[\left|\left\|\bz\right\|_{\bQ}^2-\Tr\left(\bQ\bA\right)\right|>
2\sqrt{\Tr\left[\left(\bQ^{1/2}\bA\bQ^{1/2}\right)^2\right]t}
+2\left\|\bQ^{1/2}\bA\bQ^{1/2}\right\|_2t\right]\le2e^{-t}.
\label{eq:g-gaussian-quadratic-two-sided}
\end{equation}

Define \(\mathcal E_{0,t}^{\mathrm g}\) as the event on which
\begin{equation}
\begin{aligned}
\max_{1\le j\le m}
\left\|\tbtheta_j-\btheta_j^\star\right\|_{\bQ}
&\le
\sqrt{q_+v_+}
\left(\frac{p+2\sqrt{pt}+2t}{n}\right)^{1/2},\\
\max_{J\subseteq\mathcal I_0}
\left\|\overline\btheta_J-\overline\btheta_J^\star\right\|_{\bQ}
&\le
\sqrt{q_+v_+}
\left(\frac{p+2\sqrt{pt}+2t}{n}\right)^{1/2}.
\end{aligned}
\label{eq:g-local-subset-event-bounds}
\end{equation}
For every \(j\in\left[m\right]\) and \(J\subseteq\mathcal I_0\),
\[
\tbtheta_j-\btheta_j^\star
\sim
\mathcal N\left(\bm{\zero},\frac{1}{n_j}\bSigma_j\right),
\qquad
\overline\btheta_J-\overline\btheta_J^\star
\sim
\mathcal N\left(\bm{\zero},\bV_J\right).
\]
Assumption~\ref{assump:seq-2} and \eqref{eq:g-subset-covariance-bounds-new} imply
\[
\begin{aligned}
\Tr\left(\frac{1}{n_j}\bQ\bSigma_j\right)
&\le q_+v_+\frac{p}{n},
&
\Tr\left[\left(\frac{1}{n_j}\bQ^{1/2}\bSigma_j\bQ^{1/2}\right)^2\right]
&\le q_+^2v_+^2\frac{p}{n^2},
&
\left\|\frac{1}{n_j}\bQ^{1/2}\bSigma_j\bQ^{1/2}\right\|_2
&\le q_+v_+\frac{1}{n},\\
\Tr\left(\bQ\bV_J\right)
&\le q_+v_+\frac{p}{n},
&
\Tr\left[\left(\bQ^{1/2}\bV_J\bQ^{1/2}\right)^2\right]
&\le q_+^2v_+^2\frac{p}{n^2},
&
\left\|\bQ^{1/2}\bV_J\bQ^{1/2}\right\|_2
&\le q_+v_+\frac{1}{n}.
\end{aligned}
\]
Equation~\eqref{eq:g-gaussian-quadratic-tail} therefore gives
\[
\begin{aligned}
\mathbb P\left(
\left\|\tbtheta_j-\btheta_j^\star\right\|_{\bQ}^2
>
q_+v_+\frac{p+2\sqrt{pt}+2t}{n}
\right)
&\le e^{-t},\\
\mathbb P\left(
\left\|\overline\btheta_J-\overline\btheta_J^\star\right\|_{\bQ}^2
>
q_+v_+\frac{p+2\sqrt{pt}+2t}{n}
\right)
&\le e^{-t}.
\end{aligned}
\]
A union bound over \(m\) indices \(j\) and at most \(2^{m_0}\) subsets \(J\) yields
\begin{equation}
\mathbb P\left(\mathcal E_{0,t}^{\mathrm g}\right)
\ge
1-\left(m+2^{m_0}\right)e^{-t}.
\label{eq:g-local-subset-event-probability}
\end{equation}

For \(J\subsetneq\mathcal I_0\) and \(j\in\mathcal I_0\setminus J\), define
\begin{equation}
\bW_{j,J}
=
\bV_J\left(\bV_J+\frac{1}{n_j}\bSigma_j\right)^{-1},
\qquad
\bS_{j,J}
=
\bQ^{1/2}\bW_{j,J}\bV_J\bQ^{1/2}.
\label{eq:g-candidate-matrices}
\end{equation}
The identities
\(\bW_{j,J}^{-1}=\bI_p+n_j^{-1}\bSigma_j\bV_J^{-1}\) and
\eqref{eq:g-subset-covariance-bounds-new} imply
\begin{equation}
 c_3^{\mathrm g}\left\|\bx\right\|_{\bQ}
\le
\left\|\bW_{j,J}\bx\right\|_{\bQ}
\le
 c_2^{\mathrm g}\left\|\bx\right\|_{\bQ},
\qquad
\bx\in\mathbb R^p,
\label{eq:g-weight-norm-bounds}
\end{equation}
where the constants are $c_2^{\mathrm g} = \sqrt{\frac{q_+}{q_-}}\frac{v_+}{v_-}, \quad c_3^{\mathrm g} = \sqrt{\frac{q_-}{q_+}} \left(1+\frac{m v_+}{v_-}\right)^{-1}$. Indeed,
\(\left\|\bW_{j,J}\right\|_2\le v_+/v_-\) and
\(\left\|\bW_{j,J}^{-1}\right\|_2\le1+m v_+/v_-\), and the spectral bounds for
\(\bQ\) convert these Euclidean norm inequalities into \eqref{eq:g-weight-norm-bounds}.

Decompose the candidate direction as
\begin{equation}
\begin{aligned}
\bW_{j,J}\left(\tbtheta_j-\overline\btheta_J\right)
&=\bz_{j,J}^{\mathrm g}+\bmu_{j,J}^{\mathrm g},\\
\bz_{j,J}^{\mathrm g}
&=\bW_{j,J}\left[
\left(\tbtheta_j-\btheta_j^\star\right)
-\left(\overline\btheta_J-\overline\btheta_J^\star\right)\right],\\
\bmu_{j,J}^{\mathrm g}
&=\bW_{j,J}\left(\btheta_j^\star-\overline\btheta_J^\star\right).
\end{aligned}
\label{eq:g-candidate-decomposition}
\end{equation}
Since \(\tbtheta_j\mathrel{\perp\!\!\!\perp}\overline\btheta_J\), equations
\eqref{eq:g-subset-covariance} and \eqref{eq:g-candidate-matrices} give
\begin{equation*}
\operatorname{Cov}\left(\bz_{j,J}^{\mathrm g}\right)
=
\bW_{j,J}
\left(\frac{1}{n_j}\bSigma_j+\bV_J\right)
\bW_{j,J}^{\top}
=
\bW_{j,J}\bV_J.
\end{equation*}
Equation~\eqref{eq:g-candidate-matrices} gives
\[
\bW_{j,J}\bV_J
=
\bV_J\left(\bV_J+\frac{1}{n_j}\bSigma_j\right)^{-1}\bV_J
\succeq
\bm{\zero}.
\]
Applying \eqref{eq:g-gaussian-quadratic-two-sided} with
\(\bA=\bW_{j,J}\bV_J\) to each \(\bz_{j,J}^{\mathrm g}\), followed by a union bound over the at most
\(m2^{m_0}\) pairs \(\left(J,j\right)\), gives an event \(\mathcal E_{1,t}^{\mathrm g}\) with
\begin{equation}
\mathbb P\left(\mathcal E_{1,t}^{\mathrm g}\right)
\ge
1-2m2^{m_0}e^{-t}.
\label{eq:g-candidate-event-probability}
\end{equation}
On \(\mathcal E_{1,t}^{\mathrm g}\), every admissible pair \(\left(J,j\right)\) satisfies
\begin{equation}
\frac{
\left|
\left\|\bz_{j,J}^{\mathrm g}\right\|_{\bQ}^2-
\Tr\left(\bS_{j,J}\right)
\right|
}{
\Tr\left(\bS_{j,J}\right)
}
\le
2\sqrt{\frac{\left\|\bS_{j,J}\right\|_2t}{\Tr\left(\bS_{j,J}\right)}}
+2\frac{\left\|\bS_{j,J}\right\|_2t}{\Tr\left(\bS_{j,J}\right)}
\le
2\sqrt{\frac{t}{d_{\rm eff}}}+2\frac{t}{d_{\rm eff}}.
\label{eq:g-candidate-relative-centered-calibration}
\end{equation}
Indeed, the eigenvalues of \(\bS_{j,J}\succeq\bm{\zero}\) satisfy
$
\Tr\left(\bS_{j,J}^2\right)
\le
\left\|\bS_{j,J}\right\|_2\Tr\left(\bS_{j,J}\right)
$.
Assumption~\ref{assump:seq-2} gives
$
\frac{\Tr\left(\bS_{j,J}\right)}{\left\|\bS_{j,J}\right\|_2}
\ge
d_{\rm eff}
$.
Substitution of these two inequalities into \eqref{eq:g-gaussian-quadratic-two-sided} gives
\eqref{eq:g-candidate-relative-centered-calibration}.

Equations \eqref{eq:g-candidate-matrices} and
\eqref{eq:g-subset-covariance-bounds-new} give
\begin{equation}
\Tr\left(\bS_{j,J}\right)
\ge
c_4^{\mathrm g}\frac{p}{n},
\qquad
\left\|\bS_{j,J}\right\|_2
\le
c_5^{\mathrm g}\frac{1}{n},
\qquad
c_4^{\mathrm g}=\frac{q_-v_-^2}{2m^2v_+},
\qquad
c_5^{\mathrm g}=q_+v_+.
\label{eq:sequential-local-trace-lower}
\end{equation}
Indeed, for the lower trace bound, every unit vector \(\bx\) satisfies
\[
\bx^{\top}\bW_{j,J}\bV_J\bx
=
\left(\bV_J\bx\right)^{\top}
\left(\bV_J+n_j^{-1}\bSigma_j\right)^{-1}
\left(\bV_J\bx\right)
\ge
\frac{v_-^2}{2m^2v_+n}
.\]
The spectral bound for \(\bQ\) therefore gives
$
\Tr\left(\bS_{j,J}\right)
=
\Tr\left(\bQ\bW_{j,J}\bV_J\right)
\ge
q_-p\frac{v_-^2}{2m^2v_+n}
=
c_4^{\mathrm g}\frac{p}{n}
$.
Moreover,
$
\bm{\zero}
\preceq
\bW_{j,J}\bV_J
\preceq
\bV_J
\preceq
\frac{v_+}{n}\bI_p
$
implies
\(\left\|\bS_{j,J}\right\|_2\le q_+v_+/n=c_5^{\mathrm g}/n\), which completes
\eqref{eq:sequential-local-trace-lower}.

Define the event parameters by
\[
N_1^{\mathrm g}=m+2^{m_0}+2m2^{m_0},
\qquad
t_n^{\mathrm g}
=
\max\left[
1,
\log\left(\frac{2N_1^{\mathrm g}}{\rho}\right),
\log\left(2N_1^{\mathrm g}n\delta_{\max}^2\right)
\right],
\]
and define
\(\mathcal G_n^{\mathrm g}=\mathcal E_{0,t_n^{\mathrm g}}^{\mathrm g}\cap
\mathcal E_{1,t_n^{\mathrm g}}^{\mathrm g}\) and
\(\pi_n^{\mathrm g}=N_1^{\mathrm g}e^{-t_n^{\mathrm g}}\). Equations
\eqref{eq:g-local-subset-event-probability} and
\eqref{eq:g-candidate-event-probability} imply
\begin{equation}
\mathbb P\left(\mathcal G_n^{\mathrm g}\right)
\ge
1-\pi_n^{\mathrm g},
\qquad
\pi_n^{\mathrm g}
\le
\min\left[\frac{\rho}{2},\frac{1}{2n\delta_{\max}^2}\right]
=
o\left(1\right).
\label{eq:g-unified-good-event-probability}
\end{equation}
The assumptions in the theorem imply
\begin{equation}
\frac{t_n^{\mathrm g}}{p}\longrightarrow0,
\qquad
\frac{
\max_{j\in\mathcal I_0}
\left\|\btheta_j^\star-\btheta_1^\star\right\|_{\bQ}
}{\sqrt{p/n}}
\longrightarrow0,
\qquad
\frac{\sqrt{p/n}}{
\min_{j\in\mathcal I_1}
\left\|\btheta_j^\star-\btheta_1^\star\right\|_{\bQ}}
\longrightarrow0.
\label{eq:g-three-small-ratios}
\end{equation}

When
\(\tbtheta_j^{\circ}=\tbtheta_j\),
\(\overline\btheta_J^{\circ}=\overline\btheta_J\), and
\(\bS_{j,J}^{\circ}=\bS_{j,J}\), the quantity \(\mathcal R_n^{\circ}\) in Theorem~\ref{thm:sequential-unified} equals
\[
\max_{J\subsetneq\mathcal I_0}
\max_{j\in\mathcal I_0\setminus J}
\left[
\left\|\tbtheta_j-\overline\btheta_J\right\|_{\bQ}
+
\Tr\left(\bS_{j,J}\right)^{1/2}
\right].
\]
For \(J\subsetneq\mathcal I_0\) and \(j\in\mathcal I_0\setminus J\), equations
\eqref{eq:g-local-subset-event-bounds} and \eqref{eq:g-subset-bias} give on
\(\mathcal G_n^{\mathrm g}\)
\[
\left\|\tbtheta_j-\overline\btheta_J\right\|_{\bQ}
\le
2\sqrt{q_+v_+}
\left(\frac{p+2\sqrt{pt_n^{\mathrm g}}+2t_n^{\mathrm g}}{n}\right)^{1/2}
+
\left(1+c_1^{\mathrm g}\right)
\max_{k\in\mathcal I_0}
\left\|\btheta_k^\star-\btheta_1^\star\right\|_{\bQ}.
\]
For all sufficiently large \(n\), equations \eqref{eq:g-three-small-ratios} and
\eqref{eq:sequential-local-trace-lower} imply
\[
\left\|\tbtheta_j-\overline\btheta_J\right\|_{\bQ}
\le
\left(2\sqrt{5q_+v_+}+1+c_1^{\mathrm g}\right)
\sqrt{\frac{p}{n}},
\qquad
\Tr\left(\bS_{j,J}\right)^{1/2}
\le
\sqrt{q_+v_+}\sqrt{\frac{p}{n}}.
\]
The first quantity therefore satisfies
\begin{equation}
\mathcal R_n^{\circ}
\le
k_0^{\mathrm g}\sqrt{\frac{p}{n}},
\qquad
k_0^{\mathrm g}=3\sqrt{5q_+v_+}+1+c_1^{\mathrm g}.
\label{eq:g-unified-R-bound}
\end{equation}
The second quantity equals
\[
\mathcal U_n^{\circ}
=
\max_{1\le j\le m}
\left\|\tbtheta_j-\btheta_j^\star\right\|_{\bQ}
+
\max_{J\subseteq\mathcal I_0}
\left\|\overline\btheta_J-\btheta_1^\star\right\|_{\bQ}.
\]
Equations \eqref{eq:g-local-subset-event-bounds} and \eqref{eq:g-subset-bias} give
\[
\max_{1\le j\le m}
\left\|\tbtheta_j-\btheta_j^\star\right\|_{\bQ}
\le
\sqrt{5q_+v_+}\sqrt{\frac{p}{n}},
\qquad
\max_{J\subseteq\mathcal I_0}
\left\|\overline\btheta_J-\btheta_1^\star\right\|_{\bQ}
\le
\left(\sqrt{5q_+v_+}+c_1^{\mathrm g}\right)
\sqrt{\frac{p}{n}},
\]
and hence
\begin{equation}
\mathcal U_n^{\circ}
\le
\left(2\sqrt{5q_+v_+}+c_1^{\mathrm g}\right)
\sqrt{\frac{p}{n}}
\le
k_0^{\mathrm g}\sqrt{\frac{p}{n}}.
\label{eq:g-unified-U-bound}
\end{equation}

Recall that $\bmu_{j,J}^{\mathrm g}
=
\bW_{j,J}
\left(\btheta_j^\star-\overline\btheta_J^\star\right)$. Equations \eqref{eq:g-subset-bias} and
\eqref{eq:g-weight-norm-bounds} give
\begin{equation}
\left\|\bmu_{j,J}^{\mathrm g}\right\|_{\bQ}
\le
c_2^{\mathrm g}\left(1+c_1^{\mathrm g}\right)
\max_{k\in\mathcal I_0}
\left\|\btheta_k^\star-\btheta_1^\star\right\|_{\bQ}.
\label{eq:g-candidate-bias-bound}
\end{equation}
Using the decomposition in \eqref{eq:g-candidate-decomposition}, the Cauchy-Schwarz inequality,
\eqref{eq:g-candidate-relative-centered-calibration},
\eqref{eq:sequential-local-trace-lower}, and
\eqref{eq:g-candidate-bias-bound}, we~obtain
\begin{align}
&\frac{
\left|
\left\|\bW_{j,J}\left(\tbtheta_j-\overline\btheta_J\right)\right\|_{\bQ}^2
-
\Tr\left(\bS_{j,J}\right)
\right|
}{\Tr\left(\bS_{j,J}\right)}
+
\frac{\left\|\bS_{j,J}\right\|_2}{\Tr\left(\bS_{j,J}\right)}
\notag\\
&\quad\le
2\sqrt{\frac{t_n^{\mathrm g}}{d_{\rm eff}}}
+2\frac{t_n^{\mathrm g}}{d_{\rm eff}}
+
2\left(1+2\sqrt{\frac{t_n^{\mathrm g}}{d_{\rm eff}}}
+2\frac{t_n^{\mathrm g}}{d_{\rm eff}}\right)^{1/2}
\frac{
c_2^{\mathrm g}\left(1+c_1^{\mathrm g}\right)
\max_{k\in\mathcal I_0}
\left\|\btheta_k^\star-\btheta_1^\star\right\|_{\bQ}
}{\sqrt{c_4^{\mathrm g}p/n}}
\notag\\
&\qquad+
\frac{
\left(c_2^{\mathrm g}\right)^2\left(1+c_1^{\mathrm g}\right)^2
\max_{k\in\mathcal I_0}
\left\|\btheta_k^\star-\btheta_1^\star\right\|_{\bQ}^2
}{c_4^{\mathrm g}p/n}
+
\frac{1}{d_{\rm eff}}.
\label{eq:g-complete-relative-calibration}
\end{align}
Let \(\varepsilon_n^{\mathrm g}\) equal the right-hand side of
\eqref{eq:g-complete-relative-calibration}. Equation \eqref{eq:g-three-small-ratios} and
\(d_{\rm eff}\asymp p\) imply \(\varepsilon_n^{\mathrm g}\longrightarrow0\). Therefore,~for
\[
\mathcal Z_n^{\circ}
=
\max_{J\subsetneq\mathcal I_0}
\max_{j\in\mathcal I_0\setminus J}
\frac{
\left|
\left\|\bW_{j,J}\left(\tbtheta_j-\overline\btheta_J\right)\right\|_{\bQ}^2
-
\Tr\left(\bS_{j,J}\right)
\right|
+
\left\|\bS_{j,J}\right\|_2
}{
\Tr\left(\bS_{j,J}\right)
},
\]
equation~\eqref{eq:g-complete-relative-calibration} gives
\begin{equation}
\mathcal Z_n^{\circ}
\le
\varepsilon_n^{\mathrm g}
\qquad\text{on }\mathcal G_n^{\mathrm g}.
\label{eq:g-unified-Z-bound}
\end{equation}

For the fourth-moment condition, equations \eqref{eq:g-subset-definitions} and
\eqref{eq:g-subset-covariance} give
\begin{equation}
\operatorname{Cov}
\left(\tbtheta_1-\overline\btheta_{\mathcal I_0}\right)
=
\frac{1}{n_1}\bSigma_1-\bV_{\mathcal I_0}
\preceq
\frac{v_+}{n}\bI_p.
\label{eq:g-target-oracle-difference-covariance}
\end{equation}

For a centered Gaussian vector $\bz$ with $\operatorname{Cov}\left(\bz\right)\preceq v_+\bI_p/n$, let $\bB=\bQ^{1/2}\operatorname{Cov}\left(\bz\right)\bQ^{1/2}$. Isserlis' identity \citep{isserlis1918formula} gives
\begin{equation}
\begin{aligned}
\EE\left\|\bz\right\|_{\bQ}^4
&=\Tr\left(\bB\right)^2+2\Tr\left(\bB^2\right)
\le3\Tr\left(\bB\right)^2
\le3p^2\left\|\bB\right\|_2^2
\le3q_+^2v_+^2\left(p/n\right)^2.
\end{aligned}
\label{eq:g-fourth-moment}
\end{equation}
Equation \eqref{eq:g-subset-bias}, applied to \(J=\mathcal I_0\), gives 
\(\left\|\mathbb E\left(\tbtheta_1-\overline\btheta_{\mathcal I_0}\right)\right\|_{\bQ}^4
\le\left(p/n\right)^2\) for all sufficiently large \(n\). Equations
\eqref{eq:g-target-oracle-difference-covariance} and \eqref{eq:g-fourth-moment}, together with this consequence of \eqref{eq:g-subset-bias} and
\(\left(x+y\right)^4\le8x^4+8y^4\),~yield
\begin{equation}
\mathbb E
\left\|\tbtheta_1-\overline\btheta_{\mathcal I_0}\right\|_{\bQ}^4
\le
k_4^{\mathrm g}\left(\frac{p}{n}\right)^2,
\qquad
k_4^{\mathrm g}=24q_+^2v_+^2+8.
\label{eq:g-unified-fourth-moment}
\end{equation}

Equations \eqref{eq:g-unified-good-event-probability},
\eqref{eq:g-unified-R-bound}, \eqref{eq:g-unified-U-bound},
\eqref{eq:g-unified-Z-bound}, and \eqref{eq:g-unified-fourth-moment} verify the assumptions of
Theorem~\ref{thm:sequential-unified}. Hence, for all sufficiently large \(n\),
\begin{equation}
\mathbb P\left(
\left\{j_1,\ldots,j_{m_0}\right\}=\mathcal I_0
\right)
\ge
1-\pi_n^{\mathrm g}
\ge
1-\rho,
\qquad
\mathbb E
\left\|\widehat\btheta_{\rm seq}-\overline\btheta_{\mathcal I_0}\right\|_{\bQ}^2
=
o\left(\frac{p}{n}\right).
\label{eq:g-unified-conclusions}
\end{equation}
This proves part 1 of Theorem \ref{thm:seq-large-separation-trace-condition}.

For part 2, equations \eqref{eq:g-subset-definitions} and
\eqref{eq:g-subset-covariance} give
\[
\operatorname{Cov}\left(\overline\btheta_{\mathcal I_0}\right)
=
\bV_{\mathcal I_0},
\qquad
\mathbb E\overline\btheta_{\mathcal I_0}-\btheta_1^\star
=
\bV_{\mathcal I_0}
\sum_{k\in\mathcal I_0}n_k\bSigma_k^{-1}
\left(\btheta_k^\star-\btheta_1^\star\right).
\]
Equation~\eqref{eq:g-subset-bias} gives $\left\| \mathbb E\overline\btheta_{\mathcal I_0}-\btheta_1^\star \right\|_{\bQ}^2 = o\left(\frac{p}{n}\right)$. The oracle risk therefore satisfies
\begin{equation}
\mathbb E
\left\|\overline\btheta_{\mathcal I_0}-\btheta_1^\star\right\|_{\bQ}^2
=
\left\langle\bV_{\mathcal I_0},\bQ\right\rangle
+
o\left(\frac{p}{n}\right).
\label{eq:g-oracle-risk}
\end{equation}
Expanding the quadratic risk around \(\overline\btheta_{\mathcal I_0}\), applying the Cauchy-Schwarz inequality, and using
\eqref{eq:g-unified-conclusions} and \eqref{eq:g-oracle-risk} gives
\begin{align}
&\left|
\mathbb E\left\|\widehat\btheta_{\rm seq}-\btheta_1^\star\right\|_{\bQ}^2
-
\mathbb E\left\|\overline\btheta_{\mathcal I_0}-\btheta_1^\star\right\|_{\bQ}^2
\right|
\notag\\
&\quad\le
2
\left[\mathbb E\left\|\overline\btheta_{\mathcal I_0}-\btheta_1^\star\right\|_{\bQ}^2\right]^{1/2}
\left[\mathbb E\left\|\widehat\btheta_{\rm seq}-\overline\btheta_{\mathcal I_0}\right\|_{\bQ}^2\right]^{1/2}
+
\mathbb E\left\|\widehat\btheta_{\rm seq}-\overline\btheta_{\mathcal I_0}\right\|_{\bQ}^2
=
o\left(\frac{p}{n}\right).
\label{eq:g-sequential-oracle-risk-difference}
\end{align}
Equations \eqref{eq:g-sequential-oracle-risk-difference} and
\eqref{eq:g-oracle-risk} prove part 2.

For part 3,
\[
\mathbb E
\left[
\sum_{j\in\mathcal I}\bW_j
\left(\tbtheta_j-\tbtheta_1\right)
\right]
=
\sum_{j\in\mathcal I}\bW_j
\left(\btheta_j^\star-\btheta_1^\star\right),
\qquad
\left\|
\sum_{j\in\mathcal I}\bW_j
\left(\btheta_j^\star-\btheta_1^\star\right)
\right\|_{\bQ}
=
\delta_{\mathcal I}.
\]
The covariance of the centered Gaussian component is
\begin{equation*}
\bA_{\mathcal I}
=
\sum_{j\in\mathcal I}
\bW_j\left(\frac{1}{n_j}\bSigma_j\right)\bW_j^{\top}
+
\left(\sum_{j\in\mathcal I}\bW_j\right)
\left(\frac{1}{n_1}\bSigma_1\right)
\left(\sum_{j\in\mathcal I}\bW_j\right)^{\top}.
\end{equation*}
The finite number of sources and \eqref{eq:g-weight-norm-bounds} imply
\[
\left\|\bA_{\mathcal I}\right\|_2
\le
c_6^{\mathrm g}\frac{1}{n},
\qquad
c_6^{\mathrm g}
=
v_+\left[\left(m-1\right)+\left(m-1\right)^2\right]
\left(c_2^{\mathrm g}\right)^2.
\]
Let
\(c_7^{\mathrm g}=\left(64q_+c_6^{\mathrm g}\right)^{-1}\) and
\(t_{\mathcal I}^{\mathrm g}=c_7^{\mathrm g}n\delta_{\mathcal I}^2\). The condition
\(\delta_{\mathcal I}/\sqrt{p/n}\to\infty\) implies
\(t_{\mathcal I}^{\mathrm g}/p\to\infty\). Hence, for all sufficiently large \(n\),
\begin{align*}
&\Tr\left(\bQ\bA_{\mathcal I}\right)
+2\sqrt{\Tr\left[\left(\bQ^{1/2}\bA_{\mathcal I}\bQ^{1/2}\right)^2\right]
t_{\mathcal I}^{\mathrm g}}
+2\left\|\bQ^{1/2}\bA_{\mathcal I}\bQ^{1/2}\right\|_2t_{\mathcal I}^{\mathrm g}\\
&\quad\le
q_+c_6^{\mathrm g}\frac{p+2\sqrt{pt_{\mathcal I}^{\mathrm g}}+2t_{\mathcal I}^{\mathrm g}}{n}
\le
\frac{5}{64}\delta_{\mathcal I}^2
<
\frac{1}{4}\delta_{\mathcal I}^2.
\end{align*}
Applying \eqref{eq:g-gaussian-quadratic-tail} with
\(\bA=\bA_{\mathcal I}\) and \(t=t_{\mathcal I}^{\mathrm g}\) gives
\begin{equation}
\mathbb P\left(
\left\|
\sum_{j\in\mathcal I}\bW_j
\left[
\left(\tbtheta_j-\btheta_j^\star\right)
-
\left(\tbtheta_1-\btheta_1^\star\right)
\right]
\right\|_{\bQ}
>
\frac{\delta_{\mathcal I}}{2}
\right)
\le
\exp\left(-c_7^{\mathrm g}n\delta_{\mathcal I}^2\right).
\label{eq:g-single-direction-tail}
\end{equation}

For \(s\in\left[\underline s,\overline s\right]\), $0\le s\le\overline s=\Tr\left(\bS\right)\le c_8^{\mathrm g}\frac{p}{n}, \quad c_8^{\mathrm g}=q_+v_+$, and the single-step update is
\[
\widehat\btheta_{\rm ss}-\tbtheta_1
=
\frac{s}{
\left\|\sum_{j\in\mathcal I}\bW_j
\left(\tbtheta_j-\tbtheta_1\right)\right\|_{\bQ}^2\vee s}
\sum_{j\in\mathcal I}\bW_j
\left(\tbtheta_j-\tbtheta_1\right).
\]
On the event
\[
\left\|
\sum_{j\in\mathcal I}\bW_j
\left[
\left(\tbtheta_j-\btheta_j^\star\right)
-
\left(\tbtheta_1-\btheta_1^\star\right)
\right]
\right\|_{\bQ}
\le
\frac{\delta_{\mathcal I}}{2},
\]
the triangle inequality gives $\left\|\sum_{j\in\mathcal I}\bW_j \left(\tbtheta_j-\tbtheta_1\right)\right\|_{\bQ} \ge \frac{\delta_{\mathcal I}}{2}$. For all sufficiently large \(n\),
\(\delta_{\mathcal I}^2/4>c_8^{\mathrm g}p/n\ge s\), and hence $\left\|\widehat\btheta_{\rm ss}-\tbtheta_1\right\|_{\bQ}^2 \le 4\left(c_8^{\mathrm g}\right)^2 \frac{\left(p/n\right)^2}{\delta_{\mathcal I}^2}$. For every realization, $\left\|\widehat\btheta_{\rm ss}-\tbtheta_1\right\|_{\bQ}^2 \le s \le c_8^{\mathrm g}\frac{p}{n}$. Equation~\eqref{eq:g-single-direction-tail} therefore yields
\begin{equation}
\mathbb E
\left\|\widehat\btheta_{\rm ss}-\tbtheta_1\right\|_{\bQ}^2
\le
4\left(c_8^{\mathrm g}\right)^2
\frac{\left(p/n\right)^2}{\delta_{\mathcal I}^2}
+
c_8^{\mathrm g}\frac{p}{n}
\exp\left(-c_7^{\mathrm g}n\delta_{\mathcal I}^2\right)
=
o\left(\frac{p}{n}\right).
\label{eq:g-single-update-L2}
\end{equation}
The equality follows from
\(\left(p/n\right)/\delta_{\mathcal I}^2\to0\) and
\(n\delta_{\mathcal I}^2\to\infty\).

The target estimator has risk
\begin{equation}
\mathbb E
\left\|\tbtheta_1-\btheta_1^\star\right\|_{\bQ}^2
=
\left\langle\frac{1}{n_1}\bSigma_1,\bQ\right\rangle.
\label{eq:g-target-risk}
\end{equation}
The Woodbury identity and the covariance bounds give
\begin{equation*}
\frac{1}{n_1}\bSigma_1-\bV_{\mathcal I_0}
=
\left(\frac{1}{n_1}\bSigma_1\right)
\left[
\frac{1}{n_1}\bSigma_1
+
\left(\sum_{j\in\mathcal I_0}n_j\bSigma_j^{-1}\right)^{-1}
\right]^{-1}
\left(\frac{1}{n_1}\bSigma_1\right)
\succeq
\frac{v_-^2}{2v_+n}\bI_p.
\end{equation*}
Taking the trace against $\bQ$ gives
\begin{equation}
\left\langle\frac{1}{n_1}\bSigma_1-\bV_{\mathcal I_0},\bQ\right\rangle
\ge
g_0^{\mathrm g}\frac{p}{n},
\qquad
g_0^{\mathrm g}=\frac{q_-v_-^2}{2v_+}.
\label{eq:g-oracle-trace-gap}
\end{equation}
Expanding the risk of \(\widehat\btheta_{\rm ss}\) around \(\tbtheta_1\), applying the Cauchy-Schwarz inequality, and using
\eqref{eq:g-single-update-L2} and \eqref{eq:g-target-risk} gives
\begin{equation}
\mathbb E
\left\|\widehat\btheta_{\rm ss}-\btheta_1^\star\right\|_{\bQ}^2
=
\left\langle\frac{1}{n_1}\bSigma_1,\bQ\right\rangle
+
o\left(\frac{p}{n}\right).
\label{eq:g-single-step-risk-expansion}
\end{equation}
Combining part 2 with \eqref{eq:g-oracle-trace-gap} and
\eqref{eq:g-single-step-risk-expansion} yields
\[
\mathbb E
\left\|\widehat\btheta_{\rm ss}-\btheta_1^\star\right\|_{\bQ}^2
-
\mathbb E
\left\|\widehat\btheta_{\rm seq}-\btheta_1^\star\right\|_{\bQ}^2
\ge
g_0^{\mathrm g}\frac{p}{n}+o\left(\frac{p}{n}\right)>0
\]
for all sufficiently large \(n\). This proves part 3.
\end{proof}

\begin{proof}[Proof of Theorem~\ref{thm:risk-bound-s-multiset}]
Define
\begin{equation}
\begin{aligned}
\tbtheta_{-1}
&=
\left(\sum_{j=2}^m n_j\bSigma_j^{-1}\right)^{-1}
\sum_{j=2}^m n_j\bSigma_j^{-1}\tbtheta_j,
\qquad
\bV_{-1}
=
\left(\sum_{j=2}^m n_j\bSigma_j^{-1}\right)^{-1},\\
\bW_{-1}
&=
\left(n_1\bSigma_1^{-1}+\sum_{j=2}^m n_j\bSigma_j^{-1}\right)^{-1}
\sum_{j=2}^m n_j\bSigma_j^{-1}.
\end{aligned}
\label{eq:multiset-source-aggregate-estimator-proof}
\end{equation}
Independence of the data sets gives
\begin{equation}
\tbtheta_{-1}
\sim
\mathcal N\left(
\bV_{-1}\sum_{j=2}^m n_j\bSigma_j^{-1}\btheta_j^\star,
\bV_{-1}
\right),
\qquad
\tbtheta_{-1}\perp\tbtheta_1.
\label{eq:multiset-source-aggregate-distribution-proof}
\end{equation}
The definition of \(\bW_j\) in \eqref{eq:weight-multiset} and
\eqref{eq:multiset-source-aggregate-estimator-proof} imply
\begin{equation}
\sum_{j=2}^m\bW_j
=
\bW_{-1},
\qquad
\sum_{j=2}^m\bW_j\left(\tbtheta_j-\tbtheta_1\right)
=
\bW_{-1}\left(\tbtheta_{-1}-\tbtheta_1\right).
\label{eq:multiset-direction-two-set-reduction-proof}
\end{equation}
The corresponding matrix identities are
\begin{equation}
\bW_{-1}
=
\left[\left(n_1^{-1}\bSigma_1\right)^{-1}+\bV_{-1}^{-1}\right]^{-1}\bV_{-1}^{-1},
\qquad
\bQ^{1/2}\bW_{-1}\left(n_1^{-1}\bSigma_1\right)\bQ^{1/2}
=
\sum_{j=2}^m n_1^{-1}\bQ^{1/2}\bW_j\bSigma_1\bQ^{1/2}
=
\bS.
\label{eq:multiset-S-two-set-reduction-proof}
\end{equation}
Equations \eqref{eq:multiset-source-aggregate-distribution-proof},
\eqref{eq:multiset-direction-two-set-reduction-proof}, and
\eqref{eq:multiset-S-two-set-reduction-proof} identify the estimator in
\eqref{eq:shrinkage-estimator-s-multi-set} with the two-set estimator in
\eqref{eq:shrinkage-estimator-s}, with \(\tbtheta_{-1}\) as the source estimator and
\(\bV_{-1}\) as its covariance. The weak differentiability and integrability conditions used in
Theorem~\ref{thm:risk-bound-s} are verified by
\eqref{eq:two-set-truncated-bound}--\eqref{eq:two-set-integrability-derivative}. Substituting
\eqref{eq:multiset-direction-two-set-reduction-proof} and
\eqref{eq:multiset-S-two-set-reduction-proof} into the risk inequality obtained from
\eqref{eq:two-set-outer-integrand-exact}--\eqref{eq:two-set-integrand-combined-full} gives
\begin{align*}
&\mathbb E\left[
\left\|\widehat\btheta_s-\btheta_1^\star\right\|_{\bQ}^2
\right]
-
\left\langle n_1^{-1}\bSigma_1,\bQ\right\rangle\\
&\quad\le
-s\left[2\Tr\left(\bS\right)-4\left\|\bS\right\|_2-s\right]
\mathbb E\left[
\frac{1}{
\left\|\sum_{j=2}^m\bW_j\left(\tbtheta_j-\tbtheta_1\right)\right\|_{\bQ}^2\vee s}
\right],
\end{align*}
which is \eqref{eq:risk-upper-multiset}.
\end{proof}

\begin{proof}[Proof of Theorem~\ref{thm:risk-bound-shat-multiset}]
Use \(\tbtheta_{-1}\), \(\bV_{-1}\), and \(\bW_{-1}\) from
\eqref{eq:multiset-source-aggregate-estimator-proof}. Define
\begin{equation}
\begin{aligned}
\btheta_{-1}^\star
&=
\bV_{-1}\sum_{j=2}^m n_j\bSigma_j^{-1}\btheta_j^\star,\\
\bW_{-1}\left(\btheta_{-1}^\star-\btheta_1^\star\right)
&=
\left(n_1\bSigma_1^{-1}+\sum_{k=2}^m n_k\bSigma_k^{-1}\right)^{-1}
\sum_{j=2}^m n_j\bSigma_j^{-1}
\left(\btheta_j^\star-\btheta_1^\star\right)\\
&=
\sum_{j=2}^m\bW_j\left(\btheta_j^\star-\btheta_1^\star\right).
\end{aligned}
\label{eq:multiset-population-direction-two-set-proof}
\end{equation}
The pooled estimator and its covariance equal
\begin{equation}
\begin{aligned}
\overline\btheta
&=
\left[\left(n_1^{-1}\bSigma_1\right)^{-1}+\bV_{-1}^{-1}\right]^{-1}
\left[\left(n_1^{-1}\bSigma_1\right)^{-1}\tbtheta_1+\bV_{-1}^{-1}\tbtheta_{-1}\right],\\
\operatorname{Cov}\left(\overline\btheta\right)
&=
\left(n_1\bSigma_1^{-1}+\sum_{j=2}^m n_j\bSigma_j^{-1}\right)^{-1}.
\end{aligned}
\label{eq:multiset-pooled-two-set-proof}
\end{equation}
Equations \eqref{eq:multiset-source-aggregate-distribution-proof},
\eqref{eq:multiset-direction-two-set-reduction-proof},
\eqref{eq:multiset-S-two-set-reduction-proof},
\eqref{eq:multiset-population-direction-two-set-proof}, and
\eqref{eq:multiset-pooled-two-set-proof} identify every quantity in the multi-source estimator with the corresponding two-set quantity in
Theorem~\ref{thm:risk-bound-shat}. Substituting the identities in
\eqref{eq:multiset-source-aggregate-distribution-proof}--\eqref{eq:multiset-pooled-two-set-proof} into
\eqref{eq:two-set-adaptive-alpha-risk-full} gives
\begin{align*}
\mathbb E\left[
\left\|\widehat\btheta_s-\btheta_1^\star\right\|_{\bQ}^2
\right]
&\le
\left\langle
\left(\sum_{j=1}^m n_j\bSigma_j^{-1}\right)^{-1},
\bQ
\right\rangle+
\frac{
\Tr\left(\bS\right)
\left\|\sum_{j=2}^m\bW_j\left(\btheta_j^\star-\btheta_1^\star\right)\right\|_{\bQ}^2
}{
\Tr\left(\bS\right)
+
\left\|\sum_{j=2}^m\bW_j\left(\btheta_j^\star-\btheta_1^\star\right)\right\|_{\bQ}^2
}
+
4\left\|\bS\right\|_2,
\end{align*}
which is \eqref{eq:risk-bound-shat-multiset}.
\end{proof}

\subsection{Proof of the Results in Section~\ref{sec:plugin-covariance-stability}}\label{sec:proof-plugin-covariance-stability}

\begin{proof}[Proof of Proposition~\ref{prop:plugin-stability}]
Let $\bP=n_1\bSigma_1^{-1}+n_2\bSigma_2^{-1}$ and
$\widehat\bP=n_1\widehat\bSigma_1^{-1}+n_2\widehat\bSigma_2^{-1}$.
The fixed spectral bounds and $n_j\asymp n$ give
$\left\|\bP^{-1}\right\|_2\vee\left\|\widehat\bP^{-1}\right\|_2\le C/n$.
Using
$\widehat\bSigma_j^{-1}-\bSigma_j^{-1}=\widehat\bSigma_j^{-1}\left(\bSigma_j-\widehat\bSigma_j\right)\bSigma_j^{-1}$,
we obtain, for every realization,
\[
\begin{aligned}
\max_{j=1,2}\left\|\widehat\bSigma_j^{-1}-\bSigma_j^{-1}\right\|_2&\le C\eta_n,
\qquad \left\|\widehat\bP-\bP\right\|_2\le Cn\eta_n,\\
\left\|\widehat\bP^{-1}-\bP^{-1}\right\|_2
&=\left\|\widehat\bP^{-1}\left(\bP-\widehat\bP\right)\bP^{-1}\right\|_2\le C\eta_n/n.
\end{aligned}
\]
The exact differences are
\[
\begin{aligned}
\widehat\bW_2-\bW_2
&=n_2\left(\widehat\bP^{-1}-\bP^{-1}\right)\widehat\bSigma_2^{-1}
+n_2\bP^{-1}\left(\widehat\bSigma_2^{-1}-\bSigma_2^{-1}\right),\\
\widehat\bS-\bS
&=\frac1{n_1}\bQ^{1/2}
\left[\left(\widehat\bW_2-\bW_2\right)\widehat\bSigma_1+\bW_2\left(\widehat\bSigma_1-\bSigma_1\right)\right]\bQ^{1/2}.
\end{aligned}
\]
Taking norms in these identities gives
\[
\left\|\widehat\bW_2-\bW_2\right\|_2\le C\eta_n,\qquad
\left\|\widehat\bS-\bS\right\|_2\le C\eta_n/n,\qquad
\left|\widehat s-s_0\right|\le\left(p+1\right)\left\|\widehat\bS-\bS\right\|_2\le Cp\eta_n/n.
\]
Squaring these pointwise inequalities and taking expectations proves the mean-square bounds, without any independence condition between covariance and local estimators.
\end{proof}

\begin{proof}[Proof of Theorem~\ref{thm:plugin-risk}]
The nonexpansiveness argument preceding \eqref{eq:f.3} applies to $\mathcal T_s$. In Euclidean coordinates, the corresponding radial map $\boldsymbol\tau_s\left(\by\right)=s\by/\left(\left\|\by\right\|_2^2\vee s\right)$ has derivative $\bI_p$ inside the sphere and
\begin{equation}
D\boldsymbol\tau_s\left(\by\right)=\frac{s}{\left\|\by\right\|_2^2}
\left(\bI_p-2\frac{\by\by^\top}{\left\|\by\right\|_2^2}\right),\qquad \left\|\by\right\|_2^2>s.
\label{eq:plugin-truncation-jacobian}
\end{equation}
The resulting map satisfies
\begin{equation}
\left\|\mathcal T_s\left(\bu\right)-\mathcal T_s\left(\bv\right)\right\|_{\bQ}\le\left\|\bu-\bv\right\|_{\bQ},\qquad
\left\|\mathcal T_s\left(\bu\right)\right\|_{\bQ}\le\sqrt{s}.
\label{eq:plugin-truncation-input-lipschitz}
\end{equation}
Assume without loss of generality that \(0<s\le t\). The difference
$\left\|\mathcal T_t\left(\bu\right)-\mathcal T_s\left(\bu\right)\right\|_{\bQ}$
equals zero when $\left\|\bu\right\|_{\bQ}^2\le s$, equals
$\left(\left\|\bu\right\|_{\bQ}^2-s\right)/\left\|\bu\right\|_{\bQ}$ when
$s<\left\|\bu\right\|_{\bQ}^2\le t$, and equals
$\left(t-s\right)/\left\|\bu\right\|_{\bQ}$ when $\left\|\bu\right\|_{\bQ}^2>t$.
Each expression yields
\begin{equation}
\left\|\mathcal T_t\left(\bu\right)-\mathcal T_s\left(\bu\right)\right\|_{\bQ}
\le
\frac{\left|t-s\right|}{\sqrt{s\wedge t}}.
\label{eq:plugin-truncation-scale-lipschitz}
\end{equation}

We also require a bound for a relative matrix perturbation. Let \(\bB\) be invertible and suppose~that $\varepsilon = \left\|\bQ^{1/2}\left(\bB-\bI_p\right)\bQ^{-1/2}\right\|_2 \le \frac12$. For \(\bB_t=\bI_p+t\left(\bB-\bI_p\right)\), the reverse triangle inequality gives
\begin{equation}
\left\|\bB_t\bu\right\|_{\bQ}
\ge
\left(1-t\varepsilon\right)\left\|\bu\right\|_{\bQ}
\ge
\frac12\left\|\bu\right\|_{\bQ}.
\label{eq:plugin-relative-path-lower}
\end{equation}
If \(\left\|\bu\right\|_{\bQ}\le2\sqrt{s}\), equation
\eqref{eq:plugin-truncation-input-lipschitz} gives
\(\left\|\mathcal T_s\left(\bB\bu\right)-\mathcal T_s\left(\bu\right)\right\|_{\bQ}
\le2\varepsilon\sqrt{s}\). If
\(\left\|\bu\right\|_{\bQ}>2\sqrt{s}\), equation
\eqref{eq:plugin-relative-path-lower} places the entire path in the outer region. Equations
\eqref{eq:plugin-truncation-jacobian} and \eqref{eq:plugin-relative-path-lower} then give
\[
\left\|\frac{\mathrm d}{\mathrm dt}\mathcal T_s\left(\bB_t\bu\right)\right\|_{\bQ}
\le
\frac{s}{\left\|\bB_t\bu\right\|_{\bQ}^2}
\left\|\left(\bB-\bI_p\right)\bu\right\|_{\bQ}
\le
\frac{4s\varepsilon}{\left\|\bu\right\|_{\bQ}}
<
2\varepsilon\sqrt{s}.
\]
Integration over \(t\in\left[0,1\right]\) proves
\begin{equation}
\left\|\mathcal T_s\left(\bB\bu\right)-\mathcal T_s\left(\bu\right)\right\|_{\bQ}
\le
2\varepsilon\sqrt{s}.
\label{eq:plugin-truncation-relative-matrix-bound}
\end{equation}

We next obtain deterministic lower and upper bounds for the shrinkage scales. Let
\(\bC_j=\bSigma_j/n_j\) and
\(\widehat\bC_j=\widehat\bSigma_j/n_j\). The spectral assumptions and
\(n_j\asymp n\) provide constants \(0<c<C<\infty\) such~that $\frac{c}{n}\bI_p \preceq \bC_j,\widehat\bC_j \preceq \frac{C}{n}\bI_p, \quad j\in\left\{1,2\right\}$. The parallel sum identity gives
\begin{equation*}
\begin{aligned}
\bW_2
&=
\bC_1\left(\bC_1+\bC_2\right)^{-1},
&
\bS
&=
\bQ^{1/2}\bC_1\left(\bC_1+\bC_2\right)^{-1}\bC_1\bQ^{1/2},\\
\widehat\bW_2
&=
\widehat\bC_1\left(\widehat\bC_1+\widehat\bC_2\right)^{-1},
&
\widehat\bS
&=
\bQ^{1/2}\widehat\bC_1
\left(\widehat\bC_1+\widehat\bC_2\right)^{-1}
\widehat\bC_1\bQ^{1/2}.
\end{aligned}
\end{equation*} For every unit vector \(\bx\),
\[
\bx^{\top}\bC_1\left(\bC_1+\bC_2\right)^{-1}\bC_1\bx
\ge
\frac{\lambda_{\min}\left(\bC_1\right)^2}
{\lambda_{\max}\left(\bC_1+\bC_2\right)}
\ge
\frac{c}{n},
\]
where the value of \(c\) is reduced if necessary. Moreover,
\(\bm{\zero}\preceq\bC_1\left(\bC_1+\bC_2\right)^{-1}\bC_1\preceq\bC_1\).
Congruence by \(\bQ^{1/2}\) therefore gives
\begin{equation}
\frac{c}{n}\bI_p
\preceq
\bS,\widehat\bS
\preceq
\frac{C}{n}\bI_p.
\label{eq:plugin-S-eigenvalue-bounds}
\end{equation}
The condition \(\Tr\left(\bS\right)>4\left\|\bS\right\|_2\) implies \(p>4\). For $\bM\in\left\{\bS,\widehat\bS\right\}$, $\Tr\left(\bM\right)-\left\|\bM\right\|_2 \ge\left(p-1\right)\lambda_{\min}\left(\bM\right)\ge\frac{cp}{2n}$. The shrinkage scales therefore satisfy
\begin{equation}
 c\frac{p}{n}
\le
s_0\wedge\widehat s
\le
s_0\vee\widehat s
\le
C\frac{p}{n},
\qquad
s_0=\Tr\left(\bS\right)-\left\|\bS\right\|_2,
\qquad
\widehat s=\Tr\left(\widehat\bS\right)-\left\|\widehat\bS\right\|_2.
\label{eq:plugin-scale-bounds}
\end{equation}

The identity \(\bW_2^{-1}=\bI_p+\bC_2\bC_1^{-1}\) and the spectral bounds imply
\(\left\|\bQ^{1/2}\bW_2^{-1}\bQ^{-1/2}\right\|_2\le C\). Proposition~\ref{prop:plugin-stability} then gives
\begin{equation}
\left\|
\bQ^{1/2}
\left(\widehat\bW_2\bW_2^{-1}-\bI_p\right)
\bQ^{-1/2}
\right\|_2
\le
C\eta_n.
\label{eq:plugin-relative-weight-bound}
\end{equation}
Apply \eqref{eq:plugin-truncation-relative-matrix-bound} with
\(\bB=\widehat\bW_2\bW_2^{-1}\) and
\(\bu=\bW_2\left(\tbtheta_2-\tbtheta_1\right)\) when the left-hand side of
\eqref{eq:plugin-relative-weight-bound} is at most \(1/2\). Equations
\eqref{eq:plugin-scale-bounds} and \eqref{eq:plugin-relative-weight-bound} give
\[
\left\|
\mathcal T_{s_0}\left(\widehat\bW_2\left(\tbtheta_2-\tbtheta_1\right)\right)
-
\mathcal T_{s_0}\left(\bW_2\left(\tbtheta_2-\tbtheta_1\right)\right)
\right\|_{\bQ}
\le
C\sqrt{\frac{p}{n}}\eta_n.
\]
If the left-hand side of \eqref{eq:plugin-relative-weight-bound} exceeds \(1/2\), then
\(\eta_n\ge c\). Equation \eqref{eq:plugin-truncation-input-lipschitz} gives a bound of \(2\sqrt{s_0}\) for the difference of the two truncated vectors, and \eqref{eq:plugin-scale-bounds} gives
\(2\sqrt{s_0}\le C\sqrt{p/n}\,\eta_n\) after increasing \(C\). Thus, for every realization,
\begin{equation}
\left\|
\mathcal T_{s_0}\left(\widehat\bW_2\left(\tbtheta_2-\tbtheta_1\right)\right)
-
\mathcal T_{s_0}\left(\bW_2\left(\tbtheta_2-\tbtheta_1\right)\right)
\right\|_{\bQ}
\le
C\sqrt{\frac{p}{n}}\eta_n.
\label{eq:plugin-direction-map-perturbation}
\end{equation}
Proposition~\ref{prop:plugin-stability}, \eqref{eq:plugin-truncation-scale-lipschitz}, and
\eqref{eq:plugin-scale-bounds} also give
\begin{equation}
\left\|
\mathcal T_{\widehat s}\left(\widehat\bW_2\left(\tbtheta_2-\tbtheta_1\right)\right)
-
\mathcal T_{s_0}\left(\widehat\bW_2\left(\tbtheta_2-\tbtheta_1\right)\right)
\right\|_{\bQ}
\le
\frac{\left|\widehat s-s_0\right|}{\sqrt{\widehat s\wedge s_0}}
\le
C\sqrt{\frac{p}{n}}\eta_n.
\label{eq:plugin-scale-map-perturbation}
\end{equation}
Combining \eqref{eq:plugin-direction-map-perturbation} and
\eqref{eq:plugin-scale-map-perturbation} yields
\begin{equation}
\left\|
\widehat\btheta_{\rm pl}
-
\left[
\tbtheta_1+
\mathcal T_{s_0}\left(\bW_2\left(\tbtheta_2-\tbtheta_1\right)\right)
\right]
\right\|_{\bQ}
\le
C\sqrt{\frac{p}{n}}\eta_n.
\label{eq:plugin-estimator-pointwise-perturbation}
\end{equation}
Squaring \eqref{eq:plugin-estimator-pointwise-perturbation}, taking expectations, and using
\(\left[\mathbb E\eta_n^2\right]^{1/2}=\eta_{n,2}\) proves
\begin{equation}
\left[
\mathbb E\left\|
\widehat\btheta_{\rm pl}
-
\left[
\tbtheta_1+
\mathcal T_{s_0}\left(\bW_2\left(\tbtheta_2-\tbtheta_1\right)\right)
\right]
\right\|_{\bQ}^2
\right]^{1/2}
\le
C\sqrt{\frac{p}{n}}\eta_{n,2}.
\label{eq:plugin-estimator-L2-perturbation}
\end{equation}

The size inequality in \eqref{eq:plugin-truncation-input-lipschitz}, the Gaussian risk of \(\tbtheta_1\), and
\eqref{eq:plugin-scale-bounds} imply
\begin{equation}
\mathbb E\left\|
\tbtheta_1+
\mathcal T_{s_0}\left(\bW_2\left(\tbtheta_2-\tbtheta_1\right)\right)
-
\btheta_1^\star
\right\|_{\bQ}^2
\le
2\mathbb E\left\|\tbtheta_1-\btheta_1^\star\right\|_{\bQ}^2+2s_0
\le
C\frac{p}{n}.
\label{eq:plugin-population-estimator-uniform-risk}
\end{equation}
Expanding the plug-in risk around the population covariance estimator and applying the Cauchy-Schwarz inequality gives
\begin{align}
\mathbb E\left\|\widehat\btheta_{\rm pl}-\btheta_1^\star\right\|_{\bQ}^2
&\le
\mathbb E\left\|
\tbtheta_1+
\mathcal T_{s_0}\left(\bW_2\left(\tbtheta_2-\tbtheta_1\right)\right)
-
\btheta_1^\star
\right\|_{\bQ}^2+
C\frac{p}{n}\eta_{n,2}.
\label{eq:plugin-risk-square-expansion}
\end{align}
Here \eqref{eq:plugin-estimator-L2-perturbation},
\eqref{eq:plugin-population-estimator-uniform-risk}, and
\(\eta_{n,2}=o\left(1\right)\) control the cross term and the squared perturbation term. Applying
Theorem~\ref{thm:risk-bound-shat} to the first term in
\eqref{eq:plugin-risk-square-expansion} gives the displayed risk bound in the theorem.

The bound is uniform over the source discrepancy. Indeed,
\begin{equation*}
0
\le
\frac{
\Tr\left(\bS\right)
\left\|\bW_2\left(\btheta_2^\star-\btheta_1^\star\right)\right\|_{\bQ}^2
}{
\Tr\left(\bS\right)
+
\left\|\bW_2\left(\btheta_2^\star-\btheta_1^\star\right)\right\|_{\bQ}^2
}
\le
\Tr\left(\bS\right)
\le
C\frac{p}{n},
\end{equation*}
and the other terms in the theorem are bounded by \(Cp/n\) through
\eqref{eq:plugin-S-eigenvalue-bounds} and the covariance spectral bounds. This proves the uniform
\(O\left(p/n\right)\) assertion.

Finally, the covariance identity in the proof of Theorem~\ref{thm:risk-bound-shat} gives
\begin{equation}
\frac{1}{n_1}\Tr\left(\bQ\bSigma_1\right)
=
\left\langle
\left(n_1\bSigma_1^{-1}+n_2\bSigma_2^{-1}\right)^{-1},
\bQ
\right\rangle
+
\Tr\left(\bS\right).
\label{eq:plugin-baseline-risk-decomposition}
\end{equation}
If
\(\left\|\bW_2\left(\btheta_2^\star-\btheta_1^\star\right)\right\|_{\bQ}^2
\le C_0\Tr\left(\bS\right)\), then
\[
\frac{
\Tr\left(\bS\right)
\left\|\bW_2\left(\btheta_2^\star-\btheta_1^\star\right)\right\|_{\bQ}^2
}{
\Tr\left(\bS\right)
+
\left\|\bW_2\left(\btheta_2^\star-\btheta_1^\star\right)\right\|_{\bQ}^2
}
\le
\frac{C_0}{1+C_0}\Tr\left(\bS\right).
\]
Subtracting the risk bound from \eqref{eq:plugin-baseline-risk-decomposition} yields
\begin{equation}
\mathbb E\left\|\tbtheta_1-\btheta_1^\star\right\|_{\bQ}^2
-
\mathbb E\left\|\widehat\btheta_{\rm pl}-\btheta_1^\star\right\|_{\bQ}^2
\ge
\frac{1}{1+C_0}\Tr\left(\bS\right)
-4\left\|\bS\right\|_2
-C\frac{p}{n}\eta_{n,2}.
\label{eq:plugin-strict-improvement-gap}
\end{equation}
Equation \eqref{eq:plugin-S-eigenvalue-bounds} gives
\(\Tr\left(\bS\right)\ge cp/n\). Dividing
\eqref{eq:plugin-strict-improvement-gap} by \(\Tr\left(\bS\right)\), and using
\(\left\|\bS\right\|_2/\Tr\left(\bS\right)\to0\) and
\(\eta_{n,2}\to0\), gives a limit inferior of at least
\(1/\left(1+C_0\right)>0\). The strict comparison follows for all sufficiently large \(n\).
\end{proof}

\begin{proof}[Proof of Theorem~\ref{thm:plugin-sequential}]
Fix $\rho\in\left(0,1\right)$. Apply Theorem~\ref{thm:sequential-unified} with $\tbtheta_j^{\circ}=\tbtheta_j, \quad \bC_j^{\circ}=\frac{1}{n_j}\widehat\bSigma_j, \quad j\in\left[m\right]$. The assumed eigenvalue bounds and \(n_j\asymp n\) give fixed constants \(a_->0\) and \(a_+<\infty\) such that
\begin{equation}
\frac{a_-}{n}\bI_p
\preceq
\frac{1}{n_j}\widehat\bSigma_j
\preceq
\frac{a_+}{n}\bI_p,
\qquad j\in\left[m\right],
\label{eq:p-input-spectral-bounds-revised}
\end{equation}
for every realization. Define $\eta_n= \max_{j\in\left[m\right]} \left\|\widehat\bSigma_j-\bSigma_j\right\|_2$. The identity $\widehat\bSigma_j^{-1}-\bSigma_j^{-1} = \widehat\bSigma_j^{-1} \left(\bSigma_j-\widehat\bSigma_j\right) \bSigma_j^{-1}$ and the eigenvalue bounds imply
\begin{equation}
\max_{j\in\left[m\right]}
\left\|\widehat\bSigma_j^{-1}-\bSigma_j^{-1}\right\|_2
\le C\eta_n.
\label{eq:p-input-inverse-bound-revised}
\end{equation}

For every \(J\subseteq\mathcal I_0\), define
\begin{equation}
\begin{aligned}
\bV_J
&=
\left(
 n_1\bSigma_1^{-1}
 +\sum_{k\in J}n_k\bSigma_k^{-1}
\right)^{-1},
&
\overline\btheta_J
&=
\bV_J
\left(
 n_1\bSigma_1^{-1}\tbtheta_1
 +\sum_{k\in J}n_k\bSigma_k^{-1}\tbtheta_k
\right),
\\
\widehat\bV_J
&=
\left(
 n_1\widehat\bSigma_1^{-1}
 +\sum_{k\in J}n_k\widehat\bSigma_k^{-1}
\right)^{-1},
&
\widehat{\overline\btheta}_J
&=
\widehat\bV_J
\left(
 n_1\widehat\bSigma_1^{-1}\tbtheta_1
 +\sum_{k\in J}n_k\widehat\bSigma_k^{-1}\tbtheta_k
\right).
\end{aligned}
\label{eq:p-subset-objects-revised}
\end{equation}
For \(j\in\mathcal I_0\setminus J\), let
\begin{equation}
\begin{aligned}
\bW_{j,J}
&=
\bV_J
\left(
\bV_J+\frac{1}{n_j}\bSigma_j
\right)^{-1},
&
\widehat\bW_{j,J}
&=
\widehat\bV_J
\left(
\widehat\bV_J+\frac{1}{n_j}\widehat\bSigma_j
\right)^{-1},
\\
\bu_{j,J}
&=
\bW_{j,J}
\left(\tbtheta_j-\overline\btheta_J\right),
&
\widehat\bu_{j,J}
&=
\widehat\bW_{j,J}
\left(\tbtheta_j-\widehat{\overline\btheta}_J\right),
\\
\bS_{j,J}
&=
\bQ^{1/2}\bW_{j,J}\bV_J\bQ^{1/2},
&
\widehat\bS_{j,J}
&=
\bQ^{1/2}\widehat\bW_{j,J}\widehat\bV_J\bQ^{1/2}.
\end{aligned}
\label{eq:p-candidate-objects-revised}
\end{equation}
The precision matrices in \eqref{eq:p-subset-objects-revised} have eigenvalues of order \(n\), uniformly over \(J\). Combining their inverse difference identity with \eqref{eq:p-input-inverse-bound-revised} gives
\begin{equation}
\max_{J\subseteq\mathcal I_0}
\left\|\widehat\bV_J-\bV_J\right\|_2
\le C\frac{\eta_n}{n}.
\label{eq:p-subset-covariance-perturbation-revised}
\end{equation}
For the candidate weights, subtract the two expressions in \eqref{eq:p-candidate-objects-revised} and use $\widehat\bR^{-1}-\bR^{-1} = \widehat\bR^{-1} \left(\bR-\widehat\bR\right) \bR^{-1}$ with
\(
\bR=\bV_J+n_j^{-1}\bSigma_j
\)
and
\(
\widehat\bR=\widehat\bV_J+n_j^{-1}\widehat\bSigma_j
\).
Equations \eqref{eq:p-input-inverse-bound-revised} and \eqref{eq:p-subset-covariance-perturbation-revised} then yield
\begin{equation}
\max_{J\subsetneq\mathcal I_0}
\max_{j\in\mathcal I_0\setminus J}
\left\|\widehat\bW_{j,J}-\bW_{j,J}\right\|_2
\le C\eta_n.
\label{eq:p-weight-perturbation-revised}
\end{equation}
The aggregation weights in \eqref{eq:p-subset-objects-revised} sum to \(\bI_p\). Subtracting the two aggregate representations and applying \eqref{eq:p-input-inverse-bound-revised} and \eqref{eq:p-subset-covariance-perturbation-revised} gives
\begin{equation}
\left\|\widehat{\overline\btheta}_J-\overline\btheta_J\right\|_{\bQ}
\le
C\eta_n
\max_{k\in J}
\left\|\tbtheta_k-\tbtheta_1\right\|_{\bQ},
\label{eq:p-subset-estimator-perturbation-revised}
\end{equation}
where the maximum is zero when \(J=\varnothing\). From \eqref{eq:p-candidate-objects-revised},
\[
\widehat\bu_{j,J}-\bu_{j,J}
=
\left(\widehat\bW_{j,J}-\bW_{j,J}\right)
\left(\tbtheta_j-\overline\btheta_J\right)
+
\widehat\bW_{j,J}
\left(\overline\btheta_J-\widehat{\overline\btheta}_J\right).
\]
The operator norms of \(\widehat\bW_{j,J}\) are uniformly bounded. Therefore \eqref{eq:p-weight-perturbation-revised} and \eqref{eq:p-subset-estimator-perturbation-revised} imply
\begin{equation*}
\left\|\widehat\bu_{j,J}-\bu_{j,J}\right\|_{\bQ}
\le
C\eta_n
\left[
\left\|\tbtheta_j-\overline\btheta_J\right\|_{\bQ}
+
\max_{k\in J}
\left\|\tbtheta_k-\tbtheta_1\right\|_{\bQ}
\right].
\end{equation*}
Finally, the identity
\[
\widehat\bS_{j,J}-\bS_{j,J}
=
\bQ^{1/2}
\left(\widehat\bW_{j,J}-\bW_{j,J}\right)
\widehat\bV_J\bQ^{1/2}
+
\bQ^{1/2}\bW_{j,J}
\left(\widehat\bV_J-\bV_J\right)
\bQ^{1/2}
\]
gives the matrix bound
\begin{equation*}
\max_{J\subsetneq\mathcal I_0}
\max_{j\in\mathcal I_0\setminus J}
\left\|\widehat\bS_{j,J}-\bS_{j,J}\right\|_2
\le C\frac{\eta_n}{n}.
\end{equation*}
Consequently, if
\(
s_{j,J}=\Tr\left(\bS_{j,J}\right)-\left\|\bS_{j,J}\right\|_2
\)
and
\(
\widehat s_{j,J}=\Tr\left(\widehat\bS_{j,J}\right)-\left\|\widehat\bS_{j,J}\right\|_2
\), then the trace and operator norm inequalities give
\begin{equation}
\begin{aligned}
\left|\widehat s_{j,J}-s_{j,J}\right|
&\le
\left(p+1\right)
\left\|\widehat\bS_{j,J}-\bS_{j,J}\right\|_2
\le C\frac{p\eta_n}{n},\\
\left|
\left[\Tr\left(\widehat\bS_{j,J}\right)-2\left\|\widehat\bS_{j,J}\right\|_2\right]_+
-
\left[\Tr\left(\bS_{j,J}\right)-2\left\|\bS_{j,J}\right\|_2\right]_+
\right|
&\le
\left(p+2\right)
\left\|\widehat\bS_{j,J}-\bS_{j,J}\right\|_2
\le C\frac{p\eta_n}{n}.
\end{aligned}
\label{eq:p-scale-perturbation-revised}
\end{equation}

Let \(\mathcal G_n^{\mathrm g}\) be the event defined in \eqref{eq:g-unified-good-event-probability}, and define
\begin{equation*}
\mathcal G_n^{\mathrm p}
=
\mathcal G_n^{\mathrm g}
\cap
\left\{
\eta_n\le \eta_{n,2}^{1/2}+n^{-1}
\right\}.
\end{equation*}
Since \(\mathbb E\eta_n^2=\eta_{n,2}^2\), Markov's inequality and \eqref{eq:g-unified-good-event-probability} give
\begin{equation}
\mathbb P\left(\mathcal G_n^{\mathrm p}\right)
\ge
1-\pi_n^{\mathrm g}-\eta_{n,2}
=
1-o\left(1\right).
\label{eq:p-good-event-probability-revised}
\end{equation}
On \(\mathcal G_n^{\mathrm g}\), equations \eqref{eq:g-unified-R-bound}, \eqref{eq:g-unified-U-bound}, and Assumption~\ref{assump:seq-1} imply
\begin{equation*}
\max_{J\subsetneq\mathcal I_0}
\max_{j\in\mathcal I_0\setminus J}
\left\|\tbtheta_j-\overline\btheta_J\right\|_{\bQ}
+
\max_{k\in\mathcal I_0}
\left\|\tbtheta_k-\tbtheta_1\right\|_{\bQ}
\le C\sqrt{\frac{p}{n}}.
\end{equation*}
Equations \eqref{eq:p-subset-estimator-perturbation-revised} through \eqref{eq:p-scale-perturbation-revised} therefore show that, on \(\mathcal G_n^{\mathrm p}\),
\begin{equation}
\mathcal R_n^{\circ}
\le C\sqrt{\frac{p}{n}},
\qquad
\mathcal U_n^{\circ}
\le C\sqrt{\frac{p}{n}},
\qquad
\mathcal Z_n^{\circ}
\le
\varepsilon_n^{\mathrm g}
+C\left(\eta_{n,2}^{1/2}+n^{-1}\right).
\label{eq:p-unified-quantities-revised}
\end{equation}
The bound for \(\mathcal Z_n^{\circ}\) in \eqref{eq:p-unified-quantities-revised} follows from \eqref{eq:g-unified-Z-bound} and the polarization identity
\[
\left|
\left\|\widehat\bu_{j,J}\right\|_{\bQ}^2
-
\left\|\bu_{j,J}\right\|_{\bQ}^2
\right|
\le
\left\|\widehat\bu_{j,J}-\bu_{j,J}\right\|_{\bQ}
\left(
2\left\|\bu_{j,J}\right\|_{\bQ}
+
\left\|\widehat\bu_{j,J}-\bu_{j,J}\right\|_{\bQ}
\right),
\]
and the trace lower bound \eqref{eq:sequential-local-trace-lower}. The assumptions give $\varepsilon_n^{\mathrm g} + C\left(\eta_{n,2}^{1/2}+n^{-1}\right) \longrightarrow 0$.

For the fourth moment condition in Theorem~\ref{thm:sequential-unified}, equations \eqref{eq:g-fourth-moment} and \eqref{eq:p-subset-estimator-perturbation-revised}, together with \(\eta_n\le C\), give
\begin{equation}
\mathbb E
\left\|\tbtheta_1-\widehat{\overline\btheta}_{\mathcal I_0}\right\|_{\bQ}^4
\le
C\left(\frac{p}{n}\right)^2.
\label{eq:p-fourth-moment-revised}
\end{equation}
Equations \eqref{eq:p-input-spectral-bounds-revised}, \eqref{eq:p-good-event-probability-revised}, \eqref{eq:p-unified-quantities-revised}, and \eqref{eq:p-fourth-moment-revised} verify every condition of Theorem~\ref{thm:sequential-unified}. Hence, for all sufficiently large \(n\),
\begin{equation}
\mathbb P
\left(
\left\{\widehat j_1,\ldots,\widehat j_{m_0}\right\}
=
\mathcal I_0
\right)
\ge 1-\rho,
\qquad
\mathbb E
\left\|\widehat\btheta_{\rm seq,pl}-\widehat{\overline\btheta}_{\mathcal I_0}\right\|_{\bQ}^2
=o\left(\frac{p}{n}\right).
\label{eq:p-unified-conclusion-revised}
\end{equation}

The aggregate perturbation bound \eqref{eq:p-subset-estimator-perturbation-revised}, the Cauchy-Schwarz inequality, and \eqref{eq:g-fourth-moment} imply
\begin{equation}
\mathbb E
\left\|\widehat{\overline\btheta}_{\mathcal I_0}-\overline\btheta_{\mathcal I_0}\right\|_{\bQ}^2
\le
C\eta_{n,2}\frac{p}{n}
=o\left(\frac{p}{n}\right).
\label{eq:p-oracle-aggregate-perturbation-revised}
\end{equation}
Combining \eqref{eq:p-unified-conclusion-revised}, \eqref{eq:p-oracle-aggregate-perturbation-revised}, and \eqref{eq:g-oracle-risk} gives
\begin{equation}
\mathbb E
\left\|\widehat\btheta_{\rm seq,pl}-\btheta_1^\star\right\|_{\bQ}^2
=
\left\langle\bV_{\mathcal I_0},\bQ\right\rangle
+o\left(\frac{p}{n}\right).
\label{eq:p-sequential-risk-revised}
\end{equation}

We finally compare the plug-in single-step estimator. Define the population and plug-in all-source weights by
\begin{equation*}
\bW_j
=
\left(
 n_1\bSigma_1^{-1}
 +\sum_{k\in\mathcal I}n_k\bSigma_k^{-1}
\right)^{-1}
 n_j\bSigma_j^{-1},
\qquad
\widehat\bW_j
=
\left(
 n_1\widehat\bSigma_1^{-1}
 +\sum_{k\in\mathcal I}n_k\widehat\bSigma_k^{-1}
\right)^{-1}
 n_j\widehat\bSigma_j^{-1}.
\end{equation*}
Set
\(
\bu_{\mathcal I}=\sum_{j\in\mathcal I}\bW_j\left(\tbtheta_j-\tbtheta_1\right)
\)
and
\(
\widehat\bu_{\mathcal I}=\sum_{j\in\mathcal I}\widehat\bW_j\left(\tbtheta_j-\tbtheta_1\right)
\).
The perturbation arguments leading to \eqref{eq:p-weight-perturbation-revised}, together with Gaussian fourth moment bounds, give
\begin{equation}
\left[
\mathbb E
\left\|\widehat\bu_{\mathcal I}-\bu_{\mathcal I}\right\|_{\bQ}^2
\right]^{1/2}
\le
C\left(
\eta_{n,2}\delta_{\max}
+
\eta_{n,2}^{1/2}\sqrt{\frac{p}{n}}
\right).
\label{eq:p-single-step-direction-perturbation-revised}
\end{equation}
The conditions
\(
\delta_{\mathcal I}\gg\eta_{n,2}\delta_{\max}
\)
and
\(
\delta_{\mathcal I}\gg\sqrt{p/n}
\)
show that the right side of \eqref{eq:p-single-step-direction-perturbation-revised} is \(o\left(\delta_{\mathcal I}\right)\). Equation \eqref{eq:g-single-direction-tail} and Markov's inequality then yield
\begin{equation}
\mathbb P
\left(
\left\|\widehat\bu_{\mathcal I}\right\|_{\bQ}
\ge
\frac{\delta_{\mathcal I}}{4}
\right)
=1-o\left(1\right).
\label{eq:p-single-step-direction-lower-revised}
\end{equation}
For every admissible plug-in single-step scale \(s\), equation~\eqref{eq:p-input-spectral-bounds-revised} gives $0 \le s \le C\frac{p}{n}$. On the event in \eqref{eq:p-single-step-direction-lower-revised},
\[
\left\|\widehat\btheta_{\rm ss,pl}-\tbtheta_1\right\|_{\bQ}^2
\le
\frac{s^2}{\left\|\widehat\bu_{\mathcal I}\right\|_{\bQ}^2}
\le
C\frac{\left(p/n\right)^2}{\delta_{\mathcal I}^2},
\]
whereas every realization satisfies $\left\|\widehat\btheta_{\rm ss,pl}-\tbtheta_1\right\|_{\bQ}^2 \le s \le C\frac{p}{n}$. Taking expectations and using~\eqref{eq:p-single-step-direction-lower-revised}~gives
\begin{equation}
\mathbb E
\left\|\widehat\btheta_{\rm ss,pl}-\tbtheta_1\right\|_{\bQ}^2
=o\left(\frac{p}{n}\right).
\label{eq:p-single-step-update-revised}
\end{equation}
Equations \eqref{eq:p-single-step-update-revised}, \eqref{eq:g-target-risk}, and the Cauchy-Schwarz inequality imply
\begin{equation}
\mathbb E
\left\|\widehat\btheta_{\rm ss,pl}-\btheta_1^\star\right\|_{\bQ}^2
=
\left\langle\frac{1}{n_1}\bSigma_1,\bQ\right\rangle
+o\left(\frac{p}{n}\right).
\label{eq:p-single-step-risk-revised}
\end{equation}
Subtracting \eqref{eq:p-sequential-risk-revised} from \eqref{eq:p-single-step-risk-revised} and applying \eqref{eq:g-oracle-trace-gap} gives a strictly positive difference for all sufficiently large \(n\). This proves the last assertion.
\end{proof}

\subsection{Proofs of the Results for General Loss Functions}

\begin{proof}[Proof of Proposition~\ref{prop:optimal-aggregation}]
Let $\bV_{\left\{\bW_j\right\}}=\sum_{j=1}^m n_j^{-1}\bW_j\bSigma_j\bW_j^\top$ and
$\bB_{j,n}=\bV_{\left\{\bW_j\right\}}^{-1/2}n_j^{-1/2}\bW_j\bSigma_j^{1/2}$.
Independence gives joint convergence of the $m$ standardized local estimators to independent standard normal vectors. Moreover,
\[
\bV_{\left\{\bW_j\right\}}^{-1/2}\left(\tbtheta_{\left\{\bW_j\right\}}-\btheta_1^\star\right)
=\sum_{j=1}^m\bB_{j,n}\sqrt{n_j}\bSigma_j^{-1/2}\left(\tbtheta_j-\btheta_j^\star\right),
\qquad
\sum_{j=1}^m\bB_{j,n}\bB_{j,n}^\top=\bI_p.
\]
Thus $\left\|\bB_{j,n}\right\|_2\le1$. For fixed $p,m$, every subsequence has a further subsequence along which all coefficient matrices converge. The joint limit along every such subsequence is $\mathcal N\left(\bm0,\bI_p\right)$, proving the stated convergence even when the deterministic weights vary with $n$.

Let $\bV^\star=\left(\sum_{k=1}^m n_k\bSigma_k^{-1}\right)^{-1}$. The optimal weights satisfy
$\sum_j\bW_j^\star=\bI_p$ and $n_j^{-1}\bW_j^\star\bSigma_j=\bV^\star$.
Since $\sum_j\left(\bW_j-\bW_j^\star\right)=\bm0$, the cross terms vanish in the square completion
\[
\begin{aligned}
&\sum_{j=1}^m\frac1{n_j}\bW_j\bSigma_j\bW_j^\top
-\sum_{j=1}^m\frac1{n_j}\bW_j^\star\bSigma_j\bW_j^{\star\top}=\sum_{j=1}^m\frac1{n_j}\left(\bW_j-\bW_j^\star\right)\bSigma_j\left(\bW_j-\bW_j^\star\right)^\top
\succeq\bm0.
\end{aligned}
\]
This proves Loewner minimality.
\end{proof}

\begin{lemma}\label{lem:tail-bounds}
Let \(r_0=\min\left(r,\underline\lambda/\left(4L\right),\left(B_\Theta-B^\star\right)/2\right)\), \(\gamma_1=\left\lceil c\right\rceil+11\),
$
C_{\mathrm{F1},\mathrm{s}}=\frac{16\tau}{3}\sqrt{\log 9+1}$, $C_{\mathrm{F1},\mathrm{h}}=16\tau^2\sqrt{\gamma_1+10+\frac{\log\left(9\left(1+2r\right)\right)}{\log 2}}+\frac{4L}{\sqrt{\log 2}}$, $c_{\mathrm{F1},\mathrm{e}}=5$,
and \(C_{\mathrm{F1},4}=9\left(4C_{\mathrm{F1},\mathrm{s}}/\underline\lambda\right)^4+5\left(B_\Theta+B^\star\right)^4\). Under Assumptions~\ref{assump:1}--\ref{assump:4} and \(p\log\left(n_j\right)/n_j=o\left(1\right)\), for every \(t\ge0\),
\begin{equation}\label{eq:M-score-all-t}
\PP\left(\left\|\nabla f_j\left(\btheta_j^\star\right)\right\|_2>C_{\mathrm{F1},\mathrm{s}}\sqrt{\frac{p+t}{n_j}}\right)\le2e^{-t}.
\end{equation}
There exists an event \(\mathcal E_j\) satisfying \(\PP\left(\mathcal E_j^c\right)\le c_{\mathrm{F1},\mathrm{e}}n_j^{-6}\) on which
\begin{align}
\left\|\nabla f_j\left(\btheta_j^\star\right)\right\|_2&\le C_{\mathrm{F1},\mathrm{s}}\sqrt{\frac{p+7\log\left(n_j\right)}{n_j}},\label{eq:M-score-local-lemma}\\
\sup_{\btheta\in\mathcal B_j\left(r\right)}\left\|\nabla^2f_j\left(\btheta\right)-\EE\nabla^2\ell_j\left(\btheta,\bx\right)\right\|_2&\le C_{\mathrm{F1},\mathrm{h}}\sqrt{\frac{p\log\left(n_j\right)}{n_j}},\label{eq:M-Hessian-uniform-lemma}\\
\inf_{\btheta\in\mathcal B_j\left(r_0\right)}\lambda_{\min}\left(\nabla^2f_j\left(\btheta\right)\right)&\ge\frac{\underline\lambda}{2},\label{eq:M-local-curvature-lemma}\\
\left\|\tbtheta_j-\btheta_j^\star\right\|_2&\le\frac{4}{\underline\lambda}\left\|\nabla f_j\left(\btheta_j^\star\right)\right\|_2. \label{eq:M-local-minimizer-bound-lemma}\\
\nabla f_j\left(\tbtheta_j\right)&=\bm0. \label{eq:M-local-minimizer}
\end{align}
The fourth moment also satisfies
\begin{equation}\label{eq:M-local-fourth-lemma}
\EE\left\|\tbtheta_j-\btheta_j^\star\right\|_2^4\le C_{\mathrm{F1},4}\left(\frac{p}{n_j}\right)^2.
\end{equation}
\end{lemma}

\begin{lemma}\label{lem:consistency-HhatVhat}
Under Assumptions~\ref{assump:1}--\ref{assump:4} and \(p\log\left(n_j\right)/n_j=o\left(1\right)\), let \(\gamma_2=\left\lceil c\right\rceil+6\), 
$
C_{\mathrm{F2},\mathrm{r}}=\frac{12C_{\mathrm{F1},\mathrm{s}}}{\underline\lambda}$,  $C_{\mathrm{F2},\mathrm{H}}=C_{\mathrm{F1},\mathrm{h}}+LC_{\mathrm{F2},\mathrm{r}}$,  $C_{\mathrm{F2},\mathrm{grid}}=8\tau^2\sqrt{\gamma_2+3+\frac{\log\left(9\left(1+2r\right)\right)}{\log 2}}$,
$
C_{\mathrm{F2},\mathrm{V}}=2C_{\mathrm{F2},\mathrm{grid}}+\frac{8L}{\sqrt{\log 2}}+8\sqrt2\tau^3C_{\mathrm{F2},\mathrm{r}}$, $c_{\mathrm{F2},\mathrm{e}}=8$.
Then
\[
\left\|\widehat\bH_j-\bH_j^\star\right\|_2\le C_{\mathrm{F2},\mathrm{H}}\sqrt{\frac{p\log\left(n_j\right)}{n_j}},\quad \left\|\widehat\bV_j-\bV_j^\star\right\|_2\le C_{\mathrm{F2},\mathrm{V}}\sqrt{\frac{p\log\left(n_j\right)}{n_j}},
\]
with probability at least \(1-c_{\mathrm{F2},\mathrm{e}}n_j^{-2}\).
\end{lemma}

\begin{lemma}\label{lem:deterministic-clipping}
Let \(\underline\sigma_0=\underline\lambda/\overline\lambda^2\), \(\overline\sigma_0=\overline\lambda/\underline\lambda^2\),
$
C_{\mathrm{F3},\bSigma}=\frac{10\overline\lambda C_{\mathrm{F2},\mathrm{H}}}{\underline\lambda^3}+\frac{2C_{\mathrm{F2},\mathrm{V}}}{\underline\lambda^2}$, and $C_{\mathrm{F3},\bSigma^{-1}}=\frac{2C_{\mathrm{F3},\bSigma}}{\underline\sigma_0^2}
$.
Under Assumptions~\ref{assump:1}--\ref{assump:4}, \(m=O\left(1\right)\), \(n_j\asymp n\) uniformly over \(j\in\left[m\right]\), and \(p\log\left(n\right)/n=o\left(1\right)\), there exists an event \(\mathcal E_n\) satisfying \(\PP\left(\mathcal E_n^c\right)\le8\sum_{j=1}^m n_j^{-2}\) such that, on \(\mathcal E_n\), for all $j\in\left[m\right]$ and sufficiently large \(n\),
\[
\widehat\bH_j^{\rm c}=\widehat\bH_j,\, \widehat\bSigma_j^{\rm c}=\widehat\bH_j^{-1}\widehat\bV_j\widehat\bH_j^{-1},\,
\frac{\underline\lambda}{2}\bI_p\preceq\widehat\bH_j^{\rm c}\preceq2\overline\lambda\bI_p,\, \frac{\underline\sigma_0}{2}\bI_p\preceq\widehat\bSigma_j^{\rm c}\preceq2\overline\sigma_0\bI_p,
\]
and the estimation errors satisfy
\[
\max_{j\in\left[m\right]}\left\|\widehat\bSigma_j^{\rm c}-\bSigma_j\right\|_2\le C_{\mathrm{F3},\bSigma}\max_{j\in\left[m\right]}\sqrt{\frac{p\log\left(n_j\right)}{n_j}},\quad \max_{j\in\left[m\right]}\left\|\left(\widehat\bSigma_j^{\rm c}\right)^{-1}-\bSigma_j^{-1}\right\|_2\le C_{\mathrm{F3},\bSigma^{-1}}\max_{j\in\left[m\right]}\sqrt{\frac{p\log\left(n_j\right)}{n_j}}.
\]
\end{lemma}

\begin{lemma}\label{lem:quadratic-influence}
Let \(\bZ=\sum_{a=1}^N\bY_a\), where the \(\bY_a\) are independent and mean-zero and satisfy \(\left\|\bu^\top\bY_a\right\|_{\psi_1}\le K\left(\bu^\top\bK_a\bu\right)^{1/2}\) for every \(\bu\in\mathbb S^{p-1}\), with \(\bK_a=\EE\left(\bY_a\bY_a^\top\right)\succ0\). Let \(\bK=\sum_{a=1}^N\bK_a\). For every fixed positive definite \(\bQ\),
\begin{equation}\label{eq:F4-main-bound}
\EE\left[\left\|\bZ\right\|_{\bQ}^2-\Tr\left(\bQ\bK\right)\right]^2\le48K^4\sum_{a=1}^N\Tr\left(\bQ\bK_a\right)^2+2\Tr\left(\bQ\bK\bQ\bK\right).
\end{equation}
Consequently, if \(\left\|\bK_a\right\|_2\le C_{\mathrm{F4},\mathrm{a}}n^{-2}\), \(N\le C_{\mathrm{F4},\mathrm{N}}n\), and \(\bK\preceq C_{\mathrm{F4},\mathrm{K}}n^{-1}\bI_p\), then
\begin{equation}\label{eq:F4-consequence}
\EE\left[\left\|\bZ\right\|_{\bQ}^2-\Tr\left(\bQ\bK\right)\right]^2\le2\left\|\bQ\right\|_2^2C_{\mathrm{F4},\mathrm{K}}^2\frac{p}{n^2}+48K^4\left\|\bQ\right\|_2^2C_{\mathrm{F4},\mathrm{N}}C_{\mathrm{F4},\mathrm{a}}^2\frac{p^2}{n^3}.
\end{equation}
\end{lemma}

\begin{proof}[Proof of Theorem~\ref{thm:risk-upper-general}]
Fix $\rho\in\left(0,1\right)$. Throughout the proof, $c,C>0$ depend only on the fixed constants in the assumptions, the spectral bounds for $\bQ$, and fixed bounds for $m$ and $n_j/n$. Their values may change between occurrences. All bounds indexed by $J$ and $j$ below hold uniformly over $J\subseteq\mathcal I_0$ and $j\in\mathcal I_0\setminus J$, whenever these indices occur. Choose fixed $0<\underline\kappa\le\overline\kappa<\infty$ with $\underline\kappa n\le n_j\le\overline\kappa n$ for every $j\in\left[m\right]$. Write
\[
\delta_{0,n}=\max_{j\in\mathcal I_0}\left\|\btheta_j^\star-\btheta_1^\star\right\|_{\bQ}
=o\left(\sqrt{p/n}\right),\qquad
\eta_n=\max_{j\in\left[m\right]}\left\|\widehat\bSigma_j^{\rm c}-\bSigma_j\right\|_2.
\]
We apply Theorem~\ref{thm:sequential-unified} with
$\tbtheta_j^\circ=\tbtheta_j$ and $\bC_j^\circ=n_j^{-1}\widehat\bSigma_j^{\rm c}$.
The clipping rule in \eqref{eq:clipped-sandwich} gives
\[
\frac{\varepsilon^{\rm c}}{\overline\kappa n}\bI_p
\preceq\bC_j^\circ\preceq
\frac{1}{\varepsilon^{\rm c}\underline\kappa n}\bI_p,
\qquad j\in\left[m\right],
\]
which verifies the first condition of Theorem~\ref{thm:sequential-unified} with $a_-=\varepsilon^{\rm c}/\overline\kappa$ and $a_+=1/\left(\varepsilon^{\rm c}\underline\kappa\right)$. Assumption~\ref{assump:seq-1} gives
$\delta_{\min,n}/\sqrt{p/n}\to\infty$ with
$\delta_{\min,n}=\min_{j\in\mathcal I_1}\left\|\btheta_j^\star-\btheta_1^\star\right\|_{\bQ}$.
The remaining conditions concern $\mathcal R_n^\circ$, $\mathcal U_n^\circ$, $\mathcal Z_n^\circ$, and the fourth moment of $\tbtheta_1-\overline\btheta_{\mathcal I_0}^\circ$.

Define the influence vectors, linearized estimators, and remainders by
\begin{equation}
\begin{aligned}
\bxi_{ji}&=-\bH_j^{\star-1}\nabla\ell_j\left(\btheta_j^\star,\bx_{ji}\right),\\
\tbtheta_{j,\mathrm{lin}}&=\btheta_j^\star+\frac1{n_j}\sum_{i=1}^{n_j}\bxi_{ji},
\qquad \boldsymbol r_j=\tbtheta_j-\tbtheta_{j,\mathrm{lin}}.
\end{aligned}
\label{eq:M-proof-influence-definitions}
\end{equation}
Then $\EE\bxi_{ji}=\bm0$, $\EE\left(\bxi_{ji}\bxi_{ji}^\top\right)=\bSigma_j$, and
$\operatorname{Cov}\left(\tbtheta_{j,\mathrm{lin}}\right)=\bSigma_j/n_j$.
On the event $\mathcal E_j$ in Lemma~\ref{lem:tail-bounds}, set
$\widetilde\bH_{j,\mathrm{int}}=\int_0^1\nabla^2f_j\left(\btheta_j^\star+u\left(\tbtheta_j-\btheta_j^\star\right)\right)\,\mathrm du$.
Equation~\eqref{eq:M-local-minimizer} gives
$\bm0=\nabla f_j\left(\btheta_j^\star\right)+\widetilde\bH_{j,\mathrm{int}}\left(\tbtheta_j-\btheta_j^\star\right)$.
Combining this identity with \eqref{eq:M-proof-influence-definitions},
\eqref{eq:M-score-local-lemma}, \eqref{eq:M-Hessian-uniform-lemma}, and
\eqref{eq:M-local-minimizer-bound-lemma} yields
\begin{equation}
\begin{aligned}
\boldsymbol r_j
&=\left(\bH_j^{\star-1}-\widetilde\bH_{j,\mathrm{int}}^{-1}\right)\nabla f_j\left(\btheta_j^\star\right),\\
\left\|\widetilde\bH_{j,\mathrm{int}}-\bH_j^\star\right\|_2
&\le C\sqrt{\frac{p\log n_j}{n_j}}+L\left\|\tbtheta_j-\btheta_j^\star\right\|_2
\le C\sqrt{\frac{p\log n_j}{n_j}},\\
\left\|\boldsymbol r_j\right\|_{\bQ}
&\le C\sqrt{\frac{p\log n_j}{n_j}}
\sqrt{\frac{p+7\log n_j}{n_j}}
\le C\frac{p\log n_j}{n_j}.
\end{aligned}
\label{eq:compressed-M-linearization}
\end{equation}
Equations~\eqref{eq:M-score-local-lemma} and \eqref{eq:M-local-minimizer-bound-lemma} ensure that the segment lies in $\mathcal B_j\left(r_0\right)$ for sufficiently large $n$, while \eqref{eq:M-local-curvature-lemma} gives $\left\|\widetilde\bH_{j,\mathrm{int}}^{-1}\right\|_2\le2/\underline\lambda$.

Equation~\eqref{eq:M-score-all-t} and the identity
$\tbtheta_{j,\mathrm{lin}}-\btheta_j^\star=-\bH_j^{\star-1}\nabla f_j\left(\btheta_j^\star\right)$ imply
\begin{equation}
\begin{aligned}
\PP\left(\left\|\tbtheta_{j,\mathrm{lin}}-\btheta_j^\star\right\|_{\bQ}>C\sqrt{\frac{p+t}{n_j}}\right)&\le2e^{-t},\\
\EE\left\|\tbtheta_{j,\mathrm{lin}}-\btheta_j^\star\right\|_{\bQ}^4&\le C\left(\frac p{n_j}\right)^2.
\end{aligned}
\label{eq:M-proof-linearized-tail-moment}
\end{equation}
The fourth-moment bound follows from the tail integration in \eqref{eq:M-score-fourth-revised} and $\left\|\bH_j^{\star-1}\right\|_2\le1/\underline\lambda$. Together with \eqref{eq:M-local-fourth-lemma}, it gives $\EE\left\|\boldsymbol r_j\right\|_{\bQ}^4\le C\left(p/n_j\right)^2$. Splitting over $\mathcal E_j$ and its complement yields
\[
\EE\left\|\boldsymbol r_j\right\|_{\bQ}^2
\le C\frac{p^2\log^2 n_j}{n_j^2}
+\left[\EE\left\|\boldsymbol r_j\right\|_{\bQ}^4\right]^{1/2}\PP\left(\mathcal E_j^c\right)^{1/2}
\le C\frac{p^2\log^2 n_j}{n_j^2}+C\frac p{n_j}n_j^{-3}.
\]
Since $m$ is bounded, $n_j\asymp n$, and $p\log^2n/n=o\left(1\right)$,
\begin{equation}
\begin{aligned}
\EE\max_j\left\|\boldsymbol r_j\right\|_{\bQ}^2
&\le C\frac{p^2\log^2n}{n^2}=o\left(\frac pn\right),\\
\EE\max_j\left[\left\|\boldsymbol r_j\right\|_{\bQ}^4+
\left\|\tbtheta_{j,\mathrm{lin}}-\btheta_j^\star\right\|_{\bQ}^4\right]
&\le C\left(\frac pn\right)^2.
\end{aligned}
\label{eq:compressed-M-local-moments}
\end{equation}
Equations~\eqref{eq:F3-sandwich-error} and \eqref{eq:F3-clipping-inactive} give $\eta_n\le C\sqrt{p\log n/n}$ on $\mathcal E_n$, and \eqref{eq:F3-event-probability} gives $\PP\left(\mathcal E_n^c\right)\le Cn^{-2}$. The deterministic spectral bounds in \eqref{eq:clipped-sandwich} and \eqref{eq:F3-population-spectral-bounds} then imply
\begin{equation}
\begin{gathered}
\eta_n\le C,\quad
\EE\eta_n^2\le C\left(\frac{p\log n}{n}+n^{-2}\right)=o\left(1\right),\quad
\EE\eta_n^4\le C\left(\frac{p^2\log^2 n}{n^2}+n^{-2}\right)=o\left(1\right).
\end{gathered}
\label{eq:compressed-M-covariance-moments}
\end{equation}
The identity
\[
\left(\widehat\bSigma_j^{\rm c}\right)^{-1}-\bSigma_j^{-1}
=\left(\widehat\bSigma_j^{\rm c}\right)^{-1}\left(\bSigma_j-\widehat\bSigma_j^{\rm c}\right)\bSigma_j^{-1}
\]
and the same spectral bounds imply
$\max_j\left\|\left(\widehat\bSigma_j^{\rm c}\right)^{-1}-\bSigma_j^{-1}\right\|_2\le C\eta_n$,
whose second and fourth moments are therefore controlled by \eqref{eq:compressed-M-covariance-moments}.

For $J\subseteq\mathcal I_0$ and $j\in\mathcal I_0\setminus J$, define
\begin{equation}
\begin{aligned}
\bV_J&=\left(\sum_{k\in\left\{1\right\}\cup J}n_k\bSigma_k^{-1}\right)^{-1},\\
\overline\btheta_{J,\mathrm{lin}}&=\bV_J\sum_{k\in\left\{1\right\}\cup J}n_k\bSigma_k^{-1}\tbtheta_{k,\mathrm{lin}},
\qquad\overline\btheta_J^\star=\EE\overline\btheta_{J,\mathrm{lin}},\\
\bW_{j,J}&=\bV_J\left(\bV_J+n_j^{-1}\bSigma_j\right)^{-1},
\qquad\bS_{j,J}=\bQ^{1/2}\bW_{j,J}\bV_J\bQ^{1/2},\\
\boldsymbol d_{j,J}&=\bW_{j,J}\left(\tbtheta_{j,\mathrm{lin}}-\overline\btheta_{J,\mathrm{lin}}\right).
\end{aligned}
\label{eq:M-proof-population-objects}
\end{equation}
In \eqref{eq:compressed-sequential-covariances} and \eqref{eq:compressed-sequential-geometry}, take $\bC_k^\circ=n_k^{-1}\bSigma_k$ and $\bA=\bV_J$.
Equation~\eqref{eq:F3-population-spectral-bounds} verifies the required input bounds with $a_-=\underline\sigma_0/\overline\kappa$ and $a_+=\overline\sigma_0/\underline\kappa$. These substitutions give
\begin{equation}
\begin{gathered}
\frac cn\bI_p\preceq\bV_J\preceq\frac Cn\bI_p,\quad
\left\|\bW_{j,J}\right\|_{\bQ\to\bQ}\le C, \quad
c\frac pn\le\Tr\left(\bS_{j,J}\right)\le C\frac pn,\quad
\left\|\bS_{j,J}\right\|_2\le\frac Cn.
\end{gathered}
\label{eq:compressed-M-population-geometry}
\end{equation}
From \eqref{eq:M-proof-population-objects},
$\sum_{k\in\left\{1\right\}\cup J}\bV_Jn_k\bSigma_k^{-1}=\bI_p$ and
$\left\|\bV_Jn_k\bSigma_k^{-1}\right\|_{\bQ\to\bQ}\le C$. Consequently,
\begin{equation}
\begin{aligned}
\overline\btheta_J^\star-\btheta_1^\star
&=\bV_J\sum_{k\in J}n_k\bSigma_k^{-1}\left(\btheta_k^\star-\btheta_1^\star\right),\\
\EE\boldsymbol d_{j,J}&=\bW_{j,J}\left(\btheta_j^\star-\overline\btheta_J^\star\right),\\
\left\|\overline\btheta_J^\star-\btheta_1^\star\right\|_{\bQ}
+\left\|\EE\boldsymbol d_{j,J}\right\|_{\bQ}&\le C\delta_{0,n}.
\end{aligned}
\label{eq:M-proof-population-bias}
\end{equation}
Substituting \eqref{eq:M-proof-influence-definitions} into \eqref{eq:M-proof-population-objects} gives the independent-sum representation
\begin{equation}
\boldsymbol d_{j,J}-\EE\boldsymbol d_{j,J}
=\frac1{n_j}\sum_{i=1}^{n_j}\bW_{j,J}\bxi_{ji}
-\sum_{k\in\left\{1\right\}\cup J}\sum_{i=1}^{n_k}\bW_{j,J}\bV_J\bSigma_k^{-1}\bxi_{ki},
\label{eq:compressed-M-independent-sum}
\end{equation}
with covariance
\begin{equation}
\operatorname{Cov}\left(\boldsymbol d_{j,J}\right)
=\bW_{j,J}\left(\frac{\bSigma_j}{n_j}+\bV_J\right)\bW_{j,J}^\top
=\bW_{j,J}\bV_J.
\label{eq:M-proof-direction-covariance}
\end{equation}
There are $O\left(n\right)$ summands, each with covariance norm $O\left(n^{-2}\right)$. Assumption~\ref{assump:2} and the lower spectral bound for $\bSigma_k$ imply, for every deterministic vector $\bv$, $\left\|\bv^\top\bxi_{ki}\right\|_{\psi_1} \le C\left\|\bv\right\|_2 \le C\left(\bv^\top\bSigma_k\bv\right)^{1/2}$. This verifies the directional condition of Lemma~\ref{lem:quadratic-influence} for every transformed summand in \eqref{eq:compressed-M-independent-sum}. With $\boldsymbol\mu_{j,J}=\EE\boldsymbol d_{j,J}$, the quadratic expansion is
\begin{align*}
\left\|\boldsymbol d_{j,J}\right\|_{\bQ}^2-\Tr\left(\bS_{j,J}\right)
&=\left\|\boldsymbol d_{j,J}-\boldsymbol\mu_{j,J}\right\|_{\bQ}^2-\Tr\left(\bS_{j,J}\right)
+2\boldsymbol\mu_{j,J}^{\top}\bQ\left(\boldsymbol d_{j,J}-\boldsymbol\mu_{j,J}\right)
+\left\|\boldsymbol\mu_{j,J}\right\|_{\bQ}^2,\\
\EE\left[\left(\boldsymbol\mu_{j,J}^{\top}\bQ\left(\boldsymbol d_{j,J}-\boldsymbol\mu_{j,J}\right)\right)^2\right]
&=\boldsymbol\mu_{j,J}^{\top}\bQ\bW_{j,J}\bV_J\bQ\boldsymbol\mu_{j,J}
\le\left\|\bS_{j,J}\right\|_2\left\|\boldsymbol\mu_{j,J}\right\|_{\bQ}^2.
\end{align*}
Combining \eqref{eq:F4-consequence}, \eqref{eq:compressed-M-population-geometry},
\eqref{eq:M-proof-population-bias}, and $\left(a+b+c\right)^2\le3\left(a^2+b^2+c^2\right)$ gives
\begin{equation}
\begin{aligned}
\EE\left[\left(\left\|\boldsymbol d_{j,J}\right\|_{\bQ}^2-\Tr\left(\bS_{j,J}\right)\right)^2\right]
&\le C\left[\frac p{n^2}+\frac{p^2}{n^3}
+\frac{\left\|\boldsymbol\mu_{j,J}\right\|_{\bQ}^2}{n}
+\left\|\boldsymbol\mu_{j,J}\right\|_{\bQ}^4\right],\\
\EE\left[\left(\frac{\left\|\boldsymbol d_{j,J}\right\|_{\bQ}^2-\Tr\left(\bS_{j,J}\right)}{\Tr\left(\bS_{j,J}\right)}\right)^2\right]
&\le C\left[\frac1p+\frac1n+\frac{n\delta_{0,n}^2}{p^2}
+\frac{n^2\delta_{0,n}^4}{p^2}\right]=o\left(1\right).
\end{aligned}
\label{eq:compressed-M-relative-calibration}
\end{equation}
The bound is uniform over the bounded number of pairs $\left(J,j\right)$.

For $J\subseteq\mathcal I_0$ and $j\in\mathcal I_0\setminus J$, the implemented precision-weighted aggregates and candidate directions are
\begin{equation}
\begin{aligned}
\widehat\bV_J
&=\left(\sum_{k\in\left\{1\right\}\cup J}n_k\left(\widehat\bSigma_k^{\rm c}\right)^{-1}\right)^{-1},\\
\widehat{\overline\btheta}_J
&=\widehat\bV_J\sum_{k\in\left\{1\right\}\cup J}n_k\left(\widehat\bSigma_k^{\rm c}\right)^{-1}\tbtheta_k,\\
\widehat\bW_{j,J}
&=\widehat\bV_J\left(\widehat\bV_J+n_j^{-1}\widehat\bSigma_j^{\rm c}\right)^{-1},\\
\widehat{\boldsymbol d}_{j,J}
&=\widehat\bW_{j,J}\left(\tbtheta_j-\widehat{\overline\btheta}_J\right),
\qquad\widehat\bS_{j,J}=\bQ^{1/2}\widehat\bW_{j,J}\widehat\bV_J\bQ^{1/2}.
\end{aligned}
\label{eq:M-proof-implemented-objects}
\end{equation}
Under $\tbtheta_k^\circ=\tbtheta_k$ and $\bC_k^\circ=n_k^{-1}\widehat\bSigma_k^{\rm c}$, the definitions in Theorem~\ref{thm:sequential-unified} satisfy
\[
\bV_J^\circ=\widehat\bV_J,
\quad\overline\btheta_J^\circ=\widehat{\overline\btheta}_J,
\quad\bW_{j,J}^\circ=\widehat\bW_{j,J},
\quad\boldsymbol d_{j,J}^\circ=\widehat{\boldsymbol d}_{j,J},
\quad\bS_{j,J}^\circ=\widehat\bS_{j,J}.
\]
Applying the inverse-covariance bound \eqref{eq:p-input-inverse-bound-revised} and the bounds
\eqref{eq:p-subset-covariance-perturbation-revised} and \eqref{eq:p-weight-perturbation-revised}
with $\widehat\bSigma_k=\widehat\bSigma_k^{\rm c}$~gives
\begin{equation}
\begin{aligned}
&\widehat\bV_J-\bV_J
=\widehat\bV_J\sum_{k\in\left\{1\right\}\cup J}n_k
\left[\bSigma_k^{-1}-\left(\widehat\bSigma_k^{\rm c}\right)^{-1}\right]\bV_J,\\
&n\left\|\widehat\bV_J-\bV_J\right\|_2+
\left\|\widehat\bW_{j,J}-\bW_{j,J}\right\|_2+
 n\left\|\widehat\bS_{j,J}-\bS_{j,J}\right\|_2\le C\eta_n.
\end{aligned}
\label{eq:compressed-M-matrix-perturbations}
\end{equation}
The matrix inequalities in \eqref{eq:compressed-M-matrix-perturbations} use only the spectral bounds and the inverse difference identity, and thus hold for every $J\subseteq\mathcal I$ and $j\in\mathcal I\setminus J$, with $\mathcal I=\left\{2,\ldots,m\right\}$. Define
$L_n=\max_{k\in\mathcal I_0}\left\|\tbtheta_{k,\mathrm{lin}}-\tbtheta_{1,\mathrm{lin}}\right\|_{\bQ}$.
Subtracting the aggregate formulas in \eqref{eq:M-proof-population-objects} and \eqref{eq:M-proof-implemented-objects} gives
\begin{equation}
\begin{aligned}
\widehat{\overline\btheta}_J-\overline\btheta_{J,\mathrm{lin}}
&=\boldsymbol r_1+
\sum_{k\in J}\widehat\bV_Jn_k\left(\widehat\bSigma_k^{\rm c}\right)^{-1}\left(\boldsymbol r_k-\boldsymbol r_1\right)\\
&\quad+\sum_{k\in J}\left[\widehat\bV_Jn_k\left(\widehat\bSigma_k^{\rm c}\right)^{-1}-\bV_Jn_k\bSigma_k^{-1}\right]
\left(\tbtheta_{k,\mathrm{lin}}-\tbtheta_{1,\mathrm{lin}}\right).
\end{aligned}
\label{eq:compressed-M-aggregate-subtraction}
\end{equation}
The candidate-direction formulas in \eqref{eq:M-proof-population-objects} and \eqref{eq:M-proof-implemented-objects} give
\[
\begin{aligned}
\widehat{\boldsymbol d}_{j,J}-\boldsymbol d_{j,J}
&=\left(\widehat\bW_{j,J}-\bW_{j,J}\right)
\left(\tbtheta_{j,\mathrm{lin}}-\overline\btheta_{J,\mathrm{lin}}\right)+\widehat\bW_{j,J}
\left[\boldsymbol r_j-\left(\widehat{\overline\btheta}_J-\overline\btheta_{J,\mathrm{lin}}\right)\right].
\end{aligned}
\]
Since $\left\|\tbtheta_{j,\mathrm{lin}}-\overline\btheta_{J,\mathrm{lin}}\right\|_{\bQ}\le CL_n$,
\eqref{eq:compressed-M-matrix-perturbations} and \eqref{eq:compressed-M-aggregate-subtraction} imply
\begin{equation}
\max_{J\subseteq\mathcal I_0}\left\|\widehat{\overline\btheta}_J-\overline\btheta_{J,\mathrm{lin}}\right\|_{\bQ}
+\max_{J\subsetneq\mathcal I_0}\max_{j\in\mathcal I_0\setminus J}
\left\|\widehat{\boldsymbol d}_{j,J}-\boldsymbol d_{j,J}\right\|_{\bQ}
\le C\left[\max_k\left\|\boldsymbol r_k\right\|_{\bQ}+\eta_n L_n\right].
\label{eq:compressed-M-vector-perturbations}
\end{equation}
Taking $t=p$ in \eqref{eq:M-proof-linearized-tail-moment}, applying a union bound over $j\in\left[m\right]$, and using $n_j\asymp n$ gives a fixed $C_0>0$ such that
\begin{equation}
\begin{aligned}
&\PP\left(\max_j\left\|\tbtheta_{j,\mathrm{lin}}-\btheta_j^\star\right\|_{\bQ}>C_0\sqrt{p/n}\right)
\le2m e^{-p},\\
&\EE L_n^4
\le C\sum_{k\in\left\{1\right\}\cup\mathcal I_0}
\EE\left\|\tbtheta_{k,\mathrm{lin}}-\btheta_k^\star\right\|_{\bQ}^4
+C\delta_{0,n}^4
\le C\left(p/n\right)^2.
\end{aligned}
\label{eq:M-proof-linearized-maximal}
\end{equation}
Define the deterministic sequence
\[
\varepsilon_n=
\left[\frac{p\log^2 n}{n}+\frac1p+\frac1n
+\frac{n\delta_{0,n}^2}{p^2}+\frac{n^2\delta_{0,n}^4}{p^2}\right]^{1/4}
\longrightarrow0.
\]
Let $\mathcal G_n^\circ$ be the event on which all four bounds
\begin{equation}
\begin{aligned}
\max_j\left\|\tbtheta_{j,\mathrm{lin}}-\btheta_j^\star\right\|_{\bQ}&\le C_0\sqrt{p/n},\\
\max_j\left\|\boldsymbol r_j\right\|_{\bQ}&\le\varepsilon_n\sqrt{p/n},\\
\eta_n&\le\varepsilon_n,\\
\max_{J\subsetneq\mathcal I_0}\max_{j\in\mathcal I_0\setminus J}
\frac{\left|\left\|\boldsymbol d_{j,J}\right\|_{\bQ}^2-\Tr\left(\bS_{j,J}\right)\right|}{\Tr\left(\bS_{j,J}\right)}
&\le\varepsilon_n
\end{aligned}
\label{eq:M-proof-good-event-definition}
\end{equation}
hold. Markov's inequality, \eqref{eq:compressed-M-local-moments},
\eqref{eq:compressed-M-covariance-moments}, and \eqref{eq:compressed-M-relative-calibration} yield
\begin{equation}
\begin{aligned}
\pi_n:=\PP\left(\left(\mathcal G_n^\circ\right)^c\right)
&\le2m e^{-p}+\frac{C}{\varepsilon_n^2}
\left[\frac{p\log^2 n}{n}+\frac1p+\frac1n
+\frac{n\delta_{0,n}^2}{p^2}+\frac{n^2\delta_{0,n}^4}{p^2}\right]\\
&\le2m e^{-p}+C\varepsilon_n^2=o\left(1\right).
\end{aligned}
\label{eq:audited-M-good-event}
\end{equation}
On $\mathcal G_n^\circ$, $L_n\le C\sqrt{p/n}$, and \eqref{eq:compressed-M-vector-perturbations} gives
\[
\max_{J\subseteq\mathcal I_0}\left\|\widehat{\overline\btheta}_J-\overline\btheta_{J,\mathrm{lin}}\right\|_{\bQ}
+\max_{J\subsetneq\mathcal I_0}\max_{j\in\mathcal I_0\setminus J}
\left\|\widehat{\boldsymbol d}_{j,J}-\boldsymbol d_{j,J}\right\|_{\bQ}
\le C\varepsilon_n\sqrt{p/n}.
\]
Using \eqref{eq:M-proof-population-bias}, \eqref{eq:compressed-M-vector-perturbations}, and \eqref{eq:M-proof-good-event-definition} gives
\[
\begin{gathered}
\max_j\left\|\tbtheta_j-\btheta_j^\star\right\|_{\bQ}
+\max_{J\subseteq\mathcal I_0}\left\|\widehat{\overline\btheta}_J-\btheta_1^\star\right\|_{\bQ}
\le C\sqrt{p/n},\\
\mathcal R_n^\circ\vee\mathcal U_n^\circ\le k_0\sqrt{p/n}.
\end{gathered}
\]
Here $k_0$ is fixed. The bound $\Tr\left(\widehat\bS_{j,J}\right)\asymp p/n$ follows from \eqref{eq:compressed-sequential-geometry} with $\bC_k^\circ=n_k^{-1}\widehat\bSigma_k^{\rm c}$, $\bA=\widehat\bV_J$, $a_-=\varepsilon^{\rm c}/\overline\kappa$, and $a_+=1/\left(\varepsilon^{\rm c}\underline\kappa\right)$.

For the calibration quantity, let $\boldsymbol e_{j,J}=\widehat{\boldsymbol d}_{j,J}-\boldsymbol d_{j,J}$. On the same event,
\begin{align*}
&\frac{\left|\left\|\widehat{\boldsymbol d}_{j,J}\right\|_{\bQ}^2-\Tr\left(\widehat\bS_{j,J}\right)\right|
+\left\|\widehat\bS_{j,J}\right\|_2}{\Tr\left(\widehat\bS_{j,J}\right)}\\
&\quad\le C\frac np\left[
\left|\left\|\boldsymbol d_{j,J}\right\|_{\bQ}^2-\Tr\left(\bS_{j,J}\right)\right|+\left\|\bS_{j,J}\right\|_2\right]\\
&\qquad+C\frac np\left\|\boldsymbol e_{j,J}\right\|_{\bQ}
\left(2\left\|\boldsymbol d_{j,J}\right\|_{\bQ}+\left\|\boldsymbol e_{j,J}\right\|_{\bQ}\right)
+C\frac{n\left(p+1\right)}p\left\|\widehat\bS_{j,J}-\bS_{j,J}\right\|_2\\
&\quad\le C\left(\varepsilon_n+p^{-1}+\varepsilon_n^2\right)
\le C_1\varepsilon_n.
\end{align*}
Here $\left\|\boldsymbol d_{j,J}\right\|_{\bQ}\le C\sqrt{p/n}$,
$\left\|\boldsymbol e_{j,J}\right\|_{\bQ}\le C\varepsilon_n\sqrt{p/n}$, and
$\left\|\widehat\bS_{j,J}-\bS_{j,J}\right\|_2\le C\varepsilon_n/n$ follow from
\eqref{eq:compressed-M-population-geometry}, \eqref{eq:compressed-M-matrix-perturbations}, \eqref{eq:compressed-M-vector-perturbations}, and \eqref{eq:M-proof-good-event-definition}.
Thus $\mathcal Z_n^\circ\le C_1\varepsilon_n$ on $\mathcal G_n^\circ$, where $C_1$ is fixed.
The required event condition holds with tolerance $C_1\varepsilon_n$ and failure probability $\pi_n$ in \eqref{eq:audited-M-good-event}.

Equation~\eqref{eq:M-proof-implemented-objects} gives
$\sum_{k\in\left\{1\right\}\cup\mathcal I_0}\widehat\bV_{\mathcal I_0}n_k\left(\widehat\bSigma_k^{\rm c}\right)^{-1}=\bI_p$.
The clipping bounds in \eqref{eq:clipped-sandwich} imply
$\max_k\left\|\widehat\bV_{\mathcal I_0}n_k\left(\widehat\bSigma_k^{\rm c}\right)^{-1}\right\|_{\bQ\to\bQ}\le C$ for every realization. Together with \eqref{eq:M-local-fourth-lemma}, these identities yield
\[
\left\|\tbtheta_1-\widehat{\overline\btheta}_{\mathcal I_0}\right\|_{\bQ}
\le C\left[\max_k\left\|\tbtheta_k-\btheta_k^\star\right\|_{\bQ}+\delta_{0,n}\right],\qquad
\EE\left\|\tbtheta_1-\widehat{\overline\btheta}_{\mathcal I_0}\right\|_{\bQ}^4
\le C\left(p/n\right)^2.
\]
This verifies the fourth-moment condition. Since $\pi_n\to0$ and $\rho$ is fixed, $\pi_n\le\rho$ for sufficiently large $n$. Theorem~\ref{thm:sequential-unified} therefore yields
\begin{equation}
\PP\left(\left\{\widehat j_1,\ldots,\widehat j_{m_0}\right\}=\mathcal I_0\right)\ge1-\rho,
\qquad
\EE\left\|\widehat\btheta_{\rm seq}-\widehat{\overline\btheta}_{\mathcal I_0}\right\|_{\bQ}^2
=o\left(p/n\right).
\label{eq:compressed-M-unified-conclusion}
\end{equation}

Equations~\eqref{eq:compressed-M-local-moments}, \eqref{eq:compressed-M-covariance-moments}, \eqref{eq:compressed-M-vector-perturbations}, and \eqref{eq:M-proof-linearized-maximal}, together with Cauchy--Schwarz, give
\begin{equation}
\begin{aligned}
\EE\left\|\widehat{\overline\btheta}_{\mathcal I_0}-\overline\btheta_{\mathcal I_0,\mathrm{lin}}\right\|_{\bQ}^2
&\le C\EE\max_k\left\|\boldsymbol r_k\right\|_{\bQ}^2
+C\left(\EE\eta_n^4\right)^{1/2}\left(\EE L_n^4\right)^{1/2}\\
&=o\left(p/n\right).
\end{aligned}
\label{eq:M-proof-oracle-aggregate-error}
\end{equation}
The covariance identities following \eqref{eq:M-proof-influence-definitions}, independence across data sets, and \eqref{eq:M-proof-population-objects} give
\[
\operatorname{Cov}\left(\overline\btheta_{\mathcal I_0,\mathrm{lin}}\right)=\bV_{\mathcal I_0},\qquad
\EE\overline\btheta_{\mathcal I_0,\mathrm{lin}}-\btheta_1^\star
=\bV_{\mathcal I_0}\sum_{k\in\mathcal I_0}n_k\bSigma_k^{-1}\left(\btheta_k^\star-\btheta_1^\star\right).
\]
Equation~\eqref{eq:M-proof-population-bias} gives squared bias $o\left(p/n\right)$, and hence
$\EE\left\|\overline\btheta_{\mathcal I_0,\mathrm{lin}}-\btheta_1^\star\right\|_{\bQ}^2
=\left\langle\bV_{\mathcal I_0},\bQ\right\rangle+o\left(p/n\right)$.
Equations~\eqref{eq:M-proof-oracle-aggregate-error} and \eqref{eq:compressed-M-unified-conclusion} give
\[
\begin{aligned}
\EE\left\|\widehat\btheta_{\rm seq}-\overline\btheta_{\mathcal I_0,\mathrm{lin}}\right\|_{\bQ}^2
&\le2\EE\left\|\widehat\btheta_{\rm seq}-\widehat{\overline\btheta}_{\mathcal I_0}\right\|_{\bQ}^2+2\EE\left\|\widehat{\overline\btheta}_{\mathcal I_0}-\overline\btheta_{\mathcal I_0,\mathrm{lin}}\right\|_{\bQ}^2
=o\left(p/n\right).
\end{aligned}
\]
For square-integrable estimators $\bu,\bv$, the risk comparison
\begin{equation}
\left|\EE\left\|\bu-\btheta_1^\star\right\|_{\bQ}^2-\EE\left\|\bv-\btheta_1^\star\right\|_{\bQ}^2\right|
\le\EE\left\|\bu-\bv\right\|_{\bQ}^2
+2\left[\EE\left\|\bu-\bv\right\|_{\bQ}^2\right]^{1/2}\left[\EE\left\|\bv-\btheta_1^\star\right\|_{\bQ}^2\right]^{1/2}
\label{eq:compressed-M-risk-comparison}
\end{equation}
applied with $\bu=\widehat\btheta_{\rm seq}$ and $\bv=\overline\btheta_{\mathcal I_0,\mathrm{lin}}$ therefore implies
\begin{equation}
\EE\left\|\widehat\btheta_{\rm seq}-\btheta_1^\star\right\|_{\bQ}^2
=\left\langle\bV_{\mathcal I_0},\bQ\right\rangle+o\left(p/n\right).
\label{eq:compressed-M-oracle-risk}
\end{equation}

For the single-step comparison, define the all-source weights explicitly by
\begin{equation}
\begin{aligned}
\bW_j&=\left(\sum_{k=1}^m n_k\bSigma_k^{-1}\right)^{-1}n_j\bSigma_j^{-1},\\
\widehat\bW_j&=\left(\sum_{k=1}^m n_k\left(\widehat\bSigma_k^{\rm c}\right)^{-1}\right)^{-1}n_j\left(\widehat\bSigma_j^{\rm c}\right)^{-1},
\qquad j\in\mathcal I.
\end{aligned}
\label{eq:M-proof-all-source-weights}
\end{equation}
The corresponding population-weighted and implemented directions are
\[
\bu_{\mathcal I}=\sum_{j\in\mathcal I}\bW_j\left(\tbtheta_{j,\mathrm{lin}}-\tbtheta_{1,\mathrm{lin}}\right),\qquad
\widehat\bu_{\mathcal I}=\sum_{j\in\mathcal I}\widehat\bW_j\left(\tbtheta_j-\tbtheta_1\right).
\]
Their difference separates into
\begin{align*}
\widehat\bu_{\mathcal I}-\bu_{\mathcal I}
&=\sum_{j\in\mathcal I}\left(\widehat\bW_j-\bW_j\right)\left(\btheta_j^\star-\btheta_1^\star\right)\\
&\quad+\sum_{j\in\mathcal I}\left(\widehat\bW_j-\bW_j\right)
\left[\left(\tbtheta_{j,\mathrm{lin}}-\btheta_j^\star\right)-\left(\tbtheta_{1,\mathrm{lin}}-\btheta_1^\star\right)\right]\\
&\quad+\sum_{j\in\mathcal I}\widehat\bW_j\left(\boldsymbol r_j-\boldsymbol r_1\right).
\end{align*}
The all-source weights in \eqref{eq:M-proof-all-source-weights} equal
$\bW_{j,\mathcal I\setminus\left\{j\right\}}$ and $\widehat\bW_{j,\mathcal I\setminus\left\{j\right\}}$ from
\eqref{eq:M-proof-population-objects} and \eqref{eq:M-proof-implemented-objects}, extended to $J\subseteq\mathcal I$. Thus \eqref{eq:compressed-M-matrix-perturbations} gives
$\max_j\left\|\widehat\bW_j-\bW_j\right\|_2\le C\eta_n\le C\sqrt{p\log n/n}$ on $\mathcal E_n$.
Minkowski's inequality, \eqref{eq:M-proof-linearized-tail-moment}, and \eqref{eq:compressed-M-local-moments} give
\begin{equation}
\left[\EE\left[\left\|\widehat\bu_{\mathcal I}-\bu_{\mathcal I}\right\|_{\bQ}^2\mathbf1_{\mathcal E_n}\right]\right]^{1/2}
\le C\left[\sqrt{\frac{p\log n}{n}}\delta_{\max}+\frac{p\log n}{n}\right]
=o\left(\delta_{\mathcal I}\right).
\label{eq:compressed-M-single-direction-error}
\end{equation}
For the homogeneous indices in this sum, $\left\|\btheta_j^\star-\btheta_1^\star\right\|_{\bQ}\le\delta_{\max}$ eventually by separation. The final order in \eqref{eq:compressed-M-single-direction-error} uses the additional condition in part 3, $\delta_{\mathcal I}\gg\sqrt{p/n}$, and $p\log^2n/n=o\left(1\right)$.
Also,
\[
\left\|\EE\bu_{\mathcal I}\right\|_{\bQ}=\delta_{\mathcal I},\qquad
\EE\left\|\bu_{\mathcal I}-\EE\bu_{\mathcal I}\right\|_{\bQ}^2\le C p/n.
\]
The triangle inequality, Chebyshev's inequality, and \eqref{eq:compressed-M-single-direction-error} imply
\[
\PP\left(\left\|\widehat\bu_{\mathcal I}\right\|_{\bQ}<\delta_{\mathcal I}/4\right)
\le\PP\left(\mathcal E_n^c\right)+C\frac{p}{n\delta_{\mathcal I}^2}
+\frac C{\delta_{\mathcal I}^2}\EE\left[\left\|\widehat\bu_{\mathcal I}-\bu_{\mathcal I}\right\|_{\bQ}^2\mathbf1_{\mathcal E_n}\right]
=o\left(1\right).
\]
Writing $\widehat\bV_{\mathcal I}$ as in \eqref{eq:M-proof-implemented-objects}, the parallel-sum identity gives
\[
\left(\sum_{j\in\mathcal I}\widehat\bW_j\right)\frac{\widehat\bSigma_1^{\rm c}}{n_1}
=\frac{\widehat\bSigma_1^{\rm c}}{n_1}-\widehat\bV_{\mathcal I}
\preceq\frac{\widehat\bSigma_1^{\rm c}}{n_1}.
\]
The admissible interval in Theorem~\ref{thm:risk-bound-shat-multiset} and \eqref{eq:clipped-sandwich} therefore give
$0\le s\le\Tr\left[\bQ\left(\widehat\bSigma_1^{\rm c}/n_1-\widehat\bV_{\mathcal I}\right)\right]\le Cp/n$
for every realization. Substituting this bound into \eqref{eq:shrinkage-estimator-s-multi-set} yields
\begin{equation}
\begin{aligned}
\EE\left\|\widehat\btheta_{\rm ss}-\tbtheta_1\right\|_{\bQ}^2
&\le C\frac{\left(p/n\right)^2}{\delta_{\mathcal I}^2}
+C\frac pn\PP\left(\left\|\widehat\bu_{\mathcal I}\right\|_{\bQ}<\delta_{\mathcal I}/4\right)
=o\left(p/n\right).
\end{aligned}
\label{eq:compressed-M-single-update}
\end{equation}
Applying \eqref{eq:compressed-M-risk-comparison} first to $\tbtheta_1$ and $\tbtheta_{1,\mathrm{lin}}$, then to $\widehat\btheta_{\rm ss}$ and $\tbtheta_1$, gives
\[
\EE\left\|\tbtheta_1-\btheta_1^\star\right\|_{\bQ}^2
=\left\langle n_1^{-1}\bSigma_1,\bQ\right\rangle+o\left(p/n\right),\qquad
\EE\left\|\widehat\btheta_{\rm ss}-\btheta_1^\star\right\|_{\bQ}^2
=\left\langle n_1^{-1}\bSigma_1,\bQ\right\rangle+o\left(p/n\right).
\]
Finally, \eqref{eq:F3-population-spectral-bounds}, $\mathcal I_0\ne\varnothing$, and the parallel-sum identity imply
\[
\begin{aligned}
&\frac{\bSigma_1}{n_1}-\bV_{\mathcal I_0}
=\frac{\bSigma_1}{n_1}
\left[\frac{\bSigma_1}{n_1}+\left(\sum_{j\in\mathcal I_0}n_j\bSigma_j^{-1}\right)^{-1}\right]^{-1}
\frac{\bSigma_1}{n_1}
\succeq\frac{\underline\kappa\,\underline\sigma_0^2}{2\overline\kappa^2\overline\sigma_0 n}\bI_p,\\
&\EE\left\|\widehat\btheta_{\rm ss}-\btheta_1^\star\right\|_{\bQ}^2
-\EE\left\|\widehat\btheta_{\rm seq}-\btheta_1^\star\right\|_{\bQ}^2\ge\frac{q_-\underline\kappa\,\underline\sigma_0^2}{2\overline\kappa^2\overline\sigma_0}\frac pn+o\left(p/n\right)>0
\end{aligned}
\]
for sufficiently large $n$, proving part 3.
\end{proof}

\subsection{Proofs of the Technical Lemmas}\label{sec:proof-technical-lemmas}

\begin{proof}[Proof of Lemma~\ref{lem:sure}]
Let $\bZ=\bx-\btheta$ and let $\varphi_{\bSigma}$ be its Gaussian density. The identity
$\nabla\varphi_{\bSigma}\left(\bz\right)=-\bSigma^{-1}\bz\varphi_{\bSigma}\left(\bz\right)$ gives
$z_a\varphi_{\bSigma}\left(\bz\right)=-\sum_k\Sigma_{ak}\partial_k\varphi_{\bSigma}\left(\bz\right)$.
Choose a smooth cutoff $\chi_R$ equal to one on $\left\{\bz\in\mathbb R^p\mid\left\|\bz\right\|_2\le R\right\}$, supported on $\left\{\bz\in\mathbb R^p\mid\left\|\bz\right\|_2\le2R\right\}$, with $0\le\chi_R\le1$ and $\left\|\nabla\chi_R\right\|_2\le C/R$.
Weak integration by parts yields
\[
\EE\left[Z_ag_b\left(\bx\right)\chi_R\left(\bZ\right)\right]
=\sum_k\Sigma_{ak}\EE\left[\partial_k g_b\left(\bx\right)\chi_R\left(\bZ\right)\right]
+\sum_k\Sigma_{ak}\EE\left[g_b\left(\bx\right)\partial_k\chi_R\left(\bZ\right)\right].
\]
The cutoff error satisfies
\[
\left|\EE\left[g_b\left(\bx\right)\partial_k\chi_R\left(\bZ\right)\right]\right|
\le\frac C{R^2}\EE\left[\left\|\bZ\right\|_2\left|g_b\left(\bx\right)\right|\right]
\le\frac C{R^2}\sum_{r=1}^p\EE\left|Z_rg_b\left(\bx\right)\right|\longrightarrow0.
\]
The assumed integrability and dominated convergence give $\EE\left[Z_ag_b\left(\bx\right)\right]=\sum_k\Sigma_{ak}\EE\left[\partial_k g_b\left(\bx\right)\right]$. Expanding the loss and using symmetry of $\bQ$ and $\bSigma$ now gives
\[
\begin{aligned}
\EE\left\|\bx+\bg\left(\bx\right)-\btheta\right\|_{\bQ}^2
&=\Tr\left(\bQ\bSigma\right)+2\sum_{a,b,k}Q_{ab}\Sigma_{ak}\EE\left[\partial_k g_b\left(\bx\right)\right]
+\EE\left\|\bg\left(\bx\right)\right\|_{\bQ}^2\\
&=\left\langle\bSigma,\bQ\right\rangle+
\EE\left[2\left\langle\bJ\left(\bx\right)\bSigma,\bQ\right\rangle+\left\|\bg\left(\bx\right)\right\|_{\bQ}^2\right].
\end{aligned}
\]
If $\EE\left\|\bg\left(\bx\right)\right\|_{\bQ}^2=\infty$, the inequality
$\left\|\bZ+\bg\left(\bx\right)\right\|_{\bQ}^2\ge\tfrac12\left\|\bg\left(\bx\right)\right\|_{\bQ}^2-\left\|\bZ\right\|_{\bQ}^2$
shows that both sides equal $+\infty$, and hence the same identity holds in the extended sense.
\end{proof}

\begin{proof}[Proof of Lemma~\ref{lem:tail-bounds}]
	Fix \(j\in\left[m\right]\). Let \(\mathcal N_{1/4}\) be a one-quarter Euclidean net of \(\mathbb S^{p-1}\) with \(\left|\mathcal N_{1/4}\right|\le9^p\). For every \(\bx\in\mathbb S^{p-1}\), choose \(\bv\in\mathcal N_{1/4}\) with \(\left\|\bx-\bv\right\|_2\le1/4\). Then $\left|\bx^\top\bu\right|\le\left|\bv^\top\bu\right|+\frac14\left\|\bu\right\|_2$, which~gives
	\begin{equation}\label{eq:cover-v-vector}
		\left\|\bu\right\|_2\le\frac43\max_{\bv\in\mathcal N_{1/4}}\left|\bv^\top\bu\right|.
	\end{equation}
	For every symmetric \(\bA\),
	\[
	\begin{aligned}
		\left|\bx^\top\bA\bx-\bv^\top\bA\bv\right|
		\le\left|\left(\bx-\bv\right)^\top\bA\bx\right|+\left|\bv^\top\bA\left(\bx-\bv\right)\right|\le\frac12\left\|\bA\right\|_2,
	\end{aligned}
	\]
	and taking the supremum over \(\bx\in\mathbb S^{p-1}\) gives
	\begin{equation}\label{eq:cover-v-matrix}
		\left\|\bA\right\|_2\le2\max_{\bv\in\mathcal N_{1/4}}\left|\bv^\top\bA\bv\right|.
	\end{equation}
	We use the following scalar inequalities. If \(X_1,\ldots,X_N\) are independent and mean-zero with \(\max_i\left\|X_i\right\|_{\psi_2}\le K\), then
	\begin{equation}\label{eq:F1-scalar-subGaussian}
		\PP\left[\left|\frac1N\sum_{i=1}^NX_i\right|>x\right]\le2\exp\left(-\frac{Nx^2}{16K^2}\right),\qquad x>0.
	\end{equation}
	If \(\max_i\left\|X_i\right\|_{\psi_1}\le K\), then
	\begin{equation}\label{eq:F1-scalar-Bernstein}
		\PP\left[\left|\frac1N\sum_{i=1}^NX_i\right|>x\right]\le2\exp\left[-N\min\left[\frac{x^2}{16K^2},\frac{x}{4K}\right]\right],\qquad x>0.
	\end{equation}
	
	For fixed \(\bv\in\mathcal N_{1/4}\), the variables \(\bv^\top\nabla\ell_j\left(\btheta_j^\star,\bx_{ji}\right)\) are mean-zero and have \(\psi_2\)-norm at most \(\tau\) by Assumption~\ref{assump:2}. Apply \eqref{eq:F1-scalar-subGaussian} with $x=\frac34C_{\mathrm{F1},\mathrm{s}}\sqrt{\frac{p+t}{n_j}}$. The definition of \(C_{\mathrm{F1},\mathrm{s}}\) gives $\frac{n_jx^2}{16\tau^2}=\frac{9C_{\mathrm{F1},\mathrm{s}}^2}{256\tau^2}\left(p+t\right)=\left(\log9+1\right)\left(p+t\right)\ge p\log9+t$. Equations~\eqref{eq:cover-v-vector} and \eqref{eq:F1-scalar-subGaussian}, followed by a union bound, therefore give
	\[
	\begin{aligned}
		\PP\left[\left\|\nabla f_j\left(\btheta_j^\star\right)\right\|_2>C_{\mathrm{F1},\mathrm{s}}\sqrt{\frac{p+t}{n_j}}\right]
		&\le\sum_{\bv\in\mathcal N_{1/4}}\PP\left[\left|\bv^\top\nabla f_j\left(\btheta_j^\star\right)\right|>\frac34C_{\mathrm{F1},\mathrm{s}}\sqrt{\frac{p+t}{n_j}}\right]\\
		&\le2\left|\mathcal N_{1/4}\right|\exp\left[-\left(\log9+1\right)\left(p+t\right)\right]\le2e^{-t}.
	\end{aligned}
	\]
	This proves \eqref{eq:M-score-all-t}.
	
	Let \(\gamma_1=\left\lceil c\right\rceil+11\), and let \(\mathcal Q_j\) be a maximal
	\(n_j^{-\gamma_1}\)-separated subset of \(\mathcal{B}_j\left(r\right)\).
	Then \(\mathcal Q_j\) is also an
	\(n_j^{-\gamma_1}\)-net of \(\mathcal{B}_j\left(r\right)\), and the balls
	$
	\left\{
	\mathcal{B} \left(\boldsymbol{\theta},
	n_j^{-\gamma_1}/2\right):
	\boldsymbol{\theta}\in\mathcal Q_j
	\right\}
	$
	are pairwise disjoint and contained in
	$\mathcal B_j\left(r+n_j^{-\gamma_1}/2\right)$. Comparing Lebesgue
	volumes gives
	$\left|\mathcal Q_j\right|\left(n_j^{-\gamma_1}/2\right)^p\le\left(r+n_j^{-\gamma_1}/2\right)^p$ and
	\begin{equation}\label{eq:F1-parameter-net-size}
		\left|\mathcal Q_j\right|\le\left(1+2rn_j^{\gamma_1}\right)^p\le\left(1+2r\right)^pn_j^{\gamma_1p}.
	\end{equation}
	For \(i\in\left[n_j\right]\), let \(M_{ji}=\sup_{\btheta\in\mathcal B_j\left(r\right)}\left\|\nabla^3\ell_j\left(\btheta,\bx_{ji}\right)\right\|_{\mathrm{op}}\). Assumption~\ref{assump:3} gives \(\EE M_{ji}\le Lp^c\). Markov's inequality gives
	\begin{equation}\label{eq:M-third-derivative-envelope-revised}
		\PP\left(\frac1{n_j}\sum_{i=1}^{n_j}M_{ji}>Ln_j^7p^c\right)\le\frac{n_j^{-1}\sum_{i=1}^{n_j}\EE M_{ji}}{Ln_j^7p^c}\le n_j^{-7}.
	\end{equation}
	On the complementary event, take \(\btheta\in\mathcal B_j\left(r\right)\) and \(\btheta^\circ\in\mathcal Q_j\) with \(\left\|\btheta-\btheta^\circ\right\|_2\le n_j^{-\gamma_1}\). For every \(\bv\in\mathcal N_{1/4}\), the fundamental theorem of calculus gives
	\[
	\begin{aligned}
		\left|\bv^\top\left[\nabla^2f_j\left(\btheta\right)-\nabla^2f_j\left(\btheta^\circ\right)\right]\bv\right|\le\left\|\btheta-\btheta^\circ\right\|_2\frac1{n_j}\sum_{i=1}^{n_j}\sup_{\widetilde\btheta\in\mathcal B_j\left(r\right)}\left\|\nabla^3\ell_j\left(\widetilde\btheta,\bx_{ji}\right)\right\|_{\mathrm{op}}\le Ln_j^{7-\gamma_1}p^c,
	\end{aligned}
	\]
	while Assumption~\ref{assump:3} gives
	\[
	\left|\bv^\top\EE\left[\nabla^2\ell_j\left(\btheta,\bx\right)-\nabla^2\ell_j\left(\btheta^\circ,\bx\right)\right]\bv\right|\le L\left\|\btheta-\btheta^\circ\right\|_2\le Ln_j^{-\gamma_1}.
	\]
	The condition \(p\log n_j/n_j=o\left(1\right)\) implies \(p\le n_j\) for all sufficiently large \(n_j\). Since \(\gamma_1\ge c+11\),
	\begin{equation}\label{eq:M-Hessian-grid-approximation-revised}
		\begin{aligned}
			&\max_{\bv\in\mathcal N_{1/4}}\left|\bv^\top\left[\nabla^2f_j\left(\btheta\right)-\EE\nabla^2\ell_j\left(\btheta,\bx\right)-\nabla^2f_j\left(\btheta^\circ\right)+\EE\nabla^2\ell_j\left(\btheta^\circ,\bx\right)\right]\bv\right|\\
            &\qquad\le Ln_j^{7-\gamma_1+c}+Ln_j^{-\gamma_1}\le2Ln_j^{-4}.
		\end{aligned}
	\end{equation}
	For fixed \(\btheta^\circ\in\mathcal Q_j\) and \(\bv\in\mathcal N_{1/4}\), set
	\[
	X_{ji}\left(\btheta^\circ,\bv\right)=\bv^\top\nabla^2\ell_j\left(\btheta^\circ,\bx_{ji}\right)\bv-\EE\left[\bv^\top\nabla^2\ell_j\left(\btheta^\circ,\bx_{j1}\right)\bv\right].
	\]
	Assumption~\ref{assump:2} and \(\left\|X-\EE X\right\|_{\psi_1}\le2\left\|X\right\|_{\psi_1}\) give \(\left\|X_{ji}\left(\btheta^\circ,\bv\right)\right\|_{\psi_1}\le2\tau^2\). Apply \eqref{eq:F1-scalar-Bernstein} with
	\[
	x=8\tau^2\sqrt{\gamma_1+10+\frac{\log\left[9\left(1+2r\right)\right]}{\log2}}\sqrt{\frac{p\log n_j}{n_j}}.
	\]
	The quadratic branch applies for all sufficiently large \(n_j\), and $\frac{n_jx^2}{64\tau^4}=\left(\gamma_1+10+\frac{\log\left[9\left(1+2r\right)\right]}{\log2}\right)p\log n_j$. Equations~\eqref{eq:F1-parameter-net-size} and \(\log n_j\ge\log2\), followed by a union bound over \(\mathcal Q_j\times\mathcal N_{1/4}\), give
	\begin{equation}\label{eq:M-Hessian-net-concentration-revised}
		\begin{aligned}
			&\PP\left[\max_{\btheta^\circ\in\mathcal Q_j}\max_{\bv\in\mathcal N_{1/4}}\left|\bv^\top\left[\nabla^2f_j\left(\btheta^\circ\right)-\EE\nabla^2\ell_j\left(\btheta^\circ,\bx\right)\right]\bv\right|>x\right]\\
			&\le2\exp\left[p\log\left[9\left(1+2r\right)\right]+\gamma_1p\log n_j-\left(\gamma_1+10+\frac{\log\left[9\left(1+2r\right)\right]}{\log2}\right)p\log n_j\right]\\
			&\le2n_j^{-10}.
		\end{aligned}
	\end{equation}
	On the intersection of the complementary events in \eqref{eq:M-third-derivative-envelope-revised} and \eqref{eq:M-Hessian-net-concentration-revised}, denoted by $\mathcal E_{\rm{F1}, 1}$, equations~\eqref{eq:cover-v-matrix} and \eqref{eq:M-Hessian-grid-approximation-revised} give
	\[
	\begin{aligned}
		&\sup_{\btheta\in\mathcal B_j\left(r\right)}\left\|\nabla^2f_j\left(\btheta\right)-\EE\nabla^2\ell_j\left(\btheta,\bx\right)\right\|_2\\
		&\quad\le16\tau^2\sqrt{\gamma_1+10+\frac{\log\left[9\left(1+2r\right)\right]}{\log2}}\sqrt{\frac{p\log n_j}{n_j}}+4Ln_j^{-4}\le C_{\mathrm{F1},\mathrm{h}}\sqrt{\frac{p\log n_j}{n_j}}.
	\end{aligned}
	\]
	The last inequality follows from \(n_j^{-4}\le\left(\log2\right)^{-1/2}\sqrt{p\log n_j/n_j}\) for \(p\ge1\) and \(n_j\ge2\). This proves \eqref{eq:M-Hessian-uniform-lemma} on the stated  event $\mathcal E_{\rm{F1}, 1}$.
	
	For \(\btheta\in\mathcal B_j\left(r_0\right)\), Assumption~\ref{assump:3} gives
	\begin{equation}\label{eq:M-population-Hessian-local-revised}
		\begin{aligned}
			\left\|\EE\nabla^2\ell_j\left(\btheta,\bx\right)-\bH_j^\star\right\|_2
			&=\left\|\int_0^1\EE\nabla^3\ell_j\left(\btheta_j^\star+u\left(\btheta-\btheta_j^\star\right),\bx\right)\left[\btheta-\btheta_j^\star\right]\mathrm du\right\|_2\\
			&\le L\left\|\btheta-\btheta_j^\star\right\|_2\le Lr_0\le\frac{\underline\lambda}{4}.
		\end{aligned}
	\end{equation}
	For all sufficiently large \(n_j\), \(C_{\mathrm{F1},\mathrm{h}}\sqrt{p\log n_j/n_j}\le\underline\lambda/4\). Assumption~\ref{assump:1}, \eqref{eq:M-Hessian-uniform-lemma}, \eqref{eq:M-population-Hessian-local-revised}, and Weyl's inequality then give
	\[
	\begin{aligned}
		\inf_{\btheta\in\mathcal B_j\left(r_0\right)}\lambda_{\min}\left[\nabla^2f_j\left(\btheta\right)\right]
		&\ge\lambda_{\min}\left(\bH_j^\star\right)-\sup_{\btheta\in\mathcal B_j\left(r_0\right)}\left\|\nabla^2f_j\left(\btheta\right)-\bH_j^\star\right\|_2\ge\underline\lambda-\frac{\underline\lambda}{4}-\frac{\underline\lambda}{4}=\frac{\underline\lambda}{2}.
	\end{aligned}
	\]
	This proves \eqref{eq:M-local-curvature-lemma}.
	
	Assumption~\ref{assump:4} and the definition of \(r_0\) give, for every \(p\ge1\), $\left\|\btheta_j^\star\right\|_2+r_0 \le B^\star\sqrt p+\frac{B_\Theta-B^\star}{2} <B_\Theta\sqrt p$. Hence \(\mathcal B_j\left(r_0\right)\subset\operatorname{int}\left(\Theta_p\right)\). Let \(\tbtheta_{j,\mathrm{con}}\) minimize \(f_j\) over \(\mathcal B_j\left(r_0\right)\). Let \(\bh=\tbtheta_{j,\mathrm{con}}-\btheta_j^\star\). Taylor's formula with integral remainder gives
	\[
	\begin{aligned}
		0&\ge f_j\left(\tbtheta_{j,\mathrm{con}}\right)-f_j\left(\btheta_j^\star\right)\\
		&=\nabla f_j\left(\btheta_j^\star\right)^\top\bh+\int_0^1\left(1-u\right)\bh^\top\nabla^2f_j\left(\btheta_j^\star+u\bh\right)\bh\,\mathrm du\\
		&\ge-\left\|\nabla f_j\left(\btheta_j^\star\right)\right\|_2\left\|\bh\right\|_2+\frac{\underline\lambda}{4}\left\|\bh\right\|_2^2.
	\end{aligned}
	\]
	Solving this quadratic inequality gives
	\begin{equation}\label{eq:M-constrained-minimizer-revised}
		\left\|\tbtheta_{j,\mathrm{con}}-\btheta_j^\star\right\|_2\le\frac4{\underline\lambda}\left\|\nabla f_j\left(\btheta_j^\star\right)\right\|_2.
	\end{equation}
	Let \(\mathcal E_j\) be the intersection of the event $\mathcal E_{\rm{F1}, 1}$ above and the event $\left\|\nabla f_j\left(\btheta_j^\star\right)\right\|_2\le C_{\mathrm{F1},\mathrm{s}}\sqrt{\frac{p+7\log n_j}{n_j}}$. Equation~\eqref{eq:M-score-all-t} with \(t=7\log n_j\), together with \eqref{eq:M-third-derivative-envelope-revised} and \eqref{eq:M-Hessian-net-concentration-revised}, gives
	\begin{equation}\label{eq:M-F1-event-probability-revised}
		\PP\left(\mathcal E_j^c\right)\le2n_j^{-7}+n_j^{-7}+2n_j^{-10}\le5n_j^{-6}.
	\end{equation}
	On \(\mathcal E_j\), the right side of \eqref{eq:M-constrained-minimizer-revised} is less than \(r_0/2\) for all sufficiently large \(n_j\). Thus \(\tbtheta_{j,\mathrm{con}}\) is an interior point of \(\mathcal B_j\left(r_0\right)\), and $\nabla f_j\left(\tbtheta_{j,\mathrm{con}}\right)=\bm0$. Convexity makes \(\tbtheta_{j,\mathrm{con}}\) a global minimizer over \(\Theta_p\). To prove uniqueness, suppose that \(\btheta'\ne\tbtheta_{j,\mathrm{con}}\) is another global minimizer. Choose \(0<s\le1\) such that \(\tbtheta_{j,\mathrm{con}}+s\left(\btheta'-\tbtheta_{j,\mathrm{con}}\right)\in\mathcal B_j\left(r_0\right)\) and \(s>0\). Convexity gives $f_j\left(\tbtheta_{j,\mathrm{con}}+s\left(\btheta'-\tbtheta_{j,\mathrm{con}}\right)\right)=f_j\left(\tbtheta_{j,\mathrm{con}}\right)$. Taylor's formula and \eqref{eq:M-local-curvature-lemma} give
	\[
	\begin{aligned}
		&f_j\left(\tbtheta_{j,\mathrm{con}}+s\left(\btheta'-\tbtheta_{j,\mathrm{con}}\right)\right)-f_j\left(\tbtheta_{j,\mathrm{con}}\right)\\
		&=\int_0^1\left(1-u\right)s^2\left(\btheta'-\tbtheta_{j,\mathrm{con}}\right)^\top\nabla^2f_j\left(\tbtheta_{j,\mathrm{con}}+us\left(\btheta'-\tbtheta_{j,\mathrm{con}}\right)\right)\left(\btheta'-\tbtheta_{j,\mathrm{con}}\right)\mathrm du\\
		&\ge\frac{\underline\lambda s^2}{4}\left\|\btheta'-\tbtheta_{j,\mathrm{con}}\right\|_2^2>0,
	\end{aligned}
	\]
	a contradiction. Hence \(\tbtheta_j=\tbtheta_{j,\mathrm{con}}\), and equations~\eqref{eq:M-constrained-minimizer-revised} and \eqref{eq:M-F1-event-probability-revised} prove all assertions defining \(\mathcal E_j\).
	
	It remains to prove the fourth moment. Let \(Y=\left\|\nabla f_j\left(\btheta_j^\star\right)\right\|_2\) and \(a=C_{\mathrm{F1},\mathrm{s}}\sqrt{p/n_j}\). Equation~\eqref{eq:M-score-all-t} gives, for \(x\ge a\), $\PP\left(Y>x\right)\le2\exp\left(p-\frac{n_jx^2}{C_{\mathrm{F1},\mathrm{s}}^2}\right)$. Integrating the tail bound gives
	\begin{equation}\label{eq:M-score-fourth-revised}
		\begin{aligned}
			\EE Y^4
			&=4\int_0^\infty x^3\PP\left(Y>x\right)\mathrm dx\\
			&\le a^4+8e^p\int_a^\infty x^3\exp\left(-\frac{n_jx^2}{C_{\mathrm{F1},\mathrm{s}}^2}\right)\mathrm dx\\
			&=C_{\mathrm{F1},\mathrm{s}}^4\frac{p^2}{n_j^2}+4C_{\mathrm{F1},\mathrm{s}}^4\frac{p+1}{n_j^2}\le9C_{\mathrm{F1},\mathrm{s}}^4\left(\frac p{n_j}\right)^2.
		\end{aligned}
	\end{equation}
	On \(\mathcal E_j\), equation~\eqref{eq:M-local-minimizer-bound-lemma} gives
	\[
	\EE\left[\left\|\tbtheta_j-\btheta_j^\star\right\|_2^4\mathbf1\left[\mathcal E_j\right]\right]\le\left(\frac4{\underline\lambda}\right)^4\EE Y^4\le9\left(\frac{4C_{\mathrm{F1},\mathrm{s}}}{\underline\lambda}\right)^4\left(\frac p{n_j}\right)^2.
	\]
	On \(\mathcal E_j^c\), Assumption~\ref{assump:4} and the definition of \(\Theta_p\) give $\left\|\tbtheta_j-\btheta_j^\star\right\|_2\le\left(B_\Theta+B^\star\right)\sqrt p$. Together with \eqref{eq:M-F1-event-probability-revised},
	\[
	\EE\left[\left\|\tbtheta_j-\btheta_j^\star\right\|_2^4\mathbf1\left[\mathcal E_j^c\right]\right]\le5\left(B_\Theta+B^\star\right)^4p^2n_j^{-6}\le5\left(B_\Theta+B^\star\right)^4\left(\frac p{n_j}\right)^2.
	\]
	Adding the last two bounds gives \eqref{eq:M-local-fourth-lemma} with the stated value of \(C_{\mathrm{F1},4}\).
\end{proof}

\begin{proof}[Proof of Lemma~\ref{lem:consistency-HhatVhat}]
Fix \(j\in\left[m\right]\). On \(\mathcal E_j\), equations~\eqref{eq:M-score-local-lemma} and \eqref{eq:M-local-minimizer-bound-lemma} give
\[
\left\|\tbtheta_j-\btheta_j^\star\right\|_2\le\frac{4C_{\mathrm{F1},\mathrm{s}}}{\underline\lambda}\sqrt{\frac{p+7\log n_j}{n_j}}\le\frac{12C_{\mathrm{F1},\mathrm{s}}}{\underline\lambda}\sqrt{\frac{p\log n_j}{n_j}}=C_{\mathrm{F2},\mathrm{r}}\sqrt{\frac{p\log n_j}{n_j}}.
\]
The right side converges to zero, hence \(\tbtheta_j\in\mathcal B_j\left(r\right)\) for all sufficiently large \(n_j\). Since \(\widehat\bH_j=\nabla^2f_j\left(\tbtheta_j\right)\),
\[
\begin{aligned}
\left\|\widehat\bH_j-\bH_j^\star\right\|_2
&\le\left\|\nabla^2f_j\left(\tbtheta_j\right)-\EE\nabla^2\ell_j\left(\tbtheta_j,\bx\right)\right\|_2\\
&\quad+\left\|\int_0^1\EE\nabla^3\ell_j\left(\btheta_j^\star+u\left(\tbtheta_j-\btheta_j^\star\right),\bx\right)\left[\tbtheta_j-\btheta_j^\star\right]\mathrm du\right\|_2\\
&\le C_{\mathrm{F1},\mathrm{h}}\sqrt{\frac{p\log n_j}{n_j}}+L\left\|\tbtheta_j-\btheta_j^\star\right\|_2\\
&\le C_{\mathrm{F2},\mathrm{H}}\sqrt{\frac{p\log n_j}{n_j}}.
\end{aligned}
\]
This proves the Hessian bound on \(\mathcal E_j\).

Define the empirical and population matrices by
\[
\mathcal V_{j,\mathrm{emp}}\left(\btheta\right)=\frac1{n_j}\sum_{i=1}^{n_j}\nabla\ell_j\left(\btheta,\bx_{ji}\right)\nabla\ell_j\left(\btheta,\bx_{ji}\right)^\top,
\qquad
\mathcal V_{j,\mathrm{pop}}\left(\btheta\right)=\EE\left[\nabla\ell_j\left(\btheta,\bx\right)\nabla\ell_j\left(\btheta,\bx\right)^\top\right].
\]
Then \(\widehat\bV_j=\mathcal V_{j,\mathrm{emp}}\left(\tbtheta_j\right)\) and \(\bV_j^\star=\mathcal V_{j,\mathrm{pop}}\left(\btheta_j^\star\right)\). Let $\mathcal B_{j,n}=\mathcal B_j\left(C_{\mathrm{F2},\mathrm{r}}\sqrt{\frac{p\log n_j}{n_j}}\right)$. For all sufficiently large \(n_j\), \(\mathcal B_{j,n}\subseteq\mathcal B_j\left(r\right)\). Let \(\gamma_2=\left\lceil c\right\rceil+6\), let \(\mathcal G_{j,n}\) be an \(n_j^{-\gamma_2}\)-net of \(\mathcal B_{j,n}\), and use the one-quarter net \(\mathcal N_{1/4}\) from \eqref{eq:cover-v-vector}. Since \(\mathcal B_{j,n}\subseteq\mathcal B_j\left(r\right)\),
\begin{equation}\label{eq:F2-parameter-net-size}
\left|\mathcal G_{j,n}\right|\le\left(1+2rn_j^{\gamma_2}\right)^p\le\left(1+2r\right)^pn_j^{\gamma_2p}.
\end{equation}
For \(i\in\left[n_j\right]\), let
\[
M_{ji}=\sup_{\btheta\in\mathcal B_j\left(r\right)}\left\|\nabla\ell_j\left(\btheta,\bx_{ji}\right)\otimes\nabla^2\ell_j\left(\btheta,\bx_{ji}\right)\right\|_{\mathrm{op}}.
\]
Assumption~\ref{assump:3} gives \(\EE M_{ji}\le Lp^c\). Markov's inequality gives
\begin{equation}\label{eq:M-score-Gram-envelope-revised}
\PP\left[\frac1{n_j}\sum_{i=1}^{n_j}M_{ji}>Ln_j^2p^c\right]\le\frac{n_j^{-1}\sum_{i=1}^{n_j}\EE M_{ji}}{Ln_j^2p^c}\le n_j^{-2}.
\end{equation}
On the complementary event, take \(\btheta\in\mathcal B_{j,n}\), choose \(\btheta^\circ\in\mathcal G_{j,n}\) with \(\left\|\btheta-\btheta^\circ\right\|_2\le n_j^{-\gamma_2}\), and let \(\btheta_u=\btheta^\circ+u\left(\btheta-\btheta^\circ\right)\). Direct differentiation gives
\[
\begin{aligned}
&\frac{\mathrm d}{\mathrm du}\left[\nabla\ell_j\left(\btheta_u,\bx_{ji}\right)\nabla\ell_j\left(\btheta_u,\bx_{ji}\right)^\top\right]\\
&=\nabla^2\ell_j\left(\btheta_u,\bx_{ji}\right)\left(\btheta-\btheta^\circ\right)\nabla\ell_j\left(\btheta_u,\bx_{ji}\right)^\top\\
&\quad+\nabla\ell_j\left(\btheta_u,\bx_{ji}\right)\left(\btheta-\btheta^\circ\right)^\top\nabla^2\ell_j\left(\btheta_u,\bx_{ji}\right).
\end{aligned}
\]
Thus, for every \(\bv\in\mathcal N_{1/4}\),
\[
\left|\bv^\top\left[\mathcal V_{j,\mathrm{emp}}\left(\btheta\right)-\mathcal V_{j,\mathrm{emp}}\left(\btheta^\circ\right)\right]\bv\right|\le2Ln_j^{2-\gamma_2}p^c.
\]
The expectation bound in Assumption~\ref{assump:3} gives
\[
\left|\bv^\top\left[\mathcal V_{j,\mathrm{pop}}\left(\btheta\right)-\mathcal V_{j,\mathrm{pop}}\left(\btheta^\circ\right)\right]\bv\right|\le2L p^c n_j^{-\gamma_2}.
\]
Since \(p\le n_j\) eventually and \(\gamma_2\ge c+6\),
\begin{equation}\label{eq:M-score-Gram-grid-approximation}
\begin{aligned}
&\left|\bv^\top\left[\mathcal V_{j,\mathrm{emp}}\left(\btheta\right)-\mathcal V_{j,\mathrm{pop}}\left(\btheta\right)-\mathcal V_{j,\mathrm{emp}}\left(\btheta^\circ\right)+\mathcal V_{j,\mathrm{pop}}\left(\btheta^\circ\right)\right]\bv\right|\\
&\le2Ln_j^{2-\gamma_2+c}+2Ln_j^{-\gamma_2+c}\le4Ln_j^{-4}.
\end{aligned}
\end{equation}
For fixed \(\btheta^\circ\in\mathcal G_{j,n}\) and \(\bv\in\mathcal N_{1/4}\), set
\[
X_{ji}\left(\btheta^\circ,\bv\right)=\left[\bv^\top\nabla\ell_j\left(\btheta^\circ,\bx_{ji}\right)\right]^2-\EE\left[\left(\bv^\top\nabla\ell_j\left(\btheta^\circ,\bx_{j1}\right)\right)^2\right].
\]
The identity \(\left\|X^2\right\|_{\psi_1}=\left\|X\right\|_{\psi_2}^2\), Assumption~\ref{assump:2}, and \(\left\|X-\EE X\right\|_{\psi_1}\le2\left\|X\right\|_{\psi_1}\) give \(\left\|X_{ji}\left(\btheta^\circ,\bv\right)\right\|_{\psi_1}\le2\tau^2\). Apply \eqref{eq:F1-scalar-Bernstein} with $x=C_{\mathrm{F2},\mathrm{grid}}\sqrt{\frac{p\log n_j}{n_j}}$. The quadratic branch applies for all sufficiently large \(n_j\), and $\frac{n_jx^2}{64\tau^4}=\left(\gamma_2+3+\frac{\log\left[9\left(1+2r\right)\right]}{\log2}\right)p\log n_j$. Equations~\eqref{eq:F2-parameter-net-size} and \(\log n_j\ge\log2\) give \(p\log\left[9\left(1+2r\right)\right]\le p\log n_j\,\log\left[9\left(1+2r\right)\right]/\log2\). Hence a union bound gives
\begin{equation}\label{eq:M-score-Gram-net-revised}
\PP\left[\max_{\btheta^\circ\in\mathcal G_{j,n}}\max_{\bv\in\mathcal N_{1/4}}\left|\bv^\top\left[\mathcal V_{j,\mathrm{emp}}\left(\btheta^\circ\right)-\mathcal V_{j,\mathrm{pop}}\left(\btheta^\circ\right)\right]\bv\right|>x\right]\le2n_j^{-3}.
\end{equation}
Equations~\eqref{eq:cover-v-matrix}, \eqref{eq:M-score-Gram-grid-approximation}, and \eqref{eq:M-score-Gram-net-revised} imply, outside an event of probability at most \(n_j^{-2}+2n_j^{-3}\),
\begin{equation}\label{eq:M-score-Gram-empirical-uniform}
\sup_{\btheta\in\mathcal B_{j,n}}\left\|\mathcal V_{j,\mathrm{emp}}\left(\btheta\right)-\mathcal V_{j,\mathrm{pop}}\left(\btheta\right)\right\|_2\le\left(2C_{\mathrm{F2},\mathrm{grid}}+\frac{8L}{\sqrt{\log2}}\right)\sqrt{\frac{p\log n_j}{n_j}}.
\end{equation}

For \(\btheta\in\mathcal B_{j,n}\), let \(\bh=\btheta-\btheta_j^\star\). The fundamental theorem of calculus gives
\[
\nabla\ell_j\left(\btheta,\bx\right)-\nabla\ell_j\left(\btheta_j^\star,\bx\right)=\int_0^1\nabla^2\ell_j\left(\btheta_j^\star+u\bh,\bx\right)\bh\,\mathrm du.
\]
For unit vectors \(\bv\) and \(\bw\), polarization gives
\[
\bv^\top\nabla^2\ell_j\left(\widetilde\btheta,\bx\right)\bw=\frac14\left[\left(\bv+\bw\right)^\top\nabla^2\ell_j\left(\widetilde\btheta,\bx\right)\left(\bv+\bw\right)-\left(\bv-\bw\right)^\top\nabla^2\ell_j\left(\widetilde\btheta,\bx\right)\left(\bv-\bw\right)\right].
\]
Assumption~\ref{assump:2} and homogeneity imply \(\left\|\bv^\top\nabla^2\ell_j\left(\widetilde\btheta,\bx\right)\bw\right\|_{\psi_1}\le2\tau^2\). The tail identity gives \(\left[\EE\left|X\right|^2\right]^{1/2}\le2\left\|X\right\|_{\psi_1}\), hence
\[
\left[\EE\left|\bv^\top\left[\nabla\ell_j\left(\btheta,\bx\right)-\nabla\ell_j\left(\btheta_j^\star,\bx\right)\right]\right|^2\right]^{1/2}\le4\tau^2\left\|\bh\right\|_2.
\]
Assumption~\ref{assump:2} also gives $\left[\EE\left|\bv^\top\nabla\ell_j\left(\widetilde\btheta,\bx\right)\right|^2\right]^{1/2}\le\sqrt2\tau$ for every \(\widetilde\btheta\in\mathcal B_j\left(r\right)\). Therefore
\[
\begin{aligned}
&\left|\bv^\top\left[\mathcal V_{j,\mathrm{pop}}\left(\btheta\right)-\mathcal V_{j,\mathrm{pop}}\left(\btheta_j^\star\right)\right]\bv\right|\\
&=\left|\EE\left[\bv^\top\left[\nabla\ell_j\left(\btheta,\bx\right)-\nabla\ell_j\left(\btheta_j^\star,\bx\right)\right]\left[\bv^\top\nabla\ell_j\left(\btheta,\bx\right)+\bv^\top\nabla\ell_j\left(\btheta_j^\star,\bx\right)\right]\right]\right|\\
&\le8\sqrt2\tau^3\left\|\btheta-\btheta_j^\star\right\|_2.
\end{aligned}
\]
Taking the supremum over \(\mathcal B_{j,n}\) gives
\begin{equation}\label{eq:M-score-Gram-population-uniform}
\sup_{\btheta\in\mathcal B_{j,n}}\left\|\mathcal V_{j,\mathrm{pop}}\left(\btheta\right)-\bV_j^\star\right\|_2\le8\sqrt2\tau^3C_{\mathrm{F2},\mathrm{r}}\sqrt{\frac{p\log n_j}{n_j}}.
\end{equation}
On \(\mathcal E_j\), \(\tbtheta_j\in\mathcal B_{j,n}\). Equations~\eqref{eq:M-score-Gram-empirical-uniform} and \eqref{eq:M-score-Gram-population-uniform} give
\[
\begin{aligned}
\left\|\widehat\bV_j-\bV_j^\star\right\|_2
&\le\left\|\mathcal V_{j,\mathrm{emp}}\left(\tbtheta_j\right)-\mathcal V_{j,\mathrm{pop}}\left(\tbtheta_j\right)\right\|_2+\left\|\mathcal V_{j,\mathrm{pop}}\left(\tbtheta_j\right)-\bV_j^\star\right\|_2\\
&\le C_{\mathrm{F2},\mathrm{V}}\sqrt{\frac{p\log n_j}{n_j}}.
\end{aligned}
\]
The joint failure probability is at most $5n_j^{-6}+n_j^{-2}+2n_j^{-3}\le8n_j^{-2}$. The Hessian and score Gram bounds therefore hold with probability at least \(1-c_{\mathrm{F2},\mathrm{e}}n_j^{-2}\), where \(c_{\mathrm{F2},\mathrm{e}}=8\).
\end{proof}

\begin{proof}[Proof of Lemma~\ref{lem:deterministic-clipping}]
Assumption~\ref{assump:1} gives, for every \(j\in\left[m\right]\),
\[
\underline\lambda\bI_p\preceq\bH_j^\star\preceq\overline\lambda\bI_p,\qquad \underline\lambda\bI_p\preceq\bV_j^\star\preceq\overline\lambda\bI_p.
\]
Since \(\bSigma_j=\bH_j^{\star-1}\bV_j^\star\bH_j^{\star-1}\), for every \(\bx\in\mathbb R^p\),
\[
\bx^\top\bSigma_j\bx=\left(\bH_j^{\star-1}\bx\right)^\top\bV_j^\star\left(\bH_j^{\star-1}\bx\right)\ge\underline\lambda\left\|\bH_j^{\star-1}\bx\right\|_2^2\ge\frac{\underline\lambda}{\overline\lambda^2}\left\|\bx\right\|_2^2,
\]
\[
\bx^\top\bSigma_j\bx\le\overline\lambda\left\|\bH_j^{\star-1}\bx\right\|_2^2\le\frac{\overline\lambda}{\underline\lambda^2}\left\|\bx\right\|_2^2.
\]
The covariance eigenvalues therefore satisfy
\begin{equation}\label{eq:F3-population-spectral-bounds}
\underline\sigma_0\bI_p\preceq\bSigma_j\preceq\overline\sigma_0\bI_p,
\qquad \underline\sigma_0=\frac{\underline\lambda}{\overline\lambda^2},
\qquad \overline\sigma_0=\frac{\overline\lambda}{\underline\lambda^2}.
\end{equation}
Let \(\mathcal E_n\) be the intersection over \(j\in\left[m\right]\) of the events in Lemma~\ref{lem:consistency-HhatVhat}. A union bound gives
\begin{equation}\label{eq:F3-event-probability}
\PP\left(\mathcal E_n^c\right)\le8\sum_{j=1}^mn_j^{-2}.
\end{equation}
On \(\mathcal E_n\),
\begin{equation}\label{eq:F3-HV-errors}
\max_{j\in\left[m\right]}\left\|\widehat\bH_j-\bH_j^\star\right\|_2\le C_{\mathrm{F2},\mathrm{H}}\max_{j\in\left[m\right]}\sqrt{\frac{p\log n_j}{n_j}},\qquad \max_{j\in\left[m\right]}\left\|\widehat\bV_j-\bV_j^\star\right\|_2\le C_{\mathrm{F2},\mathrm{V}}\max_{j\in\left[m\right]}\sqrt{\frac{p\log n_j}{n_j}}.
\end{equation}
The condition \(p\log n/n=o\left(1\right)\) and \(n_j\asymp n\) imply that the right sides of \eqref{eq:F3-HV-errors} converge to zero. Hence, on \(\mathcal E_n\) for all sufficiently large \(n\), Weyl's inequality and Assumption~\ref{assump:1} give
\begin{equation}\label{eq:F3-Hhat-spectral-bounds}
\frac{\underline\lambda}{2}\bI_p\preceq\widehat\bH_j\preceq2\overline\lambda\bI_p,
\qquad \left\|\widehat\bH_j^{-1}\right\|_2\le\frac2{\underline\lambda},
\qquad j\in\left[m\right].
\end{equation}
The inverse difference identity gives $\widehat\bH_j^{-1}-\bH_j^{\star-1}=\widehat\bH_j^{-1}\left(\bH_j^\star-\widehat\bH_j\right)\bH_j^{\star-1}$. Equations~\eqref{eq:F3-HV-errors} and \eqref{eq:F3-Hhat-spectral-bounds} therefore imply
\begin{equation}\label{eq:F3-Hinverse-error}
\max_{j\in\left[m\right]}\left\|\widehat\bH_j^{-1}-\bH_j^{\star-1}\right\|_2\le\frac{2C_{\mathrm{F2},\mathrm{H}}}{\underline\lambda^2}\max_{j\in\left[m\right]}\sqrt{\frac{p\log n_j}{n_j}}.
\end{equation}
The exact identity
\[
\begin{aligned}
\widehat\bH_j^{-1}\widehat\bV_j\widehat\bH_j^{-1}-\bSigma_j
&=\left(\widehat\bH_j^{-1}-\bH_j^{\star-1}\right)\widehat\bV_j\widehat\bH_j^{-1}+\bH_j^{\star-1}\left(\widehat\bV_j-\bV_j^\star\right)\widehat\bH_j^{-1}\\
&\quad+\bH_j^{\star-1}\bV_j^\star\left(\widehat\bH_j^{-1}-\bH_j^{\star-1}\right)
\end{aligned}
\]
and \(\left\|\widehat\bV_j\right\|_2\le\left\|\bV_j^\star\right\|_2+\left\|\widehat\bV_j-\bV_j^\star\right\|_2\le2\overline\lambda\) for all sufficiently large \(n\) give
\[
\begin{aligned}
\left\|\widehat\bH_j^{-1}\widehat\bV_j\widehat\bH_j^{-1}-\bSigma_j\right\|_2
&\le\left(\frac{4\overline\lambda}{\underline\lambda}+\frac{\overline\lambda}{\underline\lambda}\right)\left\|\widehat\bH_j^{-1}-\bH_j^{\star-1}\right\|_2+\frac2{\underline\lambda^2}\left\|\widehat\bV_j-\bV_j^\star\right\|_2\\
&\le\left(\frac{10\overline\lambda C_{\mathrm{F2},\mathrm{H}}}{\underline\lambda^3}+\frac{2C_{\mathrm{F2},\mathrm{V}}}{\underline\lambda^2}\right)\max_{k\in\left[m\right]}\sqrt{\frac{p\log n_k}{n_k}}.
\end{aligned}
\]
Thus, on \(\mathcal E_n\),
\begin{equation}\label{eq:F3-sandwich-error}
\max_{j\in\left[m\right]}\left\|\widehat\bH_j^{-1}\widehat\bV_j\widehat\bH_j^{-1}-\bSigma_j\right\|_2\le C_{\mathrm{F3},\bSigma}\max_{j\in\left[m\right]}\sqrt{\frac{p\log n_j}{n_j}}.
\end{equation}
For all sufficiently large \(n\), the right side of \eqref{eq:F3-sandwich-error} is at most \(\underline\sigma_0/2\) and at most \(\overline\sigma_0\). Equations~\eqref{eq:F3-population-spectral-bounds}, \eqref{eq:F3-sandwich-error}, and Weyl's inequality give
\begin{equation}\label{eq:F3-sandwich-spectral-bounds}
\frac{\underline\sigma_0}{2}\bI_p\preceq\widehat\bH_j^{-1}\widehat\bV_j\widehat\bH_j^{-1}\preceq2\overline\sigma_0\bI_p.
\end{equation}
Assumption~\ref{assump:1} requires
\[
0<\varepsilon^{\rm c}<\frac12\min\left[\underline\lambda,\frac{\underline\lambda}{\overline\lambda^2},\frac1{\overline\lambda},\frac{\underline\lambda^2}{\overline\lambda}\right]=\frac12\min\left[\underline\lambda,\underline\sigma_0,\frac1{\overline\lambda},\frac1{\overline\sigma_0}\right].
\]
Consequently, \(\varepsilon^{\rm c}<\underline\lambda/2\), \(\varepsilon^{\rm c}<\underline\sigma_0/2\), \(\left(\varepsilon^{\rm c}\right)^{-1}>2\overline\lambda\), and \(\left(\varepsilon^{\rm c}\right)^{-1}>2\overline\sigma_0\). Equations~\eqref{eq:F3-Hhat-spectral-bounds} and \eqref{eq:F3-sandwich-spectral-bounds} show that both clipping maps in \eqref{eq:clipped-sandwich} are inactive. Hence
\begin{equation}\label{eq:F3-clipping-inactive}
\widehat\bH_j^{\rm c}=\widehat\bH_j,
\qquad \widehat\bSigma_j^{\rm c}=\widehat\bH_j^{-1}\widehat\bV_j\widehat\bH_j^{-1},
\qquad j\in\left[m\right].
\end{equation}
Equations~\eqref{eq:F3-sandwich-error} and \eqref{eq:F3-clipping-inactive} prove the first covariance error bound. The inverse difference identity gives $\left(\widehat\bSigma_j^{\rm c}\right)^{-1}-\bSigma_j^{-1}=\left(\widehat\bSigma_j^{\rm c}\right)^{-1}\left(\bSigma_j-\widehat\bSigma_j^{\rm c}\right)\bSigma_j^{-1}$. Equations~\eqref{eq:F3-population-spectral-bounds} and \eqref{eq:F3-sandwich-spectral-bounds} yield \(\left\|\left(\widehat\bSigma_j^{\rm c}\right)^{-1}\right\|_2\le2/\underline\sigma_0\) and \(\left\|\bSigma_j^{-1}\right\|_2\le1/\underline\sigma_0\). Therefore
\[
\max_{j\in\left[m\right]}\left\|\left(\widehat\bSigma_j^{\rm c}\right)^{-1}-\bSigma_j^{-1}\right\|_2\le\frac{2C_{\mathrm{F3},\bSigma}}{\underline\sigma_0^2}\max_{j\in\left[m\right]}\sqrt{\frac{p\log n_j}{n_j}}=C_{\mathrm{F3},\bSigma^{-1}}\max_{j\in\left[m\right]}\sqrt{\frac{p\log n_j}{n_j}}.
\]
Together with \eqref{eq:F3-event-probability}, this proves all assertions.
\end{proof}

\begin{proof}[Proof of Lemma~\ref{lem:quadratic-influence}]
The quadratic decomposition is
\[
\left\|\bZ\right\|_{\bQ}^2-\Tr\left(\bQ\bK\right)
=\sum_{a=1}^N\left[\bY_a^\top\bQ\bY_a-\Tr\left(\bQ\bK_a\right)\right]
+2\sum_{a<b}\bY_a^\top\bQ\bY_b.
\]
The centered diagonal terms are independent. All mixed products and products of distinct off-diagonal terms have an independent mean-zero vector appearing linearly, and hence their expectations vanish. For example, if $b\ne d$ and both differ from $a$,
$\EE\left[\left(\bY_a^\top\bQ\bY_b\right)\left(\bY_a^\top\bQ\bY_d\right)\right]
=\EE\left(\bY_b^\top\bQ\bK_a\bQ\bY_d\right)=0$.
It follows that
\begin{equation}
\EE\left[\left(\left\|\bZ\right\|_{\bQ}^2-\Tr\left(\bQ\bK\right)\right)^2\right]
=\sum_a\operatorname{Var}\left(\bY_a^\top\bQ\bY_a\right)
+4\sum_{a<b}\Tr\left(\bQ\bK_b\bQ\bK_a\right).
\label{eq:compressed-quadratic-variance}
\end{equation}
The directional assumption and homogeneity give
$\left\|\bu^\top\bK_a^{-1/2}\bY_a\right\|_{\psi_1}\le K$ for every unit vector $\bu$, which implies the fourth-moment bound
\[
\EE\left|\bu^\top\bK_a^{-1/2}\bY_a\right|^4
\le8\int_0^\infty t^3e^{-t/K}\,\mathrm dt=48K^4.
\]
Let $\lambda_{a,k},\bu_{a,k}$ be the eigenvalues and orthonormal eigenvectors of $\bK_a^{1/2}\bQ\bK_a^{1/2}$. Weighted Cauchy--Schwarz yields
\[
\begin{aligned}
\operatorname{Var}\left(\bY_a^\top\bQ\bY_a\right)
&\le\EE\left[\left(\sum_k\lambda_{a,k}\left(\bu_{a,k}^\top\bK_a^{-1/2}\bY_a\right)^2\right)^2\right]\\
&\le\left(\sum_k\lambda_{a,k}\right)\sum_k\lambda_{a,k}
\EE\left|\bu_{a,k}^\top\bK_a^{-1/2}\bY_a\right|^4
\le48K^4\Tr\left(\bQ\bK_a\right)^2.
\end{aligned}
\]
The cross terms satisfy
\[
4\sum_{a<b}\Tr\left(\bQ\bK_b\bQ\bK_a\right)
=2\Tr\left(\bQ\bK\bQ\bK\right)-2\sum_a\Tr\left(\bQ\bK_a\bQ\bK_a\right)
\le2\Tr\left(\bQ\bK\bQ\bK\right).
\]
Substitution into \eqref{eq:compressed-quadratic-variance} proves \eqref{eq:F4-main-bound}. Finally,
\[
\begin{aligned}
\Tr\left(\bQ\bK\bQ\bK\right)&\le p\left\|\bQ\right\|_2^2\left\|\bK\right\|_2^2
\le\left\|\bQ\right\|_2^2C_{\mathrm{F4},\mathrm K}^2\frac p{n^2},\\
\sum_a\Tr\left(\bQ\bK_a\right)^2&\le Np^2\left\|\bQ\right\|_2^2\max_a\left\|\bK_a\right\|_2^2
\le\left\|\bQ\right\|_2^2C_{\mathrm{F4},\mathrm N}C_{\mathrm{F4},\mathrm a}^2\frac{p^2}{n^3}.
\end{aligned}
\]
These inequalities give \eqref{eq:F4-consequence} with the stated constants.
\end{proof}

\ifseparatebib
\renewcommand{\refname}{References}
\putbib[main]
\end{bibunit}
\fi

\end{document}